\documentclass[acmsmall,screen]{acmart}
\setcopyright{rightsretained}
\acmJournal{JACM}
\acmYear{2020} \acmVolume{1} \acmNumber{1} \acmArticle{1} \acmMonth{1} \acmPrice{}\acmDOI{10.1145/3380825}

\usepackage{booktabs}   %
\usepackage{subfig}

\usepackage{xcolor}
\usepackage{wrapfig}

\usepackage{etoolbox}
\BeforeBeginEnvironment{wrapfigure}{\setlength{\intextsep}{0pt}}

\usepackage{enumitem}
\makeatletter
\newcommand{\myitem}[1][]{%
\item[#1]\protected@edef\@currentlabel{#1}\ignorespaces%
}
\makeatother

\usepackage{bm}
\usepackage[bbsets,Dfprime,setrelation]{math}

\usepackage{wasysym}
\usepackage[prefixflatinterpret,bracketmodalinterpret,fixformat,silentconst,sidenotecalculus,longseqcontext,seqinsist]{logic}
\usepackage[prefixflatinterpret,bracketmodalinterpret,fixformat,silentconst,differentialdL,simplenames]{dL}
\def\leftrule{L}%
\def\rightrule{R}%
\newcommand{\bebecomes}{\mathrel{::=}}
\newcommand{\alternative}{~|~}

\newcommand{\solvar}{\varphi}

\newcommand{\I}{\dLint[state=\omega]}
\newcommand{\It}{\dLint[state=\nu]}
\newcommand{\If}{\DALint[flow=\solvar]}
\newcommand*{\Iff}[1][\zeta]{\dLint[state=\solvar({#1})]}%
\newcommand*{\Iffy}[1][\zeta]{\dLint[state=\solvar_y({#1})]}%
\newcommand*{\IffA}[1][\zeta]{\dLint[state=\solvar_1({#1})]}%
\newcommand*{\IffB}[1][\zeta]{\dLint[state=\solvar_2({#1})]}%
\newcommand*{\Ifff}[1][\zeta]{\dLint[state=\psi({#1})]}%
\newcommand*{\Ifffy}[1][\zeta]{\dLint[state=\psi_y({#1})]}%
\newcommand*{\Ifffye}[1][\zeta]{\dLint[state=\psi_y|_\epsilon({#1})]}%
\newcommand*{\Iffff}[1][\zeta]{\dLint[state=\Phi({#1})]}%

\newcommand{\solmodels}[3]{#1 \models #2 \land #3}
\newsavebox{\Rval}%
\sbox{\Rval}{$\scriptstyle\mathbb{R}$}
\newsavebox{\Rvalext}%
\sbox{\Rvalext}{$\scriptstyle\mathbb{R}_{\exp,\sin,\cos}$}
\newsavebox{\Rvalexp}%
\sbox{\Rvalexp}{$\scriptstyle\mathbb{R}_{\exp}$}
\relax

  \newdimen\linferenceRulehskipamount%
  \newdimen\lcalculuscollectionvskipamount%
\definecolor{vblue}{rgb}{.1,.15,.62}
\definecolor{vgray}{rgb}{.35,.35,.35}
\renewcommand*{\irrulename}[1]{\text{\textcolor{vgray}{\upshape#1}}}%
\renewcommand{\linferenceRuleNameSeparation}{\hspace{5pt}}
\renewcommand{\linferenceRuleSidenoteSeparation}{\quad}
\renewcommand{\linferenceRuleRefSeparation}{, }

\usepackage{amsthm}

\usepackage{prettyref}
\newcommand{\rref}[2][]{\prettyref{#2}}
\newrefformat{sec}{Section\,\ref{#1}}
\newrefformat{subsec}{Section\,\ref{#1}}
\newrefformat{def}{Def.\,\ref{#1}}
\newrefformat{thm}{Theorem\,\ref{#1}}
\newrefformat{prop}{Proposition\,\ref{#1}}
\newrefformat{lem}{Lemma\,\ref{#1}}
\newrefformat{cor}{Corollary\,\ref{#1}}
\newrefformat{ex}{Example\,\ref{#1}}
\newrefformat{tab}{Table\,\ref{#1}}
\newrefformat{fig}{Fig.\,\ref{#1}}
\newrefformat{subfig}{Fig.\,\ref{#1}}
\newrefformat{case}{case\,\ref{#1}}
\newrefformat{foot}{Footnote\,\ref{#1}}
\newrefformat{itm}{\ref{#1}}

\newrefformat{app}{Appendix\,\ref{#1}}
\newenvironment{proofsketch}[1][TODO]{\proof[Proof Summary (\ifthenelse{\equal{#1}{TODO}}{TODO}{\rref{#1}})]}{\endproof}

\newcommand{\ie}{i.e.}

\newcommand{\eg}{e.g.}

\newcommand{\cmp}{\succcurlyeq}

\renewcommand{\allvars}{\mathbb{V}}
\newcommand{\States}{\mathbb{S}}

\newcommand{\initassum}{x{=}y}
\newcommand{\notinitassum}{x{\neq}y}

\newcommand{\dprogressin}[3][]{%
  {\langle{\pevolvein{#2}{#3}}\rangle}{\ddnext} {#1}%
}

\newcommand{\ddnext}{\bigcirc}
\definecolor{highlightred}{rgb}{.7, 0.0, 0.0}

\renewcommand*{\der}[1]{\D{(#1)}}

\renewcommand*{\lie}[3][]
{\mathcal{L}_{#2}^{\ifthenelse{\equal{#1}{}}{}{^{\left(#1\right)}}}(#3)}
\renewcommand*{\lied}[3][]{\overset{\bm .}{#3}\ifthenelse{\equal{#1}{}}{}{{}^{(#1)}}}

\renewcommand{\siglied}[3][]{\overset{\bm .}{#3}{}^{\Dostar{#1}}}

\renewcommand{\Dostar}[1]{\ifthenelse{\equal{#1}{}}{(*)}{-(*)}}
\newcommand{\sigliedgt}[3][]{\siglied[#1]{#2}{#3}>0}
\newcommand{\sigliedgeq}[3][]{\siglied[#1]{#2}{#3}\geq0}

\newcommand{\sigliedzero}[3][]{\siglied[#1]{#2}{#3}=0}

\newcommand{\sigliedsai}[3][]{\siglied[#1]{#2}{#3}}

\renewcommand*{\vec}[1]{\mathbf{#1}}

\newcommand{\vecpolyn}[2]{\vec{#1}}

\newcommand{\matpolyn}[2]{#1}

\newcommand{\truncafter}[2]{#1|_{#2}}
\newcommand{\soltrunc}[1]{\truncafter{\solvar}{#1}}

\newcommand{\etermA}{e}
\newcommand{\etermB}{\tilde{e}}

\newcommand{\odeterm}{f}
\renewcommand*{\genDE}[1]{\odeterm(#1)}

\newcommand{\etermAA}{d}
\newcommand{\etermBB}{\tilde{d}}

\newcommand{\ptermA}{p}
\newcommand{\ptermB}{q}
\newcommand{\cofterm}{g}
\newcommand{\coftermC}{G}

\newcommand{\argx}{(x)}
\newcommand{\argxx}{(x,\D{x})}

\newcommand{\noef}{h}
\newcommand{\noeff}{\upsilon }
\newcommand{\noefdom}{H}
\newcommand{\noeffdom}{\Upsilon}
\newcommand{\noefff}{f}

\newcommand{\fvarA}{\phi}
\newcommand{\fvarB}{\psi}

\newcommand{\rfvar}{P}
\newcommand{\rrfvar}{R}

\newcommand{\footnotesizeoff}{}%

\newcommand{\normeuc}[1]{\left\lVert#1\right\rVert_2}
\newcommand{\normfrob}[1]{\left\lVert#1\right\rVert_F}

\newsavebox{\Lightningval}%
\sbox{\Lightningval}{\mbox{\lightning}}

\begin{document}

\title{Differential Equation Invariance Axiomatization}

\author{Andr\'e Platzer}
\orcid{0000-0001-7238-5710}
\affiliation{
  \department{Computer Science Department}
  \institution{Carnegie Mellon University}
  \streetaddress{5000 Forbes Avenue}
  \city{Pittsburgh}
  \state{PA}
  \postcode{15213}
  \country{USA}
}
\email{aplatzer@cs.cmu.edu}

\author{Yong Kiam Tan}
\orcid{0000-0001-7033-2463}
\affiliation{
  \department{Computer Science Department}
  \institution{Carnegie Mellon University}
  \streetaddress{5000 Forbes Avenue}
  \city{Pittsburgh}
  \state{PA}
  \postcode{15213}
  \country{USA}
}
\email{yongkiat@cs.cmu.edu}

\begin{abstract}
This article proves the completeness of an axiomatization for differential equation invariants described by Noetherian functions.
First, the differential equation axioms of differential dynamic logic are shown to be complete for reasoning about analytic invariants.
Completeness crucially exploits differential ghosts, which introduce additional variables that can be chosen to evolve freely along new differential equations.
Cleverly chosen differential ghosts are the proof-theoretical counterpart of dark matter.
They create new hypothetical state, whose relationship to the original state variables satisfies invariants that did not exist before.
The reflection of these new invariants in the original system then enables its analysis.

An extended axiomatization with existence and uniqueness axioms is complete for all local progress properties,
and, with a real induction axiom, is complete for all semianalytic invariants.
This parsimonious axiomatization serves as the logical foundation for reasoning about invariants of differential equations.
Indeed, it is precisely this logical treatment that enables the generalization of completeness to the Noetherian case.
\end{abstract}




\maketitle

\section{Introduction}
\label{sec:introduction}

Classically, differential equations are studied by analyzing their solutions, which is at odds with the fact that solutions are often much more complicated than the differential equations themselves.
This stark difference between the simple local description as differential equations, and the complex global behavior exhibited by their solutions is fundamental to the descriptive power of differential equations.
Poincar\'e's qualitative study of differential equations~\cite{Poincare81} calls for the exploitation of this difference by deducing properties of solutions \emph{directly from the differential equations}.
This article completes an important step in this enterprise by identifying the \emph{logical foundations for proving invariance properties of differential equations} described by Noetherian functions~\cite{MR1150568,MR1732408,MR2083248,MR3925105}.

This result exploits the differential equation axioms of differential dynamic logic (\dL)~\cite{DBLP:conf/lics/Platzer12b,DBLP:journals/jar/Platzer17}.
\dL is a logic for deductive verification of hybrid systems that are modeled by hybrid programs combining discrete computation (\eg, assignments, tests and loops), and continuous dynamics specified using systems of ordinary differential equations (ODEs).
By the continuous relative completeness theorem for \dL~\cite[Theorem 1]{DBLP:conf/lics/Platzer12b}, verification of hybrid systems reduces completely to the study of differential equations.
Thus, the hybrid systems axioms of \dL provide a way of lifting the findings of this article about differential equations to hybrid systems.
The remaining practical challenge is to find succinct ODE invariants. The \dL calculus reduces proving such an invariant to arithmetical questions, which are decidable if the invariants are in first-order real arithmetic~\cite{Bochnak1998}.

To understand the difficulty in verifying properties of ODEs, it is useful to draw an analogy between ODEs and discrete program loops.\footnote{%
In fact, this analogy can be made precise: \dL also has a converse relative completeness theorem~\cite[Theorem 2]{DBLP:conf/lics/Platzer12b} that reduces hybrid systems and their ODEs completely to discrete Euler approximation loops.}
Loops also exhibit the dichotomy between global behavior and local description.
Although the body of a loop may be simple, it is almost always impractical to reason about the global behavior of loops by unfolding all possible iterations.
Instead, the premier reasoning technique for loops is to study their loop invariants, \ie, inductive properties that are always preserved across each execution of the loop body.

Similarly, invariants of ODEs describe subsets of the state space from which solutions of the ODEs cannot escape.
The three basic \dL principles for reasoning about such invariants are: (1) \emph{differential invariants}, which enable local reasoning about local change of truth in differential form, (2) \emph{differential cuts}, which accumulate knowledge about the evolution of an ODE from multiple proofs, and (3) \emph{differential ghosts}, which add differential equations for new ghost variables to the existing system of differential equations enabling reasoning about the historical evolution of ODE systems in integral form.
These reasoning principles relate to their discrete loop counterparts as follows: (1) corresponds to loop induction by analyzing the loop body, (2) corresponds to progressive refinement of the loop guards, and (3) corresponds to adding discrete ghost variables to remember intermediate program states.
At first glance, differential ghosts seem counter-intuitive: they increase the dimension of the system, which should be adverse to analyzing it!
However, just as the addition of discrete ghosts allows the expression of new relationships between variables along executions of a program~\cite{DBLP:journals/cacm/OwickiG76}, adding differential ghosts that suitably co-evolve with the ODEs crucially allows the expression of new relationships along solutions to the differential equations.
The ramifications of these new relationships are then used to analyze the original, unaugmented ODEs.

The \dL proof calculus internalizes these reasoning principles as \emph{syntactic} axioms \cite{DBLP:journals/jar/Platzer17}.
ODE invariance proofs in \dL are syntactic derivations whose correctness relies solely on the soundness of these underlying axioms.
Crucially, this obviates the need to unfold the mathematical semantics of differential equations for proving their invariance properties.
This separation of syntax and axiomatics from semantics enables their sound implementation, e.g., in \KeYmaeraX~\cite{DBLP:conf/cade/FultonMQVP15} with \dL's uniform substitution calculus~\cite{DBLP:journals/jar/Platzer17}, and their verification in foundational theorem provers~\cite{DBLP:conf/cpp/BohrerRVVP17}.

This article extends the authors' earlier conference version~\cite{DBLP:conf/lics/PlatzerT18} beyond the polynomial setting and presents a \emph{differential equation invariance axiomatization}.
For extended term languages (and ODEs) meeting three extended term conditions, this article presents the following contributions:
\begin{enumerate}
\item\label{itm:contrib1} \emph{All} analytic invariants, \ie, finite conjunctions and disjunctions of equations between extended terms, are provable using only the three ODE axioms outlined above.
\item\label{itm:contrib2} With axioms internalizing the existence and uniqueness theorems for solutions of differential equations, all \emph{local progress} properties of ODEs are provable for all semianalytic formulas, \ie, propositional combinations of inequalities between extended terms.
\item\label{itm:contrib3} With a real induction axiom that reduces invariance to local progress, the \dL calculus is complete for proving \emph{all} semianalytic invariants of differential equations.
\item\label{itm:contrib4} These are \emph{axiomatic completeness} results: all (semi)analytic invariance and local progress questions are provably \emph{equivalent} in \dL to questions about the underlying arithmetic. This equivalence also yields disproofs when the resulting arithmetic questions are refuted.
\end{enumerate}

These results are proved constructively, yielding practical and purely logical proof-producing procedures for reducing ODE invariance questions to arithmetical questions in \dL.
The most subtle step is the construction of suitable differential ghosts that simplify the analysis as a function of both the differential equations and desired invariant.
The axiomatic approach crucially enables these contributions because the axioms internalize basic properties of ODEs and thus remain sound and complete for \emph{all} extended term languages meeting the extended term conditions.
Furthermore, the identification of a parsimonious yet complete ODE axiomatization provides the best of both worlds: \emph{parsimony} minimizes effort required in implementation and verification of the proof calculus while \emph{completeness} guarantees that all ODE invariance reasoning is possible using only syntactic proofs from the foundational axioms.
Since the completeness results prove equivalences, these advantages continue to hold whether proving or disproving invariance properties of differential equations.
This logical foundation is essential because it enables \emph{compositional} syntactic reasoning in \dL for (continuous) differential equations in isolation from other (discrete) parts of the hybrid system \cite{DBLP:conf/lics/Platzer12b}.

Since Noetherian functions from real analytic geometry~\cite{MR1150568,MR1732408,MR2083248,MR3925105} generate Noetherian rings closed under (partial) derivatives, they meet all of the extended term conditions and thus provide an ideal setting for extending term languages.
Noetherian functions include many functions of practical interest for modeling hybrid systems, e.g., real exponential and trigonometric functions, which are implicitly definable in \dL~\cite{DBLP:journals/logcom/Platzer10,DBLP:conf/lics/Platzer12b} but do not come with effective reasoning principles.\footnote{The relative decidability theorem for \dL~\cite[Theorem 11]{DBLP:conf/lics/Platzer12b} needs either an oracle for (continuous) differential equation properties or an oracle for discrete program properties.}
Making them first-class members of the term language enables their explicit use in hybrid systems models and proofs, especially in descriptions of ODE invariants.
The study of Noetherian functions is a major new contribution of this extended version, among others:
\begin{enumerate}[resume]
\item\label{itm:contrib5} Noetherian functions are shown to meet the extended term conditions. Any such extension automatically inherits all of the aforementioned completeness results.
\item\label{itm:contrib6} The authors' earlier results and proofs~\cite{DBLP:conf/lics/PlatzerT18} are generalized to extended term languages.
\item\label{itm:contrib7} A stronger proof-theoretical result is shown for algebraic (and analytic) completeness compared to the earlier result~\cite{DBLP:conf/lics/PlatzerT18} using only \emph{scalar} differential ghosts.
\end{enumerate}

Just as discrete ghosts can make a program logic relatively complete~\cite{DBLP:journals/cacm/OwickiG76}, differential ghosts achieve completeness for algebraic (and analytic) invariants in \dL.
The improvement (\rref{itm:contrib7}) is significant for conceptual, implementation, and proof purposes.
All algebraic (and analytic) invariants can now be proved using only a \emph{constant} number of scalar differential ghosts compared to the earlier result \cite{DBLP:conf/lics/PlatzerT18} which introduces a \emph{quadratic} number of ghost variables using \emph{vectorial} differential ghosts.

\section{Background: Differential Dynamic Logic}
\label{sec:background}

This section briefly reviews (differential-form) differential dynamic logic (\dL), focusing on its continuous fragment.
It also establishes the notational conventions used in this article and motivates the extended term conditions.
The interested reader is referred to the literature~\cite{DBLP:conf/lics/Platzer12b,DBLP:journals/jar/Platzer17} for a complete exposition of \dL, including its discrete and hybrid fragments.

\subsection{Syntax}
\label{subsec:background-syntax}

\emph{Terms} in \dL are given by the following grammar, where $x \in \allvars$ is a variable from the set of all variables $\allvars$, $c \in \rationals$ is a rational constant, and $\noef \in \{ \noef_1,\dots, \noef_r\}$ are $k$-ary function symbols:
\[
  \etermA,\etermB \bebecomes x \alternative c \alternative \etermA + \etermB \alternative \etermA \cdot \etermB \alternative \noef(\etermA_1,\dots,\etermA_k) \alternative \der{\etermA}
\]

Terms generated using only the first 4 clauses of this grammar correspond to polynomials over the variables under consideration.
For this article, the term language is extended with a finite number of new $k$-ary fixed function symbols, $\noef \in \{ \noef_1,\dots, \noef_r \}$, with fixed interpretations.
As a running example of such an extended term language, consider the unary function symbols $\exp,\sin,\cos$ which are always interpreted as the real exponential and trigonometric functions respectively:
\begin{align}
	\etermA,\etermB \bebecomes x \alternative c \alternative \etermA + \etermB \alternative \etermA \cdot \etermB \alternative \exp(\etermA) \alternative \sin(\etermA) \alternative \cos(\etermA) \alternative \der{\etermA}
\label{eq:extlang}
\end{align}

Of course, the syntactic extension cannot be completely arbitrary, e.g., adding functions $\noef$ whose interpretation is nowhere differentiable would fundamentally break the enterprise of studying ODEs directly by their local behavior.
These unsuitable syntactic extensions are ruled out by a set of extended term conditions.
These conditions are developed and motivated along the way, with a summary in~\rref{subsec:background-compatibility}.
The class of Noetherian functions, which meets all of these extended term conditions, is introduced and studied in~\rref{sec:noetherianfunctions}.
The Noetherian class includes the functions in the example ($\exp,\sin,\cos$) so the completeness results also apply to the extended term language~\rref{eq:extlang}.

\emph{Differentials} \(\der{\etermA}\) are used in \dL for sound differential equations reasoning~\cite{DBLP:journals/jar/Platzer17}.
The value of \(\der{\etermA}\) relates to how the value of term $\etermA$ changes with each of its variables as a function also of how those variables themselves change.
The fundamental insight is that, along the evolution of a differential equation, the value of differential $\der{\etermA}$ coincides with the analytic time derivative of $\etermA$ \cite[Lem.\,35]{DBLP:journals/jar/Platzer17}, so that proofs about \emph{equations of differentials} yield proofs about \emph{differential equations}.
It is crucial for soundness and compositionality that differentials have a local semantics defined in any state, so that they can be used correctly in any context to draw sound conclusions from syntactic manipulations mixing dynamic statements about differential equations and static statements about differentials.
The precise semantics of differentials is elaborated in~\rref{subsec:background-semantics}, while~\rref{subsec:background-differentials} explains how they can be used to obtain a syntactic representation of (semantic) time derivatives along solutions to differential equations.
Syntactically, every variable $x \in \allvars$ is assumed to have a corresponding differential variable $\D{x} \in \allvars$ which, like differential terms, are syntactic representations of the semantic time derivative of $x$ along solutions.
\rref{subsec:background-differentials} also shows that, when in the context of an ODE, differential terms (and variables) can be provably turned into terms that do not contain any differentials or differential variables.
Therefore, $\etermA,\etermB$ is used exclusively in this article to refer to such differential-free terms, e.g., $\der{\etermA}$ is the differential of $\etermA$, where $\etermA$ is a differential-free term.

For this article, we write $x$ to refer to a vector of variables $x_1,\dots,x_n$ and write $\etermA\argx,\etermB\argx$ to emphasize that these terms depend only on variables $x$ free.
When this dependency is unimportant, terms $\etermA,\etermB$ are written as usual without any dependencies.
As convenient distinguishing notation, vectors of terms are written in bold $\vecpolyn{\etermA}{x},\vecpolyn{\etermB}{x}$, with $\vecpolyn{\etermA}{x}_i,\vecpolyn{\etermB}{x}_i$ for their $i$-th components.
These vectorial terms and their corresponding dimensions are always explicitly specified when used, e.g., in~\rref{subsec:vecdarbouxeq}.
Polynomial terms are useful as familiar illustrative examples and they also enjoy special properties not necessarily shared by extended term languages.
The notation $\ptermA,\ptermB$ is reserved for polynomial terms, with dependencies $\ptermA\argx,\ptermB\argx$ added when necessary.

\emph{Formulas} of \dL are given by the following grammar, where $\sim$ is a comparison operator $=,\geq,>$ and $\alpha$ is a differential equation (or, more generally, a \emph{hybrid program}~\cite{DBLP:conf/lics/Platzer12b,DBLP:journals/jar/Platzer17}):
\[
  \fvarA,\fvarB \bebecomes \etermA \sim \etermB \alternative \fvarA \land \fvarB \alternative \fvarA \lor \fvarB \alternative \lnot{\fvarA} \alternative \lforall{x}{\fvarA} \alternative \lexists{x}{\fvarA} \alternative \dbox{\alpha}{\fvarA} \alternative\ddiamond{\alpha}{\fvarA}
\]

Formulas can be normalized such that every atomic comparison $\etermA \sim \etermB$ has $0$ on the right-hand side.
The notation $\etermA \cmp 0$ is used when there is a free choice between $\geq$ or $>$. Other logical connectives, e.g., $\limply,\lbisubjunct$ are definable as usual.
For the formula $\vecpolyn{\etermA}{x}=\vecpolyn{\etermB}{x}$ where both $\vecpolyn{\etermA}{x},\vecpolyn{\etermB}{x}$ have dimension $n$, equality is understood \emph{component-wise} as \m{\landfold_{i=1}^{n} \vecpolyn{\etermA}{x}_i=\vecpolyn{\etermB}{x}_i} and $\vecpolyn{\etermA}{x}\neq\vecpolyn{\etermB}{x}$ as $\lnot{(\vecpolyn{\etermA}{x}=\vecpolyn{\etermB}{x})}$.
The modal formula $\dibox{\alpha}\fvarA$ is true iff $\fvarA$ is true after \emph{all} runs of $\alpha$, and its dual $\didia{\alpha}\fvarA$ is true iff $\fvarA$ is true after \emph{some} run of $\alpha$.

Formulas not containing the first-order quantifiers and modal connectives are called \emph{semianalytic} formulas and are written as $\rfvar,\ivr$.
The word ``analytic'' refers to the (semantic) real analyticity~\cite{MR1916029} of terms when extended with Noetherian functions in~\rref{sec:noetherianfunctions}.
As usual, the dependencies $\rfvar\argx,\ivr\argx$ are added when necessary.
Every semianalytic formula can be normalized to one that is formed from only conjunctions and disjunctions of atomic comparison formulas.
Formulas $\rfvar,\ivr$ formed from only conjunctions and disjunctions of equalities are called \emph{analytic} formulas.
When all atomic comparisons in (semi)analytic formulas are restricted to only occur between polynomial terms $\ptermA \sim \ptermB$, the resulting formulas are also known as \emph{(semi)algebraic} formulas~\cite{Bochnak1998}.
The first-order theory of the reals with polynomial terms (and with quantifiers), is decidable by quantifier elimination~\cite{Bochnak1998}.
Thus, every first-order formula of real arithmetic is equivalent to a (quantifier-free) semialgebraic formula and no expressiveness is lost by disallowing first-order quantifiers in the semialgebraic case.
Unfortunately, quantifier elimination is impossible even for simple term language extensions like the exponential function~\cite{MR762106}.
Furthermore, even for the extended term language~\rref{eq:extlang} with trigonometric functions, arithmetic questions are already undecidable~\cite{DBLP:journals/jsyml/Richardson68}.
Therefore, special care is taken in this article to distinguish first-order properties of the real closed fields, i.e., those described by (semi)algebraic formulas~\cite{Bochnak1998}, from those properties described by (semi)analytic formulas.

The \dL modalities $\dbox{\alpha}{}$ and $\ddiamond{\alpha}{}$ are parameterized by a \emph{continuous program} $\alpha$ (more general \emph{hybrid programs} combining discrete and continuous dynamics are supported in \dL but not relevant here):
\[\alpha \bebecomes \cdots \alternative \pevolvein{\D{x}=\genDE{x}}{\ivr}\]

The continuous program $\pevolvein{\D{x}=\genDE{x}}{\ivr}$ is an autonomous (vectorial) differential equation system with LHS $\D{x} = (\D{x_1},\dots,\D{x_n})$ and RHS term $\odeterm_i(x)$ for each $\D{x_i}$.
Following the notational convention, $\odeterm_i(x)$ is differential-free so the ODE system $\D{x}=\genDE{x}$ is given in explicit form~\cite{DBLP:journals/jar/Platzer17}.
Autonomous ODEs $\D{x}=\genDE{x}$ do not depend explicitly on time on the RHS.
A standard transformation is to add a clock variable $t$ to the system with $\D{x}=\genDE{x,t},\D{t}=1$ if time dependency on the RHS is desired.
The evolution domain constraint $\ivr$ is a semianalytic formula restricting the set of states in which the ODE is allowed to evolve continuously; the ODE is simply written as \m{\pevolve{\D{x}=\genDE{x}}} when the domain constraint is $\ltrue$.
The following running example ODE is used in this article (see \rref{fig:exampleODE}):
\begin{align}
\alpha_e \mdefequiv \D{u}=-v + \frac{u}{4} (1 - u^2 - v^2), \D{v} = u + \frac{v}{4} (1 - u^2 - v^2)
\label{eq:example-ode}
\end{align}

\begin{wrapfigure}[13]{r}{0.58\textwidth}
\centering
\includegraphics[width=0.58\textwidth,trim=0 3 0 3,clip]{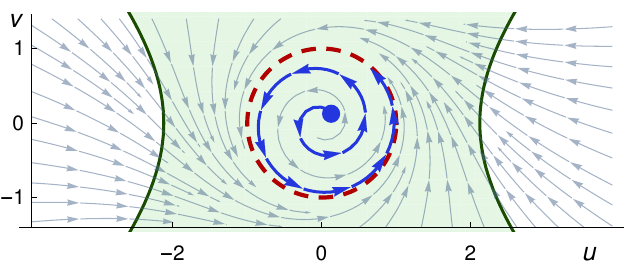}
\caption{The red dashed circle $u^2+v^2=1$ is approached by solutions of $\alpha_e$ from all points except the origin, e.g., the blue trajectory from $(\frac{1}{8},\frac{1}{8})$ spirals towards the circle. The red circle, green region $u^2 \leq v^2 + \frac{9}{2}$, and origin are invariants of the system.}
\label{fig:exampleODE}
\end{wrapfigure}
Following the analogy between ODEs and (discrete) program loops in \rref{sec:introduction}, solutions of an ODE must continuously (locally) follow its RHS.
This is visualized in~\rref{fig:exampleODE} with directional arrows corresponding to the RHS of $\alpha_e$ evaluated at points on the plane.
Even though the RHS of $\alpha_e$ are polynomials, its solutions, which must locally follow the arrows, already exhibit complex global behavior.
\rref{fig:exampleODE} suggests, \eg, that all points (except the origin) globally evolve towards the unit circle.

\subsection{Semantics}
\label{subsec:background-semantics}
A state $\iget[state]{\I} : \allvars \to \reals$ assigns a real value to each variable in $\allvars$; the set of all states is written $\States$.
The semantics of term $\etermA$ in state $\iget[state]{\I}$ is written as $\ivaluation{\I}{\etermA} \in \reals$.
It is defined as usual for the standard arithmetic operators, e.g., $\ivaluation{\I}{\etermA+\etermB} = \ivaluation{\I}{\etermA} + \ivaluation{\I}{\etermB}$.
The semantics of each fixed function symbol $\noef$ is given by a corresponding real-valued function $\noef : \reals^k \to \reals$ (using the same symbol $\noef$ for the LHS syntactic function symbol and its RHS semantic interpretation by a slight abuse of notation):
\begin{align*}
\ivaluation{\I}{\noef(\etermA_1,\dots,\etermA_k)} &= \noef(\ivaluation{\I}{\etermA_1},\dots,\ivaluation{\I}{\etermA_k})
\end{align*}

The semantics of differentials~\cite{DBLP:journals/jar/Platzer17} is the sum of partial derivatives \(\Dp[x]{\ivaluation{\I}{\etermA}}\) by all variables $x\in\allvars$ multiplied by the values of their associated differential variables $\D{x}$, where $\iget[state]{\I}(\D{x})$ selects the direction in which $x$ evolves locally and \(\Dp[x]{\ivaluation{\I}{\etermA}}\) describes how the value of $\etermA$ changes with a change of $x$:
\begin{align*}
\ivaluation{\I}{\der{\etermA}} &= \sum_{x\in\allvars} \iget[state]{\I}(\D{x}) \Dp[x]{\ivaluation{\I}{\etermA}}
\end{align*}

There are two subtleties to highlight.
First, the real-valued functions $\noef$ are required to be defined on the domain $\reals^k$ so that the term semantics are well-defined in all states.
It is possible to extend \dL with terms that are only defined within an open domain of definition rather than the entire real domain and this would allow, e.g., rational functions to be added to the term language.
This article will not pursue such an extension, although the Noetherian functions from~\rref{sec:noetherianfunctions} and \rref{prop:diffaxiomatization} give an implicit way of working with quotients of extended terms.
Second, the semantics of differentials implicitly requires that the partial derivatives $\Dp[x]{\ivaluation{\I}{\etermA}}$ exist.
In fact, partial derivatives of any order for the semantics of any term must exist because their differentials (which provably reduce to differential-free terms by \rref{subsec:background-differentials}), in turn, have differentials that must also have a well-defined semantics.
Following the standard interpretation of function symbols~\cite{DBLP:journals/jar/Platzer17}, it suffices to require that the fixed function symbols $\noef$ are interpreted as smooth $C^\infty$ functions, i.e., have partial derivatives of any order.
Since the $C^\infty$ functions are closed under addition, multiplication and function composition, the resulting term semantics are also smooth~\cite{DBLP:journals/jar/Platzer17}, as required.

The semantics of differentials $\der{\etermA}$ is well-defined for isolated states $\iget[state]{\I}$, independent of any ODEs.
Their importance for differential equations reasoning in \dL stems from the semantics of hybrid programs $\alpha$, which are transition relations, $\iaccess[\alpha]{\I} \subseteq \States \times \States$, between states. The semantics of an ODE, $\iaccess[\pevolvein{\D{x}=\genDE{x}}{\ivr}]{\I}$, is the set of all pairs of states that are connected by some solution of the ODE:
\begin{align*}
\iaccessible[\pevolvein{\D{x}=\genDE{x}}{\ivr}]{\I}{\It} ~\text{iff}~
&\text{there is a duration}~0 \leq T \in \reals~\text{and a function}~\solvar:[0,T] \to \States\\
&\text{with}~ \solvar(0)=\iget[state]{\I} ~\text{on}~ \scomplement{\{\D{x}\}}, \solvar(T)=\iget[state]{\It}, ~\text{and}~ \solmodels{\solvar}{\D{x}=\genDE{x}}{\ivr}
\end{align*}

The \m{\solmodels{\solvar}{\D{x}=\genDE{x}}{\ivr}} condition checks \m{\imodels{\Iff}{\D{x}=\genDE{x} \land \ivr}}, $\solvar(0) = \solvar(\zeta)$ on $\scomplement{\{x,\D{x}\}}$ for all $0 \leq \zeta \leq T$, and, if $T>0$, then $\D[t]{\solvar(t)(x)}(\zeta)$ exists, and is equal to $\solvar(\zeta)(\D{x})$ for all $0 \leq \zeta \leq T$. In other words, $\solvar$ is a solution of the differential equations $\D{x}=\genDE{x}$ that always stays in the evolution domain constraint $\ivr$. It is also required to hold all variables other than $x,\D{x}$ constant. Most importantly, the values of the differential variables $\D{x}$ are required to match the value of the RHS $\genDE{x}$ of the differential equation along the solution. See~\cite[Definition 7]{DBLP:journals/jar/Platzer17} for further details.

The semantics of comparison operations and first-order logical connectives are defined as usual, with $\imodel{\I}{\fvarA} \subseteq \States$ being the set of states where formula $\fvarA$ is true.
For example, $\imodels{\I}{\etermA \leq \etermB}$ iff $\ivaluation{\I}{\etermA} \leq \ivaluation{\I}{\etermB}$, and $\imodels{\I}{\fvarA \land \fvarB}$ iff $\imodels{\I}{\fvarA}$ and $\imodels{\I}{\fvarB}$.
The semantics of modal connectives are:
\begin{align*}
\imodels{\I}{\dbox{\alpha}{\fvarA}} &~\text{iff}~ \imodels{\It}{\fvarA} ~\text{for all}~ \iget[state]{\It} ~\text{such that}~ \iaccessible[\alpha]{\I}{\It}\\
\imodels{\I}{\ddiamond{\alpha}{\fvarA}} &~\text{iff there is a state}~ \iget[state]{\It} ~\text{such that}~ \iaccessible[\alpha]{\I}{\It} ~\text{and}~ \imodels{\It}{\fvarA}
\end{align*}

Formula $\fvarA$ is \emph{valid} iff it is true in all states, i.e., \(\imodel{\I}{\fvarA} = \States \).
If the formula \(\rfvar \limply \dbox{\pevolvein{\D{x}=\genDE{x}}{\ivr}}{\rfvar}\) is valid, then the formula $\rfvar$ is called an \emph{invariant} of the ODE, $\pevolvein{\D{x}=\genDE{x}}{\ivr}$. Unfolding the semantics, this means that from any initial state $\imodels{\I}{\rfvar}$, any solution $\solvar$ of $\D{x}=\genDE{x}$ starting in $\iget[state]{\I}$, which does not leave the evolution domain $\imodel{\I}{\ivr}$, stays in $\imodel{\I}{\rfvar}$ for its \emph{entire duration}.
\rref{fig:exampleODE} suggests several invariants for the ODE $\alpha_e$ from~\rref{eq:example-ode}.
The unit circle, $u^2+v^2=1$, is an equational invariant because the direction of flow on the circle is always tangential to it.
The open unit disk, $u^2+v^2 < 1$, is also invariant because trajectories within the disk spiral towards the circle but never reach it.
The green region described by $u^2 \leq v^2 + \frac{9}{2}$ is invariant but needs a careful proof.

\subsection{Axiomatization}
\label{subsec:background-axiomatizaton}

\irlabel{qear|\usebox{\Rval}}
\irlabel{qearpoly|\usebox{\Rval}} %
\irlabel{qearext|\usebox{\Rvalext}}
\irlabel{qearexp|\usebox{\Rvalexp}}
\irlabel{badqear|\usebox{\Lightningval}}

This article uses a standard, classical sequent calculus~\cite{DBLP:journals/jar/Platzer08} with the usual rules for manipulating logical connectives and sequents, \eg, \irref{orl+andr}, and \irref{cut}.
A summary of all base axioms and proof rules used in this article is in~\rref{app:axiomssummary}, along with a list of the standard propositional and first-order sequent calculus proof rules.
The semantics of \emph{sequent} \(\lsequent{\Gamma}{\fvarA}\) is equivalent to the formula \((\landfold_{\fvarB \in\Gamma} \fvarB) \limply \fvarA\) and the sequent is valid iff its corresponding formula is valid.
Formulas $\Gamma$ are \emph{antecedents} of the sequent, while formula $\fvarA$ is its \emph{succedent}.
Completed branches are marked with $\lclose$ in a sequent proof.
An axiom is \emph{sound} iff all its instances are valid, and a proof rule is \emph{sound} iff the validity of all its \emph{premises} (above the rule bar) imply the validity of its \emph{conclusion} (below rule bar).
When an implicational or equivalence axiom is used, propositional sequent manipulation steps are omitted and the proof step is directly labeled with the axiom, giving the resulting premises accordingly~\cite{DBLP:journals/jar/Platzer17}.
The calculus includes \irref{alll+existsr} rules for the reals allowing $x$ to be instantiated (or witnessed) with a term for the sequents \(\lsequent{\Gamma,\lforall{x}{\fvarA}}{\fvarB}\) and \(\lsequent{\Gamma}{\lexists{x}{\fvarA}}\) respectively.
First-order real arithmetic is decidable~\cite{Bochnak1998} so access to such a decision procedure is assumed; proof steps are labeled with \irref{qear} whenever they follow as a substitution instance of a valid formula of first-order real arithmetic.
Extra care must be taken when using \irref{qear} for extended term languages, as illustrated next.

\begin{example}[Proving with \irref{qear}]
\label{ex:provingwithR}
Consider the following two proofs in the extended term language~\rref{eq:extlang}.
The left sequent proves (as a substitution instance) by \irref{qear} because the negation of a real number is its additive inverse.
In contrast, the right sequent does not prove by \irref{qear} alone (indicated by the subscript on rule \irref{qearexp}) because it uses the fact that the real exponential function is strictly positive.
\begin{minipage}[b]{0.5\textwidth}
{\footnotesizeoff%
\begin{sequentdeduction}[default]
\linfer[qear]{
  \lclose
}
  {\lsequent{} {\exp{(x)} + (-\exp{(x)}) = 0}}
\end{sequentdeduction}
}%
\end{minipage}%
\hfill%
\begin{minipage}[b]{0.5\textwidth}
{\footnotesizeoff%
\begin{sequentdeduction}[default]
\linfer[qearexp]{
  \lclose
}
  {\lsequent{}{\exp{(x)} > 0}}
\end{sequentdeduction}
}%
\end{minipage}\\%

An alternative understanding is that rule~\irref{qear} can be used to conclude valid properties that follow \emph{only} from first-order properties of the real closed fields~\cite{Bochnak1998}.
\end{example}

All of the axiom and rule (schemata) in this article can be derived from axioms in \dL's uniform substitution calculus~\cite{DBLP:journals/jar/Platzer17}, an approach which minimizes the soundness-critical core in implementations of the logic~\cite{DBLP:conf/cade/FultonMQVP15}.
This presentation is omitted as it is not the focus for this article. Readers are referred to the authors' earlier conference version~\cite{DBLP:conf/lics/PlatzerT18} for details.

\subsubsection{Differentials and Lie Derivatives}
\label{subsec:background-differentials}
ODEs $\D{x}=\genDE{x}$ precisely specify time derivatives that their solutions must obey.
The deduction of invariants directly from these differential equations therefore relates to the study of time derivatives of the quantities that the invariants involve.
However, directly using time derivatives leads to numerous subtle sources of unsoundness because they are semantic objects that only make sense when a ``time'' axis even exists at all.
Such a continuous time axis is furnished by the domain of definition of an ODE solution, but time derivatives are not otherwise well-defined in arbitrary contexts, e.g., in isolated states or across discrete transitions.
It is of utmost importance for soundness that, unlike time derivatives, differentials have a local semantics that is well-defined in single states which enables their use in arbitrary contexts for sound syntactic manipulations~\cite{DBLP:journals/jar/Platzer17}.
The crucial differential lemma~\cite[Lem.\,35]{DBLP:journals/jar/Platzer17} shows that, along a solution of the ODE \(\pevolve{\D{x}=\genDE{x}}\), the value of the differential term $\der{\etermA}$ coincides with the time derivative $\D[t]{}$ of the value of term $\etermA$.
This relationship allows conclusions to be drawn about the differential equations directly from syntactic \dL proofs involving differentials.
The latter syntactic manipulation of differentials is achieved using the \emph{differential axioms} of \dL, which are given below.
In axiom \irref{DE}, \(\D{x}=\genDE{x}\) is understood vectorially, i.e., $x$ is a vector of variables $x_1,\dots,x_n$, $\D{x}$ the corresponding vector of differential variables $\D{x_1},\dots,\D{x_n}$, and $\genDE{x}$ a vector of terms $\odeterm_1(x),\dots,\odeterm_n(x)$.

\begin{theorem}[Differential axioms~\cite{DBLP:journals/jar/Platzer17}]
\label{thm:diffaxioms}
The following are sound axioms of \dL:

\begin{calculuscollection}
\begin{calculus}
\cinferenceRule[DE|DE]{differential effect} %
{\linferenceRule[equiv]
  {\dbox{\pevolvein{\D{x}=\genDE{x}}{\ivr\argx}}{\dbox{\Dupdate{\Dumod{\D{x}}{\genDE{x}}}}{\rfvar\argxx}}}
  {\axkey{\dbox{\pevolvein{\D{x}=\genDE{x}}{\ivr\argx}}{\rfvar\argxx}}}
}
{}%
\end{calculus}\\
\begin{calculus}
\;\quad
\cinferenceRule[Dconst|$c'$]{derive constant}
{\linferenceRule[eq]
  {0}
  {\axkey{\der{c}}}
}
{}%

\cinferenceRule[Dplus|$+'$]{derive sum}
{\linferenceRule[eq]
  {\der{\etermA}+\der{\etermB}}
  {\axkey{\der{\etermA + \etermB}}}
}
{}
\end{calculus}
\qquad
\begin{calculus}
\cinferenceRule[Dvar|$x'$]{derive variable}
{\linferenceRule[eq]
  {\D{x}}
  {\axkey{\der{x}}}
}
{}%
\cinferenceRule[Dtimes|$\cdot'$]{derive product}
{\linferenceRule[eq]
  {\der{\etermA}\cdot \etermB + \etermA \cdot\der{\etermB}}
  {\axkey{\der{\etermA \cdot \etermB}}}
}
{}
\end{calculus}
\end{calculuscollection}
\end{theorem}
\begin{proof}[Proof Sketch]
The soundness proofs for these axioms~\cite{DBLP:journals/jar/Platzer17} carry over unchanged for extended term languages since fixed function symbols $\noef$ are interpreted as smooth $C^\infty$ functions.
\end{proof}

The differential effect axiom (\irref{DE}) says that the differential variables $\D{x}$ take on the values of the RHS along solutions to an ODE.
This is expressed on its RHS with an assignment \(\Dupdate{\Dumod{\D{x}}{\genDE{x}}}\) to the differential variable $\D{x}$.
The syntax and semantics of this (discrete) assignment are defined elsewhere~\cite{DBLP:journals/jar/Platzer17} but it suffices to understand it here with the assignment axiom of \dL:
\[
\cinferenceRule[assignb|$\dibox{:=}$]{assignment / substitution axiom}
{\linferenceRule[equiv]
  {\fvarA(e)}
  {\axkey{\dbox{\pupdate{\umod{x}{\etermA}}}{\fvarA\argx}}}
}
{\text{$\etermA$ free for $x$ in $\fvarA$}}
\]

Intuitively, axiom \irref{assignb} says that property $\fvarA\argx$ is true after assigning $\etermA$ to $x$ iff $\fvarA(\etermA)$ is true right now.
Together, axioms \irref{DE+assignb} allow replacing free occurrences of $\D{x}$ in the postcondition $\rfvar\argxx$ yielding postcondition $\rfvar(x,f(x))$.
However, proofs usually involve working with differentials of terms $\der{\etermA}$ rather than differential variables directly.
This is where the differential axioms (\irref{Dconst+Dvar+Dplus+Dtimes}) are used.
Axiom \irref{Dconst} says that the differential of a constant is $0$, while axiom \irref{Dvar} says the differential of a variable $\der{x}$ is the corresponding differential variable $\D{x}$.
Axioms \irref{Dplus+Dtimes} are the sum and product rules of differentiation respectively.
Soundness of these axioms allows differential terms to be rewritten equationally in all contexts, including in the postcondition of an ODE and within sub-terms.
The differential axioms enable sound syntactic differentiation because differential terms $\der{\etermA}$ can be rewritten according to these equational axioms until no further differential sub-terms occur; any remaining differential variables are substituted away using \irref{DE+assignb} under an ODE.
Such \emph{exhaustive use of differential axioms} is simply labeled as \irref{Dder}\irlabel{Dder|$(~)'$} in proofs.
The following example shows such a derivation concretely using a polynomial from the running example:

\begin{example}[Syntactic differentiation]
\label{ex:syndifferentation}
The following \dL derivation syntactically differentiates the polynomial $v^2-u^2+\frac{9}{2}$ along the ODE $\alpha_e$ from~\rref{eq:example-ode} for any comparison operator $\sim$:
{\footnotesizeoff%
\begin{sequentdeduction}[array]
\linfer[DE]{
\linfer[Dder]{
\linfer[assignb]{
\linfer[qearpoly]{
  \lsequent{} {\dbox{\alpha_e}{4uv + \frac{1}{2}(1-u^2-v^2)(v^2-u^2) \sim 0}}
}
  {\lsequent{} {\dbox{\alpha_e}{2v(u + \frac{v}{4} (1 - u^2 - v^2))- 2u(-v + \frac{u}{4} (1 - u^2 - v^2))\sim 0}}}
}
  {\lsequent{} {\dbox{\alpha_e}{\Dusubst{\D{u}}{-v + \frac{u}{4} (1 - u^2 - v^2)}{\Dusubst{\D{v}}{u + \frac{v}{4} (1 - u^2 - v^2)}{2v\D{v}-2u\D{u} \sim 0}}}}}
}
  {\lsequent{} {\dbox{\alpha_e}{\Dusubst{\D{u}}{-v + \frac{u}{4} (1 - u^2 - v^2)}{\Dusubst{\D{v}}{u + \frac{v}{4} (1 - u^2 - v^2)}{\der{v^2-u^2+\frac{9}{2}} \sim 0}}}}}
}
  {\lsequent{} {\dbox{\alpha_e}{\der{v^2-u^2+\frac{9}{2}} \sim 0}}}
\end{sequentdeduction}
}%

The first \irref{DE} step makes available assignments on $\D{u},\D{v}$.
The \irref{Dder} step is then used to syntactically simplify $\der{v^2-u^2+\frac{9}{2}}$ yielding $2v\D{v}-2u\D{u}$.
A subsequent use of \irref{assignb} replaces the resulting differential variables $\D{u},\D{v}$ with their respective RHSes along the ODE.
Finally, rule \irref{qearpoly} is used to rearrange the calculated derivative arithmetically, which results in a (simplified) polynomial term.
\end{example}

\rref{ex:syndifferentation} suggests that differentials can always be eliminated under an ODE.
This is the case for differentials of polynomial terms $\der{p}$, but differential axioms are needed for the fixed function symbols.
Consider the case of a unary fixed function symbol $\noef$ which is semantically interpreted as the function $\noef : \reals \to \reals$.
Expanding the semantics of term $\der{\noef(\etermA)}$ and applying the chain rule:
\begin{align*}
\ivaluation{\I}{\der{\noef(e)}}
&= \sum_{x\in\allvars} \iget[state]{\I}(\D{x}) \Dp[x]{\ivaluation{\I}{\noef(\etermA)}} = \sum_{x\in\allvars} \iget[state]{\I}(\D{x}) \Dp[y]{\noef}(\ivaluation{\I}{\etermA}) \Dp[x]{\ivaluation{\I}{\etermA}} \\
&= \Dp[y]{\noef}(\ivaluation{\I}{\etermA}) \sum_{x\in\allvars} \iget[state]{\I}(\D{x}) \Dp[x]{\ivaluation{\I}{\etermA}} = \Dp[y]{\noef}(\ivaluation{\I}{\etermA}) \ivaluation{\I}{\der{e}}
\end{align*}

The RHS product between $\Dp[y]{\noef}(\ivaluation{\I}{\etermA})$ and $\ivaluation{\I}{\der{e}}$ can be represented syntactically provided that the partial derivative of $\noef$ with respect to its argument is representable as a term.
Assume (suggestively) that such a term is written as $\Dp[y]{\noef}(\etermA)$.
The easiest case is to think of $\Dp[y]{\noef}$ as another unary fixed function symbol and $\Dp[y]{\noef}(\etermA)$ as function application, hence the suggestive notation.
This is not strictly necessary, however: $\Dp[y]{\noef}$ can be another term that mentions variable $y$ free, in which case $\Dp[y]{\noef}(\etermA)$ corresponds to substituting $e$ for $y$ in that term.
The differential axiom \irref{Dnoether} for fixed function symbol $\noef$ uses a product of the (syntactic) partial derivative $\Dp[y]{\noef}$ and  differential $\der{\etermA}$ of $\etermA$:
\[
  \cinferenceRule[Dnoether|$\noef'$]{derive noether}
  {\linferenceRule[eq]
    { \Dp[y]{\noef}(\etermA) \cdot \der{\etermA} }
    {\axkey{\der{\noef(\etermA)}}}
  }
  {}
\]

\begin{example}[Unary extended differential axioms]
For the extended term language~\rref{eq:extlang}, the extended terms for the partial derivatives are as usual from calculus:
\[ \Dp[y]{\exp(y)} = \exp(y) \qquad
   \Dp[y]{\sin(y)} = \cos(y) \qquad
   \Dp[y]{\cos(y)} = -\sin(y) \]

Following the axiom schema \irref{Dnoether}, the differential axioms for these fixed function symbols are:
\[
  \cinferenceRule[Dexp|$exp'$]{derive exp}
  {\linferenceRule[eq]
    { \exp(\etermA) \cdot \der{\etermA} }
    {\axkey{\der{\exp(\etermA)}}}
  }
  {} \quad
  \cinferenceRule[Dsin|$sin'$]{derive sine}
  {\linferenceRule[eq]
    { \cos(\etermA) \cdot \der{\etermA} }
    {\axkey{\der{\sin(\etermA)}}}
  }
  {} \quad
  \cinferenceRule[Dcos|$cos'$]{derive cosine}
  {\linferenceRule[eq]
    { -\sin(\etermA) \cdot \der{\etermA} }
    {\axkey{\der{\cos(\etermA)}}}
  }
  {}
\]

Axioms \irref{Dsin+Dcos} illustrate a syntactic subtlety.
Fixed function symbols must only be introduced in a syntactically complete way with respect to differentials.
The unary function symbol $\sin$ for the trigonometric sine function cannot be added without also adding one for the cosine function because there would otherwise be no way to express the differential of $\sin$ syntactically.\footnote{Technically, $\pi$ could be added and $\cos(x)$ encoded as $\sin(x+\pi)$ but that also requires another 0-ary function symbol $\pi$.}
\end{example}

The following lemma generalizes this syntactic representation condition and gives sound differential axioms for $k$-ary fixed function symbols:

\begin{lemma}[Extended differential axioms]
\label{lem:noefdifferentials}
Let the $k$-ary fixed function symbol $\noef$ be semantically interpreted as a differentiable function $\noef:\reals^k \to \reals$.
Suppose its partial derivative \m{\Dp[y_i]{\noef}(y_1,\dots,y_k)} at $y_1,\dots,y_k$ is syntactically represented by term $\Dp[y_i]{\noef}$ for each $i$ such that \m{\ivaluation{\I}{\Dp[y_i]{\noef}(y_1,\dots,y_k)}} $=$ \m{\Dp[y_i]{\noef}(\iget[state]{\I}(y_1),\dots,\iget[state]{\I}(y_k))} for all states $\iget[state]{\I}$.
Then the differential axiom schema \irref{Dnoethergen} for $\noef$ is sound:
\[
  \cinferenceRule[Dnoethergen|$\noef'$]{derive noether}
  {\linferenceRule[eq]
    {\sum_{i=1}^{k} \Dp[y_i]{\noef}(\etermA_1,\dots,\etermA_k) \cdot \der{\etermA_i} }
    {\axkey{\der{\noef(\etermA_1,\dots,\etermA_k)}}}
  }
  {}
\]
\end{lemma}
\begin{proof}
The terms $\Dp[y_i]{\noef}(\etermA_1,\dots,\etermA_k)$ appearing on the RHS of axiom \irref{Dnoethergen} are understood as (syntactic) function application of $\Dp[y_i]{\noef}$ to the arguments $\etermA_1,\dots,\etermA_k$.
Soundness of this axiom follows from the (multivariate) chain rule and the semantics of differential terms. For any given state $\iget[state]{\I}$:
\allowdisplaybreaks%
\begin{align*}
&\ivaluation{\I}{\der{\noef(e_1,\dots,e_k)}}
= \sum_{x\in\allvars} \iget[state]{\I}(\D{x}) \Dp[x]{\ivaluation{\I}{\noef(e_1,\dots,e_k)}} = \sum_{x\in\allvars} \iget[state]{\I}(\D{x}) \sum_{i=1}^{k} \Dp[y_i]{\noef}(\ivaluation{\I}{\etermA_1},\dots,\ivaluation{\I}{\etermA_k}) \Dp[x]{\ivaluation{\I}{\etermA_i}} \\
&\qquad = \sum_{i=1}^{k} \Dp[y_i]{\noef}(\ivaluation{\I}{\etermA_1},\dots,\ivaluation{\I}{\etermA_k}) \sum_{x\in\allvars} \iget[state]{\I}(\D{x}) \Dp[x]{\ivaluation{\I}{\etermA_i}} = \sum_{i=1}^{k} \Dp[y_i]{\noef}(\ivaluation{\I}{\etermA_1},\dots,\ivaluation{\I}{\etermA_k}) \ivaluation{\I}{\der{e_i}} \\
& \qquad = \sum_{i=1}^{k} \ivaluation{\I}{\Dp[y_i]{\noef}(\etermA_1,\dots,\etermA_k)} \ivaluation{\I}{\der{e_i}} = \ivaluation{\I}{\sum_{i=1}^{k} \Dp[y_i]{\noef}(\etermA_1,\dots,\etermA_k)\der{e_i}}
\end{align*}
The penultimate step uses the fact that all partial derivatives are syntactically represented in the term language in order to replace a semantic function application with its syntactic representation.
\end{proof}

With the differential axioms of~\rref{lem:noefdifferentials}, all differential terms $\der{\etermA}$ under an ODE \(\D{x}=\genDE{x}\) can be axiomatically rewritten to another term not mentioning differentials and differential variables.
This resulting term is the \emph{Lie derivative} of term $\etermA$ along ODE $\D{x}=\genDE{x}$, succinctly written as:
\[ \lie[]{\genDE{x}}{\etermA} \mdefeq \sum_{i=1}^n \Dp[x_i]{\etermA} \cdot \odeterm_i(x) \]

Unlike (semantic) time derivatives, Lie derivatives can be written down syntactically in the extended term language.
Like time derivatives though, Lie derivatives still depend on the ODE context in which they are used, so they do not give a compositional means of defining syntactic differentiation.
The use of differentials in \dL solves this problem by giving a compositional term semantics that is defined independently of any hybrid programs or formulas.
Along an ODE $\D{x}=\genDE{x}$, however, the value of Lie derivative $\lie[]{\genDE{x}}{e}$ coincides with that of the differential $\der{e}$ and \dL allows transformation between the two by proof with the differential axioms~\cite{DBLP:journals/jar/Platzer17}.
The Lie derivative $\lie[]{\genDE{x}}{\etermA}$ is written as $\lied[]{\genDE{x}}{\etermA}$ when $\D{x}=\genDE{x}$ is clear from the context.
The $i$-th \emph{higher Lie derivative} $\lied[i]{\genDE{x}}{e}{}$ of term $\etermA$ along the ODE $\D{x}=\genDE{x}$ is defined by iterating the Lie derivation operator:
\begin{align*}
  \lied[0]{\genDE{x}}{e}{} \mdefeq e, \quad \lied[i+1]{\genDE{x}}{e}{} \mdefeq \lie{\genDE{x}}{\lied[i]{\genDE{x}}{e}{}}{}, \quad \lied[]{\genDE{x}}{e} \mdefeq \lied[1]{\genDE{x}}{e}{}
\end{align*}

\subsubsection{Differential Equation Axiomatization}

Having enabled meaningful syntactic differentiation with differentials, it suffices to give axioms for working with differential equations.
The following are the \dL axioms for differential equations~\cite[Figure 3]{DBLP:journals/jar/Platzer17} as highlighted in~\rref{sec:introduction}.
All axioms are understood for vectorial differential equations as described in \rref{thm:diffaxioms} for axiom \irref{DE}.

\begin{theorem}[Differential equation axiomatization~\cite{DBLP:journals/jar/Platzer17}]
\label{thm:diffeqax}
The following are sound axioms of \dL.
In axiom~\irref{DG}, the $\exists$ quantifier can be replaced with a $\forall$ quantifier.

\begin{calculus}
\irlabel{DI|DI}
\cinferenceRule[DIeq|DI$_=$]{differential invariant}
{\linferenceRule[impl]
  {\dbox{\pevolvein{\D{x}=\genDE{x}}{\ivr}}{\der{\etermA} = 0}
  }
  {\big(\axkey{\dbox{\pevolvein{\D{x}=\genDE{x}}{\ivr}}{\etermA = 0}} \lbisubjunct (\ivr \limply \etermA = 0)\big)}
}{}

\cinferenceRule[DIgeq|DI$_\cmp$]{differential invariant}
{\linferenceRule[impl]
  {\dbox{\pevolvein{\D{x}=\genDE{x}}{\ivr}}{\der{\etermA} \geq 0}
  }
  {\big(\axkey{\dbox{\pevolvein{\D{x}=\genDE{x}}{\ivr}}{\etermA \cmp 0}} \lbisubjunct (\ivr \limply \etermA \cmp 0)\big)}
}{\text{$\cmp$ is either $\geq$ or $>$}}

\cinferenceRule[DC|DC]{differential cut}
{\linferenceRule[impl]
  {\dbox{\pevolvein{\D{x}=\genDE{x}}{\ivr}}{\rrfvar}
  }
  {\big(\axkey{\dbox{\pevolvein{\D{x}=\genDE{x}}{\ivr}}{\rfvar}} \lbisubjunct \dbox{\pevolvein{\D{x}=\genDE{x}}{\ivr\land \rrfvar}}{\rfvar}\big)}
}{}

\cinferenceRule[DG|DG]{differential ghost}
{\linferenceRule[equiv]
  {\lexists{y}{\dbox{\pevolvein{\D{x}=\genDE{x},\D{y}=a(x)y+b(x)}{\ivr}}{\rfvar}}}
  {\axkey{\dbox{\pevolvein{\D{x}=\genDE{x}}{\ivr}}{\rfvar}}}
}{}
\end{calculus}
\end{theorem}
\begin{proof}[Proof Sketch]
The soundness proofs for these axioms~\cite{DBLP:journals/jar/Platzer17} carry over unchanged for extended term languages since fixed function symbols $\noef$ are interpreted as smooth $C^\infty$ functions.
\end{proof}

\emph{Differential invariants} (\irref{DI}) reduce questions about invariance of $\etermA=0,\etermA\cmp0$ (globally, along solutions of the ODE) to local questions about differentials.
Only two instances (\irref{DIeq+DIgeq}) of the more general \irref{DI} axiom~\cite{DBLP:journals/jar/Platzer17} are needed here.
Axiom~\irref{DIeq} says that the value of extended term $\etermA$ always stays zero if its differential $\der{\etermA}$ is always zero along the solution.
Similarly, axiom~\irref{DIgeq} says that $\etermA$ stays non-negative (or strictly positive) if its differential is non-negative.
Note that axiom \irref{DIgeq} only requires $\der{\etermA} \geq 0$ in its premise even for the $\etermA > 0$ case.
These axioms internalize the mean value theorem (see~\rref{app:diaaxioms}).
With the differential axioms of~\rref{subsec:background-differentials}, the differential $\der{\etermA}$ can be soundly and syntactically transformed into the Lie derivative $\lied[]{\genDE{x}}{\etermA}$ in proofs.
\emph{Differential cut} (\irref{DC}) expresses that if the system never leaves $\rrfvar$ while staying in $\ivr$ (the outer assumption), then $\rrfvar$ may be additionally assumed in the domain constraint when proving the postcondition $\rfvar$ (the last subformula).
Even if \irref{DC} increases the deductive power over \irref{DI}~\cite{DBLP:journals/lmcs/Platzer12}, the deductive power increases even further~\cite{DBLP:journals/lmcs/Platzer12} with the \emph{differential ghost} axiom (\irref{DG}) which adds a \emph{fresh} variable $y$ to the system of ODEs for the sake of the proof.
Since $y$ is fresh, its initial value can be either existentially (\irref{DG}) or universally (\irref{DGall}) quantified~\cite{DBLP:journals/jar/Platzer17}.
\irlabel{DGall|DG$_\forall$}%
The syntactic restriction of \irref{DG} is that the new ODE must be linear (or affine) in $y$, hence $a(x),b(x)$ are not allowed to mention $y$.
This restriction prevents the newly added equation from unsoundly restricting the duration of existence for solutions to the differential equations~\cite{DBLP:journals/lmcs/Platzer12}, e.g., the (unsound) differential ghost \(\pevolve{\D{y}=y^2}\) may cause finite-time (or early) blowup of solutions~\cite{Walter1998}.
The added differential ghost variable $y$ co-evolves along solutions and crucially enables the expression of new (integral) relationships between variables along the differential equations.
These new relationships are then used to deduce invariance properties of interest in the original system.
The systematic construction of appropriate differential ghosts is central to this article's completeness results.
For example, the equational invariance axiom~\irref{DIeq} is only complete for invariant functions~\cite{DBLP:conf/itp/Platzer12}.
In~\rref{sec:darboux}, axiom~\irref{DG} is used to extend this deductive power to all Darboux (in)equalities,~\rref{sec:analyticinvs} extends it further to \emph{all} analytic invariants.

To utilize axioms~\irref{DI+DC+DG} in proofs, the following axiom and proof rule of \dL are also used~\cite{DBLP:journals/jar/Platzer17}:
\[\dinferenceRule[dW|dW]{}
{\linferenceRule
  {\lsequent{\ivr}{\rfvar}}
  {\lsequent{\Gamma}{\dbox{\pevolvein{\D{x}=\genDE{x}}{\ivr}}{\rfvar}}}
}{}
\qquad\hspace{2cm}
\cinferenceRule[K|K]{K axiom / modal modus ponens} %
{\linferenceRule[impl]
  {\dbox{\alpha}{(\fvarA \limply \fvarB)}}
  {(\dbox{\alpha}{\fvarA}\limply\axkey{\dbox{\alpha}{\fvarB}})}
}{}
\]

\emph{Differential weakening}~(\irref{dW}) drops the ODEs entirely and proves postcondition $\rfvar$ directly from the evolution domain constraint $\ivr$.
Kripke axiom~\irref{K} is the modal modus ponens for postconditions of the box modality.
Using~\irref{dW} and~\irref{K}, the following proof rule and axiom derive:
\[\dinferenceRule[MbW|M${\dibox{'}}$]{}
{\linferenceRule
  {\lsequent{\ivr,\fvarB}{\fvarA} \qquad \lsequent{\Gamma}{\dbox{\pevolvein{\D{x}=\genDE{x}}{\ivr}}{\fvarB}}}
  {\lsequent{\Gamma}{\dbox{\pevolvein{\D{x}=\genDE{x}}{\ivr}}{\fvarA}}}
}{}
\qquad
\dinferenceRule[band|${\dibox{\cdot}\land}$]{}
{\linferenceRule[equiv]
  {\dbox{\alpha}{\fvarA} \land \dbox{\alpha}{\fvarB}}
  {\axkey{\dbox{\alpha}{(\fvarA \land \fvarB)}}}
}{}
\]

Monotonicity rule \irref{MbW} strengthens the postcondition to $\fvarB$ if it implies $\fvarA$ within the domain constraint $\ivr$ of the ODE $\pevolvein{\D{x}=\genDE{x}}{\ivr}$.
Axiom \irref{band} proves conjunctive postconditions separately, e.g.,~\irref{DIeq} derives from \irref{DIgeq} using \irref{band} with the real arithmetic equivalence $\etermA = 0 \lbisubjunct \etermA \geq 0 \land -\etermA \geq 0$.

\subsection{Extended Term Conditions}
\label{subsec:background-compatibility}

Two natural conditions on the fixed function symbols $\noef \in \{\noef_1,\dots,\noef_r\}$ and their semantics have been uncovered thus far.
For this article's completeness results, a third condition is needed.
All three \emph{extended term conditions} are assumed throughout this article:

\begin{enumerate}[label=(\Alph*)]
\myitem[(S)]\label{itm:reqsmooth}
\emph{Smoothness.}
All fixed function symbols \(\noef \in \{ \noef_1,\dots \noef_r \}\) in the extended term language are interpreted as smooth \(C^\infty$ functions $\noef : \reals^k \to \reals\).

\myitem[(P)]\label{itm:reqpartial}
\emph{Syntactic partial derivatives.}
Each partial derivative $\Dp[y_i]{\noef}$ of \(\noef(y_1,\dots,y_k)\) has a syntactic representation in the extended term language in the sense of~\rref{lem:noefdifferentials}.

\myitem[(R)]\label{itm:reqdiffradical}
\emph{Computable differential radicals.}
The extended term language has computable differential radicals, i.e., for each extended term $\etermA$ and ODE \(\D{x}=\genDE{x}\) with extended terms in its RHS $\genDE{x}$, there must computably exist an $N \geq 1$ and $N$ extended terms $\cofterm_i$ such that the higher Lie derivatives of $\etermA$ along \(\D{x}=\genDE{x}\) provably satisfy the following \emph{differential radical identity}~\cite{DBLP:conf/tacas/GhorbalP14}:
\begin{equation}
\lied[N]{\genDE{x}}{\etermA} = \sum_{i=0}^{N-1} \cofterm_i \lied[i]{\genDE{x}}{\etermA}
\label{eq:differential-rank}
\end{equation}
\end{enumerate}

Condition~\rref{itm:reqsmooth} ensures that the semantics are well-defined, while conditions~\rref{itm:reqpartial} and~\rref{itm:reqdiffradical} enable (complete) \emph{syntactic} analysis of differential equations invariance by their local (differential) behavior.
The $C^\infty$ smoothness required by \rref{itm:reqsmooth} is subtly weaker than \emph{real analyticity}~\cite{MR1916029}.
This article often gives brief but intuitive (semantic) explanations of results, and explicitly indicates when those arguments only apply \emph{in the real analytic setting}.
None of the actual proofs given in this article require real analyticity.
Condition~\rref{itm:reqdiffradical} requires an algorithm that computes and proves the identity~\rref{eq:differential-rank}.
This identity is crucially used for completeness in~\rref{sec:analyticinvs} and~\rref{sec:semianalyticinvs}, where it is also motivated logically.
It yields a finiteness property on the number of Lie derivatives that need to be analyzed for any given term $\etermA$ and ODE.
From identity~\rref{eq:differential-rank}, the first $N-1$ Lie derivatives will turn out to suffice for completely determining the local behavior of extended term $\etermA$ along the ODE \(\D{x}=\genDE{x}\).
All three extended term conditions are met by the polynomial term language without extensions.

\begin{proposition}
\label{prop:poly-compatibility}
Polynomial term languages satisfy the extended term conditions.
\end{proposition}
\begin{proof}[Proof Sketch]
A full proof is omitted because this is a corollary of a later result (\rref{thm:noetheriancompat}).
Briefly, conditions~\rref{itm:reqsmooth} and~\rref{itm:reqpartial} are met because polynomial functions are smooth (even real analytic) and the polynomials are closed under partial derivatives.
Condition~\rref{itm:reqdiffradical} generalizes different flavors of results that have been proved in the literature~\cite{MR1697373,DBLP:conf/emsoft/LiuZZ11,DBLP:conf/tacas/GhorbalP14}.
The proofs rely on the fact that polynomials form a Noetherian ring~\cite{MR1727221} so that the ascending chain of ideals\footnote{The \emph{ideal}~\cite{Bochnak1998} generated by polynomials $\ptermA_1,\dots,\ptermA_s \in \polynomials{\reals}{x}$ is the set of all their linear combination with polynomial cofactors $\cofterm_i \in \polynomials{\reals}{x}$, denoted by $\ideal{\ptermA_1,\dots,\ptermA_s} \mdefeq \{\text{\large$\Sigma$}_{i=1}^s \cofterm_i \ptermA_i \with \cofterm_i \in \reals[x]\}$.}
 formed by successive (polynomial) Lie derivatives stabilizes.
The resulting (polynomial) identity~\rref{eq:differential-rank} is a formula of real arithmetic and can therefore always be proved by the rule~\irref{qear} for decidable real arithmetic.
\end{proof}

It is less straightforward to show that an extended term language like~\rref{eq:extlang} meets these conditions.
Indeed, even the simple language extension~\rref{eq:extlang} already features exponential rings which are not Noetherian~\cite[Remark 1.4.2]{Terzo} and undecidable arithmetic over the trigonometric functions~\cite{DBLP:journals/jsyml/Richardson68}.
In the interest of a general presentation, the question of how to determine if a candidate term language extension $\{\noef_1,\dots,\noef_r\}$ meets the extended term conditions is deferred till~\rref{sec:noetherianfunctions}.
Until then, the only assumption about the extended term language is that it satisfies these three conditions.
This suffices for the completeness results, which the next section begins to show.

\section{Darboux Invariants}
\label{sec:darboux}

This section exploits differential ghosts for proving an important class of invariance properties.
These are called \emph{Darboux invariants} because they are inspired by Darboux polynomials~\cite{Darboux}.
The derived proof rule for Darboux equalities corresponds to the case $N=1$ in the differential radical identity~\rref{eq:differential-rank}, while the subsequent rule for Darboux \emph{in}equalities is a crucial step for the completeness result in~\rref{sec:analyticinvs}.
Their derivations also show how analytic and geometric notions from the theory of differential equations, such as Darboux polynomials~\cite{Darboux} and Gr\"onwall's lemma~\cite[\S29.VI]{Walter1998,DBLP:journals/mathann/Gronwall19} can be internalized syntactically with differential ghost arguments without extension to any axiom.

\subsection{Darboux Equalities}
\label{subsec:darbouxeq}
Assume that the extended term $\etermA$ satisfies the differential radical identity~\rref{eq:differential-rank} with $N=1$ and extended term cofactor $\cofterm$, i.e., $\lied[]{\genDE{x}}{\etermA} = \cofterm\etermA$.
Taking Lie derivatives on both sides gives:
\[
\lied[2]{\genDE{x}}{\etermA} = \lie[]{\genDE{x}}{\lied[]{\genDE{x}}{\etermA}} = \lie[]{\genDE{x}}{\cofterm\etermA}
= \lied[]{\genDE{x}}{\cofterm}\etermA + \cofterm\lied[]{\genDE{x}}{\etermA} = (\lied[]{\genDE{x}}{\cofterm} + \cofterm^2) \etermA
\]

By repeatedly taking Lie derivatives, \emph{all} higher Lie derivatives of $\etermA$ can be written as a product between $\etermA$ and some cofactor.
Now, consider an initial state $\iget[state]{\I}$ where $\etermA$ evaluates to $\ivaluation{\I}{\etermA} = 0$, then:
\[\ivaluation{\I}{\lied[]{\genDE{x}}{\etermA}} = \ivaluation{\I}{\cofterm\etermA} = \ivaluation{\I}{\cofterm} \cdot \ivaluation{\I}{\etermA} = 0\]
Because every higher Lie derivative is a product with $\etermA$, all of them are simultaneously $0$ in state $\iget[state]{\I}$.
Thus, \emph{in the real analytic setting},
 $\etermA=0$ stays invariant along solutions to the ODE starting at $\iget[state]{\I}$ because all its derivatives are $0$.
This motivates the following proof rule for invariance of $\etermA=0$:
\[
\dinferenceRule[dbx|dbx]{Darboux}
{\linferenceRule
  {\lsequent{\ivr} {\lied[]{\genDE{x}}{\etermA} = \cofterm\etermA}}
  {\lsequent{\etermA=0} {\dbox{\pevolvein{\D{x}=\genDE{x}}{\ivr}}{\etermA=0}}}
}{}
\]

Rule \irref{dbx} derives using differential ghosts and is a first hint at their deductive power for equational invariants.
A special case of \irref{dbx} proves invariance for \emph{Darboux polynomials}, which are polynomials $\ptermA$ satisfying the polynomial identity $\lied[]{\genDE{x}}{\ptermA} = \cofterm\ptermA$ for some polynomial cofactor $\cofterm$.
These are of significant interest in the study of (polynomial) ODEs~\cite{Darboux} and invariant generation for hybrid systems~\cite{DBLP:journals/fmsd/SankaranarayananSM08}.
In~\rref{sec:analyticinvs}, \irref{dbx} is generalized vectorially to yield proofs of \emph{all} analytic invariants with differential ghosts.
Although the rule can be derived from \irref{DG} directly, this article follows a detour through a proof rule for Darboux \emph{in}equalities instead, which is crucially used for this vectorial generalization.

\subsection{Darboux Inequalities}
\label{subsec:darbouxineq}

Assume that the extended term $\etermA$ satisfies the Darboux \emph{inequality} $\lied{\genDE{x}}{\etermA} \geq \cofterm\etermA$ for some extended term cofactor $\cofterm$.
Semantically, in an initial state $\iget[state]{\I}$ where $\ivaluation{\I}{\etermA} \geq 0$, Gr\"onwall's lemma~\cite[\S29.VI]{Walter1998,DBLP:journals/mathann/Gronwall19} implies that $\etermA \geq 0$ stays invariant along solutions starting at $\iget[state]{\I}$ because the semantic value of extended term $\etermA$ is bounded below by a (typically decaying) non-negative exponential solution of the non-autonomous linear differential equation \(\D{\etermA}=\cofterm(t)\etermA\) for the variable $e$.
Here, $\cofterm(t)$ is the time-dependent function corresponding to the value of term $\cofterm$ evaluated along the solution to the differential equations \(\pevolve{\D{x}=\genDE{x}}\) from $\iget[state]{\I}$, see~\rref{subfig:gronwall} for an illustration.
Indeed, if $\etermA$ satisfies the Darboux equality $\lied[]{\genDE{x}}{\etermA} = \cofterm\etermA$ with cofactor $\cofterm$, then it satisfies both Darboux inequalities $\lied[]{\genDE{x}}{\etermA} \geq \cofterm\etermA$ and $\lied[]{\genDE{x}}{\etermA} \leq \cofterm\etermA$, giving an alternative semantic argument for the invariance of $\etermA=0$ in rule \irref{dbx}.

Differentials and Lie derivatives along differential equations $\D{x}=\genDE{x}$ provably coincide (\rref{subsec:background-differentials}), so axiomatic Darboux inequalities assume \m{\dbox{\pevolvein{\D{x}=\genDE{x}}{\ivr}}{\der{\etermA}\geq\cofterm\etermA}} and Darboux equalities assume \m{\dbox{\pevolvein{\D{x}=\genDE{x}}{\ivr}}{\der{\etermA}=\cofterm\etermA}}
instead of \(\lied{\genDE{x}}{\etermA} \geq \cofterm\etermA\) and \(\lied{\genDE{x}}{\etermA} = \cofterm\etermA\), respectively.
Darboux (in)equality invariance can be proved purely syntactically with the differential ghost axiom \irref{DG}.

\begin{lemma}[Darboux (in)equalities are differential ghosts] \label{lem:Darboux}
The Darboux equality \irref{DBX} and Darboux inequality \irref{DBXineq} axioms derive from \irref{DG} (and \irref{DI+DC}) for any extended term cofactor $\cofterm$.

\begin{calculus}
\dinferenceRule[DBX|DBX]{Darboux equality axiom}
{\linferenceRule[impl]
  {\dbox{\pevolvein{\D{x}=\genDE{x}}{\ivr}}{\der{\etermA}=\cofterm\etermA}}
  {(\etermA=0 \limply \axkey{\dbox{\pevolvein{\D{x}=\genDE{x}}{\ivr}}{\etermA=0})}}
}{}
\dinferenceRule[DBXineq|DBX${_\cmp}$]{Darboux inequality axiom}
{\linferenceRule[impl]
  {\dbox{\pevolvein{\D{x}=\genDE{x}}{\ivr}}{\der{\etermA}\geq\cofterm\etermA}}
  {(\etermA\cmp0 \limply \axkey{\dbox{\pevolvein{\D{x}=\genDE{x}}{\ivr}}{\etermA\cmp0}})}
  \quad
}{\text{$\cmp$ is either $\geq$ or $>$}}
\end{calculus}
\end{lemma}
\begin{proof}
Axiom~\irref{DBXineq} is derived first before axiom~\irref{DBX} is derived as a corollary.
After propositional normalization, the derivation starts with a \irref{DG+DGall} step, introducing a new ghost variable $y$ satisfying a carefully chosen differential equation \(\D{y}=-\cofterm y\).
Next, \irref{existsr+alll} pick an initial value for $y$. It suffices to pick any $y>0$.
The augmented ODE is abbreviated with $\alpha_y \mnodefequiv \D{x}=\genDE{x}\syssep\D{y}=-\cofterm y$:
{\footnotesizeoff%
\begin{sequentdeduction}[array]
\linfer[DG+DGall]
  {
  \linfer[existsr+alll]
  {
    \lsequent{\dbox{\pevolvein{\alpha_y}{\ivr}}{\der{\etermA}\geq\cofterm\etermA}, \etermA\cmp0,y>0} {\dbox{\pevolvein{\alpha_y}{\ivr}}{\etermA\cmp0}}
  }
    {\lsequent{\lforall{y}{\dbox{\pevolvein{\alpha_y}{\ivr}}{\der{\etermA}\geq\cofterm\etermA}}, \etermA\cmp0} {\lexists{y}{\dbox{\pevolvein{\alpha_y}{\ivr}}{\etermA\cmp0}}}}
  }
  {\lsequent{\dbox{\pevolvein{\D{x}=\genDE{x}}{\ivr}}{\der{\etermA}\geq\cofterm\etermA}, \etermA\cmp0} {\dbox{\pevolvein{\D{x}=\genDE{x}}{\ivr}}{\etermA\cmp0}}}
\end{sequentdeduction}
}%

The augmented ODE $\alpha_y$ has a new provable invariant relationship $\etermA y \cmp 0$ (see~\rref{fig:gronwallcounterweight} and discussion after this proof).
To deduce the original property of interest ($\etermA \cmp 0$) from this new relationship, it suffices to prove $y > 0$ invariant because the formula $\etermA y \cmp 0 \land y > 0 \limply \etermA \cmp 0$ is provable by \irref{qear}.
Axiom \irref{DC} is used to prove $y>0$ separately (right premise abbreviated with \textcircled{1}) and assume it in the evolution domain constraints of the left premise.
Subsequently, monotonicity rule \irref{MbW} and \irref{qear} strengthen the postcondition to $\etermA y \cmp 0$ using the newly added domain constraint $y>0$.
{\footnotesizeoff%
\begin{sequentdeduction}[array]
\linfer[DC]{
\linfer[MbW+qear]{
  \lsequent{\dbox{\pevolvein{\alpha_y}{\ivr \land y > 0}}{\der{\etermA}\geq\cofterm\etermA}, \etermA\cmp0,y>0} {\dbox{\pevolvein{\alpha_y}{\ivr \land y>0}}{\etermA y \cmp 0}}
}
  {\lsequent{\dbox{\pevolvein{\alpha_y}{\ivr \land y > 0}}{\der{\etermA}\geq\cofterm\etermA}, \etermA\cmp0,y>0} {\dbox{\pevolvein{\alpha_y}{\ivr \land y>0}}{\etermA \cmp 0}} }
        !
    \textcircled{1}
    }%
    {\lsequent{\dbox{\pevolvein{\alpha_y}{\ivr}}{\der{\etermA}\geq\cofterm\etermA}, \etermA\cmp0,y>0} {\dbox{\pevolvein{\alpha_y}{\ivr}}{\etermA\cmp0}}}
\end{sequentdeduction}
}%

From the left premise, a \irref{cut+qear} step adds $\etermA y \cmp 0$ to the assumptions using the provable arithmetic formula $\etermA \cmp 0 \land y > 0 \limply \etermA y \cmp 0$.
Axiom \irref{DIgeq} is used to prove the inequational invariant $\etermA y\cmp0$ and the resulting differential $\der{\etermA y}$ simplifies with~\irref{Dder} from \rref{subsec:background-differentials}.
An additional~\irref{DE+assignb} step replaces the differential variable $\D{y}$ according to the augmented ODE $\alpha_y$, before a monotonicity~\irref{MbW} (with a \irref{cut}) and \irref{qear} step closes the derivation using the domain constraint $y>0$.
The differential ghost \m{\D{y}=-\cofterm y} is specifically crafted so that this final arithmetic step proves with~\irref{qear}.
{\footnotesizeoff\renewcommand{\arraystretch}{1.2}%
\begin{sequentdeduction}[array]
\linfer[cut+qear] {
\linfer[DIgeq] {
\linfer[Dder] {
\linfer[DE+assignb] {
\linfer[MbW] {
\linfer[qear] {
  \lclose
}
  {\lsequent{\der{\etermA}\geq\cofterm\etermA, y>0} {\der{\etermA}y + \etermA(-\cofterm y) \geq 0}}
}
  {\lsequent{\dbox{\pevolvein{\alpha_y}{\ivr \land y > 0}}{\der{\etermA}\geq\cofterm\etermA}}{\dbox{\pevolvein{\alpha_y}{\ivr\land y>0}}{\der{\etermA} y + \etermA(-\cofterm y) \geq 0}}}
}
  {\lsequent{\dbox{\pevolvein{\alpha_y}{\ivr \land y > 0}}{\der{\etermA}\geq\cofterm\etermA}}{\dbox{\pevolvein{\alpha_y}{\ivr\land y>0}}{\der{\etermA} y + \etermA \D{y} \geq 0}}}
}
  {\lsequent{\dbox{\pevolvein{\alpha_y}{\ivr \land y > 0}}{\der{\etermA}\geq\cofterm\etermA}}{\dbox{\pevolvein{\alpha_y}{\ivr\land y>0}}{\der{\etermA y} \geq 0}}}
}
  {\lsequent{\dbox{\pevolvein{\alpha_y}{\ivr \land y > 0}}{\der{\etermA}\geq\cofterm\etermA}, \etermA y\cmp0} {\dbox{\pevolvein{\alpha_y}{\ivr\land y>0}}{\etermA y\cmp0}}}
}
  {\lsequent{\dbox{\pevolvein{\alpha_y}{\ivr \land y > 0}}{\der{\etermA}\geq\cofterm\etermA}, \etermA\cmp0, y>0} {\dbox{\pevolvein{\alpha_y}{\ivr\land y>0}}{\etermA y\cmp0}}}
\end{sequentdeduction}
}%

The derivation continues from premise \textcircled{1} with a second differential ghost \(\D{z}=\frac{\cofterm}{2}z\) analogously:
{\footnotesizeoff\renewcommand{\arraystretch}{1.2}%
\begin{sequentdeduction}[array]
\linfer[DG]
{\linfer[existsr+MbW+qear]
  {\linfer[DIeq+Dder]
    {\linfer[DE+assignb]
      {\linfer[dW]
        {\linfer[qear]
          {\lclose}
          {\lsequent{\ivr} {(-\cofterm y)z^2+2yz(\frac{\cofterm}{2}z)=0}}
        }%
        {\lsequent{} {\dbox{\pevolvein{\D{x}=\genDE{x}\syssep\D{y}=-\cofterm y\syssep\D{z}=\frac{\cofterm}{2}z}{\ivr}}{(-\cofterm y)z^2+2yz(\frac{\cofterm}{2}z)=0}}}
      }%
      {\lsequent{} {\dbox{\pevolvein{\D{x}=\genDE{x}\syssep\D{y}=-\cofterm y\syssep\D{z}=\frac{\cofterm}{2}z}{\ivr}}{\D{y}z^2+2yz\D{z}=0}}}
    }%
    {\lsequent{yz^2=1} {\dbox{\pevolvein{\D{x}=\genDE{x}\syssep\D{y}=-\cofterm y\syssep\D{z}=\frac{\cofterm}{2}z}{\ivr}}{yz^2=1}}}
  }%
  {\lsequent{y>0} {\lexists{z}{\dbox{\pevolvein{\D{x}=\genDE{x}\syssep\D{y}=-\cofterm y\syssep\D{z}=\frac{\cofterm}{2}z}{\ivr}}{y>0}}}}
}%
{\lsequent{y>0} {\dbox{\pevolvein{\D{x}=\genDE{x}\syssep\D{y}=-\cofterm y}{\ivr}}{y>0}}}
\end{sequentdeduction}
}%

In the \irref{existsr+MbW+qear} step, observe that if $y>0$ initially, then there exists $z$ such that $yz^2=1$.
Moreover, $yz^2=1$ is sufficient to imply $y > 0$ in the postcondition.
Rule \irref{qear} again applies here since both of these are properties of real arithmetic.
The differential ghost \m{\D{z}=\frac{\cofterm}{2}z} is specifically constructed so that $yz^2=1$ can be proved invariant along the differential equation.

Axiom \irref{DBX} derives using the derived axiom \irref{band}, the equivalence \m{\etermA=0 \lbisubjunct \etermA\geq0 \land -\etermA \geq 0} by \irref{qear}, and the equality \m{\der{-\etermA}=-\der{\etermA}} provable by \irref{Dder}.
The ODE is abbreviated with \(\alpha_x \mnodefequiv {\pevolvein{\D{x}=\genDE{x}}{\ivr}}\):
{\footnotesizeoff%
\begin{sequentdeduction}[array]
\linfer[MbW+qear]
{\linfer[band]
  {\linfer[DBXineq+andl+andr]
    {\lclose}
    {\lsequent{\dbox{\alpha_x}{\der{\etermA}\geq\cofterm\etermA}\land\dbox{\alpha_x}{\der{-\etermA}\geq\cofterm(-\etermA)}, \etermA\geq0 \land -\etermA\geq0} {\dbox{\alpha_x}{\etermA\geq0}\land\dbox{\alpha_x}{{-}\etermA\geq0}}}
  }%
  {\lsequent{\dbox{\alpha_x}{(\der{\etermA}\geq\cofterm\etermA\land\der{-\etermA}\geq\cofterm(-\etermA))}, \etermA\geq0\land-\etermA\geq0} {\dbox{\alpha_x}{(\etermA\geq0\land-\etermA\geq0)}}}
}%
{\lsequent{\dbox{\pevolvein{\D{x}=\genDE{x}}{\ivr}}{\der{\etermA}=\cofterm\etermA}, \etermA=0} {\dbox{\pevolvein{\D{x}=\genDE{x}}{\ivr}}{\etermA=0}}}
\\[-\normalbaselineskip]\tag*{\qedhere}
\end{sequentdeduction}
}%
\end{proof}

\begin{figure}[tb]
\centering
\subfloat[Gr\"onwall's lemma lower bounds ${\lied[]{\genDE{x}}{\etermA}} \geq \cofterm\etermA$]{
\includegraphics[width=.4\textwidth]{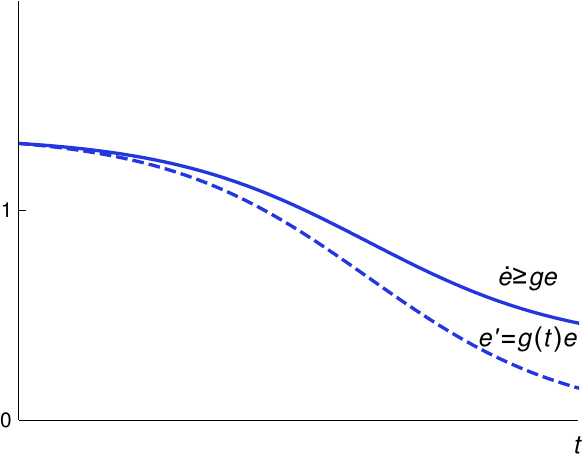}
\label{subfig:gronwall}
}
\qquad
\subfloat[Differential ghost $\D{y}=-\cofterm y$ for $\D{\etermA}=\cofterm(t)\etermA$]{
\includegraphics[width=.4\textwidth]{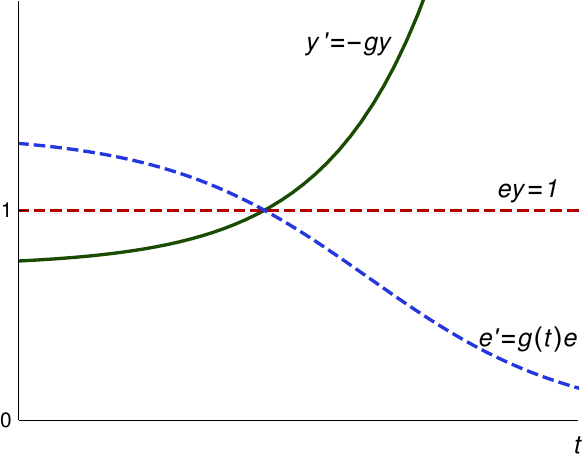}
\label{subfig:counterweight1}
}
\vskip\baselineskip
\subfloat[Differential ghost $\D{y}=-\cofterm y$ for ${\lied[]{\genDE{x}}{\etermA}} \geq \cofterm\etermA$]{
\includegraphics[width=.4\textwidth]{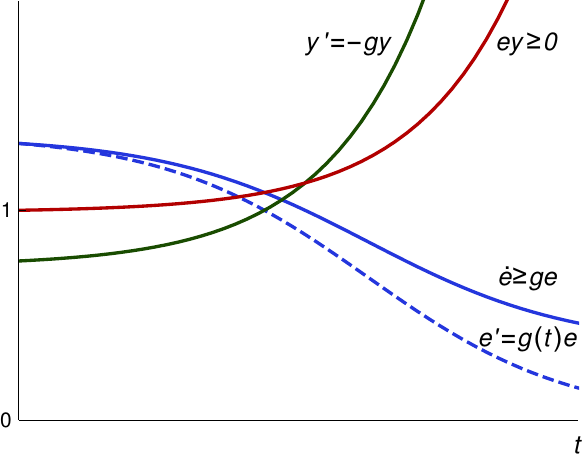}
\label{subfig:counterweight2}
}
\qquad
\subfloat[Differential ghost $\D{z}=\frac{\cofterm}{2}z$ for $\D{y}=-\cofterm y$]{
\includegraphics[width=.4\textwidth]{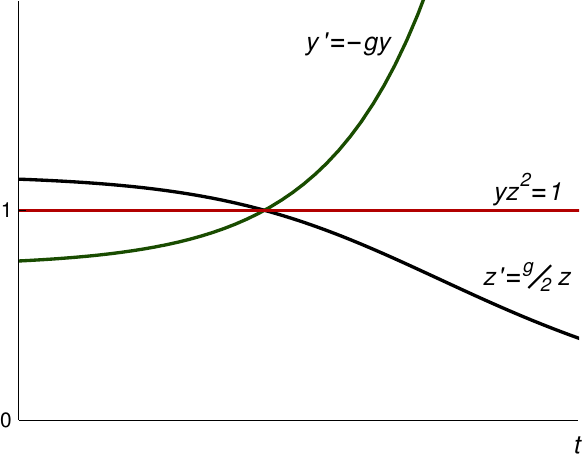}
\label{subfig:counterweight3}
}

\caption{
The horizontal axis tracks the evolution of time $t$ along solutions.
Dashed lines indicate steps based on semantical arguments while solid lines indicate constructions used in the syntactical proof of~\rref{lem:Darboux}.
(a) Solutions of \({\lied[]{\genDE{x}}{\etermA}} \geq \cofterm\etermA\) (solid blue) are bounded below by those of the non-autonomous linear differential equation $\D{\etermA} = \cofterm(t)\etermA$ (dashed blue) by Gr\"onwall's lemma.
(b) The differential ghost \(\D{y}=-\cofterm y\) (solid green) balances out \(\D{\etermA} = \cofterm(t)\etermA\) so that the value of $\etermA y$ (dashed red) remains constant at $1$.
(c) The same ghost \(\D{y}=-\cofterm y\) also balances out \({\lied[]{\genDE{x}}{\etermA}} \geq \cofterm\etermA\), where the value of $\etermA y$ (solid red) remains non-negative but not necessarily constant.
(d) A second differential ghost \(\D{z}=\frac{\cofterm}{2}z\) (solid black) balances out \(\D{y}=-\cofterm y\) so that the value of $yz^2$ (solid red) remains constant at $1$.
The constant $1$ in the RHS of $\etermA y = 1$ and $yz^2 = 1$ in (b) and (d) respectively is chosen for simplicity.
Any positive constant suffices with appropriate initial values of the differential ghosts.
}
\label{fig:gronwallcounterweight}
\end{figure}

The first two syntactic derivation steps in the derivation of~\irref{DBXineq} do not appear to have changed the sequent much, but they correspond to a significant geometric transformation of the problem, as illustrated in~\rref{subfig:counterweight1} and~\rref{subfig:counterweight2}.
In the system extended with differential ghost $y$, there is now a \emph{new} invariant $\etermA y \cmp 0$ which can be observed along solutions!
While the value of $\etermA$ decays (dangerously) towards $0$, the chosen differential equation $\D{y}=-\cofterm y$ yields an (integral) value for $y$ that counteracts this change, ensuring that their product still always stays non-negative along all solutions.
In fact, the value of $\etermA y$ even remains constant when the extended term $\etermA$ satisfies the equational identity $\lied[]{\genDE{x}}{\etermA} = \cofterm\etermA$.
The second differential ghost \m{\D{z}=\frac{\cofterm}{2}z} in the proof is similarly constructed so that $yz^2=1$ can be proved invariant \emph{along the differential equation}.
The geometric transformation from this second syntactic differential ghost is illustrated in~\rref{subfig:counterweight3}.
Since the first differential ghost $y$ satisfies a differential \emph{equation}, the second ghost $z$ exactly balances it out with the value of $yz^2$ remaining (provably) constant and positive at 1 along solutions (similarly to~\rref{subfig:counterweight1}).

The derivation of~\irref{DBXineq} illustrates how the ODE axioms of \dL (\irref{DI+DC+DG}) complement each other in proofs of ODE invariance.
For brevity, the same derivation is used for both $\geq$ and $>$ cases of \irref{DBXineq} even though the latter only needs one ghost (using $\D{y}=-\frac{\cofterm}{2}y$ and invariant $\etermA y^2 > 0$ instead).
Axiom \irref{DBX} also derives directly (similarly to~\irref{DBXineq}, using the invariant $\etermA y = 0$ instead) using just two differential ghosts rather than the four incurred with \irref{band}.

\begin{corollary}[Darboux (in)equality rules] \label{cor:Darboux}
The Darboux equality \irref{dbx} and Darboux inequality \irref{dbxineq} proof rules derive from \irref{DG} (and \irref{DI+DC}) for any extended cofactor term $\cofterm$.
\[
\dinferenceRule[dbxagain|dbx]{Darboux}
{\linferenceRule
  {\lsequent{\ivr} {\lied[]{\genDE{x}}{\etermA} = \cofterm\etermA}}
  {\lsequent{\etermA=0} {\dbox{\pevolvein{\D{x}=\genDE{x}}{\ivr}}{\etermA=0}}}
}{}
\qquad
\dinferenceRule[dbxineq|dbx${_\cmp}$]{Darboux inequality}
{\linferenceRule
  {\lsequent{\ivr} {\lied[]{\genDE{x}}{\etermA}\geq \cofterm\etermA}}
  {\lsequent{\etermA\cmp0} {\dbox{\pevolvein{\D{x}=\genDE{x}}{\ivr}}{\etermA\cmp0}}}
}{\text{$\cmp$ is either $\geq$ or $>$}}
\]
\end{corollary}

\begin{proof}
The \irref{dbx} proof rule derives from axiom~\irref{DBX} (and rule \irref{dbxineq} from axiom \irref{DBXineq}) using an additional~\irref{Dder+DE} step to turn the differential $\der{\etermA}$ into the Lie derivative followed by a~\irref{dW} step:

\begin{minipage}[b]{0.5\textwidth}
{\footnotesizeoff%
\begin{sequentdeduction}[array]
\linfer[DBX]{
\linfer[Dder+DE+assignb]{
\linfer[dW]{
  \lsequent{\ivr} {\lied[]{\genDE{x}}{\etermA}=\cofterm\etermA}
}
  {\lsequent{} {\dbox{\pevolvein{\D{x}=\genDE{x}}{\ivr}}{\lied[]{\genDE{x}}{\etermA}=\cofterm\etermA}}}
}
  {\lsequent{} {\dbox{\pevolvein{\D{x}=\genDE{x}}{\ivr}}{\der{\etermA}=\cofterm\etermA}}}
}
  {\lsequent{\etermA=0} {\dbox{\pevolvein{\D{x}=\genDE{x}}{\ivr}}{\etermA=0}}}
\end{sequentdeduction}
}%
\end{minipage}%
\hfill%
\begin{minipage}[b]{0.5\textwidth}
{\footnotesizeoff%
\begin{sequentdeduction}[array]
\linfer[DBXineq]{
\linfer[Dder+DE+assignb]{
\linfer[dW]{
  \lsequent{\ivr} {\lied[]{\genDE{x}}{\etermA}\geq\cofterm\etermA}
}
  {\lsequent{} {\dbox{\pevolvein{\D{x}=\genDE{x}}{\ivr}}{\lied[]{\genDE{x}}{\etermA}\geq\cofterm\etermA}}}
}
  {\lsequent{} {\dbox{\pevolvein{\D{x}=\genDE{x}}{\ivr}}{\der{\etermA}\geq\cofterm\etermA}}}
}
  {\lsequent{\etermA\cmp0} {\dbox{\pevolvein{\D{x}=\genDE{x}}{\ivr}}{\etermA\cmp0}}}
\\[-\normalbaselineskip]\tag*{\qedhere}
\end{sequentdeduction}
}%
\end{minipage}%
\end{proof}

Axioms~\irref{DBX+DBXineq} yield particularly efficient proofs within \dL's uniform substitution calculus~\cite{DBLP:journals/jar/Platzer17}.
They derive once-and-for-all, independently of the ODE $\D{x}=\genDE{x}$.
Subsequently substituting~\cite{DBLP:journals/jar/Platzer17} for specific ODE instances means that only the final Lie derivative calculation steps~\irref{DE+Dder+assignb} are needed for each concrete derived instance of~\irref{dbx+dbxineq}.

The following example shows a concrete proof utilizing the newly derived proof rules.

\begin{example}[Proving ODE properties in \dL]
\label{ex:continuousproperties}
Judging by the plot (\rref{fig:exampleODE}) of the ODE $\alpha_e$ from~\rref{eq:example-ode}, trajectories from within the open (or closed) disk stay trapped within the disk.
Rather than relying (informally) on a potentially incorrect plot though, this fact can be shown formally by proving that $\etermA \cmp 0$, with $\etermA \mnodefeq 1-u^2-v^2$, is an invariant of $\alpha_e$.
By calculating the Lie derivative:
\[ \lie[]{\alpha_e}{\etermA} = -2u(-v + \frac{u}{4} (1 - u^2 - v^2)) -2v(u + \frac{v}{4} (1 - u^2 - v^2)) = -\frac{1}{2}(u^2+v^2) \etermA \]
Thus, $\etermA$ satisfies the (polynomial) inequality $\lied[]{\alpha_e}{\etermA} \geq \cofterm \etermA$ with polynomial cofactor $\cofterm\mnodefeq-\frac{1}{2}(u^2+v^2)$.
The following derivation with \irref{dbxineq} proves invariance of $1-u^2-v^2 \cmp 0$:
{\footnotesizeoff%
\begin{sequentdeduction}[array]
\linfer[dbxineq]{
  \linfer[qearpoly]{\lclose}{\lsequent{} {\lie[]{\alpha_e}{1 - u^2 - v^2} \geq -\frac{1}{2}(u^2+v^2) (1 - u^2 - v^2)}}
}
  {\lsequent{1 - u^2 - v^2 \cmp 0} {\dbox{\alpha_e}{1 - u^2 - v^2 \cmp 0}}}
\end{sequentdeduction}
}%

The term $\etermA$ obeys the special case $\lied[]{\alpha_e}{\etermA} = \cofterm \etermA$ (\rref{subfig:counterweight1}) in which the seemingly innocuous syntactic introduction of a differential ghost $\D{y}=-\cofterm y$ even \emph{exactly} balances out the complicated (decaying) evolution of $\etermA$ geometrically.
Indeed, in this case, $e = 0$ can also be proved invariant for the ODE $\alpha_e$ using rule \irref{dbx}.
This proves the observation from~\rref{fig:exampleODE} that the unit circle is also invariant for $\alpha_e$.
\end{example}

The derivations of axioms~\irref{DBX+DBXineq} give constructive choices of differential ghosts when the invariant is a Darboux (in)equality.
The derived rule~\irref{dbxineq} exceeds the deductive power of \irref{DI+DC} because the formula \(y > 0 \limply \dbox{\pevolve{\D{y}=-y}}{y > 0}\) is easily provable by~\irref{dbxineq} using the Darboux equality \(\lied[]{\genDE{x}}{y}=-y\), but is \emph{not} provable with~\irref{DI+DC} alone~\cite{DBLP:journals/lmcs/Platzer12}.
The next section builds on these constructions, showing that the deductive power afforded by axiom \irref{DG} extends to \emph{all} true analytic invariants.

\section{Analytic Invariants}
\label{sec:analyticinvs}

Analytic formulas are formed from finite conjunctions and disjunctions of equalities, but, over $\reals$, can be normalized to a single equality $\etermA=0$ using the provable real arithmetic equivalences:
\(\etermA=0 \land \etermB =0 \lbisubjunct \etermA^2 + \etermB^2 = 0\) and \(\etermA=0 \lor \etermB=0 \lbisubjunct \etermA\etermB = 0\).
Thus, it suffices to restrict attention to equational formulas \(\etermA=0\) when proving completeness for analytic invariants.

The key to completeness is the differential radical identity~\rref{eq:differential-rank} for $\etermA$ with arbitrary rank $N \geq 1$, which analyzes \emph{all} higher Lie derivatives simultaneously.
Suppose that extended term $\etermA$ satisfies identity~\rref{eq:differential-rank} with rank $N$ and some cofactors $\cofterm_i$.
Taking Lie derivatives on both sides of~\rref{eq:differential-rank} yields:
\begin{align*}
\lied[N+1]{\genDE{x}}{\etermA} =&~\lie[]{\genDE{x}}{\lied[N]{\genDE{x}}{\etermA}} = \lie[]{\genDE{x}}{\sum_{i=0}^{N-1} \cofterm_i \lied[i]{\genDE{x}}{\etermA}} = \sum_{i=0}^{N-1}  \lie[]{\genDE{x}}{\cofterm_i \lied[i]{\genDE{x}}{\etermA}}
=~\sum_{i=0}^{N-1}  \left(\lied[]{\genDE{x}}{\cofterm_i} \lied[i]{\genDE{x}}{\etermA} + \cofterm_i\lied[i+1]{\genDE{x}}{\etermA}\right) \\
=&\sum_{i=0}^{N-1}  \left(\lied[]{\genDE{x}}{\cofterm_i} \lied[i]{\genDE{x}}{\etermA}\right) + \sum_{i=0}^{N-2}   \left( \cofterm_i\lied[i+1]{\genDE{x}}{\etermA}\right) + \cofterm_{N-1}\lied[N]{\genDE{x}}{\etermA}
=\sum_{i=0}^{N-1}  \left(\lied[]{\genDE{x}}{\cofterm_i} \lied[i]{\genDE{x}}{\etermA}\right) + \sum_{i=0}^{N-2}   \left( \cofterm_i\lied[i+1]{\genDE{x}}{\etermA}\right) + \cofterm_{N-1}\left(\sum_{i=0}^{N-1} \cofterm_i \lied[i]{\genDE{x}}{\etermA}\right)
\end{align*}

The last step follows using~\rref{eq:differential-rank} to expand $\lied[N]{\genDE{x}}{\etermA}$.
Observe that the resulting expression for $\lied[N+1]{\genDE{x}}{\etermA}$ is again a sum over the lower Lie derivatives $\lied[i]{\genDE{x}}{\etermA}$ for $i = 0,\dots,N-1$ multiplied by appropriate cofactors.
By repeatedly taking Lie derivatives on both sides, the higher Lie derivatives $\lied[N]{\genDE{x}}{\etermA}, \lied[N+1]{\genDE{x}}{\etermA}, \dots$ can all be written as sums over these lower Lie derivatives with appropriate cofactors.
Thus, \emph{in the real analytic setting}, for initial states $\iget[state]{\I}$ where $\ivaluation{\I}{\etermA},\ivaluation{\I}{\lied[]{\genDE{x}}{\etermA}},\dots, \ivaluation{\I}{\lied[N-1]{\genDE{x}}{\etermA}}$ all simultaneously evaluate to $0$, $\etermA = 0$ (\emph{and similarly for all its higher Lie derivatives}) stays invariant along solutions to the ODE.

This suggests that rule~\irref{dbx} should be generalized by considering higher Lie derivatives.
The canonical technique for generalizing to higher derivatives comes from the study of ODEs.
All (explicit form) ordinary differential equations involving higher derivatives can be transformed into vectorial systems of differential equations involving only first derivatives but possibly over a vector of variables~\cite[\S11.I]{Walter1998}.
This transformation can be done syntactically and is precisely the idea used to derive the (complete) proof rule for analytic invariants by reduction to a suitable vectorial generalization of rule~\irref{dbx}.
This crucial vectorial generalization is derived first.

\subsection{Vectorial Darboux Equalities}
\label{subsec:vecdarbouxeq}

Suppose that the $m$-dimensional vector of extended terms $\vecpolyn{\etermA}{x}$ satisfies the vectorial identity \(\lied[]{\genDE{x}}{\vecpolyn{\etermA}{x}}=\matpolyn{\coftermC}{x}\itimes\vecpolyn{\etermA}{x}\), where $\matpolyn{\coftermC}{x}$ is an $m \times m$ matrix of extended terms and $\lied[]{\genDE{x}}{\vecpolyn{\etermA}{x}}$ denotes component-wise Lie derivatives of vector $\vecpolyn{\etermA}{x}$ along \(\D{x}=\genDE{x}\), just like $\der{\vecpolyn{\etermA}{x}}$ denotes component-wise differentials.
If all components of $\vecpolyn{\etermA}{x}$ evaluate to $0$ in an initial state, then they all always stay at $0$ along \m{\pevolve{\D{x}=\genDE{x}}} because their component-wise Lie derivatives \emph{all} evaluate to $0$ in that initial state.

\begin{lemma}[Vectorial Darboux equalities are diffrential ghosts] \label{lem:vdbxscalar}
The vectorial Darboux axiom~\irref{VDBX} derives from \irref{DG} (and \irref{DI+DC}), where $\coftermC$ is an $m \times m$ cofactor matrix of extended terms and $\vecpolyn{\etermA}{x}$ is an $m$-dimensional vector of extended terms.
\[
\dinferenceRule[VDBX|VDBX]{vectorial Darboux axiom}
{\linferenceRule[impl]
  {\dbox{\pevolvein{\D{x}=\genDE{x}}{\ivr}}{\der{\vecpolyn{\etermA}{x}}=\matpolyn{\coftermC}{x}\itimes\vecpolyn{\etermA}{x}}}
  {(\vecpolyn{\etermA}{x}=0 \limply \axkey{\dbox{\pevolvein{\D{x}=\genDE{x}}{\ivr}}{\vecpolyn{\etermA}{x}=0}})}
}{}
\]
\end{lemma}
\begin{proof}
First, observe that the formula $\vecpolyn{\etermA}{x}=0$ is provably equivalent in real arithmetic to the formula ${-\normeuc{\vecpolyn{\etermA}{x}}^2} \geq 0$, where the term $\normeuc{\vecpolyn{\etermA}{x}}^2 \mdefeq \sum_{i=1}^m \vecpolyn{\etermA}{x}_i^2$ is the \emph{squared} Euclidean norm of vector $\vecpolyn{\etermA}{x}$.
The derivation starts with~\irref{MbW}, \irref{cut} and~\irref{qear} to rephrase $\vecpolyn{\etermA}{x}=0$ using this equivalence:
{\footnotesizeoff\renewcommand*{\arraystretch}{1.3}%
\begin{sequentdeduction}[array]
  \linfer[MbW+cut+qear]
  {
  \linfer[DBXineq]
  {
  \linfer[MbW]
  {
  \linfer[Dder+qear]
  {
    \lclose
  }
  {\lsequent{\der{\vecpolyn{\etermA}{x}}=\matpolyn{\coftermC}{x}\itimes\vecpolyn{\etermA}{x}}
{\der{-\normeuc{\vecpolyn{\etermA}{x}}^2} \geq g\itimes(-\normeuc{\vecpolyn{\etermA}{x}}^2)}}
  }
    {\lsequent{\dbox{\pevolvein{\D{x}=\genDE{x}}{\ivr}}{\der{\vecpolyn{\etermA}{x}}=\matpolyn{\coftermC}{x}\itimes\vecpolyn{\etermA}{x}}}{\dbox{\pevolvein{\D{x}=\genDE{x}}{\ivr}}{\der{-\normeuc{\vecpolyn{\etermA}{x}}^2} \geq g\itimes(-\normeuc{\vecpolyn{\etermA}{x}}^2)}}}
  }%
  {\lsequent{\dbox{\pevolvein{\D{x}=\genDE{x}}{\ivr}}{\der{\vecpolyn{\etermA}{x}}=\matpolyn{\coftermC}{x}\itimes\vecpolyn{\etermA}{x}},~ {-\normeuc{\vecpolyn{\etermA}{x}}^2} \geq 0} {\dbox{\pevolvein{\D{x}=\genDE{x}}{\ivr}}{{-\normeuc{\vecpolyn{\etermA}{x}}^2} \geq 0}}}
  }
  {\lsequent{\dbox{\pevolvein{\D{x}=\genDE{x}}{\ivr}}{\der{\vecpolyn{\etermA}{x}}=\matpolyn{\coftermC}{x}\itimes\vecpolyn{\etermA}{x}},~ \vecpolyn{\etermA}{x}=0} {\dbox{\pevolvein{\D{x}=\genDE{x}}{\ivr}}{\vecpolyn{\etermA}{x}=0}}}
\end{sequentdeduction}
}%

Thanks to this rephrasing, the sequent no longer contains vectorial quantities and the derivation is completed using a (scalar) \irref{DBXineq} step with the extended term cofactor $\cofterm \mnodefeq \normfrob{\coftermC}^2+1$, where the term $\normfrob{\coftermC}^2 \mdefeq \sum_{i=1}^m\sum_{j=1}^m \coftermC_{ij}^2$ is the squared Frobenius norm of matrix $\coftermC$.
All that remains is to justify the final \irref{Dder+qear} step after \irref{MbW} by showing that the following arithmetic formula is provable:
\begin{equation}
\der{\vecpolyn{\etermA}{x}}=\matpolyn{\coftermC}{x}\itimes\vecpolyn{\etermA}{x}
\,\limply\,
\der{-\normeuc{\vecpolyn{\etermA}{x}}^2} \geq g\itimes(-\normeuc{\vecpolyn{\etermA}{x}}^2)
\label{eq:vdbximp}
\end{equation}

The differential \(\der{-\normeuc{\vecpolyn{\etermA}{x}}^2}\) is calculated (and proved via \irref{Dder} from \rref{subsec:background-differentials}) as follows, where $\vec{u} \stimes \vec{v}$ denotes the dot product of vectors $\vec{u},\vec{v}$.
The last step uses \(\der{\vecpolyn{\etermA}{x}}=\matpolyn{\coftermC}{x}\itimes\vecpolyn{\etermA}{x}\):
\begin{align*}
\der{-\normeuc{\vecpolyn{\etermA}{x}}^2}
&= -\der{ \sum_{i=1}^m \vecpolyn{\etermA}{x}_i^2}
= -2\sum_{i=1}^m \vecpolyn{\etermA}{x}_i \der{\vecpolyn{\etermA}{x}_i}
= -2 (\vecpolyn{\etermA}{x} \stimes \der{\vecpolyn{\etermA}{x}})
= -2 (\vecpolyn{\etermA}{x} \stimes (\matpolyn{\coftermC}{x}\itimes\vecpolyn{\etermA}{x}))
\end{align*}

Thus, it suffices to prove the validity of \(-2(\vecpolyn{\etermA}{x} \stimes (\matpolyn{\coftermC}{x}\itimes\vecpolyn{\etermA}{x})) \geq \cofterm \itimes (-\normeuc{\vecpolyn{\etermA}{x}}^2)\), i.e., its truth in all states $\iget[state]{\I}$.
Validity is first shown semantically.
For ease of notation, let $\ivaluation{\I}{\vecpolyn{\etermA}{x}},\ivaluation{\I}{\matpolyn{\coftermC}{x}}$ stand for the respective real vector and matrix values of $\vecpolyn{\etermA}{x}$ and $\matpolyn{\coftermC}{x}$ evaluated component-wise in state $\iget[state]{\I}$.
The notation $\normeuc{\cdot},\normfrob{\cdot}$ denotes the (real-valued) Euclidean and Frobenius norms for vectors and matrices respectively.

By the Cauchy-Schwarz inequality~\cite[\S28.I]{Walter1998}, the dot product between vectors $\ivaluation{\I}{\vecpolyn{\etermA}{x}}$ and $\ivaluation{\I}{\matpolyn{\coftermC}{x}}\itimes\ivaluation{\I}{\vecpolyn{\etermA}{x}}$ is bounded by the product of their norms:
\[ \ivaluation{\I}{\vecpolyn{\etermA}{x}} \stimes (\ivaluation{\I}{\matpolyn{\coftermC}{x}}\itimes\ivaluation{\I}{\vecpolyn{\etermA}{x}}) \leq \normeuc{\ivaluation{\I}{\vecpolyn{\etermA}{x}}} \, \normeuc{\ivaluation{\I}{\matpolyn{\coftermC}{x}}\itimes\ivaluation{\I}{\vecpolyn{\etermA}{x}}} \]
The norm $\normeuc{\ivaluation{\I}{\matpolyn{\coftermC}{x}}\itimes\ivaluation{\I}{\vecpolyn{\etermA}{x}}}$ of this matrix-vector product is bounded by the product of their matrix and vector norms because the Euclidean and Frobenius norms are compatible~\cite[\S14.II]{Walter1998}:
\[ \normeuc{\ivaluation{\I}{\matpolyn{\coftermC}{x}}\itimes\ivaluation{\I}{\vecpolyn{\etermA}{x}}} \leq \normfrob{\ivaluation{\I}{\matpolyn{\coftermC}{x}}} \, \normeuc{\ivaluation{\I}{\vecpolyn{\etermA}{x}}} \]
Expanding the (square) inequality \(0 \leq (\normfrob{ \ivaluation{\I}{\matpolyn{\coftermC}{x}}}-1)^2\) yields an upper bound on the Frobenius norm $\normfrob{ \ivaluation{\I}{\matpolyn{\coftermC}{x}}} $ by its squared value:
\[ 2\normfrob{\ivaluation{\I}{\matpolyn{\coftermC}{x}}} \leq \normfrob{\ivaluation{\I}{\matpolyn{\coftermC}{x}}}^2 + 1 \]
Chaining these (in)equalities yields:
\begin{align*}
\ivaluation{\I}{-2 (\vecpolyn{\etermA}{x} \stimes (\matpolyn{\coftermC}{x}\itimes\vecpolyn{\etermA}{x}))}
&= -2 (\ivaluation{\I}{\vecpolyn{\etermA}{x}} \stimes (\ivaluation{\I}{\matpolyn{\coftermC}{x}}\itimes\ivaluation{\I}{\vecpolyn{\etermA}{x}}))
\geq -2 \normeuc{\ivaluation{\I}{\vecpolyn{\etermA}{x}}} \, \normeuc{\ivaluation{\I}{\matpolyn{\coftermC}{x}}\itimes\ivaluation{\I}{\vecpolyn{\etermA}{x}}} \\
& \geq -2 \normeuc{\ivaluation{\I}{\vecpolyn{\etermA}{x}}} \, \normfrob{\ivaluation{\I}{\matpolyn{\coftermC}{x}}} \, \normeuc{\ivaluation{\I}{\vecpolyn{\etermA}{x}}}
= -2 \normfrob{\ivaluation{\I}{\matpolyn{\coftermC}{x}}} \, \normeuc{\ivaluation{\I}{\vecpolyn{\etermA}{x}}}^2 \\
& \geq (\normfrob{\ivaluation{\I}{\matpolyn{\coftermC}{x}}}^2 {+} 1)(-\normeuc{\ivaluation{\I}{\vecpolyn{\etermA}{x}}}^2) = \ivaluation{\I}{\cofterm (-\normeuc{\vecpolyn{\etermA}{x}}^2)}
\end{align*}
where $\normfrob{\ivaluation{\I}{\matpolyn{\coftermC}{x}}}^2 + 1$ is precisely the semantic value of cofactor $\cofterm$ in state $\iget[state]{\I}$.
Since this semantic argument for the validity of implication~\rref{eq:vdbximp} only depends on first-order properties of the real closed fields, which is decidable~\cite{Bochnak1998},  formula~\rref{eq:vdbximp} is provable syntactically by \irref{Dder+qear}.
\end{proof}

\begin{corollary}[Vectorial Darboux equality rule] \label{cor:vdbxscalar}
The vectorial Darboux equality proof rule~\irref{vdbx} derives from  \irref{DG} (and \irref{DI+DC}), where $\coftermC$ is an $m \times m$ cofactor matrix of extended terms and $\vecpolyn{\etermA}{x}$ is an $m$-dimensional vector of extended terms.
\[
\dinferenceRule[vdbx|vdbx]{vectorial Darboux}
{\linferenceRule
  {\lsequent{\ivr} {\lied[]{\genDE{x}}{\vec{\etermA}}=\matpolyn{\coftermC}{x}\itimes\vecpolyn{\etermA}{x}}}
  {\lsequent{\vecpolyn{\etermA}{x}=0} {\dbox{\pevolvein{\D{x}=\genDE{x}}{\ivr}}{\vecpolyn{\etermA}{x}=0}}}
}{}
\]
\end{corollary}
\begin{proof}
Rule \irref{vdbx} derives from derived axiom \irref{VDBX} using~\irref{Dder+DE+assignb} to provably transform between $\der{\vecpolyn{\etermA}{x}}$ and $\lied[]{\genDE{x}}{\vec{\etermA}}$, just like rule \irref{dbx} derives from derived axiom \irref{DBX} in \rref{cor:Darboux}.
\end{proof}

The use of axiom~\irref{DBXineq} in the derivation of axiom~\irref{VDBX} corresponds to Gr\"onwall's lemma~\cite[\S29.VI]{Walter1998,DBLP:journals/mathann/Gronwall19}, as illustrated in~\rref{subfig:gronwall}.
In case $\etermA$ starts with value $0$ initially and satisfies the Darboux inequality $\lied[]{\genDE{x}}{\etermA} \geq \cofterm\etermA$, the constant zero solution of the differential equation $\D{\etermA} = \cofterm(t)\etermA$ bounds it from below.
In~\rref{subfig:gronwall}, this corresponds to the case where both blue lines lie exactly on the horizontal axis.
The proof uses the (squared) Euclidean and Frobenius norms to reduce a vectorial equality ($\vecpolyn{\etermA}{x}=0$) to a scalar inequality (${-\normeuc{\vecpolyn{\etermA}{x}}^2} \geq 0$), which enables further analysis using \emph{scalar} differential ghosts.
The convenient choice of compatible norms ensures that all syntactic proof steps are done within the extended term language.
Since all norms are equivalent on finite-dimensional vector spaces~\cite[\S10.III]{Walter1998}, this reduction can also be done using other norms with suitable syntactic representations.
Convenient choices of norms are a common technique in the study of differential equations~\cite{Walter1998}.

An alternative derivation of rule~\irref{vdbx} is given in the authors' earlier conference version~\cite{DBLP:conf/lics/PlatzerT18} based on Liouville's formula~\cite[\S15.III]{Walter1998}.
That alternative derivation has an alternative geometric interpretation as a continuous change of basis that is expressed purely syntactically~\cite{DBLP:conf/lics/PlatzerT18} but requires the use of \emph{vectorial} differential ghosts and a number of ghost variables that is quadratic in the dimension.
The new derivation in~\rref{lem:vdbxscalar} uses exactly $2$ scalar differential ghosts in the \irref{DBXineq} step independent of dimension and relies only on basic properties of real arithmetic.
In fact, just like the scalar Darboux axioms, axiom \irref{VDBX} for $m$-dimensional extended terms $\vecpolyn{\etermA}{x}$ derives once-and-for-all so no differential ghosts are needed for its subsequent use.

\subsection{Completeness for Analytic Invariants}
\label{subsec:completenessanalytic}

Returning to extended terms $\etermA$ satisfying the differential radical identity~\rref{eq:differential-rank}, a proof rule for invariance of $\etermA=0$ based on \emph{higher} Lie derivatives derives as a direct instance of derived rule \irref{vdbx}:

\begin{theorem}[Differential radical invariants are vectorial Darboux] \label{thm:DRI}
The differential radical invariant proof rule~\irref{dRI} derives from \irref{vdbx} (which in turn derives from \irref{DG}).
\[
\dinferenceRule[dRI|dRI]{differential radical invariants}
{\linferenceRule
  {\lsequent{\Gamma,\ivr} {\landfold_{i=0}^{N-1}  \lied[i]{\genDE{x}}{\etermA} = 0} & \lsequent{\ivr} {\lied[N]{\genDE{x}}{\etermA} = \sum_{i=0}^{N-1} \cofterm_i \lied[i]{\genDE{x}}{\etermA}}}
  {\lsequent{\Gamma} {\dbox{\pevolvein{\D{x}=\genDE{x}}{\ivr}}{\etermA=0}}}
}{}
\]
\end{theorem}
\begin{proofsketch}[app:alginvariants]
Rule \irref{dRI} derives from rule \irref{vdbx} by transforming identity~\rref{eq:differential-rank} involving higher Lie derivatives of $\etermA$ into a vectorial Darboux equality involving only first Lie derivatives of the extended term vector.
The proof uses the following choice of cofactor matrix $\matpolyn{\coftermC}{x}$:
{\footnotesizeoff%
\[\matpolyn{\coftermC}{x}= \left(\begin{array}{ccccc}
0      & 1      & 0      & \dots & 0      \\
0      & 0      & \ddots & \ddots & \vdots \\
\vdots & \vdots & \ddots & \ddots & 0      \\
0      & 0      & \dots & 0      & 1\\
\cofterm_0    & \cofterm_1    & \dots & \cofterm_{N-2}& \cofterm_{N-1} \end{array}\right),
\quad
\vecpolyn{\etermA}{x} = \left(\begin{array}{l}\etermA\\ \lied[1]{\genDE{x}}{\etermA}\\ \vdots \\\lied[N-2]{\genDE{x}}{\etermA}\\ \lied[N-1]{\genDE{x}}{\etermA}\end{array}\right)
\]}%
The matrix $\matpolyn{\coftermC}{x}$ has $1$ on its superdiagonal and the $\cofterm_i$ cofactors in the last row.
The left premise of \irref{dRI} is used to show $\vecpolyn{\etermA}{x} = 0$ initially, while its right premise is used to show the premise of \irref{vdbx}.
\end{proofsketch}

For any extended term $\etermA$ in the LHS of normalized equation $\etermA=0$, the computable differential radicals condition~\rref{itm:reqdiffradical} requires that the differential radical identity~\rref{eq:differential-rank} (computably) exists and proves with associated rank $N$ and cofactors $\cofterm_i$ for $\etermA$.
The resulting (provable) identity~\rref{eq:differential-rank} proves the right premise of \irref{dRI}.\footnote{%
\rref{thm:DRI} shows $\ivr$ can be assumed when proving the right premise. A finite rank must exist either way, but assuming $\ivr$ may reduce the number of higher Lie derivatives of $\etermA$ that need to be considered for the proof (as in \rref{ex:expressivity}).}
The succedent in the remaining left premise of \irref{dRI} thus gives a finitary characterization for when \emph{all} Lie derivatives of $\etermA$ evaluate to zero in the initial state.
This motivates the following definition of a \emph{finite} formula summarizing that \emph{all} higher Lie derivatives of $\etermA$ are zero:

\begin{definition}[Differential radical formula]
\label{def:diffradfml}
The \emph{differential radical formula} \(\sigliedzero{\genDE{x}}{\etermA}\) for extended term $\etermA$ of rank $N \geq 1$ from identity \rref{eq:differential-rank} with Lie derivatives along \m{\D{x}=\genDE{x}} is defined to be:
\begin{align*}
\sigliedzero{\genDE{x}}{\etermA} ~\mdefequiv~ \landfold_{i=0}^{N-1}  \lied[i]{\genDE{x}}{\etermA} = 0
\end{align*}
\end{definition}

The finiteness of \(\sigliedzero{\genDE{x}}{\etermA}\) depends on Lie derivatives along the particular differential equation $\D{x}=\genDE{x}$ of interest, because, without considering the ODE, \emph{no} corresponding chain of higher-order differentials would stabilize.
The rest of this article focuses on Lie derivatives to utilize this finiteness property, but relies under the hood on \dL's axiomatic proof transformation from differentials.

The completeness of derived rule \irref{dRI} can be proved semantically by extending earlier arguments~\cite{DBLP:conf/tacas/GhorbalP14} to extended term languages.
Even better: the following equivalent characterization in arithmetic of the future truth of analytic formulas along differential equations (including analytic invariants) derives axiomatically using the extensions developed in \rref{sec:extaxioms}.\footnote{%
With these axiomatic extensions, the requirement in~\rref{thm:algcomplete} that $\ivr$ is formed from strict inequalities is not necessary.
A derived equivalence axiom for analytic invariance with arbitrary semianalytic domain constraint $\ivr$ is given in~\rref{thm:algcompletedom}.}
In contrast to the semantic completeness argument, this syntactic characterization enables complete proofs and \emph{complete disproofs} of analytic invariance within the \dL calculus.
In other words, \emph{disproving} the RHS of the characterization under assumptions $\Gamma$, yields a \dL proof of $\lsequent{\Gamma}{\lnot{\dbox{\pevolvein{\D{x}=\genDE{x}}{\ivr}}{\etermA=0}}}$.

\begin{theorem}[Analytic completeness]
\label{thm:algcomplete}
The differential radical invariant axiom~\irref{DRI} derives in \dL when $\ivr$ is a semianalytic formula formed from conjunctions and disjunctions of strict inequalities:
\[
\dinferenceRule[DRI|DRI]{differential radical invariant axiom}
{\linferenceRule[equiv]
  {\big(\ivr \limply \sigliedzero{\genDE{x}}{\etermA}\big)}
  {\axkey{\dbox{\pevolvein{\D{x}=\genDE{x}}{\ivr}}{\etermA=0}}}
}{}
\]
\end{theorem}
\begin{proofsketch}[app:alginvariants]
The ``$\lylpmi$" direction derives (for any $\ivr$) by an application of derived rule \irref{dRI}, whose right premise closes by \rref{eq:differential-rank}.
The ``$\limply$" direction relies on existence and uniqueness of solutions to differential equations, which are internalized as axioms in \rref{sec:extaxioms}.
\end{proofsketch}

For the proof of \rref{thm:algcomplete}, the additional axioms are \emph{only required} for syntactically deriving the ``$\limply$" direction (completeness) of \irref{DRI}.
The ``$\lylpmi$`` direction (soundness) derives using~\irref{dRI}, which, by \rref{thm:DRI}, can be derived using only \irref{DI+DC+DG}.
Thus, the base \dL axiomatization with differential ghosts is \emph{complete} for proving properties of the form $\dbox{\pevolvein{\D{x}=\genDE{x}}{\ivr}}{\etermA=0}$ because \irref{dRI} provably reduces all such questions to $\ivr \limply \sigliedzero{\genDE{x}}{\etermA}$.
The validity of this resulting semianalytic formula is a purely arithmetical question.
In fact, the base \dL axiomatization \emph{decides} $\dbox{\pevolvein{\D{x}=\genDE{x}}{\ivr}}{\ptermA=0}$ in the case where $\D{x}=\genDE{x}$ is polynomial and $\ivr$ is semialgebraic, because the resulting RHS of~\irref{DRI} is semialgebraic, and hence, a formula of decidable real arithmetic~\cite{Bochnak1998}.
The same applies for the next result, which is a corollary of \rref{thm:algcomplete} but applies beyond the continuous fragment of \dL.

\begin{corollary}[Analytic hybrid program completeness]
\label{cor:testfree}
For analytic formulas $\rfvar$ and analytic hybrid programs $\alpha$, i.e., whose tests and domain constraints are negations of analytic formulas, it is possible to compute an extended term $\etermA$ such that the equivalence \(\dbox{\alpha}{\rfvar} \lbisubjunct \etermA=0\) is derivable in \dL, provided that the term language is Noetherian (defined in~\rref{app:hybridprograms}).
\end{corollary}
\begin{proofsketch}[app:hybridprograms]
By structural induction on $\alpha$ analogous to \cite[Thm.\ 1]{DBLP:conf/lics/Platzer12b}, using \rref{thm:algcomplete} for the differential equations case and the Noetherian term language for loops.
\end{proofsketch}

The Noetherian condition of~\rref{cor:testfree} implies the computable differential radicals condition~\rref{itm:reqdiffradical}.
Polynomial term languages are Noetherian so~\rref{cor:testfree} shows that \dL decides $\dbox{\alpha}{\rfvar}$ where $\rfvar$ and $\alpha$ are both algebraic.
However, the Noetherian condition fails even for simple extended term languages such as~\rref{eq:extlang}.
The stronger Noetherian condition is only required when the analytic hybrid program $\alpha$ contains loops.
Otherwise, the weaker condition~\rref{itm:reqdiffradical} suffices for~\rref{cor:testfree}.

\section{Extended Axiomatization}
\label{sec:extaxioms}

This section presents the axiomatic extension that is used for the rest of this article.
The purpose of this axiomatic extension is to internalize standard properties of differential equations, such as existence and uniqueness~\cite[\S10.VI]{Walter1998} as \emph{syntactic} reasoning principles.
The extension requires that the system \m{\D{x}=\genDE{x}} \emph{locally evolves $x$}, i.e., it has no fixpoint at which $\genDE{x}$ is the 0 vector.
This can be ensured syntactically, e.g., by requiring that the system contains a clock variable $\D{x_1}=1$ that tracks the passage of time.
Such a clock can always first be added using axiom \irref{DG} if necessary.

\subsection{Existence, Uniqueness, and Continuity}
\label{subsec:existenceuniqcont}

The differential equations of \dL are smooth.
Hence, the Picard-Lindel\"{o}f theorem~\cite[\S10.VI]{Walter1998} guarantees that for any initial state $\iget[state]{\I}$, a \emph{unique} solution of the system $\pevolve{\D{x}=\genDE{x}}$, i.e., $\solvar : [0,T] \to \States$ with $\solvar(0) = \iget[state]{\I}$, \emph{exists} for some duration $T > 0$.
The solution $\solvar$ can be extended (uniquely) to its maximal open interval of existence~\cite[\S10.IX]{Walter1998} and $\solvar(\zeta)$ is smooth with respect to $\zeta$.

\begin{lemma}[Continuous existence, uniqueness, and differential adjoints]
\label{lem:uniqcont}
The following axioms are sound.
In \irref{Cont} and \irref{Dadjoint}, $y$ are fresh variables (so not in \m{\pevolvein{\D{x}=\genDE{x}}{\ivr(x)}} or $\etermA$).

\begin{calculus}
\cinferenceRule[Uniq|Uniq]{uniqueness}
{\linferenceRule[equiv]
  {\big(\ddiamond{\pevolvein{\D{x}=\genDE{x}}{\ivr_1}}{\rfvar}\big) \land
  \big(\ddiamond{\pevolvein{\D{x}=\genDE{x}}{\ivr_2}}{\rfvar}\big)}
  {\axkey{\ddiamond{\pevolvein{\D{x}=\genDE{x}}{\ivr_1 \land \ivr_2}}{\rfvar}}}
}{}

\cinferenceRule[Cont|Cont]{continuous existence}
{
 \linferenceRule[impl]
  {x = y}
  {\big(\axkey{\ddiamond{\pevolvein{\D{x}=\genDE{x}}{\etermA > 0}}{x \not= y}} \lbisubjunct \etermA>0\big)}
}{}

\cinferenceRule[Dadjoint|Dadj]{differential adjoints}
{\linferenceRule[equiv]
  {\ddiamond{\pevolvein{\D{y}=-\genDE{y}}{\ivr(y)}}{\,y=x}}
  {\axkey{\ddiamond{\pevolvein{\D{x}=\genDE{x}}{\ivr(x)}}{\,x=y}}}
}{}
\end{calculus}%
\end{lemma}
\begin{proofsketch}[app:extaxiomatization]
\irref{Uniq} internalizes uniqueness, \irref{Cont} internalizes continuity of the values of $\etermA$ and existence of solutions, and \irref{Dadjoint} internalizes differential adjoints by the group action of time on ODE solutions, which is another consequence of existence and uniqueness.
\end{proofsketch}

The \emph{uniqueness axiom} \irref{Uniq} says that if a state has two solutions $\solvar_1,\solvar_2$ respectively staying in evolution domains $\ivr_1,\ivr_2$ and whose endpoints satisfy $\rfvar$, then, by uniqueness, one of $\solvar_1$ or $\solvar_2$ is a prefix of the other, and therefore, that prefix stays in both evolution domains $\ivr_1 \land \ivr_2$ and satisfies $\rfvar$ at its endpoint.
The \emph{continuous existence axiom} \irref{Cont} expresses a notion of \emph{local progress} for differential equations.
It says that from an initial state satisfying $x=y$, the system can locally evolve to another state satisfying $x \neq y$ while still staying in the \emph{open set} of states characterized by $\etermA>0$ iff the initial state is already in that open set.
This uses the assumption that the system locally evolves $x$ at all.
The \emph{differential adjoints} axiom \irref{Dadjoint} expresses that $x$ can flow forward to $y$ iff $y$ can flow backward to $x$ along the negated ODE.
It is at the heart of the ``there and back again" axiom that equivalently expresses properties of differential equations with evolution domains in terms of properties of forwards and backwards differential equations without evolution domains~\cite{DBLP:conf/lics/Platzer12b}.

Although all three axioms are stated as (conditional\footnote{Axiom~\irref{Cont} is sound even without the condition from assumption $x=y$. It is stated conditionally to align with the intuition of local evolution from an initial state satisfying $x=y$.}) equivalences to support intuition, the main properties of interest are their ``$\lylpmi$'' directions.
For example, the ``$\limply$'' direction of~\irref{Uniq} derives from domain constraint monotonicity for the diamond modality (derived rule~\irref{gddR} below).
These diamond monotonicity principles are given below as they are useful for working with the newly introduced axioms.
They derive by duality from the usual \dL box modality principles:

\begin{corollary}[Derived diamond modality domain rules and axioms]
\label{cor:diadiffeqax}
The following axiom and its corollary proof rule derive in \dL:

\begin{calculus}
\dinferenceRule[dDR|DR${\didia{\cdot}}$]{}
{\linferenceRule[impl]
  {\dbox{\pevolvein{\D{x}=\genDE{x}}{\rrfvar}}{\ivr}}
  {\big(\ddiamond{\pevolvein{\D{x}=\genDE{x}}{\rrfvar}}{\rfvar} \limply \axkey{\ddiamond{\pevolvein{\D{x}=\genDE{x}}{\ivr}}{\rfvar}}\big)}
}{}

\dinferenceRule[gddR|dRW${\didia{\cdot}}$]{}
{\linferenceRule
  {
  \lsequent{\rrfvar}{\ivr} &
  \lsequent{\Gamma}{\ddiamond{\pevolvein{\D{x}=\genDE{x}}{\rrfvar}}{\rfvar}}
  }
  {\lsequent{\Gamma}{\ddiamond{\pevolvein{\D{x}=\genDE{x}}{\ivr}}{\rfvar}}}
}{}
\end{calculus}
\end{corollary}
\begin{proofsketch}[app:diaaxioms]
Axiom \irref{dDR} is the diamond version of  \dL's domain refinement axiom that underlies axiom \irref{DC}; if solutions never leave $\ivr$ when staying in $\rrfvar$ (first assumption), then any solution staying in $\rrfvar$ (second assumption) must also stay in $\ivr$ (conclusion).
The rule \irref{gddR} derives from \irref{dDR} using rule \irref{dW} on its first assumption.
\end{proofsketch}

\subsection{Real Induction}
\label{subsec:realind}
The final axiomatic extension is based on the real induction principle~\cite{doi:10.1080/0025570X.2019.1549902}, which is briefly recalled:

\begin{definition}[Inductive subset~\cite{doi:10.1080/0025570X.2019.1549902}]
\label{def:indsubset}
The subset $S\subseteq\interval{[a,b]}$ is called an \emph{inductive} subset of the compact interval $\interval{[a,b]}$ iff for all $a \leq \zeta \leq b$ and $\interval{[a,\zeta)} \subseteq S$,
\begin{enumerate}
\item [\textcircled{1}] $\zeta \in S$.
\item [\textcircled{2}] If $\zeta < b$ then $\interval{(\zeta,\zeta+\epsilon]} \subseteq S$ for some $\epsilon>0$.
\end{enumerate}
Here, $\interval{[a,a)}$ is the empty interval, hence \textcircled{1} requires $a \in S$.
\end{definition}

\begin{proposition}[Real induction~\cite{doi:10.1080/0025570X.2019.1549902}]
\label{prop:realindR}
The subset $S \subseteq [a,b]$ is inductive iff $S=[a,b]$.
\end{proposition}
\begin{proof}
In the ``$\mylpmi$'' direction, $S = [a,b]$ is inductive by definition.
For the ``$\mimply$'' direction, let $S \subseteq [a,b]$ be inductive.
Suppose $S \neq [a,b]$, so that the complement set $\scomplement{S} = [a,b] \setminus S$ is nonempty.
Let $\zeta$ be the infimum of $\scomplement{S}$, then $\zeta \in [a,b]$ since $[a,b]$ is left-closed.
First, note that $[a,\zeta) \subseteq S$.
Otherwise, $\zeta$ is not an infimum of $\scomplement{S}$, because there would exist $a \leq \tau < \zeta$, such that $\tau \in \scomplement{S}$.
By \textcircled{1}, $\zeta \in S$.
Next, if $\zeta = b$, then $S = [a,b]$, contradiction.
Thus, $\zeta < b$, and by \textcircled{2}, $(\zeta,\zeta+\epsilon] \subseteq S$ for some $\epsilon > 0$.
Since $\zeta \in S$, this implies that $\zeta + \epsilon$ is a greater lower bound of $\scomplement{S}$ than $\zeta$, contradiction.
\end{proof}

\rref{prop:realindR} is based on the completeness of the reals~\cite{doi:10.1080/0025570X.2019.1549902} for compact intervals $[a,b]$ of $\reals$.
Applying it to the time axis of ODE solutions yields real induction \emph{along} solutions of differential equations.
For brevity, only the real induction axiom for systems without evolution domain constraints is presented here, leaving the general version to~\rref{app:extaxiomatization}, since evolution domains are definable in \dL \cite{DBLP:conf/lics/Platzer12b}.

\begin{lemma}[Real induction]
\label{lem:realindODE}
The real induction axiom~\irref{RealInd} is sound, where $y$ is fresh in \m{\dbox{\pevolve{\D{x}=\genDE{x}}}{\rfvar}}.
\[
  \cinferenceRule[RealInd|RI]{real induction axiom}
  {\linferenceRule[equiv]
  {\lforall{y}{\dbox{\pevolvein{\D{x}=\genDE{x}}{\rfvar \lor x=y}}{
  \big(x=y \limply \rfvar \land \ddiamond{\pevolvein{\D{x}=\genDE{x}}{\rfvar \lor x=y}}{x \neq y}\big)}}}
  {\axkey{\dbox{\pevolve{\D{x}=\genDE{x}}}{\rfvar}}}
  }{}
\]
\end{lemma}
\begin{proofsketch}[app:extaxiomatization]
The \irref{RealInd} axiom follows from the real induction principle~\cite{doi:10.1080/0025570X.2019.1549902} and the Picard-Lindel\"{o}f theorem~\cite[\S10.VI]{Walter1998}.
\end{proofsketch}

\begin{wrapfigure}[17]{r}{0.4\textwidth}
\centering
\includegraphics[width=0.38\textwidth,trim=0 10 0 10,clip]{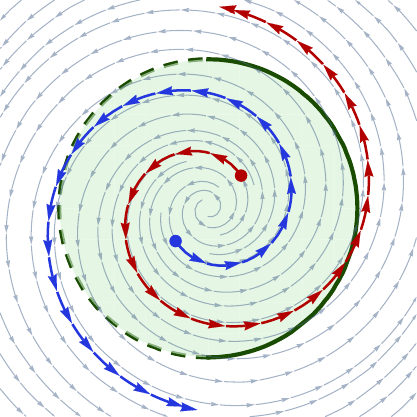}
\caption{The half-open green disk is not invariant for the ODE $\alpha_e$ from~\rref{eq:example-ode} because the red and blue trajectories spiral out of it towards the unit circle at a closed (solid green) or open (dashed green) boundary, respectively.}
\label{fig:realinduct}
\end{wrapfigure}

Real induction axiom~\irref{RealInd} can be understood in relation to~\rref{def:indsubset}: its RHS is true in a state iff the subset of times at which the solution satisfies $\rfvar$ is inductive.
First, $\lforall{y}{\dbox{\dots}{\big(x=y \limply \dots\big)}}$ can be understood as quantifying over all final states ($x=y$) reached by trajectories staying within $\rfvar$ except possibly at the endpoint $x=y$.
This corresponds to $[a,\zeta) \subseteq S$ in~\rref{def:indsubset}.
The left conjunct ($\rfvar$) under the box modality expresses that $\rfvar$ is still true at such an endpoint, corresponding to \textcircled{1} in~\rref{def:indsubset}.
The right conjunct ($\ddiamond{\pevolvein{\D{x}=\genDE{x}}{\rfvar \lor x=y}}{x \neq y}$) expresses that $\rfvar$ continues to remain true locally when following the ODE for a short time, corresponding to \textcircled{2} in~\rref{def:indsubset}.

To see the topological significance of \irref{RealInd}, recall the ODE $\alpha_e$ from~\rref{eq:example-ode} and consider a set of points that is \emph{not invariant}.
Figure~\ref{fig:realinduct} illustrates two trajectories that leave the half-open candidate invariant disk characterized by the disjunctive formula: $u^2+v^2 < \frac{1}{4} \lor u^2+v^2 = \frac{1}{4}\land u \geq 0$.
Trajectories starting in the disk leave it through its boundary but only in one of two ways: either at a point which is also in the disk (red trajectory exiting right) or which is not in the disk (blue trajectory exiting left).
The left conjunct of~\irref{RealInd} rules out trajectories like the blue one exiting left in \rref{fig:realinduct}, while the right conjunct rules out trajectories like the red trajectory exiting right.

The right conjunct of axiom~\irref{RealInd} also suggests a way to use it: axiom~\irref{RealInd} reduces proofs of invariance to local progress properties under the box modality.
This motivates the following syntactic modality abbreviation for \emph{local progress} into an evolution domain $\ivr$:
\[ \dprogressin{\D{x}=\genDE{x}}{\ivr}{} \mdefequiv \ddiamond{\pevolvein{\D{x}=\genDE{x}}{\ivr \lor x=y}}{\,x\neq y} \]

All proofs in this article use the modality with an initial assumption $x=y$, where $y$ is fresh.
In this case, where \(\ivaluation{\I}{x}=\ivaluation{\I}{y}\), since the ODE locally evolves $x$, the $\ddnext$ modality has this semantics:
\begin{align*}
&\imodels{\I}{\dprogressin{\D{x}=\genDE{x}}{\ivr}{}}~\text{iff}~\text{there is a function}~\solvar:[0,T] \to \States \text{~with~} T>0, \solvar(0)=\iget[state]{\I},\\
&\solvar~\text{solves the ODE}~\D{x}=\genDE{x} ~\text{and}~ \solvar(\zeta) \in \imodel{\I}{\ivr}~\text{for all $\zeta$ in the half-open interval}~(0,T]
\end{align*}

Thus, the abbreviation $\ddnext$ is a continuous-time version of the \emph{next} modality of temporal logic~\cite{DBLP:reference/mc/2018} for differential equations.
Conventionally, such a next state operator is excluded from continuous-time generalizations of temporal logic~\cite{DBLP:reference/mc/2018} because there is no unique ``next'' state in the continuous setting.
The local progress modality $\ddnext$ overcomes this by instead quantifying over \emph{some} time interval $(0,T]$, with $T > 0$, of states along the solution.
Intuitively, the exclusion of time $0$ is because the $\ddnext$ modality describes what solutions will do \emph{next} (or \emph{locally}) rather than what it is doing \emph{now}.
A precise (topological) explanation is provided in~\rref{app:localprogress}.
A complete characterization of local progress for all semianalytic formulas is derived in~\rref{sec:semianalyticinvs}.
As a corollary this shows that, like its discrete counterpart, the $\ddnext$ modality is self-dual for semianalytic $\ivr$.

The final derived rule \irref{realind} shows what the added axioms and local progress provide: axiom \irref{RealInd} reduces global invariance properties of ODEs to local progress properties.
These local progress properties are provable using \irref{Cont+Uniq} and the existing \dL axioms, as shown in the next section.
\begin{corollary}[Real induction rule]
\label{cor:realind}
The real induction proof rule~\irref{realind} derives from \irref{RealInd+Dadjoint}.
Variables $y$ are fresh in the ODE $\D{x}=\genDE{x}$ and formula $\rfvar$.

\[
\dinferenceRule[realind|rI]{}
{\linferenceRule
  {
   \lsequent{\initassum, \rfvar}{\dprogressin{\D{x}=\genDE{x}}{\rfvar}} \qquad
   \lsequent{\initassum, \lnot{\rfvar}}{\dprogressin{\D{x}=-\genDE{x}}{\lnot{\rfvar}}}
  }
  {\lsequent{\rfvar}{\dbox{\pevolve{\D{x}=\genDE{x}}}{\rfvar}}}
}{}
\]
\end{corollary}
\begin{proofsketch}[app:diaaxioms]
Rule~\irref{realind} derives from axiom \irref{RealInd}, where the left/right premises of the rule correspond respectively to the right/left conjunct of the RHS of~\irref{RealInd}.
Axiom \irref{Dadjoint} is used to syntactically flip signs in the right premise.
\end{proofsketch}

\section{Semianalytic Invariants}
\label{sec:semianalyticinvs}

From now on, assume domain constraint $\ivr \equiv \ltrue$ since it is definable \cite{DBLP:conf/lics/Platzer12b} and not central to the core idea of this section.
Using the generalizations of \irref{RealInd+realind} from~\rref{app:axiomatization}, the case of semianalytic invariants for ODEs with arbitrary semianalytic evolution domain $\ivr$ is given in~\rref{app:completeness}.

The first step in invariance proofs for semianalytic $\rfvar$ is to use derived rule \irref{realind}, which yields premises of the form \(\lsequent{\initassum, \rfvar}{\dprogressin{\D{x}=\genDE{x}}{\rfvar}}\) (modulo sign changes and negation).
These premises express local progress properties of the ODE $\D{x}=\genDE{x}$.
Analogously to the equivalent arithmetic reduction of equational properties of differential equations in~\rref{thm:algcomplete} using the \emph{finite} differential radical formula (\rref{def:diffradfml}), the key insight is that local progress for any semianalytic formula is also (provably) completely characterized by a \emph{finite} semianalytic progress formula.

\subsection{Local Progress}
\label{subsec:localprogress}

The characterization of local progress was previously used implicitly for semialgebraic invariants~\cite{DBLP:conf/emsoft/LiuZZ11,DBLP:journals/cl/GhorbalSP17}.
This section shows how it can be derived syntactically in \dL for extended term languages and proves the completeness of this characterization.
The derivation is built up systematically, starting from the base case of atomic inequalities before moving on to the full semianalytic case.
Interesting properties of this characterization, e.g., self-duality, are observed along the way.

\subsubsection{Atomic Inequalities}
Consider the atomic inequality $\etermA \cmp 0$.
Intuitively, to show \emph{local progress} into such an inequality, it is sufficient to locally consider the \emph{first} (significant) Lie derivative of $\etermA$ because the sign of a smooth function is locally dominated by the sign of its first non-zero derivative (if one exists).
The key to its syntactic rendition is the following lemma for non-strict inequalities:

\begin{lemma}[Local progress step]
The local progress step axiom~\irref{Lpgeq} derives from \irref{Cont}.
Variables $y$ are fresh in the ODE $\D{x}=\genDE{x}$ and extended term $\etermA$.
\[
\dinferenceRule[Lpgeq|LPi$_\geq$]{}
{
\linferenceRule[impl]
  {\initassum}
  {\Big( \etermA \geq 0 \land \big(\etermA=0 \limply
    \ddiamond{\pevolvein{\D{x}=\genDE{x}}{\lied[]{\genDE{x}}{\etermA} \geq 0}}{\notinitassum}\big)
    \limply
    \axkey{\ddiamond{\pevolvein{\D{x}=\genDE{x}}{\etermA \geq 0}}{\notinitassum}}\Big)}
}{}
\]
\end{lemma}
\begin{proof}
The proof starts with a \irref{orl} case split since $\etermA \geq 0$ is equivalent to \(\etermA > 0 \lor \etermA = 0\) by~\irref{qear}.
The resulting premises are respectively abbreviated \textcircled{1} for $\etermA > 0$ and \textcircled{2} for $\etermA = 0$, and continued below.
{\footnotesizeoff\renewcommand*{\arraystretch}{1.3}%
\begin{sequentdeduction}[array]
\linfer[qear+orl]{
  \textcircled{1} ! \textcircled{2}
  }
  {\lsequent{\initassum, \etermA \geq 0, \etermA = 0 \limply \ddiamond{\pevolvein{\D{x}=\genDE{x}}{\lied[]{\genDE{x}}{\etermA} \geq 0}}{\notinitassum}}{\ddiamond{\pevolvein{\D{x}=\genDE{x}}{\etermA \geq 0}}{\notinitassum}}}
\end{sequentdeduction}
}%

From premise~\textcircled{1}, since the value of $\etermA$ is already positive initially, it must \emph{locally} stay positive.
Using~\irref{gddR}, the non-strict inequality in the domain constraint of the succedent is strengthened to a strict one, after which~\irref{Cont} closes the derivation.
{\footnotesizeoff%
\begin{sequentdeduction}[array]
\linfer[gddR]{
\linfer[Cont]{
  \lclose
}
  {\lsequent{\initassum, \etermA > 0}{\ddiamond{\pevolvein{\D{x}=\genDE{x}}{\etermA > 0}}{\notinitassum}}}
}
  {\lsequent{\initassum, \etermA > 0}{\ddiamond{\pevolvein{\D{x}=\genDE{x}}{\etermA \geq 0}}{\notinitassum}}}
\end{sequentdeduction}
}%

From premise~\textcircled{2}, the local sign of $\etermA$ cannot be determined from its initial value alone.
The proof looks to the Lie derivative of $\etermA$, which is assumed to be locally non-negative (in the implication $\etermA=0 \limply \dots$).
Axiom~\irref{dDR} reduces this to a box modality, after which~\irref{DI} finishes the proof.
{\footnotesizeoff\renewcommand*{\arraystretch}{1.3}%
\begin{sequentdeduction}[array]
\linfer[implyl]{
  \linfer[dDR]{
  \linfer[DI+Dder+DE+assignb]{
    \lclose
  }
  {\lsequent{\etermA = 0}{\dbox{\pevolvein{\D{x}=\genDE{x}}{\lied[]{\genDE{x}}{\etermA} \geq 0}}{\etermA \geq 0}}}
  }
  {\lsequent{\etermA = 0, \ddiamond{\pevolvein{\D{x}=\genDE{x}}{\lied[]{\genDE{x}}{\etermA} \geq 0}}{\notinitassum}}{\ddiamond{\pevolvein{\D{x}=\genDE{x}}{\etermA \geq 0}}{\notinitassum}}}
}
  {\lsequent{\etermA = 0, \etermA = 0 \limply \ddiamond{\pevolvein{\D{x}=\genDE{x}}{\lied[]{\genDE{x}}{\etermA} \geq 0}}{\notinitassum}}{\ddiamond{\pevolvein{\D{x}=\genDE{x}}{\etermA \geq 0}}{\notinitassum}}}
\\[-\normalbaselineskip]\tag*{\qedhere}
\end{sequentdeduction}
}%
\end{proof}

Similar to~\irref{DBX+DBXineq}, and \irref{VDBX}, a version of~\irref{Lpgeq} derives once-and-for-all using differentials ($\der{\etermA} \geq 0$) in domain constraints.
This presentation is omitted to keep with this article's syntactic convention that domain constraints are always (differential-free) semianalytic formulas (\rref{subsec:background-syntax}).
Mathematically, to conclude that $\etermA$ is locally non-negative, it is important that the Lie derivative $\lied[]{\genDE{x}}{\etermA}$ is assumed to be \emph{locally} non-negative rather than just \emph{initially} non-negative.
Just as the local sign of $\etermA$ cannot be determined (directly) when its initial value is zero, the same is true for $\lied[]{\genDE{x}}{\etermA}$.
This difference drives the use of \emph{higher} Lie derivatives when~\irref{Lpgeq} is generalized below.
Syntactically, this difference manifests in both~\irref{dDR+DI} proof steps which crucially rely on the formula $\lied[]{\genDE{x}}{\etermA} \geq 0$ appearing in their respective domain constraints rather than simply as an initial assumption.

Observe that \irref{Lpgeq} allows derivations to pass from reasoning about local progress\footnote{The local progress property used in~\irref{Lpgeq} is syntactically simpler than for the $\ddnext$ modality (no $x=y$ in the domain constraints). For non-strict inequalities, the two are equivalent but the syntactic simplification in~\irref{Lpgeq} allows its re-use as a lemma in proving $\ddnext$ local progress for both non-strict and strict inequalities.} for $\etermA \geq 0$ to local progress for its (first) Lie derivative $\lied[]{\genDE{x}}{\etermA} \geq 0$ whilst accumulating $\etermA=0$ in the antecedent.
This is reminiscent of derivative tests from elementary calculus used for testing the local behavior around a given stationary point of a (sufficiently) smooth function.
The difference is that (syntactic) Lie derivatives have to be used for soundness instead of analytic time derivatives, but these notions are provably equatable along ODEs using differentials~\cite[Lem.\,35]{DBLP:journals/jar/Platzer17}.
Similar to these derivative tests, if the first Lie derivative is indeterminate as well, then the derivation can look to the second higher Lie derivative, and so on.
Deductively, this is done by repeated use of derived axiom~\irref{Lpgeq}:
{\footnotesizeoff\renewcommand{\linferPremissSeparation}{~~~}%
\begin{sequentdeduction}[array]
\linfer[Lpgeq]{
  \lsequent{\Gamma}{\etermA \geq 0} !
  \linfer[Lpgeq]{
  \lsequent{\Gamma,\etermA = 0}{\lied[]{\genDE{x}}{\etermA} \geq 0} !
  \linfer[Lpgeq]{
  \lsequent{\Gamma,\initassum,\etermA = 0,\dots, \lied[k-1]{\genDE{x}}{\etermA} = 0}{\ddiamond{\pevolvein{\D{x}=\genDE{x}}{\lied[k]{\genDE{x}}{\etermA} \geq 0}}{\notinitassum}}
}
  {\dots}
}
  {\lsequent{\Gamma,\initassum,\etermA = 0}{\ddiamond{\pevolvein{\D{x}=\genDE{x}}{\lied[]{\genDE{x}}{\etermA} \geq 0}}{\notinitassum}}}
}
  {\lsequent{\Gamma,\initassum}{\ddiamond{\pevolvein{\D{x}=\genDE{x}}{\etermA \geq 0}}{\notinitassum}}}
\end{sequentdeduction}
}%

The rightmost premise closes whenever the (strict) inequality $\lied[k]{\genDE{x}}{\etermA} > 0$ can be proved from the accumulated antecedents.
The local sign of $\etermA$ (and of all its Lie derivatives below the $k$-th one) is dominated by that of $\lied[k]{\genDE{x}}{\etermA}$ because all of the lower (Lie) derivatives have indeterminate sign.
The use of \irref{Cont+gddR} finishes that proof because the solution must then locally enter $\lied[k]{\genDE{x}}{\etermA} > 0$:
{\footnotesizeoff%
\renewcommand{\linferPremissSeparation}{~~~}%
\begin{sequentdeduction}[array]
\linfer[cut]{
  \lsequent{\Gamma,\initassum,\etermA = 0,\dots,\lied[k-1]{\genDE{x}}{\etermA}= 0}{\lied[k]{\genDE{x}}{\etermA} > 0} !
  \linfer[gddR]{
  \linfer[Cont]{
    \lclose
  }
    {\lsequent{\initassum,\lied[k]{\genDE{x}}{\etermA} > 0}{\ddiamond{\pevolvein{\D{x}=\genDE{x}}{\lied[k]{\genDE{x}}{\etermA} > 0}}{\notinitassum}}}
  }
  {\lsequent{\initassum,\lied[k]{\genDE{x}}{\etermA} > 0}{\ddiamond{\pevolvein{\D{x}=\genDE{x}}{\lied[k]{\genDE{x}}{\etermA} \geq 0}}{\notinitassum}}}
}
  \lsequent{\Gamma,\initassum,\etermA = 0,\dots,\lied[k-1]{\genDE{x}}{\etermA} = 0}{\ddiamond{\pevolvein{\D{x}=\genDE{x}}{\lied[k]{\genDE{x}}{\etermA} \geq 0}}{\notinitassum}}
\end{sequentdeduction}
}%

The extended term conditions of smoothness~\rref{itm:reqsmooth} and syntactic partial derivatives~\rref{itm:reqpartial} guarantee that all of the (infinitely many) higher Lie derivatives of $\etermA$ are well-defined semantically and syntactically.
Derivations, on the other hand, are finite syntactic objects and can only mention \emph{finitely many} Lie derivatives.
Thus, one might suspect they are insufficient (hence incomplete) when, e.g., none of the higher Lie derivatives has a definite sign or if (infinitely many) different choices of $k$ are needed in the proof depending on the initial state that satisfies assumptions $\Gamma$.

This is where the third, computable differential radicals condition~\rref{itm:reqdiffradical} is crucially used.
When $N$ is the rank of $\etermA$ according to identity~\rref{eq:differential-rank}, then once the derivation has gathered $\etermA=0,\dots,\lied[N-1]{\genDE{x}}{\etermA}=0$, i.e., $\sigliedzero{\genDE{x}}{\etermA}$ in the antecedents, derived rule \irref{dRI} proves the invariant $\etermA=0$ and ODEs always locally progress in invariants.
Furthermore, this shows (mathematically) that it is unnecessary to go to higher Lie derivatives when proving local progress for $\etermA > 0$ because none of the signs of the higher Lie derivatives of $\etermA$ beyond $N$ will be sign-definite.
Thus, the rank provides a uniform and finite bound for the number of Lie derivatives of $\etermA$ that need to be analyzed in any state, regardless of assumptions $\Gamma$.
This finiteness property motivates the following definition, which gathers the above open premises to obtain a finite formula characterizing the \emph{first significant Lie derivative} of $\etermA$:

\begin{definition}[First significant Lie derivative]
The \emph{progress formula} \m{\sigliedgt{\genDE{x}}{\etermA}} for extended term $\etermA$ of rank $N \geq 1$ from identity \rref{eq:differential-rank} with Lie derivatives along \m{\D{x}=\genDE{x}} is defined to be:
\begin{align*}
\sigliedgt{\genDE{x}}{\etermA} \mdefequiv & \etermA \geq 0 \land \big(\etermA = 0 \limply \lied[]{\genDE{x}}{\etermA} \geq 0\big) \land \big(\etermA = 0 \land \lied[]{\genDE{x}}{\etermA} = 0  \limply \lied[2]{\genDE{x}}{\etermA} \geq 0\big) \\
\land& \dots \land \big(\etermA = 0 \land \lied[]{\genDE{x}}{\etermA} = 0 \land \dots \land \lied[N-3]{\genDE{x}}{\etermA} = 0 \limply \lied[N-2]{\genDE{x}}{\etermA} \geq 0\big) \\
\land& \big(\etermA = 0 \land \lied[]{\genDE{x}}{\etermA} = 0 \land \dots \land \lied[N-2]{\genDE{x}}{\etermA} = 0 \limply \lied[N-1]{\genDE{x}}{\etermA} > 0\big)
\end{align*}
The \emph{progress formula} $\sigliedgeq{\genDE{x}}{\etermA}$ is defined to be \m{\sigliedgt{\genDE{x}}{\etermA} \lor \sigliedzero{\genDE{x}}{\etermA}}.
The formulas $\sigliedgt[-]{\genDE{x}}{\etermA}$ (or $\sigliedgeq[-]{\genDE{x}}{\etermA}$) are identical except their Lie derivatives are along ODE \(\D{x}=-\genDE{x}\) instead.
\end{definition}

\begin{lemma}[Local progress $\cmp$]
\label{lem:localprogresscmp}
The local progress inequality axioms~\irref{Lpgeqfull+Lpgtfull} derive from \irref{Lpgeq} and thus from~\irref{Cont}.
Variables $y$ are fresh in the ODE $\D{x}=\genDE{x}$ and extended term $\etermA$.

\begin{calculus}
\dinferenceRule[Lpgeqfull|LP$_{\geq^*}$]{Progress Conditions}
{
\linferenceRule[impl]
  {\initassum}
  {\big( \sigliedgeq{\genDE{x}}{\etermA} \limply \axkey{\dprogressin{\D{x}=\genDE{x}}{\etermA \geq 0}} \big)}
}{}

\dinferenceRule[Lpgtfull|LP$_{>^*}$]{Progress Conditions}
{
\linferenceRule[impl]
  {\initassum}
  {\big( \sigliedgt{\genDE{x}}{\etermA} \limply \axkey{\dprogressin{\D{x}=\genDE{x}}{\etermA > 0}}\big)}
}{}
\end{calculus}
\end{lemma}
\begin{proofsketch}[app:localprogress]
Both axioms derive after unfolding the syntactic abbreviation of the $\ddnext$ modality.
Axiom~\irref{Lpgeqfull} derives by the preceding discussion with iterated use of derived axioms \irref{Lpgeq} and \irref{dRI}.
Axiom~\irref{Lpgtfull} derives similarly, but with an additional tweak to weaken the strict inequality $\etermA > 0$ so that axiom~\irref{Lpgeq} can be used.
\end{proofsketch}

The difference between the derivations of~\irref{Lpgeqfull} and~\irref{Lpgtfull} is mainly technical and boils down to the handling of the assumptions about the initial state, and in particular, $\initassum$ (see Appendix~\ref{app:diaaxioms} and~\ref{app:localprogress}).
Intuitively, the difference arises from the fact that the formula $\etermA \geq 0$ characterizes a topologically closed set while $\etermA > 0$ characterizes an open set.
To locally progress into a set from initial state $\iget[state]{\I}$, the state $\iget[state]{\I}$ must already be in the topological closure of that set.
Closed sets are equal to their closure so, e.g., $\iget[state]{\I}$ must already satisfy $\etermA \geq 0$ in order to locally progress into it.
Sets that are not closed (e.g., open sets) are not equal to their closure as they lack points on their topological boundary.
An example of this is the half-open disk illustrated in~\rref{fig:realinduct}.
Thus, it is possible to locally progress into such sets from their topological boundary without $\iget[state]{\I}$ already starting in the set.

\subsubsection{Semianalytic Formulas}
Semianalytic formulas $\rfvar$ normalize propositionally to the following disjunctive normal form with extended terms $\etermA_{ij},\etermB_{ij}$:
\begin{equation}
\rfvar \mequiv \lorfold_{i=0}^{M} \Big(\landfold_{j=0}^{m(i)} \etermA_{ij} \geq 0 \land \landfold_{j=0}^{n(i)} \etermB_{ij} > 0\Big)
\label{eq:normalform}
\end{equation}
Progress formulas lift homomorphically to semianalytic formulas in normal form:
\begin{definition}[Semianalytic progress formula]
The \emph{semianalytic progress formula} $\sigliedsai{\genDE{x}}{\rfvar}$ for a semianalytic formula $\rfvar$ in normal form~\rref{eq:normalform} and Lie derivatives along \m{\D{x}=\genDE{x}} is defined to be:
\begin{align*}
\sigliedsai{\genDE{x}}{\rfvar} ~\mdefequiv~ \lorfold_{i=0}^{M} \Big(\landfold_{j=0}^{m(i)} \sigliedgeq{\genDE{x}}{\etermA_{ij}} \land \landfold_{j=0}^{n(i)} \sigliedgt{\genDE{x}}{\etermB_{ij}}\Big)
\end{align*}
The formula $\sigliedsai[-]{\genDE{x}}{\rfvar}$ takes Lie derivatives along ODE \(\D{x}=-\genDE{x}\) instead.
\end{definition}

A mention of the notation $\sigliedsai{\genDE{x}}{\rfvar}$ is understood as the progress formula for semianalytic formula $\rfvar$ after it is rewritten propositionally to any equivalent normal form~\rref{eq:normalform}.

\begin{lemma}[Semianalytic local progress]
\label{lem:localprogresssemialg}
The local progress formula axiom~\irref{LpRfull} derives from $\irref{Cont+Uniq}$.
Variables $y$ are fresh in the ODE \(\D{x}=\genDE{x}\) and semianalytic formula $\rfvar$.
\[
\dinferenceRule[LpRfull|LP\usebox{\Rval}]{Progress Condition}
{\linferenceRule[impl]
  {\initassum}
  {\big( \sigliedsai{\genDE{x}}{\rfvar} \limply \axkey{\dprogressin{\D{x}=\genDE{x}}{\rfvar}} \big)}
}{}
\]
\end{lemma}
\begin{proofsketch}[app:localprogress]
The shape of the semianalytic progress formula $\sigliedsai{\genDE{x}}{\rfvar}$ guides the proof.
The derivation is sketched at a high level here for the representative example formula:
\begin{align*}
\rfvar &\mequiv (\etermA_1 \geq 0 \land \etermB_1 > 0) \lor (\etermA_2 \geq 0 \land \etermB_2 > 0)  \\
\sigliedsai{\genDE{x}}{\rfvar} &\mequiv ( \sigliedgeq{\genDE{x}}{\etermA_1} \land \sigliedgt{\genDE{x}}{\etermB_1}) \lor (\sigliedgeq{\genDE{x}}{\etermA_2} \land \sigliedgt{\genDE{x}}{\etermB_2})
\end{align*}

To show local progress into a \emph{disjunction}, it suffices to show local progress into either disjunct.
The derivation starts by decomposing $\sigliedsai{\genDE{x}}{\rfvar}$ according to its (outermost) disjunction and accordingly decomposing $\rfvar$ in the local progress succedent with \irref{gddR}.
The premise for the second disjunct resulting from the~\irref{orl} step is symmetric and omitted here.
{\footnotesizeoff%
\begin{sequentdeduction}[array]
\linfer[orl]{
\linfer[gddR]{
  \lsequent{\initassum,\sigliedgeq{\genDE{x}}{\etermA_1} \land \sigliedgt{\genDE{x}}{\etermB_1} }{\dprogressin{\D{x}=\genDE{x}}{\etermA_1 \geq 0 \land \etermB_1 > 0}}
}
  {\lsequent{\initassum, \sigliedgeq{\genDE{x}}{\etermA_1} \land \sigliedgt{\genDE{x}}{\etermB_1} }{\dprogressin{\D{x}=\genDE{x}}{\rfvar}}}
}
  {\lsequent{\initassum, \sigliedsai{\genDE{x}}{\rfvar} }{\dprogressin{\D{x}=\genDE{x}}{\rfvar}}}
\end{sequentdeduction}
}%

To show local progress into a \emph{conjunction}, by \irref{Uniq}, it suffices to show local progress into both conjuncts separately.
The derivation continues using \irref{Uniq+andr} to split the conjunctive local progress succedent before the derived axioms \irref{Lpgeqfull+Lpgtfull} are used to finish the proofs of the atomic cases:
{\footnotesizeoff\renewcommand{\linferPremissSeparation}{~~~}%
\begin{sequentdeduction}[array]
\linfer[Uniq+andr]{
  \linfer[Lpgeqfull]{
    \lclose
  }
  {\lsequent{\initassum,\sigliedgeq{\genDE{x}}{\etermA_1}}{\dprogressin{\D{x}=\genDE{x}}{\etermA_1 {\geq} 0}}} !
  \linfer[Lpgtfull]{
    \lclose
  }
  {\lsequent{\initassum,\sigliedgt{\genDE{x}}{\etermB_1}}{\dprogressin{\D{x}=\genDE{x}}{\etermB_1 {>} 0}}}
}
  {\lsequent{\initassum,\sigliedgeq{\genDE{x}}{\etermA_1}, \sigliedgt{\genDE{x}}{\etermB_1} }{\dprogressin{\D{x}=\genDE{x}}{\etermA_1 \geq 0 \land \etermB_1 > 0}}}
\\[-\normalbaselineskip]\tag*{\qedhere}
\end{sequentdeduction}
}%
\end{proofsketch}

Completeness could potentially be lost in several steps of the proof of~\rref{lem:localprogresssemialg}, e.g., the use of~\irref{orl} at the start of the derivation, or the implicational axioms~\irref{Lpgeqfull+Lpgtfull}.
The converse (completeness) direction of axiom~\irref{LpRfull} therefore does not follow immediately from~\rref{lem:localprogresssemialg}.
Instead, the axiom~\irref{LpRfull} can be re-used to derive its own strengthening to an equivalence.
This equivalence justifies the syntactic abbreviation $\ddnext$, recalling that the $\ddnext$ modality of temporal logic is self-dual.
It also shows that the progress formulas are congruent over equivalences.

\begin{theorem}[Local progress completeness]
\label{thm:localprogresscomplete}
The local progress axiom~\irref{Lpiff} derives from $\irref{Cont}$, $\irref{Uniq}$.
Variables $y$ are fresh in the ODE \(\D{x}=\genDE{x}\) and semianalytic formula $\rfvar$.
\[\dinferenceRule[Lpiff|LP]{Iff Progress Condition}
{\linferenceRule[impl]
  {\initassum}
  {\big( \axkey{\dprogressin{\D{x}=\genDE{x}}{\rfvar}} \lbisubjunct \sigliedsai{\genDE{x}}{\rfvar} \big)}
}{}\]
\end{theorem}
\begin{corollary}[Duality and congruence]
\label{cor:localprogresscompletedualcong}
The duality axiom~\irref{duality} and congruence proof rule for progress formulas~\irref{CLP} derive from~\irref{Lpiff} and thus from $\irref{Cont+Uniq}$.
Variables $y$ are fresh in the ODE \(\D{x}=\genDE{x}\) and semianalytic formulas $\rfvar, \rrfvar$.

\begin{calculus}
\dinferenceRule[duality|$\lnot{\ddnext}$]{Duality}
{\linferenceRule[impl]
  {\initassum}
  {\big( \axkey{\dprogressin{\D{x}=\genDE{x}}{\rfvar}} \lbisubjunct \lnot{\dprogressin{\D{x}=\genDE{x}}{\lnot{\rfvar}}} \big)}
}{}
\dinferenceRule[CLP|CLP]{Congruence of local progress}
{\linferenceRule
  {\rfvar \lbisubjunct \rrfvar}
  {\sigliedsai{\genDE{x}}{\rfvar} \lbisubjunct \sigliedsai{\genDE{x}}{\rrfvar}}
}{}
\end{calculus}

\end{corollary}

\begin{proof}[Proof Summary for~\rref{thm:localprogresscomplete} and~\rref{cor:localprogresscompletedualcong} (\rref{app:localprogress})]
The axioms (and proof rule) follow from the homomorphic definition of semianalytic progress formulas which implies that any semianalytic formula $\rfvar$ in normal form \rref{eq:normalform} has a corresponding normal form for $\lnot{\rfvar}$ such that the equivalence $\lnot{(\sigliedsai{\genDE{x}}{\rfvar})} \lbisubjunct \sigliedsai{\genDE{x}}{(\lnot{\rfvar})}$ is provable.
Classically, in any state, either formula $\sigliedsai{\genDE{x}}{\rfvar}$ or $\lnot{(\sigliedsai{\genDE{x}}{\rfvar})}$ is true.
Therefore, by~\irref{LpRfull}, the ODE must (exclusively, by uniqueness) either locally progress into $\rfvar$ or $\lnot{\rfvar}$ from this state.
Both axioms~\irref{Lpiff+duality} are derivable consequences of this fact, as shown syntactically in~\rref{app:localprogress}.
Rule~\irref{CLP} follows from~\irref{Lpiff} by congruential equivalence~\cite{DBLP:journals/jar/Platzer17}.
\end{proof}

Congruence rule~\irref{CLP} shows that \emph{any} equivalent choice of normal form~\rref{eq:normalform} for semianalytic formula $\rfvar$ gives a local progress formula that is (provably) equivalent to $\sigliedsai{\genDE{x}}{\rfvar}$.
Note that the rule works for all (semianalytic) equivalences, including arithmetical ones, e.g., $\exp{(x)}=1 \lbisubjunct x=0$ from~\rref{eq:extlang}, so \irref{CLP} does not follow immediately from the homomorphic definition of progress formulas.

\subsection{Completeness for Semianalytic Invariants}
\label{subsec:completenesssemianalytic}

Combining derived axiom~\irref{Lpiff} and derived rule~\irref{realind} yields an effective proof rule which reduces a semianalytic invariance question to questions involving purely arithmetic formulas:
\begin{theorem}[Semianalytic invariants] \label{thm:sAI}
The semianalytic invariant proof rule~\irref{sAI} derives from \irref{RealInd+Dadjoint+Cont+Uniq} for semianalytic formula $\rfvar$.
\[
\dinferenceRule[sAI|sAI]{semianalytic invariants}
{\linferenceRule
  {
  \lsequent{\rfvar} {\sigliedsai{\genDE{x}}{\rfvar}} &
  \lsequent{\lnot{\rfvar}} {\sigliedsai[-]{\genDE{x}}{(\lnot{\rfvar})}}
  }
  {\lsequent{\rfvar}{\dbox{\pevolve{\D{x}=\genDE{x}}}{\rfvar}}}
}{}
\]
\end{theorem}
\begin{proof}
This follows immediately by rewriting the premises of rule~\irref{realind} with the equivalence~\irref{Lpiff}.
\end{proof}

Completeness of \irref{sAI} was first proved semantically for polynomial terms languages~\cite{DBLP:conf/emsoft/LiuZZ11}, making crucial use of semialgebraic sets and real analytic solutions to polynomial ODE systems.
The proof rule \irref{sAI} \emph{derives} syntactically in \dL and generalizes to semianalytic invariants for extended term languages.
Its completeness derives syntactically too, which yields \dL disproofs of semianalytic invariance when arithmetic counterexamples can be found.

\begin{theorem}[Semianalytic invariant completeness]
\label{thm:semialgcompleteness}
The semianalytic invariant axiom~\irref{SAI} derives from \irref{RealInd+Dadjoint+Cont+Uniq} for semianalytic formula $\rfvar$.
\[\dinferenceRule[SAI|SAI]{Semianalytic invariant axiom}
{\linferenceRule[equiv]
  {
    \lforall{x}{\big(\rfvar \limply \sigliedsai{\genDE{x}}{\rfvar}\big)}
    \land
    \lforall{x}{\big(\lnot{\rfvar} \limply \sigliedsai[-]{\-genDE{x}}{(\lnot{\rfvar})}\big)}
   }
  {\axkey{\lforall{x}{(\rfvar \limply \dbox{\pevolve{\D{x}=\genDE{x}}}{\rfvar})}}}
}{}\]
\end{theorem}
In~\rref{app:completeness}, a generalization of \rref{thm:semialgcompleteness} is proven that handles semianalytic evolution domains $\ivr$ using \irref{Lpiff} and a corresponding generalization of axiom \irref{RealInd}.
The same appendix also proves the following generalization of~\rref{thm:algcomplete} for semianalytic evolution domains:
\begin{theorem}[Analytic completeness with semianalytic domains]
\label{thm:algcompletedom}
The differential radical invariant axiom~\irref{DRIQ} derives from \irref{Cont+Uniq} for semianalytic formula $\ivr$.
\[\dinferenceRule[DRIQ|DRI{$\&$}]{differential radical invariant axiom with domain}
{\linferenceRule[equiv]
  {\big(\ivr \limply \etermA = 0 \land (\sigliedsai{\genDE{x}}{\ivr} \limply \sigliedzero{\genDE{x}}{\etermA} ) \big)}
  {\axkey{\dbox{\pevolvein{\D{x}=\genDE{x}}{\ivr}}{\etermA=0}}}
}{}\]
\end{theorem}

Thus, \dL is complete for proving invariance of \emph{all} (semi)analytic $\rfvar$ of differential equations because it reduces all such questions equivalently to first-order formulas, e.g., on the RHS of derived axiom \irref{SAI}.
In addition, \dL decides invariance properties for all first-order real arithmetic formulas $\rfvar$, because quantifier elimination~\cite{Bochnak1998} can equivalently rewrite $\rfvar$ to (semialgebraic) normal form \rref{eq:normalform} first.
Unlike for \rref{thm:algcomplete} and its generalization \rref{thm:algcompletedom}, which equivalently reduce the future truth of analytic postconditions directly, \rref{thm:semialgcompleteness} and its generalized version in~\rref{app:completeness} are only equivalences for invariants $\rfvar$, the search of which is the only remaining challenge.

Of course, the complete proof rule \irref{sAI} can be used to prove all of the suggested invariants for the ODE $\alpha_e$ from~\rref{eq:example-ode}.
However,~\rref{ex:continuousproperties} gives a significantly simpler proof for the invariance of $1-u^2-v^2 \cmp 0$ with \irref{dbxineq}.
This has implications for implementations of \irref{sAI}: simpler proofs help minimize dependence on real arithmetic decision procedures.
For semianalytic formulas (that are not semialgebraic), proof rules resulting in simpler arithmetic premises might even be preferable because validity of these arithmetic premises is undecidable in general~\cite{DBLP:journals/jsyml/Richardson68}.
Logically, when $\rfvar$ is formed from only strict (resp.\ non-strict) inequalities then the left (resp.\ right) premise of \irref{sAI} closes trivially.
This logical fact corresponds to the topological fact that the set $\rfvar$ characterizes is topologically open (resp.\ closed) so only one of the two exit trajectories in \rref{subsec:realind} can occur.

\section{Noetherian Functions}
\label{sec:noetherianfunctions}

This section studies the class of Noetherian functions which meets all of the extended term conditions required in~\rref{subsec:background-compatibility} and therefore inherits all soundness and completeness results of the preceding sections, including~\rref{thm:algcomplete} and~\rref{thm:semialgcompleteness}.

\subsection{Mathematical Preliminaries}
\label{subsec:mathematicalprelim}

The following definition of Noetherian functions is standard, although the parameters that are used for studying the complexity of these functions~\cite{MR1732408,MR2083248,MR3925105} have been omitted.
The notation $\noef : \noefdom \subseteq \reals^k \to \reals$ is used for real-valued functions with domain $\noefdom$, i.e., an open, connected subset of $\reals^k$.
With a slight abuse of notation, polynomials $\ptermA \in \polynomials{\reals}{x}$ over indeterminates $x=(x_1,\dots,x_n)$ and their corresponding polynomial functions $\ptermA(x_1,\dots,x_n)$ in $\reals^n \to \reals$ are used interchangeably.

\begin{definition}[Noetherian chain and Noetherian function]
\label{def:noetherianchain}
A \emph{Noetherian chain} is a sequence of real analytic functions $\noef_1, \dots, \noef_r : \noefdom \subseteq \reals^k \to \reals$ such that all partial derivatives in $\noefdom$ for all $i = 1,\dots,k$ and $j = 1,\dots,r$ have the following form, where each $\ptermB_{ij} \in \reals[y,z]$ is a polynomial in $k+r$ indeterminates with $y=(y_1,\dots,y_k), z=(z_1,\dots,z_r)$:
\begin{equation}
\Dp[y_i]{\noef_j} (y) = \ptermB_{ij}(y,\noef_1(y),\dots,\noef_r(y))
\label{eq:noetherian-chain-def}
\end{equation}
The function $\noef : \noefdom \subseteq \reals^k \to \reals$ is \emph{Noetherian} iff it can be written as $\noef(y) = \ptermA(y,\noef_1(y),\dots,\noef_r(y))$, where $\ptermA \in \reals[y,z]$ is a polynomial in $k+r$ indeterminates and $\noef_1,\dots,\noef_r$ is a Noetherian chain.
In that case, $\noef$ is said to be \emph{generated} by this polynomial and Noetherian chain respectively but the choice of generating chain and polynomial for $\noef$ is not unique.
\end{definition}

For the term language extension~\rref{eq:extlang}, $\exp$ is a 1-element Noetherian chain because $\Dp[y]{\exp(y)} = \exp(y)$, while $\sin,\cos$ form a 2-element Noetherian chain.
All three functions together form a 3-element Noetherian chain.
More generally, the union of any (finite) number of Noetherian chains is a Noetherian chain.
By definition, any element of a Noetherian chain is itself a Noetherian function so $\exp,\sin,\cos$ are also Noetherian functions.
It is often useful to consider Noetherian functions over a larger domain than the generating chain, e.g., $\noef(x,y) \mnodefeq \exp(y) + \sin(x)$ with $\noef : \reals^2 \to \reals$.
In this case, the domain of definition of the generating chain is implicitly extended by treating them as functions over the dimensionally larger domain, e.g., with $\exp(x,y),\sin(x,y):\reals^2 \to \reals$ which ignore their first and second argument respectively.
This is compatible with~\rref{def:noetherianchain} because the partial derivatives with respect to the ignored arguments is trivially zero.

\rref{prop:noetherianchainnoether} gives important closure properties of the Noetherian functions generated by the same Noetherian chain, which are crucial later and explain the name \emph{Noetherian} function~\cite{MR1150568,MR2083248}.

\begin{proposition}[\cite{MR1732408,MR2083248,MR3925105}]
\label{prop:noetherianchainnoether}
The set $R$ of Noetherian functions generated by a given Noetherian chain \(\noef_1,\dots,\noef_r : \noefdom \subseteq \reals^k \to \reals\) is a Noetherian ring that is closed under partial derivatives.
\end{proposition}
\begin{proof}
Let $y = (y_1,\dots,y_k), z=(z_1,\dots,z_r)$ abbreviate indeterminates as in~\rref{def:noetherianchain}.
The set $R$ is a ring under the usual addition and multiplication of real-valued functions because the corresponding generating polynomials form a ring.
Now, $R$ is Noetherian because it is a finitely generated algebra~\cite[\S2.11, Corollary 3]{MR1727221} over the Noetherian polynomial ring $\reals[y]$.
The following constructive proof yields a computational method that is used later.

Consider an ascending chain of ideals $I_0 \subseteq I_1 \subseteq \cdots$ in $R$.
For each $I_i$, associate the set of generating polynomials $J_i \mdefeq \{\ptermA~|~\ptermA(y,\noef_1(y),\dots,\noef_r(y)) \in I_i\} \subseteq \reals[y,z]$ with respect to the generating Noetherian chain.
Each $J_i$ is an ideal in $\reals[y,z]$ because the corresponding $I_i$ are themselves ideals.
By construction, $J_i \subseteq J_{i+1}$ because $I_i \subseteq I_{i+1}$ for all $i$.
Since $J_0 \subseteq J_1 \subseteq \cdots$ is an ascending chain of ideals in $\reals[y,z]$, which is a Noetherian polynomial ring, it must stabilize at some $N$ with $J_{N} = J_{N+1} = \cdots$.
Correspondingly, the chain of ideals $I_i$ stabilizes (at the latest) at $N$ so $R$ is Noetherian.
The chain $I_0 \subseteq I_1 \subseteq \cdots$ may stabilize earlier than the corresponding $J_i$ chain but that is not important here.

To show that $R$ is closed under partial derivatives, let $\noef(y)=\ptermA(y,\noef_1(y),\dots,\noef_r(y)) \in R$ with $\ptermA \in \reals[y,z]$.
For the partial derivative of $\noef$ with respect to $y_i$, applying the chain rule yields:
\[
  \Dp[y_i]{\noef}(y) = \Dp[y_i]{\ptermA(y,\noef_1(y),\dots,\noef_r(y))}
  = \Dp[y_i]{\ptermA} (y,\noef_1(y),\dots,\noef_r(y)) + \sum_{j=1}^{r} \Dp[z_j]{\ptermA} (y,\noef_1(y),\dots,\noef_r(y)) \Dp[y_i]{\noef_j}(y)
\]

By definition, $\Dp[y_i]{\ptermA} (y,\noef_1(y),\dots,\noef_r(y)) \in R$ since $\Dp[y_i]{\ptermA}$ is a polynomial in $\reals[y,z]$.
Each summand $\Dp[z_j]{\ptermA} (y,\noef_1(y),\dots,\noef_r(y)) \in R$ since $\Dp[z_j]{\ptermA}$ is a polynomial in $\reals[y,z]$,
and $\Dp[y_i]{\noef_j}(y) \in R$ by definition because $\noef_1, \dots, \noef_r$ is a Noetherian chain.
Hence, all RHS sub-terms are in $R$, and so $\Dp[y_i]{\noef}(y)\in R$.
\end{proof}

\rref{prop:noetherianchainnoether} implies that adding Noetherian functions to their generating chains yields another Noetherian chain generating the same Noetherian ring $R$ of Noetherian functions because $R$ is closed under ring addition and multiplication.
Beyond closure properties for a single Noetherian chain, the class of all Noetherian functions is also closed under other mathematical operations, including function composition, multiplicative inverses, and function inverses (with appropriate assumptions)~\cite{MR3925105}.
Closure under function composition is proved constructively as it is used later.

\begin{proposition}[\cite{MR3925105}]
\label{prop:noetheriancompose}
If \(\noef : \noefdom \subseteq \reals^k \to \reals\) is Noetherian and \(\noeff : \noeffdom \subseteq \reals^l \to \reals^k\) has a compatible image \(\noeff(\noeffdom) \subseteq \noefdom\) where each component $\noeff_i : \noeffdom \subseteq\reals^l \to \reals$ for $i=1,\dots,k$ is Noetherian,
then the function composition $\noefff \mnodefeq \noef(\noeff_1,\dots,\noeff_k) : \noeffdom \subseteq \reals^l \to \reals$ is Noetherian.
\end{proposition}
\begin{proof}
Let $y = (y_1,\dots,y_k), z=(z_1,\dots,z_r), \gamma = (\gamma_1,\dots,\gamma_l)$ abbreviate indeterminates.
The composed function $\noefff$ is well-defined on $\noeffdom$ since $\noeff(\noeffdom) \subseteq \noefdom$.
By assumption, $\noef(y) = \ptermA(y,\noef_1(y),\dots,\noef_r(y))$ for some generating Noetherian chain $\noef_1,\dots,\noef_r : \noefdom \subseteq \reals^k \to \reals$ and polynomial $\ptermA \in [y,z]$.
Since the union of Noetherian chains is Noetherian and by~\rref{prop:noetherianchainnoether}, assume without loss of generality, that the Noetherian functions $\noeff_i$ for $i=1,\dots,k$ are members of the same generating Noetherian chain $\noeff_1, \dots, \noeff_s : \noeffdom \subseteq \reals^l \to \reals$ with $k \leq s$.
Putting these together, $\noefff$ can be written as:
\[ \noefff = \ptermA(\noeff, \noefff_1,\dots, \noefff_r ) \]
where $\noeff = (\noeff_1,\dots,\noeff_k)$ and the function compositions $\noefff_i \mdefeq \noef_i(\noeff_1,\dots,\noeff_k)$ for $i=1,\dots,r$.
From this representation, $\noefff$ is generated by polynomial $\ptermA$ over the sequence:
\begin{equation}
\noeff_1,\dots,\noeff_s, \noefff_1,\dots,\noefff_r
\label{eq:composednoether}
\end{equation}

In order to show that~\rref{eq:composednoether} is a Noetherian chain, it suffices to check that the $\noefff_1\dots,\noefff_r$ obey the condition on partial derivatives~\rref{eq:noetherian-chain-def} because $\noeff_1,\dots,\noeff_s$ is already a Noetherian chain.
For each $\noefff_i(\gamma) : \noeffdom \subseteq \reals^l \to \reals$, taking the partial derivative with respect to $\gamma_j$ and applying the chain rule:
\begin{align*}
\Dp[\gamma_j]{\noefff_i} (\gamma) &= \Dp[\gamma_j]{\noef_i(\noeff_1(\gamma),\dots,\noeff_k(\gamma))}
= \sum_{l=1}^k \Dp[y_l]{\noef_i}(\noeff_1(\gamma),\dots,\noeff_k(\gamma)) \Dp[\gamma_j]{\noeff_l}(\gamma)
\end{align*}

It suffices to check that each sub-term appearing on the RHS sum are generated as polynomials over the sequence~\rref{eq:composednoether}.
The case for each $\Dp[\gamma_j]{\noeff_l}(\gamma)$ follows immediately because $\noeff_1,\dots,\noeff_s$ is a Noetherian chain.
Since $\noef_1,\dots,\noef_r$ is a Noetherian chain, each $\Dp[y_l]{\noef_i}$ is a polynomial combination $\Dp[y_l]{\noef_i} = t_{il}(y,\noef_1,\dots,\noef_r)$ for some polynomial $t_{il} \in \reals[y,z]$ and, thus, $\Dp[y_l]{\noef_i}$ is generated by chain~\rref{eq:composednoether}:
\[ \Dp[y_l]{\noef_i}(\noeff_1,\dots,\noeff_k) = t_{il}(\noeff,\noef_1(\noeff),\dots,\noef_r(\noeff))  = t_{il}(\noeff,f_1,\dots,f_r)
\qedhere\]
\end{proof}

Before turning to the study of Noetherian functions in \dL, it is helpful to first understand how they help with its differential equations reasoning.
Polynomial ODEs are very expressive and earlier results~\cite{DBLP:conf/fm/0009ZZZ15,GRACA2008330,DBLP:journals/logcom/Platzer10} make use of polynomial ODEs to implicitly characterize (and thus, eliminate) real analytic functions appearing in initial value problems (IVPs).
IVPs are specified by a system of ODEs, $\D{x}=\genDE{x}$, defined over domain $D$ with RHS $\genDE{x} : D \subseteq \reals^n \to \reals^n$ and real initial value $X_0 \in D \subseteq \reals^n$.
The IVP is called Noetherian (resp. polynomial) when all components of the RHS $\genDE{x}$ are Noetherian functions (resp. polynomials).
Both Noetherian and polynomial functions are analytic and therefore continuously differentiable.
Under the assumption of continuously differentiable RHS, the Picard-Lindel\"{o}f theorem~\cite[\S10.VI]{Walter1998} guarantees that the IVP has a unique maximal solution $\solvar(t) : (\alpha,\beta) \to \reals^n$ with $-\infty \leq \alpha < 0 < \beta \leq \infty$ such that $\solvar(0) = X_0$ and $\D[t]{\solvar(t)} = \genDE{\solvar(t)}$.
Uniqueness and maximality here means that every solution of the IVP is a truncation of $\solvar$ to a smaller existence interval.
The following generalizes aforementioned results~\cite{DBLP:conf/fm/0009ZZZ15,GRACA2008330,DBLP:journals/logcom/Platzer10} to the Noetherian setting:

\begin{proposition}
\label{prop:diffaxiomatization}
Function \(\solvar : (\alpha,\beta) \to \reals^n\) with $-\infty {\leq} \alpha {<} 0 {<} \beta {\leq} \infty$ is the (coordinate-projected) solution of a Noetherian IVP iff it is the (coordinate-projected) solution of a polynomial IVP.
\end{proposition}
\begin{proof}
In the (trivial) converse ``$\mylpmi$'' direction, suppose function $\solvar$ solves the polynomial IVP $\D{x}=\ptermA(x,y),\D{y}=\ptermB(x,y)$ with initial values $X_0 \in \reals^n, Y_0 \in \reals^r$.
Let $\solvar_x,\solvar_y$ denote the projection onto the $x$ and $y$ coordinates of $\solvar$ respectively.
Every solution of polynomial ODEs is a univariate Noetherian function~\cite{MR2083248}.
Therefore, the Noetherian IVP given by $\D{x}=\ptermA(\solvar_x(\tau),\solvar_y(\tau)), \D{\tau}=1$ with the same initial value for $x$ and $0$ for $\tau$ trivially has the (unique) solution $(\solvar_x(t),t) : (\alpha,\beta) \to \reals^n \times \reals$.

In the (nontrivial) ``$\mimply$'' direction, suppose that $\solvar$ is the solution to the Noetherian IVP $\D{x}=\genDE{x}$ where each $\odeterm_i(x) : D \subseteq \reals^n \to \reals$ is Noetherian, and with initial value $X_0 \in D$.
By uniqueness of solutions, it suffices to construct a polynomial IVP so that $\solvar(t)$ solves it in the $x$ coordinates.

Since the union of Noetherian chains is itself a Noetherian chain, assume without loss of generality that the functions $f_1, \dots, f_n$ are generated by the same Noetherian chain $\noef_1, \dots, \noef_r$ and that $f_i = p_i(x,\noef_1(x),\dots,\noef_r(x))$ for some polynomials $p_i \in \reals[x,y]$ in $n+r$ indeterminates for $i=1,\dots,n$.
Introduce new variables $y_j$ for $j=1,\dots,r$ which are meant to take on the respective value of $\noef_j$ along solutions to the ODE.
Accordingly, the RHS of the Noetherian ODE is rewritten by replacing each $f_i$ with $p_i(x,y)$, i.e., the desired polynomial ODEs for $x$ is $\D{x} \mnodefeq p(x,y)$.

It remains to ensure that each of these newly introduced variables $y_j$ take on their intended values $\noef_j(\solvar(t))$ along the solution $\solvar$.
By~\rref{eq:noetherian-chain-def}, the partial derivatives for each $\noef_j$ can be written as polynomials $q_{ij} \in \reals[x,y]$ over the generating Noetherian chain. By the chain rule:
\[\D[t]{\noef_j(\solvar(t))}
= \sum_{i=1}^{n} \Dp[x_i]{\noef_j(x)} (\solvar(t)) \D[t]{\solvar_i(t)}
= \sum_{i=1}^{n} q_{ij} \big(\solvar(t),\noef_1(\solvar(t)),\dots,\noef_r(\solvar(t))\big) \D[t]{\solvar_i(t)} \]

Back-substituting into the RHS of this equation using the intended values for $y_j$ and the new ODEs for $x$, yields the following additional ODEs for $y$:
\[ \D{y_j} \mnodefeq \sum_{i=1}^{n} q_{ij} (x,y) p_i(x,y) \]

The RHS of these additional ODEs are polynomials in $\reals[x,y]$, which completes the desired polynomial IVP with the initial values \(Y_0 \mdefeq (\noef_1(x_0), \dots, \noef_r(x_0)) \in \reals^r\) for $y$.
The construction of this polynomial IVP is correct-by-construction because of the mechanical chain rule computation.
A solution to this IVP is given by the pair $(\solvar(t),y(t)) : (\alpha,\beta) \to \reals^n \times \reals^r$ where $y_j(t) \mdefeq \noef_j(\solvar(t)) : (\alpha,\beta) \to \reals$.
By uniqueness of solutions, this completes the proof of the ``$\mimply$'' direction.\qedhere
\end{proof}

In the ``$\mimply$'' direction of~\rref{prop:diffaxiomatization}, the constructed polynomial IVP may involve additional ODEs over the variables $y$ (with their respective initial values).
The number of additional equations required in this construction is the length of the shortest Noetherian chain that generates the RHS of the input Noetherian ODE.
The polynomial IVP may have a larger maximal interval of existence than the input Noetherian IVP if it leaves the domain $D$ of the input RHS.
In the ``$\mylpmi$'' direction, only one additional time variable $\tau$ is required.
Consequently, the solution of \emph{any} $n$-dimensional IVP that is the coordinate projection of the solution of a polynomial IVP (of potentially much larger dimension) is the coordinate projection of the solution of an $(n+1)$-dimensional Noetherian IVP.

The constructive proof of the ``$\mimply$'' direction in~\rref{prop:diffaxiomatization} yields an approach for transforming input Noetherian IVPs to polynomial IVPs assuming that the Noetherian functions can be effectively associated with generating Noetherian chains and polynomials.
\begin{example}[Flight dynamics~{\cite[Equation 1]{DBLP:journals/logcom/Platzer10}}]
\label{ex:diffaxiomatization}
A simple planar model of curved aircraft motion is given by the following ODE system, where $(x,y)$ are the aircraft's planar coordinates, $\theta$ its angular orientation, and $\nu,\omega$ its linear and angular velocity respectively~\cite[Equation 1]{DBLP:journals/logcom/Platzer10}:
\[ \D{x} = \nu \cos{(\theta)}, \quad \D{y} = \nu \sin{(\theta)}, \quad \D{\theta} = \omega \]

Consider an IVP for this ODE with initial values $x=X_0,y=Y_0, \theta=\Theta_0 \in \reals$.
The linear and angular velocities $\nu,\omega$ are left as symbolic constants in this model.

The RHS of the ODE is generated by the Noetherian chain: $\sin{(\theta)}, \cos{(\theta)}$.
Introducing additional variables $z_1,z_2$ for the elements of this chain, and replacing the RHS for $\D{x},\D{y}$ according to their generating polynomials with respect to the chain gives:
\[ \D{x} = \nu z_1, \quad \D{y} = \nu z_2, \quad \D{\theta} = \omega \]

A symbolic calculation (see~\rref{prop:diffaxiomatization}) yields the following ODEs that $z_1,z_2$ must obey:
\[ \D{z_1} = \omega z_2, \quad \D{z_2} = -\omega z_1 \]

To finish constructing the polynomial IVP, set the initial values $z_1=\sin{(\Theta_0)},z_2=\cos{(\Theta_0)}$.
The resulting ODE has higher dimension but a polynomial RHS.
\end{example}

The utility of adding Noetherian functions to \dL is \emph{not} an increase in expressiveness of the differential equations.
Rather, it allows Noetherian ODEs to be written down naturally instead of relying on implicit polynomial characterization such as in~\rref{ex:diffaxiomatization}.
More importantly, they make it possible to use formulas as ODE invariants that are \emph{not} semialgebraic.
By \rref{thm:algcomplete} and \rref{thm:semialgcompleteness}, the \dL ODE axiomatization provides an effective and complete calculus for (dis)proving the resulting semianalytic ODE invariants involving Noetherian functions.
This requires Noetherian functions to meet the extended term conditions from \rref{subsec:background-compatibility}, which is shown next.

\subsection{Extended Term Conditions for Noetherian Functions}
\label{subsec:noetheriancompat}

Assume from now on that the fixed $k$-ary function symbols $\noef_1,\dots,\noef_r$ are interpreted semantically as members of a Noetherian chain $\noef_1,\dots,\noef_r : \reals^k \to \reals$ respectively.
Recall that extended \dL terms are formed syntactically from these function symbols according to the grammar (\rref{subsec:background-syntax}).
As discussed in~\rref{subsec:background-semantics}, the choice of domain $\reals^k$ for these functions ensures that the term semantics is well-defined in all states.
The first two extended term conditions are straightforward to check:
\begin{enumerate}
\item[\rref{itm:reqsmooth}] All Noetherian functions are, by definition, $C^\infty$ smooth (even real analytic) so the semantics of differentials are well-defined.
\item[\rref{itm:reqpartial}] The partial derivative of each $\noef_j(y_1,\dots,y_k): \reals^k \to \reals$ with respect to $y_i$ satisfies~\rref{eq:noetherian-chain-def} for some polynomial $q_{ij} \in \reals[y,z]$.
Since polynomials are generated by addition and multiplication, these partial derivatives $\Dp[y_i]{\noef}(y_1,\dots,y_k)$ are syntactically represented by the extended term:
\[q_{ij}\big(y_1,\dots,y_k,\noef_1(y_1,\dots,y_k),\dots,\noef_r(y_1,\dots,y_k)\big)\]
Thus,~\rref{lem:noefdifferentials} adds the (sound) differential axioms for each fixed function symbol $\noef_j$ and therefore, all Lie derivatives are representable in the extended term language.
\end{enumerate}

The final condition~\rref{itm:reqdiffradical} is more involved and relies crucially on closure properties of Noetherian functions.
A syntactic subtlety arises for extended terms with nested function applications such as $\exp(\exp(x))$.
Its semantics is the iterated real exponential function generated by the 2-element Noetherian chain $\exp(x), \exp(\exp(x))$.
Thus, even though the fixed function symbols $\noef_1,\dots,\noef_r$ form a Noetherian chain, the extended term grammar could produce extended terms that \emph{do not} correspond to Noetherian functions generated by that chain.
The following lemma resolves this issue by computing another (syntactic) Noetherian chain that generates it instead:

\begin{lemma}
\label{lem:noetheriansyntax}
The semantics of every extended term $\etermA$ over Noetherian functions is a Noetherian function and $\etermA$ can be effectively associated with a (syntactic) Noetherian chain that generates it.
\end{lemma}
\begin{proof}
By structural induction on extended \dL term $\etermA$.
The cases for variables and constants are obvious, while the cases for addition and multiplication follow inductively from closure under ring operations (\rref{prop:noetherianchainnoether}) and the fact that finite unions of Noetherian chains are Noetherian chains.
The only difficult case is when $\etermA$ is a function composition \(\noef(\etermA_1,\dots,\etermA_k)\), where $\etermA_1,\dots,\etermA_k$ are extended terms.
Inductively, each $\etermA_1, \dots,\etermA_k$ semantically is a Noetherian function.
Moreover, $\noef$ is (semantically) a Noetherian function by the assumption that the interpretation of all fixed function symbols is an element of some Noetherian chain.
Thus, \rref{prop:noetheriancompose} implies that the semantics of their composition is also a Noetherian function.
Let $\noeff_1,\dots,\noeff_s$ be the union of Noetherian chains obtained inductively for $\etermA_1,\dots,\etermA_k$.
The (constructive) proof of \rref{prop:noetheriancompose} shows that the Noetherian chain~\rref{eq:composednoether} given by $\noeff_1,\dots,\noeff_s, f_1,\dots,f_r$ generates $\noef(\etermA_1,\dots,\etermA_k)$, where $\noeff_1,\dots,\noeff_s$ are syntactically represented by extended terms using the induction hypothesis on $\etermA_1,\dots,\etermA_k$ and each $f_i \mdefeq \noef_i(\etermA_1,\dots,\etermA_k)$ is an extended term by construction.
\end{proof}

\rref{lem:noetheriansyntax} makes it possible to unambiguously refer to ``the'' generating Noetherian chain and polynomial for any extended term $\etermA$ by giving an effective procedure for finding their syntactic representations in the extended term language.
Together with~\rref{prop:noetherianchainnoether}, this suffices to prove that the extended term language has the computable differential radicals condition~\rref{itm:reqdiffradical}.

\begin{theorem}
\label{thm:noetheriancompat}
Term languages with Noetherian functions satisfy the extended term conditions.
\end{theorem}
\begin{proof}
Conditions~\rref{itm:reqsmooth} and~\rref{itm:reqpartial} have already been shown above.
It remains to show condition~\rref{itm:reqdiffradical}, i.e.,
any ODE $\D{x}=\genDE{x}$ and extended term $\etermA$ has a computable (and provable) differential radical identity~\rref{eq:differential-rank}.
By~\rref{lem:noetheriansyntax}, the terms $\genDE{x},\etermA$ are (semantically) Noetherian and so, by taking the union of Noetherian chains, are defined by the same generating Noetherian chain $\noeff_1,\dots,\noeff_s$.
The ring $R$ generated by this chain is Noetherian by~\rref{prop:noetherianchainnoether} and is closed under partial derivatives.
Recall that the Lie derivative of $\etermA$ along $\D{x}=\genDE{x}$ is given by:
\[ \lie[]{\genDE{x}}{\etermA} \mdefeq \sum_{i=1}^n \Dp[x_i]{\etermA} \cdot \odeterm_i(x)\]

Every sub-term on the RHS of its Lie derivative is contained in the ring $R$ because it already contains the RHS of the ODEs, $\genDE{x}$, and is closed under the partial derivatives of $\etermA$.
Inductively, all higher Lie derivatives $\lied[i]{\genDE{x}}{\etermA}$ for $i=0,1,\dots$ are contained in $R$ and are therefore generated by the chain $\noeff_1,\dots,\noeff_s$ with $\lied[i]{\genDE{x}}{\etermA} = \ptermA_i(x,\noeff_1,\dots,\noeff_s)$ for polynomials $\ptermA_i \in \reals[x,y]$ and $y = (y_1,\dots,y_s)$.

Following the proof of~\rref{prop:noetherianchainnoether}, consider this ascending chain of polynomial ideals:
\[ \ideal{\ptermA_0} \subseteq \ideal{\ptermA_0,\ptermA_1} \subseteq \ideal{\ptermA_0,\ptermA_1,\ptermA_2} \subseteq \cdots  \]

This chain stabilizes with the provable polynomial identity $\ptermA_N = \sum_{i=0}^{N-1} \ptermB_i \ptermA_i$ for some polynomial cofactors $\ptermB_i \in \reals[x,y]$ and $N \geq 1$.
The rank $N$ and the polynomial cofactors $\ptermB_i$ are computable by successive ideal membership checks~\cite{DBLP:conf/tacas/GhorbalP14,DBLP:conf/emsoft/LiuZZ11,DBLP:journals/cl/GhorbalSP17}.
Mapping this back into elements of $R$ gives the required provable identity for the Lie derivatives of $\etermA$ by choosing $\cofterm_i \mdefeq \ptermB_i(x,\noeff_1,\dots,\noeff_s)$:
\[ \lied[N]{\genDE{x}}{\etermA} = \sum_{i=0}^{N-1} \cofterm_i \lied[i]{\genDE{x}}{\etermA} \qedhere \]
\end{proof}

An immediate corollary is that term language extensions with Noetherian functions inherit all earlier soundness and completeness results, e.g., from~\rref{subsec:completenessanalytic} and~\rref{subsec:completenesssemianalytic}.
\begin{corollary}[Noetherian invariant completeness]
The \dL proof calculus is complete for (semi)analytic invariants of ODEs for term languages extended with Noetherian functions.
\end{corollary}
\begin{proof}
This follows from the extended term conditions for Noetherian functions (\rref{thm:noetheriancompat}) and the completeness theorems~\rref{thm:algcomplete} and~\rref{thm:semialgcompleteness}.
\end{proof}

\subsection{Extended Term Language Example}

This section illustrates constructions from~\rref{subsec:mathematicalprelim} and~\rref{subsec:noetheriancompat} using the extended term language~\rref{eq:extlang}.
The first example shows the computations from~\rref{lem:noetheriansyntax} and~\rref{thm:noetheriancompat}:

\begin{example}[Syntactic manipulation of Noetherian functions]
\label{ex:synmanip}
Consider the ODE $\D{x}=\exp(\sin(x))$ and the term $\etermA \mnodefeq x+x^2$.
The Noetherian chain for $\etermA$ is empty because it is already a polynomial while the Noetherian chain associated with $\exp(\sin(x))$ is $\noeff_1 \mnodefeq \sin(x),\noeff_2 \mnodefeq \cos(x),\noeff_3 \mnodefeq \exp(\sin(x))$.
The higher Lie derivatives of $\etermA$ are all extended terms generated by the chain $\noeff_1,\noeff_2,\noeff_3$:
\begin{align*}
\lied[1]{\genDE{x}}{\etermA} = \noeff_3 + 2 \noeff_3 x \quad \lied[2]{\genDE{x}}{\etermA} = \noeff_3^2\noeff_2 + 2(\noeff_3^2\noeff_2x + \noeff_3^2) = (2\noeff_3+\noeff_2\noeff_3+2\noeff_2\noeff_3x)( (1+2x)\lied[1]{\genDE{x}}{\etermA} - 4\noeff_3\etermA )
\end{align*}
The (polynomial) identity for $\lied[2]{\genDE{x}}{\etermA}$ in terms of $\lied[1]{\genDE{x}}{\etermA},\etermA$ and their cofactors is obtained computationally by ideal membership checks for the polynomial ring $\reals[x,y_1,y_2,y_3]$ (the indeterminate $y_i$ corresponds to $\noeff_i$ for $i=1,2,3$), following~\rref{prop:noetherianchainnoether}.
\end{example}

The next example illustrates how the extended term language allows effective proofs of more invariants than possible with polynomial term languages.

\begin{example}[Expressivity of Noetherian invariants]
\label{ex:expressivity}
The polynomial invariant $1-u^2-v^2=0$ was proved for the ODE $\alpha_e$ from~\rref{eq:example-ode} in~\rref{ex:continuousproperties}.
With respect to~\rref{fig:exampleODE}, this means that a trajectory starting at the point $(1,0)$ stays on the circle.
However, this invariant yields no information about how fast the trajectory loops around the circle or whether it revolves clockwise or anti-clockwise.
In the extended term language, the most precise invariant can be proved, namely the solution to the ODEs from this initial point.
The solution is a trigonometric function of time (given below), and so cannot be expressed as a polynomial (or semialgebraic) invariant~\cite{Bochnak1998}.
The precise solution also shows that the motion is anti-clockwise, as suggested by~\rref{fig:exampleODE}.

The following derivation uses a \irref{DC} to add the known polynomial invariant $1-u^2-v^2=0$ which proves by \irref{dbx} as in~\rref{ex:continuousproperties}.
The right premise after the differential cut assumes $1-u^2-v^2=0$ in the ODE's evolution domain constraint.
It is abbreviated \textcircled{1} and continued below.
{\footnotesizeoff%
\begin{sequentdeduction}[array]
\linfer[DC]{
  \linfer[dbx+qear]{
    \lclose
  }
  {\lsequent{u=1,v=0,t=0} {\dbox{\alpha_e,\D{t}=1}{1-u^2-v^2 = 0}}} !
  \textcircled{1}
}
  {\lsequent{u=1,v=0,t=0} {\dbox{\alpha_e,\D{t}=1}{ ( u {-}{\cos(t)} = 0 \land v {-}{\sin(t)} = 0 ) }}}
\end{sequentdeduction}%
}%

From \textcircled{1}, first calculate the Lie derivatives, abbreviating $c \mnodefeq u-\cos(t), s \mnodefeq v-\sin(t)$:
\begin{align*}
\lie[]{\alpha_e,\D{t}=1}{c} &= \lie[]{\alpha_e,\D{t}=1}{u-\cos(t)} = -v + \frac{u}{4} (1 - u^2 - v^2) + \sin(t) = -s + \frac{u}{4} (1 - u^2 - v^2) \\
\lie[]{\alpha_e,\D{t}=1}{s} &= \lie[]{\alpha_e,\D{t}=1}{v-\sin(t)} = u + \frac{v}{4} (1 - u^2 - v^2) - \cos(t) = c + \frac{u}{4} (1 - u^2 - v^2)
\end{align*}

Under the domain constraint assumption $1-u^2-v^2=0$, the additional $ \frac{u}{4} (1 - u^2 - v^2)$ term in both Lie derivatives simplifies to $0$.
The derivation starts with a \irref{cut} of the postcondition $c=0 \land s=0$.
This arithmetic premise, abbreviated \textcircled{2}, is discussed afterwards.
Continuing on the right premise, the \irref{vdbx} step closes successfully using real arithmetic manipulations only:
{\footnotesizeoff%
\begin{sequentdeduction}[array]
  \linfer[cut]{
    \textcircled{2} !
    \linfer[vdbx]{
    \linfer[qear]{
      \lclose
    }
      {\lsequent{1-u^2-v^2=0} {
        \left(\begin{array}{l}
          \lied[]{}{c} \\
          \lied[]{}{s}
        \end{array}\right) =
        \left(\begin{array}{cc}
          0  & -1 \\
          1  & 0
        \end{array}\right)
        \left(\begin{array}{l}
          c \\
          s
        \end{array}\right)
      }
    }
    }
    {\lsequent{ c = 0 \land s = 0 } {\dbox{\pevolvein{\alpha_e,\D{t}=1}{1-u^2-v^2=0}}{ (c=0 \land s=0) }}}
  }
  {\lsequent{u=1,v=0,t=0} {\dbox{\pevolvein{\alpha_e,\D{t}=1}{1-u^2-v^2=0}}{ (c=0 \land s=0) }}}
\end{sequentdeduction}
}%

The premise \textcircled{2} is valid, but involves properties of trigonometric functions ($\cos(0)=1,\sin(0)=0$) so it cannot be proved using \irref{qear}.
Instead, extended arithmetic \irref{qearext} is needed:
{\footnotesizeoff%
\begin{sequentdeduction}[array]
  \linfer[qearext]{
    \lclose
  }
  {\lsequent{u=1,v=0,t=0} {u-\cos(t) = 0 \land v-\sin(t) = 0}}
\end{sequentdeduction}
}%

The extended arithmetic theory is undecidable in general~\cite{DBLP:journals/jsyml/Richardson68} and so, unlike~\irref{qear}, rule \irref{qearext} cannot be implemented via an underlying decision procedure.
Yet, simple arithmetic questions such as~\textcircled{2} which just involve the evaluation of trigonometric functions can be easily checked.

The above derivation takes advantage of a known Darboux equality for $1-u^2-v^2$ to simplify the proof using \irref{vdbx}.
The proof could have instead directly made use of~\rref{thm:algcomplete} by encoding $c=0 \land s=0$ as $c^2+s^2=0$ and then calculating the rank of $c^2+s^2$ (which involve trigonometric functions) according to~\rref{thm:noetheriancompat}.
This also works, but $c^2+s^2$ has rank 3, and the resulting cofactors are too large to even fit on this page.
\end{example}

The final example highlights an important insight from~\rref{prop:diffaxiomatization}.
Even though this article only considers extended term languages with terms that are defined everywhere, it is possible to use logical formulas to implicitly characterize more terms, making use of closure properties of the Noetherian functions~\cite{MR3925105}.
The following example illustrates implicit characterization of quotients which are defined everywhere in the domain of interest:

\begin{example}[Implicit characterization of quotients]
\label{ex:implicit}
The trigonometric tangent function $\tan(x)$ is Noetherian and defined on the interval $(-\frac{\pi}{2},\frac{\pi}{2})$.
Consider the following ``formula'' where $x$ is restricted in the domain constraint to never reach a point where the RHS $\tan(x)$ is undefined:
\[ x = \frac{1}{2} \limply \dbox{\pevolvein{\D{x}=\tan(x)}{{-}1 \leq x \leq 1}}{x \geq \frac{1}{2}} \]

This ``formula'' is not formally in the syntax of \dL formulas because $\tan$ is not defined everywhere.
However,~\rref{prop:diffaxiomatization} can be used to ask an equivalent question in \dL.
Recall from calculus:
\[\tan(x) = \frac{\sin(x)}{\cos(x)} \qquad \Dp[x]{\frac{1}{\cos(x)}} = \frac{\sin(x)}{(\cos(x))^2} \]

Thus, $\sin(x),\cos(x),\frac{1}{\cos{x}}$ forms a 3-element Noetherian chain that generates $\tan(x)$.
For brevity, by partially following the IVP construction of~\rref{prop:diffaxiomatization}, the ``formula'' is rephrased as an actual \dL formula with $y$ representing $\frac{1}{\cos{x}}$ along the ODE.
After replacing $\D{x} = \sin(x)y$, the required differential equation for $y$ is calculated symbolically with $\D{y} = \sin(x)y^2 (\sin(x)y) = \sin^2(x)y^3$.
\[ x = \frac{1}{2}  \land \cos(x) y - 1 = 0 \limply \dbox{\pevolvein{\D{x}=\sin(x) y, \D{y}=\sin^2(x)y^3}{{-}1 \leq x \leq 1}}{x \geq \frac{1}{2}} \]

For non-zero denominator, the initial value $\frac{1}{\cos(x)}$ of $y$ is logically characterized by the formula $\cos(x)y-1 = 0$.
The following Lie derivative calculation shows that $\cos(x) y - 1$ satisfies a Darboux equality and so $\cos(x)y-1=0$ can be proven invariant along the ODE (abbreviated as $\alpha$) by \irref{dbx}.
\[ \lie[]{\alpha}{\cos(x) y - 1}  = -\sin(x)(\sin(x) y) y + \cos(x)(\sin^2(x)y^3) = \sin^2(x)y^2 (\cos(x) y - 1)\]

The rephrased formula proves after a \irref{DC} with this Darboux invariant for $y$ using the ODE invariant $x \geq \frac{1}{2}$ and rule~\irref{sAI} (or its generalization with domain constraints from~\rref{app:completeness}).
Briefly, the invariance of $x \geq \frac{1}{2}$ provably reduces to the following arithmetic premise which is valid and falls within a \emph{decidable} fragment of arithmetic with trigonometric functions~\cite{DBLP:journals/jsc/McCallumW12}:
\[ \lsequent{-1\leq x \leq 1, \cos(x)y-1=0, x = \frac{1}{2}}{\sin(x) y > 0} \]
\end{example}

\section{Related Work}
\label{sec:relatedwork}

This related work discussion focuses on deductive verification of hybrid systems.
An overview of approaches to hybrid systems verification is available elsewhere~\cite{DBLP:reference/mc/2018}.
Readers interested in ODEs~\cite{Walter1998}, real analysis~\cite{MR1916029,Walter1998}, algebra~\cite{MR1727221}, and real algebraic geometry~\cite{Bochnak1998} are referred to the cited textbooks.
The orthogonal task of efficiently generating invariants is investigated elsewhere~\cite{DBLP:conf/tacas/GhorbalP14,DBLP:conf/emsoft/LiuZZ11,DBLP:journals/fmsd/SankaranarayananSM08}.

\paragraph{Proof Rules for ODE Invariants}
Numerous useful but incomplete proof rules for ODE invariants~\cite{DBLP:conf/hybrid/PrajnaJ04,DBLP:journals/fmsd/SankaranarayananSM08,DBLP:conf/fsttcs/TalyT09,DBLP:journals/lmcs/Platzer12} are surveyed elsewhere~\cite{DBLP:journals/cl/GhorbalSP17}.
The soundness and completeness theorems for~\irref{dRI} and~\irref{sAI} were previously proved semantically~\cite{DBLP:conf/tacas/GhorbalP14,DBLP:conf/emsoft/LiuZZ11}.
These earlier results are limited to (semi)algebraic invariants as they depend on specific semantic properties limited to polynomials.
The extended term conditions (\rref{subsec:background-compatibility}) and Noetherian functions (\rref{sec:noetherianfunctions}) generalize these results, showing that all (semi)analytic invariance questions reduce completely to arithmetic.

In their original presentation~\cite{DBLP:conf/tacas/GhorbalP14,DBLP:conf/emsoft/LiuZZ11}, \irref{dRI} and \irref{sAI} are \emph{algorithmic procedures} for checking invariance of semialgebraic sets, requiring \eg, checking ideal membership for all polynomials in the semialgebraic decomposition.
This makes them difficult to implement soundly within a small, trusted axiomatic core~\cite{DBLP:conf/cade/FultonMQVP15}.
This article shows that, by relying on the logic \dL, these rules can be \emph{derived} from a small set of axiomatic principles.
Although these derivations also leverage ideal computations, they are only used in \emph{derived rules}.
With the aid of a theorem prover, derived rules can be implemented as tactics that crucially remain \emph{outside} its soundness-critical axiomatic core.

\paragraph{Deductive Power and Proof Theory}
The derivations shown in this article are fully general, which is necessary for completeness of the resulting derived rules.
The number of conjuncts in the progress and differential radical formula for an extended term $\etermA$ is equal to the rank of $\etermA$.
Known upper bounds for the rank, even in the case of polynomials in $n$ variables, are doubly exponential in $n^2 \ln{n}$~\cite{MR1697373}.
Many simpler classes of invariants can be proved using simpler derivations, as exemplified by Examples~\ref{ex:continuousproperties} and~\ref{ex:expressivity}.
This is where a study of the deductive power of sound, but incomplete, proof rules~\cite{DBLP:journals/cl/GhorbalSP17} is essential.
For ODE invariants of a simpler class, it suffices to use a proof rule that is complete for just that class.
This intuition is echoed in an earlier study~\cite{DBLP:journals/lmcs/Platzer12} of the relative deductive power of differential invariants (\irref{DI}), differential cuts (\irref{DC}), and differential ghosts (\irref{DG}).
The first completeness result (\rref{thm:algcomplete}) shows that \dL with \irref{DG} is complete for algebraic and analytic invariants.
Other proof-theoretical studies of \dL~\cite{DBLP:conf/lics/Platzer12b} reveal surprising correspondences between its hybrid, continuous, and discrete aspects in the sense that each aspect can be axiomatized completely and effectively relative to any other aspect.
\rref{cor:testfree} constructively exploits their combination.

\paragraph{Noetherian Functions}
This article only touched on basic properties of Noetherian functions.
The model-theoretic study of Noetherian functions and the related Pfaffian functions is fascinating in its own right~\cite{MR1732408,MR2083248,MR3925105}.
Pfaffian functions are generated by chains satisfying~\rref{eq:noetherian-chain-def} except with triangular dependencies in their partial derivatives~\cite{MR2083248}, and notably, the expansion of the real field with Pfaffian functions is o-minimal~\cite{MR1676876,MR1740677,MR1633348}.
Such o-minimal expansions have been studied in reachability analysis for \emph{o-minimal hybrid systems}~\cite{DBLP:journals/mcss/LafferrierePS00,DBLP:conf/csl/KorovinaV04} because they admit the construction of finite bisimulations for reachability analysis algorithms.
In contrast, expansions with (more general) Noetherian functions, e.g.,~(unrestricted) trigonometric sine and cosine, are not o-minimal because they can be used to characterize the natural numbers.
This is a barrier to the construction of finite bisimulations~\cite{DBLP:journals/mcss/LafferrierePS00} but not for deductive approaches, as long as the relevant arithmetic is provable.

Undecidability of arithmetic is a delicate issue~\cite{DBLP:journals/jsyml/Richardson68}, but this article's completeness results show that ODE invariance verification \emph{completely} reduces to arithmetic!
Many (necessarily incomplete) approaches and tools for handling special functions are available, e.g., resolution with upper and lower bounds as implemented in MetiTarski~\cite{DBLP:journals/jar/AkbarpourP10}, $\delta$-decidability as implemented in dReal~\cite{DBLP:conf/lics/GaoAC12,DBLP:conf/cade/GaoKC13}, or heuristic inference-based approaches as implemented in Polya~\cite{DBLP:journals/jar/AvigadLR16}.
Specialized decision procedures may also be applicable for restricted fragments of arithmetic~\cite{DBLP:journals/jsc/McCallumW12}, as in Examples~\ref{ex:synmanip} and~\ref{ex:implicit}.
Even in settings where all of these automated tools fail to verify an arithmetic question, the system designer can provide further mathematical intuition with an interactive proof in \KeYmaeraX~\cite{DBLP:conf/cade/FultonMQVP15}.

The findings of this article identify Noetherian functions as a more general unifying theme behind earlier results in continuous/hybrid systems verification.
Besides the completeness results for invariants,~\rref{prop:diffaxiomatization} also generalizes earlier results~\cite{DBLP:conf/fm/0009ZZZ15,GRACA2008330,DBLP:journals/logcom/Platzer10} to the Noetherian setting.
This idea is called \emph{differential axiomatization}~\cite{DBLP:journals/logcom/Platzer10} because it axiomatizes ODEs involving special functions that have undecidable arithmetic using polynomial ODEs.
Similarly,~\cite[Proposition 1]{DBLP:conf/fm/0009ZZZ15} gives an algorithm for replacing a fixed set of functions appearing in IVPs with polynomials ones.
The result from~\cite[Theorem 4]{GRACA2008330} only applies in the case of univariate Noetherian functions.

\section{Conclusion}
\label{sec:conclusion}

This article demonstrates the impressive deductive power of differential ghosts: they prove \emph{all} Darboux invariants and, as a consequence, \emph{all} analytic invariants for extended term languages with the extended term conditions.
Even \emph{scalar} differential ghosts suffice for this result, but the question of whether their deductive power extends to even larger classes of invariants is left open.

The article then introduces extensions to the \dL axiomatization and shows how they can be used to extend completeness to semianalytic invariance.
The case of (semi)algebraic invariants is even decidable, but the results prove completeness for much larger classes of (semi)analytic invariants.
\rref{tab:axproperties} gives an instructive overview of the key mathematical properties of solutions and terms that the soundness of each differential equation axiom rests on.
With these axioms, mathematical reasoning for differential equations can be carried out \emph{syntactically} and \emph{axiomatically} within the \dL proof calculus.
This concise and foundational axiomatization of mathematical properties is precisely what enables generalizations of the authors' earlier results~\cite{DBLP:conf/lics/PlatzerT18} to the (semi)analytic setting with Noetherian functions.
A subtle question is left open: the extended term conditions in~\rref{subsec:background-compatibility} do not require that the fixed function symbols $\noef$ be real analytic even if Noetherian functions are always real analytic.
This suggests that there may still be a gap between the extended term conditions and Noetherian functions.
Are there $C^\infty$ smooth (or even real analytic) functions that meet the extended term conditions but are \emph{not} Noetherian functions?
In other words, are Noetherian functions exactly the class of functions for which completeness results are possible?
Certainly, this article's completeness results continue to hold for any functions meeting those conditions, which would make both positive and negative results interesting.

\begin{table}[tbh]
\caption{Properties of ODE solutions underlying the differential equation axioms of \dL.}
\label{tab:axproperties}
\begin{tabular}{ll}
\hline
\textbf{ODE Axiom}   & \textbf{Mathematical Property}  \\ \hline
\irref{DI}       & Mean value theorem\\
\irref{DC}       & Prefix-closure of solutions\\
\irref{DG}       & Picard-Lindel\"of theorem\\
\irref{Cont}     & Existence of solutions\\
\irref{Uniq}     & Uniqueness of solutions\\
\irref{Dadjoint} & Group action on solutions\\
\irref{RealInd}  & Completeness of field $\reals$\\
\hline
\end{tabular}
\end{table}

\begin{acks}
We thank the associate editor for handling this article, and Brandon Bohrer and the anonymous reviewers for their insightful comments and feedback.
We also thank Khalil Ghorbal, Andrew Sogokon, and the LICS'18 anonymous reviewers for their detailed feedback on the earlier conference version.
This material is based upon work supported by the National Science Foundation under NSF CAREER Award CNS-1054246 and an Alexander von Humboldt fellowship.
The second author was also supported by A*STAR, Singapore.

Any opinions, findings, and conclusions or recommendations expressed in this publication are those of the author(s) and do not necessarily reflect the views of the National Science Foundation.
\end{acks}

\bibliographystyle{ACM-Reference-Format}
\bibliography{diffaxiomatic-arXiv}

%%% -*-BibTeX-*-
%%% Do NOT edit. File created by BibTeX with style
%%% ACM-Reference-Format-Journals [18-Jan-2012].

\begin{thebibliography}{00}

%%% ====================================================================
%%% NOTE TO THE USER: you can override these defaults by providing
%%% customized versions of any of these macros before the \bibliography
%%% command.  Each of them MUST provide its own final punctuation,
%%% except for \shownote{}, \showDOI{}, and \showURL{}.  The latter two
%%% do not use final punctuation, in order to avoid confusing it with
%%% the Web address.
%%%
%%% To suppress output of a particular field, define its macro to expand
%%% to an empty string, or better, \unskip, like this:
%%%
%%% \newcommand{\showDOI}[1]{\unskip}   % LaTeX syntax
%%%
%%% \def \showDOI #1{\unskip}           % plain TeX syntax
%%%
%%% ====================================================================

\ifx \showCODEN    \undefined \def \showCODEN     #1{\unskip}     \fi
\ifx \showDOI      \undefined \def \showDOI       #1{#1}\fi
\ifx \showISBNx    \undefined \def \showISBNx     #1{\unskip}     \fi
\ifx \showISBNxiii \undefined \def \showISBNxiii  #1{\unskip}     \fi
\ifx \showISSN     \undefined \def \showISSN      #1{\unskip}     \fi
\ifx \showLCCN     \undefined \def \showLCCN      #1{\unskip}     \fi
\ifx \shownote     \undefined \def \shownote      #1{#1}          \fi
\ifx \showarticletitle \undefined \def \showarticletitle #1{#1}   \fi
\ifx \showURL      \undefined \def \showURL       {\relax}        \fi
% The following commands are used for tagged output and should be
% invisible to TeX
\providecommand\bibfield[2]{#2}
\providecommand\bibinfo[2]{#2}
\providecommand\natexlab[1]{#1}
\providecommand\showeprint[2][]{arXiv:#2}

\bibitem[\protect\citeauthoryear{Akbarpour and Paulson}{Akbarpour and
  Paulson}{2010}]%
        {DBLP:journals/jar/AkbarpourP10}
\bibfield{author}{\bibinfo{person}{Behzad Akbarpour} {and}
  \bibinfo{person}{Lawrence~C. Paulson}.} \bibinfo{year}{2010}\natexlab{}.
\newblock \showarticletitle{MetiTarski: An Automatic Theorem Prover for
  Real-Valued Special Functions}.
\newblock \bibinfo{journal}{{\em J. Autom. Reasoning\/}} \bibinfo{volume}{44},
  \bibinfo{number}{3} (\bibinfo{year}{2010}), \bibinfo{pages}{175--205}.
\newblock
\showDOI{%
\url{https://doi.org/10.1007/s10817-009-9149-2}}


\bibitem[\protect\citeauthoryear{Avigad, Lewis, and Roux}{Avigad
  et~al\mbox{.}}{2016}]%
        {DBLP:journals/jar/AvigadLR16}
\bibfield{author}{\bibinfo{person}{Jeremy Avigad}, \bibinfo{person}{Robert~Y.
  Lewis}, {and} \bibinfo{person}{Cody Roux}.} \bibinfo{year}{2016}\natexlab{}.
\newblock \showarticletitle{A Heuristic Prover for Real Inequalities}.
\newblock \bibinfo{journal}{{\em J. Autom. Reasoning\/}} \bibinfo{volume}{56},
  \bibinfo{number}{3} (\bibinfo{year}{2016}), \bibinfo{pages}{367--386}.
\newblock
\showDOI{%
\url{https://doi.org/10.1007/s10817-015-9356-y}}


\bibitem[\protect\citeauthoryear{Binyamini}{Binyamini}{2019}]%
        {MR3925105}
\bibfield{author}{\bibinfo{person}{Gal Binyamini}.}
  \bibinfo{year}{2019}\natexlab{}.
\newblock \showarticletitle{Density of Algebraic Points on {N}oetherian
  Varieties}.
\newblock \bibinfo{journal}{{\em Geom. Funct. Anal.\/}} \bibinfo{volume}{29},
  \bibinfo{number}{1} (\bibinfo{year}{2019}), \bibinfo{pages}{72--118}.
\newblock
\showISSN{1016-443X}
\showDOI{%
\url{https://doi.org/10.1007/s00039-019-00475-7}}


\bibitem[\protect\citeauthoryear{Bochnak, Coste, and Roy}{Bochnak
  et~al\mbox{.}}{1998}]%
        {Bochnak1998}
\bibfield{author}{\bibinfo{person}{Jacek Bochnak}, \bibinfo{person}{Michel
  Coste}, {and} \bibinfo{person}{Marie-Fran{\c{c}}oise Roy}.}
  \bibinfo{year}{1998}\natexlab{}.
\newblock \bibinfo{booktitle}{{\em Real Algebraic Geometry}}.
\newblock \bibinfo{publisher}{Springer}, \bibinfo{address}{Heidelberg}.
\newblock
\showISBNx{978-3-540-64663-1}
\showDOI{%
\url{https://doi.org/10.1007/978-3-662-03718-8}}


\bibitem[\protect\citeauthoryear{Bohrer, Rahli, Vukotic, V{\"{o}}lp, and
  Platzer}{Bohrer et~al\mbox{.}}{2017}]%
        {DBLP:conf/cpp/BohrerRVVP17}
\bibfield{author}{\bibinfo{person}{Brandon Bohrer}, \bibinfo{person}{Vincent
  Rahli}, \bibinfo{person}{Ivana Vukotic}, \bibinfo{person}{Marcus V{\"{o}}lp},
  {and} \bibinfo{person}{Andr{\'{e}} Platzer}.}
  \bibinfo{year}{2017}\natexlab{}.
\newblock \showarticletitle{Formally Verified Differential Dynamic Logic}. In
  \bibinfo{booktitle}{{\em CPP}}, \bibfield{editor}{\bibinfo{person}{Yves
  Bertot} {and} \bibinfo{person}{Viktor Vafeiadis}} (Eds.).
  \bibinfo{publisher}{{ACM}}, \bibinfo{address}{New York},
  \bibinfo{pages}{208--221}.
\newblock
\showDOI{%
\url{https://doi.org/10.1145/3018610.3018616}}


\bibitem[\protect\citeauthoryear{Bourbaki}{Bourbaki}{1998}]%
        {MR1727221}
\bibfield{author}{\bibinfo{person}{Nicolas Bourbaki}.}
  \bibinfo{year}{1998}\natexlab{}.
\newblock \bibinfo{booktitle}{{\em Commutative Algebra. {C}hapters 1--7}}.
\newblock \bibinfo{publisher}{Springer}, \bibinfo{address}{Berlin}.
\newblock
\showISBNx{3-540-64239-0}


\bibitem[\protect\citeauthoryear{Clark}{Clark}{2019}]%
        {doi:10.1080/0025570X.2019.1549902}
\bibfield{author}{\bibinfo{person}{Pete~L. Clark}.}
  \bibinfo{year}{2019}\natexlab{}.
\newblock \showarticletitle{The Instructor's Guide to Real Induction}.
\newblock \bibinfo{journal}{{\em Math. Mag.\/}} \bibinfo{volume}{92},
  \bibinfo{number}{2} (\bibinfo{year}{2019}), \bibinfo{pages}{136--150}.
\newblock
\showDOI{%
\url{https://doi.org/10.1080/0025570X.2019.1549902}}


\bibitem[\protect\citeauthoryear{Clarke, Henzinger, Veith, and Bloem}{Clarke
  et~al\mbox{.}}{2018}]%
        {DBLP:reference/mc/2018}
\bibfield{editor}{\bibinfo{person}{Edmund~M. Clarke},
  \bibinfo{person}{Thomas~A. Henzinger}, \bibinfo{person}{Helmut Veith}, {and}
  \bibinfo{person}{Roderick Bloem}} (Eds.). \bibinfo{year}{2018}\natexlab{}.
\newblock \bibinfo{booktitle}{{\em Handbook of Model Checking}}.
\newblock \bibinfo{publisher}{Springer}, \bibinfo{address}{Cham}.
\newblock
\showISBNx{978-3-319-10574-1}
\showDOI{%
\url{https://doi.org/10.1007/978-3-319-10575-8}}


\bibitem[\protect\citeauthoryear{Darboux}{Darboux}{1878}]%
        {Darboux}
\bibfield{author}{\bibinfo{person}{Jean-Gaston Darboux}.}
  \bibinfo{year}{1878}\natexlab{}.
\newblock \showarticletitle{M{\'e}moire sur les {\'e}quations
  diff{\'e}rentielles alg{\'e}briques du premier ordre et du premier
  degr{\'e}}.
\newblock \bibinfo{journal}{{\em Bull. Sci. Math.\/}} \bibinfo{volume}{2},
  \bibinfo{number}{1} (\bibinfo{year}{1878}), \bibinfo{pages}{151--200}.
\newblock


\bibitem[\protect\citeauthoryear{Fulton, Mitsch, Quesel, V{\"{o}}lp, and
  Platzer}{Fulton et~al\mbox{.}}{2015}]%
        {DBLP:conf/cade/FultonMQVP15}
\bibfield{author}{\bibinfo{person}{Nathan Fulton}, \bibinfo{person}{Stefan
  Mitsch}, \bibinfo{person}{Jan{-}David Quesel}, \bibinfo{person}{Marcus
  V{\"{o}}lp}, {and} \bibinfo{person}{Andr{\'{e}} Platzer}.}
  \bibinfo{year}{2015}\natexlab{}.
\newblock \showarticletitle{KeYmaera {X:} An Axiomatic Tactical Theorem Prover
  for Hybrid Systems}. In \bibinfo{booktitle}{{\em CADE}} {\em
  (\bibinfo{series}{LNCS})}, \bibfield{editor}{\bibinfo{person}{Amy~P. Felty}
  {and} \bibinfo{person}{Aart Middeldorp}} (Eds.), Vol.~\bibinfo{volume}{9195}.
  \bibinfo{publisher}{Springer}, \bibinfo{address}{Cham},
  \bibinfo{pages}{527--538}.
\newblock
\showDOI{%
\url{https://doi.org/10.1007/978-3-319-21401-6_36}}


\bibitem[\protect\citeauthoryear{Gabrielov and Khovanskii}{Gabrielov and
  Khovanskii}{1998}]%
        {MR1732408}
\bibfield{author}{\bibinfo{person}{Andrei Gabrielov} {and}
  \bibinfo{person}{Askold Khovanskii}.} \bibinfo{year}{1998}\natexlab{}.
\newblock \showarticletitle{Multiplicity of a {N}oetherian Intersection}.
\newblock In \bibinfo{booktitle}{{\em Geometry of Differential Equations}}.
  \bibinfo{publisher}{Amer. Math. Soc.}, \bibinfo{address}{Providence},
  \bibinfo{pages}{119--130}.
\newblock
\showDOI{%
\url{https://doi.org/10.1090/trans2/186/03}}


\bibitem[\protect\citeauthoryear{Gabrielov and Vorobjov}{Gabrielov and
  Vorobjov}{2004}]%
        {MR2083248}
\bibfield{author}{\bibinfo{person}{Andrei Gabrielov} {and}
  \bibinfo{person}{Nicolai Vorobjov}.} \bibinfo{year}{2004}\natexlab{}.
\newblock \showarticletitle{Complexity of Computations with {P}faffian and
  {N}oetherian Functions}.
\newblock In \bibinfo{booktitle}{{\em Normal Forms, Bifurcations and Finiteness
  Problems in Differential Equations}}. \bibinfo{publisher}{Kluwer Acad.
  Publ.}, \bibinfo{address}{Netherlands}, \bibinfo{pages}{211--250}.
\newblock


\bibitem[\protect\citeauthoryear{Gao, Avigad, and Clarke}{Gao
  et~al\mbox{.}}{2012}]%
        {DBLP:conf/lics/GaoAC12}
\bibfield{author}{\bibinfo{person}{Sicun Gao}, \bibinfo{person}{Jeremy Avigad},
  {and} \bibinfo{person}{Edmund~M. Clarke}.} \bibinfo{year}{2012}\natexlab{}.
\newblock \showarticletitle{Delta-Decidability over the Reals}. In
  \bibinfo{booktitle}{{\em LICS}}. \bibinfo{publisher}{{IEEE} Computer
  Society}, \bibinfo{pages}{305--314}.
\newblock
\showDOI{%
\url{https://doi.org/10.1109/LICS.2012.41}}


\bibitem[\protect\citeauthoryear{Gao, Kong, and Clarke}{Gao
  et~al\mbox{.}}{2013}]%
        {DBLP:conf/cade/GaoKC13}
\bibfield{author}{\bibinfo{person}{Sicun Gao}, \bibinfo{person}{Soonho Kong},
  {and} \bibinfo{person}{Edmund~M. Clarke}.} \bibinfo{year}{2013}\natexlab{}.
\newblock \showarticletitle{dReal: An {SMT} Solver for Nonlinear Theories over
  the Reals}. In \bibinfo{booktitle}{{\em CADE}} {\em
  (\bibinfo{series}{LNCS})}, \bibfield{editor}{\bibinfo{person}{Maria~Paola
  Bonacina}} (Ed.), Vol.~\bibinfo{volume}{7898}. \bibinfo{publisher}{Springer},
  \bibinfo{address}{Heidelberg}, \bibinfo{pages}{208--214}.
\newblock
\showDOI{%
\url{https://doi.org/10.1007/978-3-642-38574-2\_14}}


\bibitem[\protect\citeauthoryear{Ghorbal and Platzer}{Ghorbal and
  Platzer}{2014}]%
        {DBLP:conf/tacas/GhorbalP14}
\bibfield{author}{\bibinfo{person}{Khalil Ghorbal} {and}
  \bibinfo{person}{Andr{\'{e}} Platzer}.} \bibinfo{year}{2014}\natexlab{}.
\newblock \showarticletitle{Characterizing Algebraic Invariants by Differential
  Radical Invariants}. In \bibinfo{booktitle}{{\em TACAS}} {\em
  (\bibinfo{series}{LNCS})}, \bibfield{editor}{\bibinfo{person}{Erika
  {\'{A}}brah{\'{a}}m} {and} \bibinfo{person}{Klaus Havelund}} (Eds.),
  Vol.~\bibinfo{volume}{8413}. \bibinfo{publisher}{Springer},
  \bibinfo{address}{Heidelberg}, \bibinfo{pages}{279--294}.
\newblock
\showDOI{%
\url{https://doi.org/10.1007/978-3-642-54862-8_19}}


\bibitem[\protect\citeauthoryear{Ghorbal, Sogokon, and Platzer}{Ghorbal
  et~al\mbox{.}}{2017}]%
        {DBLP:journals/cl/GhorbalSP17}
\bibfield{author}{\bibinfo{person}{Khalil Ghorbal}, \bibinfo{person}{Andrew
  Sogokon}, {and} \bibinfo{person}{Andr{\'e} Platzer}.}
  \bibinfo{year}{2017}\natexlab{}.
\newblock \showarticletitle{A Hierarchy of Proof Rules for Checking Positive
  Invariance of Algebraic and Semi-Algebraic Sets}.
\newblock \bibinfo{journal}{{\em Comput. Lang. Syst. Str.\/}}
  \bibinfo{volume}{47}, \bibinfo{number}{1} (\bibinfo{year}{2017}),
  \bibinfo{pages}{19--43}.
\newblock
\showDOI{%
\url{https://doi.org/10.1016/j.cl.2015.11.003}}


\bibitem[\protect\citeauthoryear{Gra\c{c}a, Campagnolo, and Buescu}{Gra\c{c}a
  et~al\mbox{.}}{2008}]%
        {GRACA2008330}
\bibfield{author}{\bibinfo{person}{Daniel~S. Gra\c{c}a},
  \bibinfo{person}{Manuel~L. Campagnolo}, {and} \bibinfo{person}{Jorge
  Buescu}.} \bibinfo{year}{2008}\natexlab{}.
\newblock \showarticletitle{Computability with Polynomial Differential
  Equations}.
\newblock \bibinfo{journal}{{\em Adv. Appl. Math.\/}} \bibinfo{volume}{40},
  \bibinfo{number}{3} (\bibinfo{year}{2008}), \bibinfo{pages}{330 -- 349}.
\newblock
\showISSN{0196-8858}
\showDOI{%
\url{https://doi.org/10.1016/j.aam.2007.02.003}}


\bibitem[\protect\citeauthoryear{Gr{\"o}nwall}{Gr{\"o}nwall}{1919}]%
        {DBLP:journals/mathann/Gronwall19}
\bibfield{author}{\bibinfo{person}{Thomas~H. Gr{\"o}nwall}.}
  \bibinfo{year}{1919}\natexlab{}.
\newblock \showarticletitle{Note on the Derivatives with Respect to a Parameter
  of the Solutions of a System of Differential Equations}.
\newblock \bibinfo{journal}{{\em Ann. Math.\/}} \bibinfo{volume}{20},
  \bibinfo{number}{4} (\bibinfo{year}{1919}), \bibinfo{pages}{292--296}.
\newblock
\showDOI{%
\url{https://doi.org/10.2307/1967124}}


\bibitem[\protect\citeauthoryear{Korovina and Vorobjov}{Korovina and
  Vorobjov}{2004}]%
        {DBLP:conf/csl/KorovinaV04}
\bibfield{author}{\bibinfo{person}{Margarita~V. Korovina} {and}
  \bibinfo{person}{Nicolai Vorobjov}.} \bibinfo{year}{2004}\natexlab{}.
\newblock \showarticletitle{Pfaffian Hybrid Systems}. In
  \bibinfo{booktitle}{{\em CSL}} {\em (\bibinfo{series}{LNCS})},
  \bibfield{editor}{\bibinfo{person}{Jerzy Marcinkowski} {and}
  \bibinfo{person}{Andrzej Tarlecki}} (Eds.), Vol.~\bibinfo{volume}{3210}.
  \bibinfo{publisher}{Springer}, \bibinfo{address}{Heidelberg},
  \bibinfo{pages}{430--441}.
\newblock
\showDOI{%
\url{https://doi.org/10.1007/978-3-540-30124-0\_33}}


\bibitem[\protect\citeauthoryear{Krantz and Parks}{Krantz and Parks}{2002}]%
        {MR1916029}
\bibfield{author}{\bibinfo{person}{Steven~G. Krantz} {and}
  \bibinfo{person}{Harold~R. Parks}.} \bibinfo{year}{2002}\natexlab{}.
\newblock \bibinfo{booktitle}{{\em A Primer of Real Analytic Functions\/}
  (\bibinfo{edition}{second} ed.)}.
\newblock \bibinfo{publisher}{Birkh\"{a}user}, \bibinfo{address}{Boston}.
\newblock
\showISBNx{0-8176-4264-1}
\showDOI{%
\url{https://doi.org/10.1007/978-0-8176-8134-0}}


\bibitem[\protect\citeauthoryear{Lafferriere, Pappas, and Sastry}{Lafferriere
  et~al\mbox{.}}{2000}]%
        {DBLP:journals/mcss/LafferrierePS00}
\bibfield{author}{\bibinfo{person}{Gerardo Lafferriere},
  \bibinfo{person}{George~J. Pappas}, {and} \bibinfo{person}{Shankar Sastry}.}
  \bibinfo{year}{2000}\natexlab{}.
\newblock \showarticletitle{O-Minimal Hybrid Systems}.
\newblock \bibinfo{journal}{{\em Math. Control Signals Systems\/}}
  \bibinfo{volume}{13}, \bibinfo{number}{1} (\bibinfo{year}{2000}),
  \bibinfo{pages}{1--21}.
\newblock
\showDOI{%
\url{https://doi.org/10.1007/PL00009858}}


\bibitem[\protect\citeauthoryear{Liu, Zhan, and Zhao}{Liu
  et~al\mbox{.}}{2011}]%
        {DBLP:conf/emsoft/LiuZZ11}
\bibfield{author}{\bibinfo{person}{Jiang Liu}, \bibinfo{person}{Naijun Zhan},
  {and} \bibinfo{person}{Hengjun Zhao}.} \bibinfo{year}{2011}\natexlab{}.
\newblock \showarticletitle{Computing Semi-Algebraic Invariants for Polynomial
  Dynamical Systems}. In \bibinfo{booktitle}{{\em EMSOFT}},
  \bibfield{editor}{\bibinfo{person}{Samarjit Chakraborty},
  \bibinfo{person}{Ahmed Jerraya}, \bibinfo{person}{Sanjoy~K. Baruah}, {and}
  \bibinfo{person}{Sebastian Fischmeister}} (Eds.). \bibinfo{publisher}{{ACM}},
  \bibinfo{address}{New York}, \bibinfo{pages}{97--106}.
\newblock
\showDOI{%
\url{https://doi.org/10.1145/2038642.2038659}}


\bibitem[\protect\citeauthoryear{Liu, Zhan, Zhao, and Zou}{Liu
  et~al\mbox{.}}{2015}]%
        {DBLP:conf/fm/0009ZZZ15}
\bibfield{author}{\bibinfo{person}{Jiang Liu}, \bibinfo{person}{Naijun Zhan},
  \bibinfo{person}{Hengjun Zhao}, {and} \bibinfo{person}{Liang Zou}.}
  \bibinfo{year}{2015}\natexlab{}.
\newblock \showarticletitle{Abstraction of Elementary Hybrid Systems by
  Variable Transformation}. In \bibinfo{booktitle}{{\em FM}} {\em
  (\bibinfo{series}{LNCS})}, \bibfield{editor}{\bibinfo{person}{Nikolaj
  Bj{\o}rner} {and} \bibinfo{person}{Frank~S. de~Boer}} (Eds.),
  Vol.~\bibinfo{volume}{9109}. \bibinfo{publisher}{Springer},
  \bibinfo{address}{Cham}, \bibinfo{pages}{360--377}.
\newblock
\showDOI{%
\url{https://doi.org/10.1007/978-3-319-19249-9\_23}}


\bibitem[\protect\citeauthoryear{McCallum and Weispfenning}{McCallum and
  Weispfenning}{2012}]%
        {DBLP:journals/jsc/McCallumW12}
\bibfield{author}{\bibinfo{person}{Scott McCallum} {and}
  \bibinfo{person}{Volker Weispfenning}.} \bibinfo{year}{2012}\natexlab{}.
\newblock \showarticletitle{Deciding Polynomial-Transcendental Problems}.
\newblock \bibinfo{journal}{{\em J. Symb. Comput.\/}} \bibinfo{volume}{47},
  \bibinfo{number}{1} (\bibinfo{year}{2012}), \bibinfo{pages}{16--31}.
\newblock
\showDOI{%
\url{https://doi.org/10.1016/j.jsc.2011.08.004}}


\bibitem[\protect\citeauthoryear{Novikov and Yakovenko}{Novikov and
  Yakovenko}{1999}]%
        {MR1697373}
\bibfield{author}{\bibinfo{person}{Dimitri Novikov} {and}
  \bibinfo{person}{Sergei Yakovenko}.} \bibinfo{year}{1999}\natexlab{}.
\newblock \showarticletitle{Trajectories of Polynomial Vector Fields and
  Ascending Chains of Polynomial Ideals}.
\newblock \bibinfo{journal}{{\em Ann. I. Fourier\/}} \bibinfo{volume}{49},
  \bibinfo{number}{2} (\bibinfo{year}{1999}), \bibinfo{pages}{563--609}.
\newblock
\showISSN{0373-0956}
\showDOI{%
\url{https://doi.org/10.5802/aif.1683}}


\bibitem[\protect\citeauthoryear{Owicki and Gries}{Owicki and Gries}{1976}]%
        {DBLP:journals/cacm/OwickiG76}
\bibfield{author}{\bibinfo{person}{Susan~S. Owicki} {and}
  \bibinfo{person}{David Gries}.} \bibinfo{year}{1976}\natexlab{}.
\newblock \showarticletitle{Verifying Properties of Parallel Programs: An
  Axiomatic Approach}.
\newblock \bibinfo{journal}{{\em Commun. {ACM}\/}} \bibinfo{volume}{19},
  \bibinfo{number}{5} (\bibinfo{year}{1976}), \bibinfo{pages}{279--285}.
\newblock
\showDOI{%
\url{https://doi.org/10.1145/360051.360224}}


\bibitem[\protect\citeauthoryear{Platzer}{Platzer}{2008}]%
        {DBLP:journals/jar/Platzer08}
\bibfield{author}{\bibinfo{person}{Andr{\'e} Platzer}.}
  \bibinfo{year}{2008}\natexlab{}.
\newblock \showarticletitle{Differential Dynamic Logic for Hybrid Systems}.
\newblock \bibinfo{journal}{{\em J. Autom. Reasoning\/}} \bibinfo{volume}{41},
  \bibinfo{number}{2} (\bibinfo{year}{2008}), \bibinfo{pages}{143--189}.
\newblock
\showISSN{0168-7433}
\showDOI{%
\url{https://doi.org/10.1007/s10817-008-9103-8}}


\bibitem[\protect\citeauthoryear{Platzer}{Platzer}{2010}]%
        {DBLP:journals/logcom/Platzer10}
\bibfield{author}{\bibinfo{person}{Andr{\'e} Platzer}.}
  \bibinfo{year}{2010}\natexlab{}.
\newblock \showarticletitle{Differential-Algebraic Dynamic Logic for
  Differential-Algebraic Programs}.
\newblock \bibinfo{journal}{{\em J. Log. Comput.\/}} \bibinfo{volume}{20},
  \bibinfo{number}{1} (\bibinfo{year}{2010}), \bibinfo{pages}{309--352}.
\newblock
\showDOI{%
\url{https://doi.org/10.1093/logcom/exn070}}


\bibitem[\protect\citeauthoryear{Platzer}{Platzer}{2012a}]%
        {DBLP:conf/lics/Platzer12b}
\bibfield{author}{\bibinfo{person}{Andr{\'e} Platzer}.}
  \bibinfo{year}{2012}\natexlab{a}.
\newblock \showarticletitle{The Complete Proof Theory of Hybrid Systems}. In
  \bibinfo{booktitle}{{\em LICS}}. \bibinfo{publisher}{{IEEE} Computer
  Society}, \bibinfo{pages}{541--550}.
\newblock
\showISBNx{978-1-4673-2263-8}
\showDOI{%
\url{https://doi.org/10.1109/LICS.2012.64}}


\bibitem[\protect\citeauthoryear{Platzer}{Platzer}{2012b}]%
        {DBLP:conf/itp/Platzer12}
\bibfield{author}{\bibinfo{person}{Andr{\'{e}} Platzer}.}
  \bibinfo{year}{2012}\natexlab{b}.
\newblock \showarticletitle{A Differential Operator Approach to Equational
  Differential Invariants}. In \bibinfo{booktitle}{{\em ITP}} {\em
  (\bibinfo{series}{LNCS})}, \bibfield{editor}{\bibinfo{person}{Lennart
  Beringer} {and} \bibinfo{person}{Amy~P. Felty}} (Eds.),
  Vol.~\bibinfo{volume}{7406}. \bibinfo{publisher}{Springer},
  \bibinfo{address}{Heidelberg}, \bibinfo{pages}{28--48}.
\newblock
\showDOI{%
\url{https://doi.org/10.1007/978-3-642-32347-8\_3}}


\bibitem[\protect\citeauthoryear{Platzer}{Platzer}{2012c}]%
        {DBLP:journals/lmcs/Platzer12}
\bibfield{author}{\bibinfo{person}{Andr{\'e} Platzer}.}
  \bibinfo{year}{2012}\natexlab{c}.
\newblock \showarticletitle{The Structure of Differential Invariants and
  Differential Cut Elimination}.
\newblock \bibinfo{journal}{{\em Log. Meth. Comput. Sci.\/}}
  \bibinfo{volume}{8}, \bibinfo{number}{4} (\bibinfo{year}{2012}),
  \bibinfo{pages}{1--38}.
\newblock
\showDOI{%
\url{https://doi.org/10.2168/LMCS-8(4:16)2012}}


\bibitem[\protect\citeauthoryear{Platzer}{Platzer}{2017}]%
        {DBLP:journals/jar/Platzer17}
\bibfield{author}{\bibinfo{person}{Andr{\'e} Platzer}.}
  \bibinfo{year}{2017}\natexlab{}.
\newblock \showarticletitle{A Complete Uniform Substitution Calculus for
  Differential Dynamic Logic}.
\newblock \bibinfo{journal}{{\em J. Autom. Reasoning\/}} \bibinfo{volume}{59},
  \bibinfo{number}{2} (\bibinfo{year}{2017}), \bibinfo{pages}{219--265}.
\newblock
\showDOI{%
\url{https://doi.org/10.1007/s10817-016-9385-1}}


\bibitem[\protect\citeauthoryear{Platzer and Tan}{Platzer and Tan}{2018}]%
        {DBLP:conf/lics/PlatzerT18}
\bibfield{author}{\bibinfo{person}{Andr{\'{e}} Platzer} {and}
  \bibinfo{person}{Yong~Kiam Tan}.} \bibinfo{year}{2018}\natexlab{}.
\newblock \showarticletitle{Differential Equation Axiomatization: The
  Impressive Power of Differential Ghosts}. In \bibinfo{booktitle}{{\em LICS}},
  \bibfield{editor}{\bibinfo{person}{Anuj Dawar} {and} \bibinfo{person}{Erich
  Gr{\"{a}}del}} (Eds.). \bibinfo{publisher}{ACM}, \bibinfo{address}{New York},
  \bibinfo{pages}{819--828}.
\newblock
\showISBNx{978-1-4503-5583-4}
\showDOI{%
\url{https://doi.org/10.1145/3209108.3209147}}


\bibitem[\protect\citeauthoryear{Poincar{\'{e}}}{Poincar{\'{e}}}{1881}]%
        {Poincare81}
\bibfield{author}{\bibinfo{person}{Henri Poincar{\'{e}}}.}
  \bibinfo{year}{1881}\natexlab{}.
\newblock \showarticletitle{M{\'{e}}moire sur les courbes d{\'{e}}finies par
  une {\'{e}}quation diff{\'{e}}rentielle}.
\newblock \bibinfo{journal}{{\em J. Math. Pures Appl.\/}}
  (\bibinfo{year}{1881}).
\newblock


\bibitem[\protect\citeauthoryear{Prajna and Jadbabaie}{Prajna and
  Jadbabaie}{2004}]%
        {DBLP:conf/hybrid/PrajnaJ04}
\bibfield{author}{\bibinfo{person}{Stephen Prajna} {and} \bibinfo{person}{Ali
  Jadbabaie}.} \bibinfo{year}{2004}\natexlab{}.
\newblock \showarticletitle{Safety Verification of Hybrid Systems Using Barrier
  Certificates}. In \bibinfo{booktitle}{{\em HSCC}} {\em
  (\bibinfo{series}{LNCS})}, \bibfield{editor}{\bibinfo{person}{Rajeev Alur}
  {and} \bibinfo{person}{George~J. Pappas}} (Eds.),
  Vol.~\bibinfo{volume}{2993}. \bibinfo{publisher}{Springer},
  \bibinfo{address}{Heidelberg}, \bibinfo{pages}{477--492}.
\newblock
\showDOI{%
\url{https://doi.org/10.1007/978-3-540-24743-2_32}}


\bibitem[\protect\citeauthoryear{Richardson}{Richardson}{1968}]%
        {DBLP:journals/jsyml/Richardson68}
\bibfield{author}{\bibinfo{person}{Daniel Richardson}.}
  \bibinfo{year}{1968}\natexlab{}.
\newblock \showarticletitle{Some Undecidable Problems Involving Elementary
  Functions of a Real Variable}.
\newblock \bibinfo{journal}{{\em J. Symb. Log.\/}} \bibinfo{volume}{33},
  \bibinfo{number}{4} (\bibinfo{year}{1968}), \bibinfo{pages}{514--520}.
\newblock
\showDOI{%
\url{https://doi.org/10.2307/2271358}}


\bibitem[\protect\citeauthoryear{Sankaranarayanan, Sipma, and
  Manna}{Sankaranarayanan et~al\mbox{.}}{2008}]%
        {DBLP:journals/fmsd/SankaranarayananSM08}
\bibfield{author}{\bibinfo{person}{Sriram Sankaranarayanan},
  \bibinfo{person}{Henny~B. Sipma}, {and} \bibinfo{person}{Zohar Manna}.}
  \bibinfo{year}{2008}\natexlab{}.
\newblock \showarticletitle{Constructing Invariants for Hybrid Systems}.
\newblock \bibinfo{journal}{{\em Form. Methods Syst. Des.\/}}
  \bibinfo{volume}{32}, \bibinfo{number}{1} (\bibinfo{year}{2008}),
  \bibinfo{pages}{25--55}.
\newblock
\showDOI{%
\url{https://doi.org/10.1007/s10703-007-0046-1}}


\bibitem[\protect\citeauthoryear{Speissegger}{Speissegger}{1999}]%
        {MR1676876}
\bibfield{author}{\bibinfo{person}{Patrick Speissegger}.}
  \bibinfo{year}{1999}\natexlab{}.
\newblock \showarticletitle{The {P}faffian Closure of an O-Minimal Structure}.
\newblock \bibinfo{journal}{{\em J. Reine Angew. Math.\/}}
  \bibinfo{volume}{508} (\bibinfo{year}{1999}), \bibinfo{pages}{189--211}.
\newblock
\showISSN{0075-4102}
\showDOI{%
\url{https://doi.org/10.1515/crll.1999.026}}


\bibitem[\protect\citeauthoryear{Taly and Tiwari}{Taly and Tiwari}{2009}]%
        {DBLP:conf/fsttcs/TalyT09}
\bibfield{author}{\bibinfo{person}{Ankur Taly} {and} \bibinfo{person}{Ashish
  Tiwari}.} \bibinfo{year}{2009}\natexlab{}.
\newblock \showarticletitle{Deductive Verification of Continuous Dynamical
  Systems}. In \bibinfo{booktitle}{{\em FSTTCS}} {\em
  (\bibinfo{series}{LIPIcs})}, \bibfield{editor}{\bibinfo{person}{Ravi Kannan}
  {and} \bibinfo{person}{K.~Narayan Kumar}} (Eds.), Vol.~\bibinfo{volume}{4}.
  \bibinfo{publisher}{Schloss Dagstuhl}, \bibinfo{address}{Dagstuhl},
  \bibinfo{pages}{383--394}.
\newblock
\showDOI{%
\url{https://doi.org/10.4230/LIPIcs.FSTTCS.2009.2334}}


\bibitem[\protect\citeauthoryear{Terzo}{Terzo}{2007}]%
        {Terzo}
\bibfield{author}{\bibinfo{person}{Giuseppina Terzo}.}
  \bibinfo{year}{2007}\natexlab{}.
\newblock {\em \bibinfo{title}{Consequences of Schanuel's Conjecture in
  Exponential Algebra}}.
\newblock \bibinfo{thesistype}{Ph.D. Dissertation}. \bibinfo{school}{University
  of Naples Federico II}.
\newblock


\bibitem[\protect\citeauthoryear{Tougeron}{Tougeron}{1991}]%
        {MR1150568}
\bibfield{author}{\bibinfo{person}{Jean-Claude Tougeron}.}
  \bibinfo{year}{1991}\natexlab{}.
\newblock \showarticletitle{Alg\`ebres analytiques topologiquement
  noeth\'{e}riennes. {T}h\'{e}orie de {K}hovanski\u{\i}}.
\newblock \bibinfo{journal}{{\em Ann. I. Fourier\/}} \bibinfo{volume}{41},
  \bibinfo{number}{4} (\bibinfo{year}{1991}), \bibinfo{pages}{823--840}.
\newblock
\showISSN{0373-0956}
\showDOI{%
\url{https://doi.org/10.5802/aif.1275}}


\bibitem[\protect\citeauthoryear{van~den Dries}{van~den Dries}{1984}]%
        {MR762106}
\bibfield{author}{\bibinfo{person}{Lou van~den Dries}.}
  \bibinfo{year}{1984}\natexlab{}.
\newblock \showarticletitle{Remarks on {T}arski's Problem Concerning (R, +, *,
  exp)}.
\newblock In \bibinfo{booktitle}{{\em Logic Colloquium '82}},
  \bibfield{editor}{\bibinfo{person}{Gabriele Lolli}, \bibinfo{person}{Giuseppe
  Longo}, {and} \bibinfo{person}{Annalisa Marcja}} (Eds.).
  Vol.~\bibinfo{volume}{112}. \bibinfo{publisher}{North-Holland},
  \bibinfo{address}{Amsterdam}, \bibinfo{pages}{97--121}.
\newblock
\showDOI{%
\url{https://doi.org/10.1016/S0049-237X(08)71811-1}}


\bibitem[\protect\citeauthoryear{van~den Dries}{van~den Dries}{1998}]%
        {MR1633348}
\bibfield{author}{\bibinfo{person}{Lou van~den Dries}.}
  \bibinfo{year}{1998}\natexlab{}.
\newblock \bibinfo{booktitle}{{\em Tame Topology and O-Minimal Structures}}.
\newblock \bibinfo{publisher}{Cambridge University Press},
  \bibinfo{address}{Cambridge}.
\newblock
\showISBNx{0-521-59838-9}
\showDOI{%
\url{https://doi.org/10.1017/CBO9780511525919}}


\bibitem[\protect\citeauthoryear{Walter}{Walter}{1998}]%
        {Walter1998}
\bibfield{author}{\bibinfo{person}{Wolfgang Walter}.}
  \bibinfo{year}{1998}\natexlab{}.
\newblock \bibinfo{booktitle}{{\em Ordinary Differential Equations}}.
\newblock \bibinfo{publisher}{Springer}, \bibinfo{address}{New York}.
\newblock
\showISBNx{978-0-387-98459-9}
\showDOI{%
\url{https://doi.org/10.1007/978-1-4612-0601-9}}


\bibitem[\protect\citeauthoryear{Wilkie}{Wilkie}{1999}]%
        {MR1740677}
\bibfield{author}{\bibinfo{person}{Alex~J. Wilkie}.}
  \bibinfo{year}{1999}\natexlab{}.
\newblock \showarticletitle{A Theorem of the Complement and Some New O-Minimal
  Structures}.
\newblock \bibinfo{journal}{{\em Sel. Math. New Ser.\/}} \bibinfo{volume}{5},
  \bibinfo{number}{4} (\bibinfo{year}{1999}), \bibinfo{pages}{397--421}.
\newblock
\showISSN{1022-1824}
\showDOI{%
\url{https://doi.org/10.1007/s000290050052}}


\end{thebibliography}

\appendix

\section{Differential Dynamic Logic Axiomatization}
\label{app:axiomatization}

\subsection{Extended Axiomatization Soundness}
\label{app:extaxiomatization}

This section proves the soundness of the axiomatic extension from~\rref{sec:extaxioms}.
For the solution $\solvar : [0,T] \to \States$, its truncation to the interval $[0,t]$ for some $0 \leq t \leq T$ is denoted $\soltrunc{t} : [0,t] \to \States$, with $\soltrunc{t}(\zeta)=\solvar(\zeta)$ for $\zeta \in [0,t]$.
The shorthand notation $\imodels{\Iff[{[a,b]}]}{\rfvar}$ means $\imodels{\Iff}{\rfvar}$ for all $a \leq \zeta \leq b$, where the interval $[a,b]$ is required to be a closed subinterval of the interval $[0,T]$.
Analogously, $\solvar((a,b))$ is used when the interval is open, and similarly for the half-open cases.

As explained in \rref{sec:extaxioms}, the soundness of the extended axioms requires that the ODE system \m{\D{x}=\genDE{x}} always locally evolves $x$.
An easy syntactic check ensuring this condition is if the system already contains an equation \(\D{x_1}=1\) that tracks the passage of time, which can be added using axiom \irref{DG} if necessary before using the axioms.
However, the soundness proofs are more general and only use the assumption that the ODE system locally evolves $x$, whether by \(\D{x_1}=1\) or otherwise.

The soundness proofs make use of \dL's coincidence lemmas~\cite[Lemmas 10,11]{DBLP:journals/jar/Platzer17}:
\begin{lemma}[Coincidence for terms and formulas~\cite{DBLP:journals/jar/Platzer17}]
\label{lem:coincide}
The following coincidence properties hold for \dL, where free variables $\freevars{\etermA},\freevars{\fvarA}$ are defined as expected~\cite[Sections 2.3 and 2.4]{DBLP:journals/jar/Platzer17}.
\begin{itemize}
\item If the states $\iget[state]{\I}, \iget[state]{\It}$ agree on the free variables of term $\etermA$ (\m{\freevars{\etermA}}), then $\ivaluation{\I}{\etermA} = \ivaluation{\It}{\etermA}$.
\item If the states $\iget[state]{\I}, \iget[state]{\It}$ agree on the free variables of formula $\fvarA$ (\m{\freevars{\fvarA}}), then $\imodels{\I}{\fvarA}$ iff $\imodels{\It}{\fvarA}$.
\end{itemize}
\end{lemma}

\subsubsection{Existence, Uniqueness, and Continuity}
First, the axioms from~\rref{lem:uniqcont} internalizing basic existence and uniqueness properties of solutions of differential equations are proved sound.

\begin{proof}[Proof of \rref{lem:uniqcont}]
Let $\iget[state]{\I}$ be an arbitrary initial state.
When interpreted as a function of the variables $x$, the RHS $\genDE{x}$ of the ODE system $\D{x}=\genDE{x}$ is continuously differentiable.
Therefore, by the Picard-Lindel\"{o}f theorem~\cite[\S10.VI]{Walter1998}, from $\iget[state]{\I}$, there is an interval $[0,\tau), \tau > 0$ on which there is a unique, continuous solution $\solvar : [0,\tau)\to \States$ with $\solvar(0) = \iget[state]{\I}$ on $\scomplement{\{\D{x}\}}$.
The solution may be uniquely extended in time (to the right), up to its maximal open interval of existence~\cite[\S10.IX]{Walter1998}.

\begin{description}
\item[\irref{Uniq}] The ``$\limply$'' direction follows directly from monotonicity of domain constraints because of the propositional tautology $\ivr_1 \land \ivr_2 \limply \ivr_1$ (and similarly for $\ivr_2$).
For the ``$\lylpmi$'' direction, suppose that initial state $\iget[state]{\I}$ satisfies both conjuncts with $\imodels{\I}{\ddiamond{\pevolvein{\D{x}=\genDE{x}}{\ivr_1}}{\rfvar}}$ and $\imodels{\I}{\ddiamond{\pevolvein{\D{x}=\genDE{x}}{\ivr_2}}{\rfvar}}$.
Expanding the definition of the diamond modality, there exist two solutions $\solvar_1 : [0,T_1] \to \States$, $\solvar_2 : [0,T_2] \to \States$ from $\iget[state]{\I}$ such that $\solmodels{\solvar_1}{\D{x}=\genDE{x}}{\ivr_1}$ and $\solmodels{\solvar_2}{\D{x}=\genDE{x}}{\ivr_2}$, with both $\imodels{\IffA[T_1]}{\rfvar}$ and $\imodels{\IffB[T_2]}{\rfvar}$.
Suppose $T_1 \leq T_2$.
Since $\imodels{\IffB[{[0,T_2]}]}{\ivr_2}$ and, by uniqueness, $\solvar_1$ is a truncation of $\solvar_2$ to a smaller existence interval, $\solmodels{\solvar_1}{\D{x}=\genDE{x}}{(\ivr_1 \land \ivr_2)}$.
At time $T_1$, the solution satisfies $\imodels{\IffA[T_1]}{\rfvar}$, so $\imodels{\I}{\ddiamond{\pevolvein{\D{x}=\genDE{x}}{\ivr_1 \land \ivr_2}}{\rfvar}}$, as required.
The case for $T_2<T_1$ is similar, except with $\solmodels{\solvar_2}{\D{x}=\genDE{x}}{(\ivr_1 \land \ivr_2)}$ and satisfying $\imodels{\IffB[T_2]}{\rfvar}$ at time $T_2$ instead.

\item[\irref{Cont}] Assume that $\iget[state]{\I}$ satisfies the outermost implication, i.e., $\imodels{\I}{x=y}$.
The (inner) ``$\limply$'' direction follows by definition because in order for there to be a solution staying in $\etermA > 0$ at all, the initial state $\iget[state]{\I}$ must already satisfy $\etermA > 0$ (evolution domains are differential-free).
For the (inner) ``$\lylpmi$'' direction, suppose further that $\imodels{\I}{\etermA > 0}$.
Since $\D{x}\notin\etermA$ as $\etermA$ is differential-free (\rref{subsec:background-syntax}), coincidence (\rref{lem:coincide}) implies $\imodels{\Iff[0]}{\etermA>0}$.
As a composition of continuous evaluation \cite[Definition 5]{DBLP:journals/jar/Platzer17} with the continuous solution $\iget[flow]{\If}$, $\ivaluation{\Iff[t]}{\etermA}$ is a continuous function of time $t$.
Thus, $\imodels{\Iff[0]}{\etermA>0}$ implies $\imodels{\Iff[{[0,T]}]}{\etermA>0}$ for some $0 < T \leq \tau$ and the truncated solution $\truncafter{\iget[flow]{\If}}{T}$ satisfies $\solmodels{\truncafter{\iget[flow]{\If}}{T}}{\D{x}=\genDE{x}}{\etermA>0}$.
Since $y$ is constant for the ODE but \(\pevolve{\D{x}=\genDE{x}}\) was assumed to locally evolve (for example with \(\D{x_1}=1\)), there is a time $0<\epsilon\leq T$ at which $\imodels{\Iff[\epsilon]}{x\neq y}$.
The truncation \(\truncafter{\iget[flow]{\If}}{\epsilon}\) witnesses \(\imodels{\I}{\ddiamond{\pevolvein{\D{x}=\genDE{x}}{\etermA>0}}{x\neq y}}\).

\item[\irref{Dadjoint}] The ``$\lylpmi$'' direction follows immediately from the ``$\limply$'' direction by swapping the names $x,y$, because \(-(-\genDE{x})=\genDE{x}\). Therefore, it suffices to prove the ``$\limply$'' direction.
Suppose $\imodels{\I}{\ddiamond{\pevolvein{\D{x}=\genDE{x}}{\ivr\argx}}{\,x=y}}$.
Unfolding the semantics, there is a solution $\solvar : [0,T] \to \States$, of the system $\D{x}=\genDE{x}$, with $\solvar(0) = \iget[state]{\I}$ on $\scomplement{\{\D{x}\}}$, with \(\imodels{\Iff[{[0,T]}]}{\ivr\argx}\) and $\imodels{\Iff[T]}{x=y}$.
Since the variables $y$ do not appear in the differential equations \(\pevolve{\D{x}=\genDE{x}}\), their values are constant along the solution $\solvar$.
Consider the time- and variable-reversal $\psi : [0,T] \to \States$, where:
\[\psi(\tau)(z) \mdefeq
  \begin{cases}
  \solvar(T-\tau)(x_i)  & z = y_i \\
  -\solvar(T-\tau)(\D{x_i}) & z = \D{y_i} \\
  \iget[state]{\I}(z)             & \text{otherwise}
  \end{cases}\]

By construction, $\psi(0)$ agrees with $\iget[state]{\I}$ on $\scomplement{\{\D{y}\}}$ because \(\imodels{\Iff[T]}{x=y}\).
The signs of the differential variables $\D{y_i}$ are negated along $\psi$.
By uniqueness, the solutions of $\D{x}=-\genDE{x}$ are the time-reversed solutions of $\D{x}=\genDE{x}$.
As constructed, $\psi$ is the time-reversed solution for $\D{x}=\genDE{x}$ except the $x$ were replaced by $y$ instead.
Moreover, since $\imodels{\Iff[{[0,T]}]}{\ivr\argx}$, by construction and coincidence (\rref{lem:coincide}), $\imodels{\Ifff[{[0,T]}]}{\ivr(y)}$.
Therefore, $\solmodels{\psi}{\D{y}=-\genDE{y}}{\ivr(y)}$.
Finally, observe that $\psi(T)(y) = \solvar(0)(x)$, but $\psi$ holds the values of $x$ constant, thus $\psi(T)(x) = \iget[state]{\I}(x) = \solvar(0)(x)$ and so $\imodels{\Ifff[T]}{y=x}$ and $\psi$ witnesses $\imodels{\I}{\ddiamond{\pevolvein{\D{y}=-\genDE{y}}{\ivr(y)}}{y=x}}$
\qedhere
\end{description}
\end{proof}

\subsubsection{Real Induction}
The following real induction axiom with domain constraints is proved sound.
Axiom~\irref{RealInd} from~\rref{lem:realindODE} follows as an instance with no domain constraint, i.e., $\ivr \mnodefequiv \ltrue$.
\begin{align*}
\cinferenceRule[RealIndIn|RI{$\&$}]{}
{
\axkey{\dbox{\pevolvein{\D{x}=\genDE{x}}{\ivr}}{\rfvar}} &\lbisubjunct \lforall{y}{\dbox{\pevolvein{\D{x}=\genDE{x}}{\ivr \land (\rfvar \lor x=y)}}{\Big( \initassum \limply\\
&\underbrace{\vphantom{\big(\big)}\rfvar}_{\makebox[0pt]{\textcircled{a}}} \land \underbrace{\big(\ddiamond{\pevolvein{\D{x}=\genDE{x}}{\ivr \lor x=y}}{x\neq y} \limply \ddiamond{\pevolvein{\D{x}=\genDE{x}}{\rfvar \lor x=y}}{x\neq y}\big)}_{\textcircled{b}}\Big)
}}
}{}
\end{align*}

Similar to axiom~\irref{RealInd}, the axiom~\irref{RealIndIn} is based on the real induction principle~\cite{doi:10.1080/0025570X.2019.1549902} but also accounts for an arbitrary domain constraint $\ivr$.
Its RHS conjuncts labeled \textcircled{a} and \textcircled{b} correspond to \textcircled{1} and \textcircled{2} in~\rref{def:indsubset} respectively.
The quantification $\lforall{y}{\dbox{\pevolvein{\dots}{\ivr}}{\big(x=y \limply \dots\big)}}$ now only considers final states ($x=y$) reachable by trajectories that \emph{always} stay within $\ivr$, and within $\rfvar$ except possibly at the endpoint $x=y$.
The conjunct \textcircled{a} expresses that $\rfvar$ is still true at such an endpoint.
The conjunct \textcircled{b} expresses that $\rfvar$ continues to remain true locally but only when $\ivr$ itself remains true locally.
This added assumption for $\ivr$ corresponds to the ``If $\zeta < b$ then $\dots$'' assumption in \textcircled{2} of~\rref{def:indsubset}.
The conjunct \textcircled{b} can be rewritten succinctly with the local progress $\ddnext$ modality as:
\[  \dprogressin{\D{x}=\genDE{x}}{\ivr}{} \limply \dprogressin{\D{x}=\genDE{x}}{\rfvar}{} \]

With completeness for local progress (\rref{thm:localprogresscomplete}), this gives a first hint at how~\irref{RealIndIn} will be used to obtain a complete proof rule for semianalytic invariants with domain constraints in~\rref{app:completeness}.

\begin{lemma}[Real induction with domain constraints]
\label{lem:realindODEin}
The real induction axiom~\irref{RealIndIn} is sound, where $y$ is fresh in \m{\dbox{\pevolvein{\D{x}=\genDE{x}}{\ivr}}{\rfvar}}.
\end{lemma}
\begin{proof}[Proof (implies \rref{lem:realindODE})]
The conjuncts on the RHS of \irref{RealIndIn} are labeled as \textcircled{a} and \textcircled{b} respectively, as shown above.
Consider an initial state $\iget[state]{\I}$, both directions of the axiom are proved separately.
\begin{enumerate}
\item[``$\limply$'']
Assume the LHS of~\irref{RealIndIn} is true initially with \textcircled{$\star$} $\imodels{\I}{\dbox{\pevolvein{\D{x}=\genDE{x}}{\ivr}}{\rfvar}}$.
Unfolding the quantification and box modality on the RHS, let $\iget[state]{\I}_y$ be identical to $\iget[state]{\I}$ except where the values for $y$ are replaced with any arbitrary values $d \in \reals^n$.
Consider any solution $\solvar_y : [0,T] \to \States$ where $\solmodels{\solvar_y}{\D{x}=\genDE{x}}{\big(\ivr \land (\rfvar \lor x=y)\big)}$, $\solvar_y(0) = \iget[state]{\I}_y$ on $\scomplement{\{\D{x}\}}$, and $\imodels{\Iffy[T]}{x = y}$.

The following similar solution $\solvar : [0,T] \to \States$ keeps $y$ constant at their initial values in $\iget[state]{\I}$:
\[ \solvar(t)(z) \mdefeq
  \begin{cases}
  \solvar_y(t)(z)  & z \in  \scomplement{\{y\}}\\
  \iget[state]{\I}(z)   & z \in \{y\}
  \end{cases} \]
By construction, $\solvar(0)$ is identical to $\iget[state]{\I}$ on $\scomplement{\{\D{x}\}}$ and $\solvar$ is identical to $\solvar_y$ on $\scomplement{\{y\}}$.
Since $y$ is fresh in \(\pevolvein{\D{x}=\genDE{x}}{\ivr}\), by coincidence (\rref{lem:coincide}) the latter implies that $\solmodels{\solvar}{\D{x}=\genDE{x}}{\ivr}$.
By assumption \textcircled{$\star$}, $\imodels{\Iff[T]}{\rfvar}$, which implies that $\imodels{\Iffy[T]}{\rfvar}$ by coincidence (\rref{lem:coincide}) since $y$ is fresh in $\rfvar$. This proves conjunct \textcircled{a}.
Unfolding the implication and diamond modality of conjunct \textcircled{b}, assume there is another solution $\psi_y : [0,\tau] \to \States$ from $\solvar_y(T)$ with $\solmodels{\psi_y}{\D{x}=\genDE{x}}{(\ivr \lor x = y)}$ and $\imodels{\Ifffy[\tau]}{x \neq y}$.
Note that $\psi_y(0) = \solvar_y(T)$ \emph{exactly} rather than just on $\scomplement{\{\D{x}\}}$, because both states have the same values for the differential variables.
To show the RHS of the implication in \textcircled{b}, i.e., that $\imodels{\Iffy[T]}{\ddiamond{\pevolvein{\D{x}=\genDE{x}}{\rfvar \lor x=y}}{x\neq y}}$, it suffices to show: $\solmodels{\psi_y}{\D{x}=\genDE{x}}{\rfvar}$, because $\rfvar$ propositionally implies $\rfvar \lor x=y$.
In particular, since $\psi_y$ already satisfies the requisite differential equations and $ \imodels{\Ifffy[\tau]}{x \neq y}$, it remains to show that $\psi_y$ stays in the evolution domain $\rfvar$ for its entire duration, \ie, $\imodels{\Ifffy[{[0,\tau]}]}{\rfvar}$.
Let $0 \leq \zeta \leq \tau$ and consider the concatenated solution $\Phi : [0,T+\zeta] \to \States$ defined by:
\[\Phi(t)(z) \mdefeq
  \begin{cases}
  \solvar_y(t)(z)  & t \leq T, z \in \scomplement{\{y\}} \\
  \psi_y(t-T)(z)     & t > T, z \in \scomplement{\{y\}} \\
  \iget[state]{\I}(z)        & z \in \{y\}
  \end{cases}\]
As with $\solvar$, the solution $\Phi$ is constructed to keep $y$ constant at their initial values in $\iget[state]{\I}$.
Since $\psi_y$ must uniquely extend $\solvar_y$~\cite[\S10.IX]{Walter1998}, the concatenated solution $\Phi$ is a solution starting from $\iget[state]{\I}$, solving the system $\D{x}=\genDE{x}$.
It stays in $\ivr$ for its entire duration by coincidence~(\rref{lem:coincide}) because $\imodels{\Iffy[T]}{\ivr}$ and all states satisfying $x=y$ agree with $\solvar_y(T)$ on the free variables of formula $\ivr$.
In other words, $\solmodels{\Phi}{\D{x}=\genDE{x}}{\ivr}$.
By \textcircled{$\star$}, $\imodels{\Iffff[T+\zeta]}{\rfvar}$, which implies $\imodels{\Ifff[\zeta]}{\rfvar}$ by coincidence~(\rref{lem:coincide}) and so $\imodels{\Ifffy[{[0,\zeta]}]}{\rfvar}$, as required.

\item[``$\lylpmi$'']
Assume the RHS of~\irref{RealIndIn} is true in initial state $\iget[state]{\I}$ and show the LHS.
Consider an arbitrary solution $\solvar : [0,T] \to \States$ starting from $\iget[state]{\I}$ such that $\solmodels{\solvar}{\D{x}=\genDE{x}}{\ivr}$.
To show $\imodels{\Iff[{[0,T]}]}{\rfvar}$, using the real induction principle (\rref{prop:realindR}), it suffices to show that the set of times $S \mdefeq \{\zeta : \solvar(\zeta) \in \imodel{\I}{\rfvar}\}$ is an inductive subset of $[0,T]$, i.e., it satisfies properties \textcircled{1} and \textcircled{2} in \rref{def:indsubset}.
So, assume that $[0,\zeta) \subseteq S$ for some time $0 \leq \zeta \leq T$.

The proof instantiates quantified variables $y$ on the RHS of~\irref{RealIndIn} to match the values of $x$ at $\solvar(\zeta)$.
Since $y$ is constant for the ODE, this allows properties of $\solvar(\zeta)$ to be deduced using the RHS (namely~\textcircled{a}, \textcircled{b}) by mediating between $\solvar$ and its augmentation $\solvar_y$ below.
More precisely, consider the state $\iget[state]{\I}_y$ identical to $\iget[state]{\I}$, except where the values for variables $y$ are replaced with the corresponding values of $x$ in $\solvar(\zeta)$.
Correspondingly, consider the solution $\solvar_y : [0,\zeta] \to \States$ identical to $\solvar$ but which keeps $y$ constant at those initial values in $\iget[state]{\I}_y$ rather than in $\iget[state]{\I}$:
\[
\iget[state]{\I}_y(z) \mdefeq
  \begin{cases}
  \iget[state]{\I}(z)  & z \in \scomplement{\{y\}}\\
  \solvar(\zeta)(x_i)  & z = y_i
  \end{cases}
  \qquad
  \solvar_y(t)(z) \mdefeq
  \begin{cases}
  \solvar(t)(z)  & z \in \scomplement{\{y\}}\\
  \iget[state]{\I}_y(z)    & z \in \{y\}
  \end{cases}
  \]
By construction and coincidence (\rref{lem:coincide}), $\solvar_y$ is a solution from initial state $\iget[state]{\I}_y$, solving $\solmodels{\solvar_y}{\D{x}=\genDE{x}}{\ivr}$ and $\imodels{\Iffy[\zeta]}{x=y}$.
By assumption and coincidence (\rref{lem:coincide}), $\imodels{\Iffy[{[0,\zeta)}]}{\rfvar}$.
Therefore, $\imodels{\Iffy[{[0,\zeta]}]}{\ivr \land (\rfvar \lor x=y)}$.
Unfolding the quantification, box modality and implication on the RHS yields $\imodels{\Iffy[\zeta]}{\textcircled{a} \land \textcircled{b}}$.

\begin{itemize}
\item[\textcircled{1}] By \textcircled{a}, $\imodels{\Iffy[\zeta]}{\rfvar}$ so by coincidence (\rref{lem:coincide}), $\imodels{\Iff[\zeta]}{\rfvar}$ as required for~\textcircled{1}.

\item[\textcircled{2}] Further assume that $\zeta < T$ and show $ \imodels{\Iff[{(\zeta,\zeta+\epsilon]}]}{\rfvar}$ for some $\epsilon > 0$.
Observe that since $\zeta < T$, there is a solution that extends from state $\solvar(\zeta)$, \ie, $\psi : [0,T-\zeta] \to \States$, where $\psi(\tau) \mdefeq \solvar(\tau+\zeta)$ and with $\solmodels{\psi}{\D{x}=\genDE{x}}{\ivr}$.
Construct the corresponding solution $\psi_y : [0,T-\zeta] \to \States$ that extends from state $\solvar_y(\zeta)$ and still keeps $y$ constant at $\iget[state]{\I}_y$:
\[ \psi_y(t)(z) \mdefeq
  \begin{cases}
  \psi(t)(z)  & z \in  \scomplement{\{y\}}\\
  \solvar_y(\zeta)(z)    & z \in \{y\}
  \end{cases} \]
By coincidence (\rref{lem:coincide}), $\solmodels{\psi_y}{\D{x}=\genDE{x}}{\ivr}$, so by weakening the domain constraint, $\solmodels{\psi_y}{\D{x}=\genDE{x}}{(\ivr \lor x=y)}$.
Since $\imodels{\Iffy[\zeta]}{x=y}$ by construction and the differential equation is assumed to always locally evolve (for example with \(\D{x_1}=1\)), there must be some duration $0 < \delta < T-\zeta$ (recall $T-\zeta > 0$) after which the value of $x$ has changed from its initial value held constant in $y$, i.e., $\imodels{\Ifffy[\delta]}{x\neq y}$.
The truncation \(\truncafter{\psi_y}{\delta}\) witnesses the LHS of the implication in \textcircled{b} with: $\imodels{\Iffy[\zeta]}{\ddiamond{\pevolvein{\D{x}=\genDE{x}}{\ivr \lor x=y}}{x\neq y}}$.
Using this with the implication in \textcircled{b} yields
$ \imodels{\Iffy[\zeta]}{\ddiamond{\pevolvein{\D{x}=\genDE{x}}{\rfvar \lor x=y}}{x\neq y}}$.
Unfolding the semantics, this gives a solution which, by uniqueness, is a truncation $\truncafter{\psi_y}{\epsilon}$ of $\psi_y$, for some $\epsilon > 0$, that satisfies $\imodels{\Ifffye[{[0,\epsilon]}]}{\rfvar \lor x=y}$.
From \textcircled{a} and coincidence (\rref{lem:coincide}), all states satisfying $x=y$ agree with $\solvar(\zeta)$ on the free variables of formula $\rfvar$ thus $\imodels{\Ifffye[{[0,\epsilon]}]}{\rfvar}$.
By construction, $\truncafter{\psi_y}{\epsilon}(\tau)$ coincides with $\solvar(\tau+\zeta)$ on $x$ for all $0 \leq \tau \leq \epsilon$, which implies $\imodels{\Iff[{(\zeta,\zeta+\epsilon]}]}{\rfvar}$ by \rref{lem:coincide}. \qedhere
\end{itemize}
\end{enumerate}
\end{proof}

Conjunct~\textcircled{b} of~\irref{RealIndIn} can also be written as $\ddiamond{\pevolvein{\D{x}=\genDE{x}}{\ivr}}{x\neq y} \limply \ddiamond{\pevolvein{\D{x}=\genDE{x}}{\rfvar}}{x\neq y}$ because $\ivr$ and $\rfvar$ can be assumed true in the context where the conjunct appears.
This flexibility can be seen from its soundness proof above and will be made explicit syntactically in~\rref{cor:bigsmallequiv}.

\subsection{Derived Rules and Axioms}
\label{app:diaaxioms}
This section derives additional \dL axioms and proof rules that are used in subsequent derivations in the appendices.
Most of these are already derived elsewhere~\cite{DBLP:conf/lics/Platzer12b,DBLP:journals/jar/Platzer17} so their proofs are omitted.

\subsubsection{Basic Rules and Axioms}

This section presents basic \dL axioms and proof rules and derived versions that do not rely on the axiomatic extension of~\rref{sec:extaxioms} but are used in the appendices.

\begin{theorem}[Base axioms and proof rules~\cite{DBLP:conf/lics/Platzer12b,DBLP:journals/jar/Platzer17}]
\label{thm:baseeqaxioms}
The following are sound axioms and proof rules of \dL.

\begin{calculuscollection}
\begin{calculus}
\cinferenceRule[diamond|$\didia{\cdot}$]{diamond axiom}
{\linferenceRule[equiv]
  {\lnot\dbox{\alpha}{\lnot \fvarA}}
  {\axkey{\ddiamond{\alpha}{\fvarA}}}
}
{}
\cinferenceRule[K|K]{K axiom / modal modus ponens} %
{\linferenceRule[impl]
  {\dbox{\alpha}{(\fvarA \limply \fvarB)}}
  {(\dbox{\alpha}{\fvarA}\limply\axkey{\dbox{\alpha}{\fvarB}})}
}{}
\end{calculus}
\qquad
\begin{calculus}
\cinferenceRule[G|G]{$\dbox{}{}$ generalization} %
{\linferenceRule[formula]
  {\lsequent{}{\fvarA}}
  {\lsequent{\Gamma}{\dbox{\alpha}{\fvarA}}}
}{}
\cinferenceRule[V|V]{vacuous $\dbox{}{}$}
 {\linferenceRule[impl]
   {\fvarA}
   {\axkey{\dbox{\alpha}{\fvarA}}}
 }{\text{no free variable of $\fvarA$ is bound by $\alpha$}}
\end{calculus}\\
\begin{calculus}
\cinferenceRule[dBarcan|B$'$]{}
{\linferenceRule[equiv]
  {\lexists{y}{\ddiamond{\pevolvein{\D{x}=\genDE{x}}{\ivr\argx}}{\rfvar(x,y)}}}
  {\axkey{\ddiamond{\pevolvein{\D{x}=\genDE{x}}{\ivr\argx}}{\exists{y}\rfvar(x,y)}}}
}{\text{$y \not\in x$}}

\cinferenceRule[DW|DW]{}
{\axkey{\dbox{\pevolvein{\D{x}=\genDE{x}}{\ivr}}{\ivr}}
}{}

\cinferenceRule[DX|DX]{}
{\linferenceRule[equiv]
  {(\ivr \limply \rfvar \land \dbox{\pevolvein{\D{x}=\genDE{x}}{\ivr}}{\rfvar})}
  {\axkey{\dbox{\pevolvein{\D{x}=\genDE{x}}{\ivr}}{\rfvar}}}
  \hfill
}{\text{$\D{x} \not\in \rfvar,\ivr$}}

\cinferenceRule[DMP|DMP]{differential modus ponens}
{\linferenceRule[impl]
  {\dbox{\pevolvein{\D{x}=\genDE{x}}{\ivr}}{(\ivr\limply \rrfvar)}}
  {(\dbox{\pevolvein{\D{x}=\genDE{x}}{\rrfvar}}{\rfvar} \limply \axkey{\dbox{\pevolvein{\D{x}=\genDE{x}}{\ivr}}{\rfvar}})}
}{}

\end{calculus}
\end{calculuscollection}
\end{theorem}

The first three axioms~\irref{diamond+K+V} and proof rule~\irref{G} are standard reasoning principles for dynamic logics~\cite{DBLP:conf/lics/Platzer12b}.
They apply generally for any hybrid program $\alpha$.
Axiom~\irref{diamond} expresses the duality between the diamond and box modalities, allowing conversion between the two with a double negation.
Kripke axiom~\irref{K} is the modal modus ponens for postconditions of the box modality.
Vacuous axiom~\irref{V} says if no free variable of $\fvarA$ is changed by hybrid program $\alpha$, then the truth value of $\fvarA$ is also unchanged.
The G\"odel generalization rule~\irref{G} reduces proofs of $\dbox{\alpha}{\fvarA}$ to validity of $\fvarA$ but must discard all assumptions in antecedent $\Gamma$ for soundness.

The ODE Barcan axiom~\irref{dBarcan} specializes the Barcan axiom~\cite{DBLP:conf/lics/Platzer12b} to ODEs in the diamond modality, allowing an existential quantifier $\lexists{y}{}$ to be commuted with the diamond modality.
The quantified variables $y$ are required to be fresh in the ODE $\D{x}=\genDE{x}$ (i.e., $y \not\in x$).
The \emph{differential weakening} axiom~\irref{DW} expresses that domain constraints are always obeyed along ODE solutions.
It underlies the derived proof rules~\irref{dW+MbW}, which ease manipulation of domain constraints.
The \emph{differential skip} axiom~\irref{DX} expresses a reflexivity property of differential equation solutions.
If domain constraint $\ivr$ is false in an initial state $\iget[state]{\I}$, then the formula $\dbox{\pevolvein{\D{x}=\genDE{x}}{\ivr}}{\rfvar}$ is trivially true in $\iget[state]{\I}$ because no solution of the ODE starting from $\iget[state]{\I}$ stays in the domain constraint.
Conversely, if $\ivr$ is true in $\iget[state]{\I}$, then the postcondition $\rfvar$ must already be true in $\iget[state]{\I}$ because of the trivial solution of duration zero.
The condition $\D{x} \not\in \rfvar,\ivr$ of axiom~\irref{DX} is met as $\rfvar,\ivr$ are differential-free (\rref{subsec:background-syntax}).
Axiom~\irref{DMP} is the modus ponens principle for domain constraints of ODEs which underlies differential cuts~\irref{DC}~\cite{DBLP:journals/jar/Platzer17}.

\begin{proof}[Proof of~\rref{thm:baseeqaxioms}]
The soundness of all axioms and proof rules in~\rref{thm:baseeqaxioms} are proved elsewhere~\cite{DBLP:journals/jar/Platzer17}, except~\irref{DX+DMP}, which are proved here since they are written differently elsewhere.

\begin{description}
\item[\irref{DX}] Let $\iget[state]{\I}$ be an initial state. Classically, either $\imodels{\I}{\ivr}$ or not.
If $\imodels{\I}{\ivr}$, then, propositionally, it suffices to assume $\imodels{\I}{\dbox{\pevolvein{\D{x}=\genDE{x}}{\ivr}}{\rfvar}}$ and show $\imodels{\I}{\rfvar}$.
Since $\imodels{\I}{\ivr}$, there is a trivial solution $\solvar : [0,0] \to \States$ where $\solmodels{\solvar}{\D{x}=\genDE{x}}{\ivr}$ and $\solvar(0) = \omega$ on $\scomplement{\{\D{x}\}}$.
By assumption, $\imodels{\Iff[0]}{\rfvar}$.
Since $\D{x}\notin\rfvar$, coincidence (\rref{lem:coincide}) implies $\imodels{\I}{\rfvar}$.
Conversely, if $\imodels{\I}{\lnot{\ivr}}$, then, propositionally, it suffices to show $\imodels{\I}{\dbox{\pevolvein{\D{x}=\genDE{x}}{\ivr}}{\rfvar}}$.
The box modality is vacuous because, by definition, no solution $\solvar : [0,T] \to \States$ can exist for any $T \geq 0$ with $\solmodels{\solvar}{\D{x}=\genDE{x}}{\ivr}$.
Any such solution would require $\imodels{\Iff[0]}{\ivr}$ by definition.
However, because $\D{x}\notin\ivr$, coincidence (\rref{lem:coincide}) with state $\iget[state]{\I}$ gives $\imodels{\I}{\ivr}$, contradiction.

\item[\irref{DMP}] Let $\iget[state]{\I}$ be an initial state satisfying both formulas on the left of the implications in~\irref{DMP}, i.e., \textcircled{1} $\imodels{\I}{\dbox{\pevolvein{\D{x}=\genDE{x}}{\ivr}}{(\ivr \limply \rrfvar)}}$ and \textcircled{2} $\imodels{\I}{\dbox{\pevolvein{\D{x}=\genDE{x}}{\rrfvar}}{\rfvar}}$.
Consider any solution $\solvar : [0,T] \to \States$ where $\solvar(0) = \omega$ on $\scomplement{\{\D{x}\}}$, and $\solmodels{\solvar}{\D{x}=\genDE{x}}{\ivr}$.
By definition, $\imodels{\Iff[\zeta]}{\ivr}$ for all $\zeta \in [0,T]$, and so by \textcircled{1}, $\imodels{\Iff[\zeta]}{\ivr \limply \rrfvar}$ for all $\zeta \in [0,T]$.
Therefore, $\imodels{\Iff[\zeta]}{\rrfvar}$ for all $\zeta \in [0,T]$, and thus $\solmodels{\solvar}{\D{x}=\genDE{x}}{\rrfvar}$.
By \textcircled{2}, $\imodels{\Iff[T]}{\rfvar}$ as required.
\qedhere
\end{description}
\end{proof}

The axioms and proof rules of~\rref{thm:baseeqaxioms} combine to derive further useful axioms and proof rules.
Diamond Kripke axiom \irref{Kd} derives from \irref{K} by dualizing its inner implication with \irref{diamond} \cite{DBLP:conf/lics/Platzer12b}.
\[
\dinferenceRule[Kd|K${\didia{\cdot}}$]{}
{\linferenceRule[impl]
  {\dbox{\alpha}{(\fvarA \limply \fvarB)}}
  {(\ddiamond{\alpha}{\fvarA} \limply \axkey{\ddiamond{\alpha}{\fvarB}})}
}{}
\]

Axiom~\irref{diamond} also yields dual readings for the ODE axioms and proof rules.
For example, the~\irref{DIgeq} axiom internalizes the mean value theorem~\cite[Appendix B.I]{Walter1998}.
\begin{corollary}[Mean value theorem]
\label{cor:meanvalue}
The mean value theorem axiom~\irref{MVT} derives from \irref{DIgeq}:
\[
\dinferenceRule[MVT|MVT]{}
{\linferenceRule[impl]
  {\etermA \geq 0 \land \ddiamond{\pevolvein{\D{x}=\genDE{x}}{\ivr}}{\etermA < 0}}
  {\ddiamond{\pevolvein{\D{x}=\genDE{x}}{\ivr}}{\der{\etermA}<0}}
}{}
\]
\end{corollary}

\begin{proof}
The derivation takes contrapositives (dualizing with \irref{diamond}) before~\irref{DIgeq} finishes it.
{\footnotesizeoff%
\begin{sequentdeduction}[array]
\linfer[diamond+notl+notr]{
\linfer[DIgeq]{
  \lclose
}
  {\lsequent{\etermA \geq 0, \dbox{\pevolvein{\D{x}=\genDE{x}}{\ivr}}{\der{\etermA}\geq0}}{\dbox{\pevolvein{\D{x}=\genDE{x}}{\ivr}}{\etermA\geq0}}}
}
  {\lsequent{\etermA \geq 0, \ddiamond{\pevolvein{\D{x}=\genDE{x}}{\ivr}}{\etermA < 0}}{\ddiamond{\pevolvein{\D{x}=\genDE{x}}{\ivr}}{\der{\etermA}<0}}}
\\[-\normalbaselineskip]\tag*{\qedhere}
\end{sequentdeduction}
}%
\end{proof}

Axiom~\irref{V} is particularly useful when working with \emph{constant assumptions}.
If formula $\rrfvar(y)$ is true initially and $y$ has no differential equation in $\D{x}=\genDE{x}$, then it continues to be true along solutions to the differential equations from the initial state because $y$ remains constant along these solutions and the truth value of $\rrfvar(y)$ depends only on the value of its free variables $y$~\cite{DBLP:journals/jar/Platzer17}.
Axiom \irref{V} proves this for box modalities in succedents and, by duality, for diamond modalities in antecedents, e.g.:
{\footnotesizeoff%
\begin{sequentdeduction}[array]
\linfer[DC]{
  \linfer[V]{
    \lclose
  }
  {\lsequent{\rrfvar(y)}{\dbox{\pevolvein{\D{x}=\genDE{x}}{\ivr}}{\rrfvar(y)}}} !
  \lsequent{\Gamma}{\dbox{\pevolvein{\D{x}=\genDE{x}}{\ivr \land \rrfvar(y)}}{\rfvar}}
}
  {\lsequent{\Gamma,\rrfvar(y)}{\dbox{\pevolvein{\D{x}=\genDE{x}}{\ivr}}{\rfvar}}}
\end{sequentdeduction}
}%

Conversely, if a constant assumption $\rrfvar(y)$ is true in a final state reachable by an ODE $\D{x}=\genDE{x}$, then it must already be true initially.
This is shown formally by the derivation below which uses a classical case split with a \irref{cut} on whether the formula $\rrfvar(y)$ is already true initially:
{\footnotesizeoff%
\begin{sequentdeduction}[array]
\linfer[cut]{
\linfer[orl]{
  \linfer[]{
    \lclose
  }
  {\lsequent{\Gamma, \rrfvar(y), \ddiamond{\pevolvein{\D{x}=\genDE{x}}{\ivr}}{(\rfvar \land \rrfvar(y))}}{\rrfvar(y)}}
    !
    \lsequent{\Gamma, \lnot{\rrfvar(y)}, \ddiamond{\pevolvein{\D{x}=\genDE{x}}{\ivr}}{(\rfvar \land \rrfvar(y))}}{\rrfvar(y)}
}
  {\lsequent{\Gamma, \rrfvar(y) \lor \lnot{\rrfvar(y)}, \ddiamond{\pevolvein{\D{x}=\genDE{x}}{\ivr}}{(\rfvar \land \rrfvar(y))}}{\rrfvar(y)}}
}
  {\lsequent{\Gamma,\ddiamond{\pevolvein{\D{x}=\genDE{x}}{\ivr}}{(\rfvar \land \rrfvar(y))}}{\rrfvar(y)}}
\end{sequentdeduction}
}%

The left premise closes trivially. For the right premise, a contradiction is derived with \irref{diamond} as follows, where the~\irref{V+MbW} step uses the propositional tautology $\lnot{\rrfvar(y)} \limply \lnot{(\rfvar \land \rrfvar(y))}$:
{\footnotesizeoff%
\begin{sequentdeduction}[array]
\linfer[diamond+notl]{
\linfer[V+MbW]{
  \lclose
}
  {\lsequent{\lnot{\rrfvar(y)}}{\dbox{\pevolvein{\D{x}=\genDE{x}}{\ivr}}{\lnot{(\rfvar \land \rrfvar(y))}}}}
}
  {\lsequent{\lnot{\rrfvar(y)}, \ddiamond{\pevolvein{\D{x}=\genDE{x}}{\ivr}}{(\rfvar \land \rrfvar(y))}}{\lfalse}}
\end{sequentdeduction}
}%

In the sequel, these routine steps are omitted and proof steps manipulating constant assumptions are simply labeled with \irref{V} directly.

Axiom~\irref{DMP} derives the diamond refinement axiom~\irref{dDR} and its corollary~\irref{gddR} from \rref{cor:diadiffeqax}, which provide tools for working with ODE domain constraints in the diamond modality:
\begin{proof}[Proof of \rref{cor:diadiffeqax}]
Axiom \irref{dDR} derives from \irref{DMP} (the roles of $\ivr$ and $\rrfvar$ are flipped) by dualizing with the \irref{diamond} axiom.
The last~\irref{K+dW} step uses the propositional tautology $\ivr \limply (\rrfvar \limply \ivr)$.
{\footnotesizeoff%
\begin{sequentdeduction}[array]
\linfer[diamond+notr+notl]{
\linfer[DMP]{
\linfer[K+dW]{
  \lclose
}
  {\lsequent{\dbox{\pevolvein{\D{x}=\genDE{x}}{\rrfvar}}{\ivr}}{\dbox{\pevolvein{\D{x}=\genDE{x}}{\rrfvar}}{(\rrfvar \limply \ivr)}}}
}
  {\lsequent{\dbox{\pevolvein{\D{x}=\genDE{x}}{\rrfvar}}{\ivr},\dbox{\pevolvein{\D{x}=\genDE{x}}{\ivr}}{\lnot{\rfvar}}}{\dbox{\pevolvein{\D{x}=\genDE{x}}{\rrfvar}}{\lnot{\rfvar}}}}
}
  {\lsequent{\dbox{\pevolvein{\D{x}=\genDE{x}}{\rrfvar}}{\ivr},\ddiamond{\pevolvein{\D{x}=\genDE{x}}{\rrfvar}}{\rfvar}}{\ddiamond{\pevolvein{\D{x}=\genDE{x}}{\ivr}}{\rfvar}}}
\end{sequentdeduction}
}%

Rule \irref{gddR} derives from \irref{dDR} by simplifying its left premise with rule \irref{dW}.
\end{proof}

\subsubsection{Extended Derived Rules and Axioms}

This section derives additional rules and axioms that make use of the axiomatic extensions from~\rref{sec:extaxioms}.

\paragraph{Local Progress Properties}
The local progress modality $\ddnext$ excludes the initial state ($x=y$) in the domain constraint when expressing local progress for formula $\ivr$.
Recall:
\[ \dprogressin{\D{x}=\genDE{x}}{\ivr}{} \mdefequiv \ddiamond{\pevolvein{\D{x}=\genDE{x}}{\ivr \lor x=y}}{\,x\neq y} \]

The disjunct $x=y$ in the domain constraint makes local progress an interesting question for formulas characterizing sets that are not topologically closed (e.g., open sets as characterized by the formula $\etermA > 0$).
As axiom~\irref{Cont} shows, the formula $\ddiamond{\pevolvein{\D{x}=\genDE{x}}{\etermA > 0}}{x \neq y}$ which \emph{does not} exclude $x=y$ in the evolution domain constraint is already \emph{equivalent} to $\etermA > 0$.
A precise syntactic characterization of this difference is shown by the following derived axiom:

\begin{corollary}[Initial state inclusion]
\label{cor:bigsmallequiv}
The following axiom derives in \dL.
Variables $y$ are fresh in the ODE $\D{x}=\genDE{x}$ and formula $\ivr$.
\[\dinferenceRule[bigsmallequiv|Init]{}
{
\linferenceRule[impl]
  {\initassum}
  {\big(\axkey{\ddiamond{\pevolvein{\D{x}=\genDE{x}}{\ivr}}{x \neq y}} \lbisubjunct \ivr \land \dprogressin{\D{x}=\genDE{x}}{\ivr}\big)}
}
{}\]
\end{corollary}
\begin{proof}
First, by dualizing via \irref{diamond} both sides of axiom~\irref{DX}, the following equivalence is derived:
\[
  \ddiamond{\pevolvein{\D{x}=\genDE{x}}{\ivr}}{\rfvar}
  \lbisubjunct
  (\ivr \land (\rfvar \lor \ddiamond{\pevolvein{\D{x}=\genDE{x}}{\ivr}}{\rfvar}))
 \]

The derivation of~\irref{bigsmallequiv} starts by using this derived equivalence (\irref{DX+diamond}), followed by a series of equivalent propositional rewrites that simplify the logical structure of the succedent.
The propositional steps are shown below, first removing the disjunct $x \neq y$ using the assumption $\initassum$, and then pulling out the common conjunct $\ivr$ as an antecedent assumption.
{\footnotesizeoff%
\begin{sequentdeduction}[array]
  \linfer[DX+diamond]{
  \linfer[]{
  \linfer[]{
    \lsequent{\initassum, \ivr}{\ddiamond{\pevolvein{\D{x}=\genDE{x}}{\ivr}}{x \neq y} \lbisubjunct \dprogressin{\D{x}=\genDE{x}}{\ivr}}
  }
    {\lsequent{\initassum}{\ivr \land \ddiamond{\pevolvein{\D{x}=\genDE{x}}{\ivr}}{x \neq y} \lbisubjunct \ivr \land \dprogressin{\D{x}=\genDE{x}}{\ivr}}}
  }
    {\lsequent{\initassum}{\ivr \land (x \neq y \lor \ddiamond{\pevolvein{\D{x}=\genDE{x}}{\ivr}}{x \neq y}) \lbisubjunct \ivr \land \dprogressin{\D{x}=\genDE{x}}{\ivr}}}
  }
  {\lsequent{\initassum}{\ddiamond{\pevolvein{\D{x}=\genDE{x}}{\ivr}}{x \neq y} \lbisubjunct \ivr \land \dprogressin{\D{x}=\genDE{x}}{\ivr}}}
\end{sequentdeduction}
}%

Both directions of the resulting equivalence are proved separately by unfolding the abbreviation $\ddnext$.
In the ``$\limply$'' direction, a \irref{gddR} step suffices, because $\ivr \limply \ivr \lor \initassum$ is a propositional tautology:
{\footnotesizeoff%
\begin{sequentdeduction}[array]
  \linfer[gddR]{
    \lclose
  }
  {\lsequent{\initassum,\ivr,\ddiamond{\pevolvein{\D{x}=\genDE{x}}{\ivr}}{x\neq y}}{\ddiamond{\pevolvein{\D{x}=\genDE{x}}{\ivr \lor \initassum}}{x \neq y}}}
\end{sequentdeduction}
}%

In the ``$\lylpmi$'' direction, the derivation starts with a \irref{dDR} step which reduces to the box modality.
Since the formulas $x=y$ and $\ivr$ are true initially, a~\irref{V+DC} step introduces the constant assumption $\ivr(y)$ into the domain constraint, which is $\ivr$ with $y$ in place of $x$.
The derivation closes with \irref{dW}.
{\footnotesizeoff%
\begin{sequentdeduction}[array]
  \linfer[dDR]{
  \linfer[V+DC]{
  \linfer[dW]{
    \lclose
  }
    \lsequent{}{\dbox{\pevolvein{\D{x}=\genDE{x}}{(\ivr \lor \initassum) \land \ivr(y)}}{\ivr}}
  }
    {\lsequent{\initassum,\ivr}{\dbox{\pevolvein{\D{x}=\genDE{x}}{\ivr \lor \initassum}}{\ivr}}}
  }
  {\lsequent{\initassum,\ivr,\ddiamond{\pevolvein{\D{x}=\genDE{x}}{\ivr \lor \initassum}}{x \neq y}}{\ddiamond{\pevolvein{\D{x}=\genDE{x}}{\ivr}}{x\neq y}}}
\\[-\normalbaselineskip]\tag*{\qedhere}
\end{sequentdeduction}
}%
\end{proof}

It is not possible to locally progress into both formula $\rfvar$ and its negation $\lnot\rfvar$ simultaneously, by uniqueness.
This is the ``$\limply$'' direction of the duality axiom~\irref{duality} for local progress from \rref{cor:localprogresscompletedualcong}.
The converse ``$\lylpmi$'' is more involved and relies on the characterization axiom~\irref{Lpiff} later.

\begin{corollary}[Local progress duality ``$\limply$'']
\label{cor:dualityimp}
The following axiom derives from~\irref{Uniq}.
Variables $y$ are fresh in the ODE $\D{x}=\genDE{x}$ and formula $\rfvar$.
\[
\dinferenceRule[dualityimp|$\lnot{\ddnext}_\limply$]{Duality}
{\linferenceRule[impl]
  {\initassum}
  {\big( \axkey{\dprogressin{\D{x}=\genDE{x}}{\rfvar}} \limply \lnot{\dprogressin{\D{x}=\genDE{x}}{\lnot{\rfvar}}} \big)}
}{}
\]
\end{corollary}
\begin{proof}
The derivation starts with~\irref{notr}, after which the resulting local progress antecedents are combined by \irref{Uniq}, giving a conjunction of their domain constraints because $(\rfvar \lor \initassum) \land (\lnot{\rfvar}\lor \initassum)$ is propositionally equivalent to $(\rfvar \land \lnot{\rfvar}) \lor \initassum$.
The conjunction $\rfvar \land \lnot{\rfvar}$ in the domain constraint is propositionally equivalent to $\lfalse$ and no local progress is possible into an empty set of states.
{\footnotesizeoff%
\begin{sequentdeduction}[array]
\linfer[notr]{
\linfer[Uniq]{
  \lsequent{\dprogressin{\D{x}=\genDE{x}}{\rfvar \land \lnot{\rfvar}}}{\lfalse}
}
  {\lsequent{\dprogressin{\D{x}=\genDE{x}}{\rfvar},\dprogressin{\D{x}=\genDE{x}}{\lnot{\rfvar}}}{\lfalse}}
}
  {\lsequent{\dprogressin{\D{x}=\genDE{x}}{\rfvar}}{\lnot{\dprogressin{\D{x}=\genDE{x}}{\lnot{\rfvar}}}}}
\end{sequentdeduction}
}%

The derivation is completed by unfolding the $\ddnext$ syntactic abbreviation, and shifting to the box modality by \irref{diamond} duality.
The final step after using~\irref{dW} is a propositional tautology:
{\footnotesizeoff%
\begin{sequentdeduction}[array]
\linfer[]{
\linfer[diamond+notl]{
\linfer[dW]{
\linfer{%
  \lclose
}
  {\lsequent{(\rfvar \land \lnot{\rfvar}) \lor x=y}{x=y}}
}
  {\lsequent{}{\dbox{\pevolvein{\D{x}=\genDE{x}}{(\rfvar \land \lnot{\rfvar}) \lor x=y}}{x=y}}}
}
  {\lsequent{\ddiamond{\pevolvein{\D{x}=\genDE{x}}{(\rfvar \land \lnot{\rfvar}) \lor x=y}}{x\not=y}}{\lfalse}}
}
  {\lsequent{\dprogressin{\D{x}=\genDE{x}}{\rfvar \land \lnot{\rfvar}}}{\lfalse}}
\\[-\normalbaselineskip]\tag*{\qedhere}
\end{sequentdeduction}
}%
\end{proof}

\paragraph{Reflection}
The next two derived axioms~\irref{diareflect} and~\irref{reflect} internalize a mathematical property of ODE invariants.
Namely, the formula $\rfvar$ is invariant for the forwards ODE \(\D{x}=\genDE{x}\) iff its negation $\lnot\rfvar$ is invariant for the backwards ODE \(\D{x}=-\genDE{x}\).
This invariant reflection principle is used in~\rref{app:completeness} for proving completeness for semianalytic invariants, and to flip the signs in the second premise of rule \irref{realind}.
It is useful in its own right as it allows freely switching between proving invariance for either the forwards or backwards ODEs, e.g., if one direction yields simpler arithmetic.

\begin{corollary}[Reflection]
The reflection axioms~\irref{diareflect+reflect} derive from \irref{Dadjoint}:

\begin{calculus}
\dinferenceRule[diareflect|rfl$\didia{\cdot}$]{}
{\linferenceRule[equiv]
  {\lexists{x}{(\rrfvar(x) \land \ddiamond{\pevolvein{\D{x}=-\genDE{x}}{\ivr(x)}}{\rfvar(x)})}}
  {\axkey{\lexists{x}{(\rfvar(x) \land \ddiamond{\pevolvein{\D{x}=\genDE{x}}{\ivr(x)}}{\rrfvar(x)})}}}
}{}

\dinferenceRule[reflect|rfl]{}
{\linferenceRule[equiv]
  {\lforall{x}{\big(\lnot{\rfvar(x)} \limply \dbox{\pevolvein{\D{x}=-\genDE{x}}{\ivr(x)}}{\lnot{\rfvar(x)}}\big)}}
  {\axkey{\lforall{x}{\big(\rfvar(x) \limply \dbox{\pevolvein{\D{x}=\genDE{x}}{\ivr(x)}}{\rfvar(x)}\big)}}}
}{}
\end{calculus}
\end{corollary}
\begin{proof}
Axiom~\irref{reflect} derives from \irref{diareflect} by instantiating with $\rrfvar(x)\mdefequiv \lnot\rfvar(x)$ and negating both sides of the equivalence with \irref{diamond}.
The diamond reflection axiom~\irref{diareflect} is derived from~\irref{Dadjoint}.
Both implications are proved separately and the ``$\lylpmi$'' direction follows by instantiating the proof of the ``$\limply$'' direction, since $-(-\genDE{x}) = \genDE{x}$.
The ``$\limply$'' direction is proved below.

In the derivation below, the formulas are bound renamed~\cite{DBLP:journals/jar/Platzer17} for clarity.
After Skolemizing, the first \irref{Kd+dW} step introduces an existentially quantified $y$ under the diamond modality in the antecedent by monotonicity using the provable first-order formula $\rrfvar(x) \limply \lexists{y}{(x=y \land \rrfvar(y))}$.
{\footnotesizeoff%
\begin{sequentdeduction}[array]
\linfer[existsl+andl]{
\linfer[Kd+dW]{
  \lsequent{\rfvar(x), \ddiamond{\pevolvein{\D{x}=\genDE{x}}{\ivr(x)}}{\lexists{y}{(x=y \land \rrfvar(y))}}} {\lexists{y}{(\rrfvar(y) \land \ddiamond{\pevolvein{\D{y}=-\genDE{y}}{\ivr(y)}}{\rfvar(y)})}}
}
  {\lsequent{\rfvar(x), \ddiamond{\pevolvein{\D{x}=\genDE{x}}{\ivr(x)}}{\rrfvar(x)}} {\lexists{y}{(\rrfvar(y) \land \ddiamond{\pevolvein{\D{y}=-\genDE{y}}{\ivr(y)}}{\rfvar(y)})}}}
}
  {\lsequent{\lexists{x}{(\rfvar(x) \land \ddiamond{\pevolvein{\D{x}=\genDE{x}}{\ivr(x)}}{\rrfvar(x)})}} {\lexists{y}{(\rrfvar(y) \land \ddiamond{\pevolvein{\D{y}=-\genDE{y}}{\ivr(y)}}{\rfvar(y)})}}}
\end{sequentdeduction}
}%

The ODE Barcan \irref{dBarcan} axiom moves the existentially quantified $y$ out of the diamond modality since $y$ is not in $\D{x}=\genDE{x}$.
A subsequent \irref{V} step also moves the postcondition $\rrfvar(y)$ out from the diamond modality into the antecedents.
{\footnotesizeoff%
\begin{sequentdeduction}[array]
\linfer[dBarcan]{
\linfer[existsl]{
\linfer[V]{
  \lsequent{\rfvar(x),\rrfvar(y), \ddiamond{\pevolvein{\D{x}=\genDE{x}}{\ivr(x)}}{x=y}} {\lexists{y}{(\rrfvar(y) \land \ddiamond{\pevolvein{\D{y}=-\genDE{y}}{\ivr(y)}}{\rfvar(y)})}}
}
  {\lsequent{\rfvar(x), \ddiamond{\pevolvein{\D{x}=\genDE{x}}{\ivr(x)}}{(x=y \land \rrfvar(y))}}{\lexists{y}{(\rrfvar(y) \land \ddiamond{\pevolvein{\D{y}=-\genDE{y}}{\ivr(y)}}{\rfvar(y)})}}}
}
  {\lsequent{\rfvar(x), \lexists{y}{\ddiamond{\pevolvein{\D{x}=\genDE{x}}{\ivr(x)}}{(x=y \land \rrfvar(y))}}}{\lexists{y}{(\rrfvar(y) \land \ddiamond{\pevolvein{\D{y}=-\genDE{y}}{\ivr(y)}}{\rfvar(y)})}}}
}
  {\lsequent{\rfvar(x), \ddiamond{\pevolvein{\D{x}=\genDE{x}}{\ivr(x)}}{\lexists{y}{(x=y \land \rrfvar(y))}}} {\lexists{y}{(\rrfvar(y) \land \ddiamond{\pevolvein{\D{y}=-\genDE{y}}{\ivr(y)}}{\rfvar(y)})}}}
\end{sequentdeduction}
}%

The derivation continues using differential adjoints~\irref{Dadjoint} to syntactically flip the antecedent differential equations from evolving $x$ forwards to evolving $y$ backwards.
The \irref{V+Kd} step then strengthens the postcondition to $\rfvar(y)$ exploiting that the (negated) ODE does not modify $x$ so that $\rfvar(x)$ remains true along the ODE.
This completes the proof using $y$ as a witness for $\exists{y}$.
{\footnotesizeoff%
\begin{sequentdeduction}[array]
\linfer[Dadjoint]{
\linfer[V+Kd]{
\linfer[existsr]{
  \lclose
}
  {\lsequent{\rrfvar(y), \ddiamond{\pevolvein{\D{y}=-\genDE{y}}{\ivr(y)}}{\rfvar(y)}} {\lexists{y}{(\rrfvar(y) \land \ddiamond{\pevolvein{\D{y}=-\genDE{y}}{\ivr(y)}}{\rfvar(y)})}}}
}
  {\lsequent{\rfvar(x),\rrfvar(y), \ddiamond{\pevolvein{\D{y}=-\genDE{y}}{\ivr(y)}}{y=x}} {\lexists{y}{(\rrfvar(y) \land \ddiamond{\pevolvein{\D{y}=-\genDE{y}}{\ivr(y)}}{\rfvar(y)})}}}
}
  {\lsequent{\rfvar(x),\rrfvar(y), \ddiamond{\pevolvein{\D{x}=\genDE{x}}{\ivr(x)}}{x=y}} {\lexists{y}{(\rrfvar(y) \land \ddiamond{\pevolvein{\D{y}=-\genDE{y}}{\ivr(y)}}{\rfvar(y)})}}}
\\[-\normalbaselineskip]\tag*{\qedhere}
\end{sequentdeduction}
}%
\end{proof}

\paragraph{Real Induction Rule}
The real induction rule with domain constraints corresponding to axiom \irref{RealIndIn} is derived next.
It is stated with the $\ddnext$ modality from~\rref{sec:extaxioms}.
The real induction rule~\irref{realind} from~\rref{cor:realind} derives as an instance with domain constraint $\ivr \mnodefequiv \ltrue$.

\begin{corollary}[Real induction rule with domain constraints]
\label{cor:realindin}
The real induction proof rule~\irref{realindin} (with two stacked premises) derives from \irref{RealIndIn+Dadjoint+Uniq}.
Variables $y$ are fresh in the ODE $\D{x}=\genDE{x}$ and formulas $\rfvar,\ivr$.

\[\dinferenceRule[realindin|rI{$\&$}]{}
{\linferenceRule
  {
  \begin{aligned}
  \lsequent{\initassum,\rfvar,\ivr,\dprogressin{\D{x}=\genDE{x}}{\ivr}}{\dprogressin{\D{x}=\genDE{x}}{\rfvar}}\quad\; \\
  \lsequent{\initassum,\lnot{\rfvar},\ivr,\dprogressin{\D{x}=-\genDE{x}}{\ivr}}{\dprogressin{\D{x}=-\genDE{x}}{\lnot{\rfvar}}}
  \end{aligned}
  }
  {\lsequent{\rfvar}{\dbox{\pevolvein{\D{x}=\genDE{x}}{\ivr}}{\rfvar}}}
}{}\]
\end{corollary}
\begin{proof}[Proof (implies~\rref{cor:realind})]
The derivation starts by rewriting the succedent with \irref{RealIndIn}, the resulting right conjunct is abbreviated with $\rrfvar \mdefequiv \dprogressin{\D{x}=\genDE{x}}{\ivr} \limply \dprogressin{\D{x}=\genDE{x}}{\rfvar}$.
The \irref{MbW} step rewrites the postcondition equivalently using the propositional tautology $\rfvar \land \rrfvar \lbisubjunct \rfvar \land (\rfvar \limply \rrfvar)$ which allows the left conjunct $\rfvar$ to be assumed when proving the right conjunct $\rrfvar$ (the implication $\initassum$ is also distributed over the conjunction).
The two conjuncts are then split by~\irref{band+andr}, with the resulting two premises labeled \textcircled{1} and \textcircled{2} respectively. These are shown and proved below.
{\footnotesizeoff%
\begin{sequentdeduction}[array]
\linfer[RealIndIn]{
\linfer[allr]{
\linfer[MbW]{
\linfer[band+andr]{
  \textcircled{1} ! \textcircled{2}
}
  {\lsequent{\rfvar}{\dbox{\pevolvein{\D{x}=\genDE{x}}{\ivr \land (\rfvar \lor \initassum)}}{\big( (\initassum \limply \rfvar) \land (\initassum \land \rfvar \limply \rrfvar)\big)}}}
}
  {\lsequent{\rfvar}{\dbox{\pevolvein{\D{x}=\genDE{x}}{\ivr \land (\rfvar \lor \initassum)}}{(\initassum \limply \rfvar \land \rrfvar)}}}
}
  {\lsequent{\rfvar}{\lforall{y}{\dbox{\pevolvein{\D{x}=\genDE{x}}{\ivr \land (\rfvar \lor \initassum)}}{(\initassum \limply \rfvar \land \rrfvar)}}}}
}
  {\lsequent{\rfvar}{\dbox{\pevolvein{\D{x}=\genDE{x}}{\ivr}}{\rfvar}}}
\end{sequentdeduction}
}%

The premise~\textcircled{2} yields the top premise of rule~\irref{realindin} directly (unfolding the abbreviation for $\rrfvar$):
{\footnotesizeoff%
\begin{sequentdeduction}[array]
\linfer[dW]{
\linfer[implyr+andl]{
  \lsequent{\initassum,\rfvar,\ivr,\dprogressin{\D{x}=\genDE{x}}{\ivr}}{\dprogressin{\D{x}=\genDE{x}}{\rfvar}}
}
  {\lsequent{\ivr}{(\initassum \land \rfvar \limply \rrfvar)}}
}
  {\lsequent{\rfvar}{\dbox{\pevolvein{\D{x}=\genDE{x}}{\ivr \land (\rfvar \lor \initassum)}}{(\initassum \land \rfvar \limply \rrfvar)}}}
\end{sequentdeduction}
}%

Continuing from premise \textcircled{1}, the derivation splits classically on whether $\initassum$ is true initially, yielding two further premises labeled \textcircled{3} when $\initassum$ and \textcircled{4} when $x\neq y$.
{\footnotesizeoff%
\begin{sequentdeduction}[array]
\linfer[cut]{
\linfer[orl]{
  \textcircled{3} ! \textcircled{4}
}
  {\lsequent{\initassum \lor x\neq y,\rfvar}{\dbox{\pevolvein{\D{x}=\genDE{x}}{\ivr \land (\rfvar \lor \initassum)}}{(\initassum \limply \rfvar)}}}
}
  {\lsequent{\rfvar}{\dbox{\pevolvein{\D{x}=\genDE{x}}{\ivr \land (\rfvar \lor \initassum)}}{(\initassum \limply \rfvar)}}}
\end{sequentdeduction}
}%

From \textcircled{3}, the antecedents $\initassum$ and $\rfvar$ imply that $\rfvar(y)$ is true initially by a~\irref{cut}.
Since $y$ is held constant by the ODE $\D{x}=\genDE{x}$, a monotonicity step \irref{MbW} followed by \irref{V} completes the proof:
{\footnotesizeoff%
\begin{sequentdeduction}[array]
  \linfer[cut+MbW]{
  \linfer[V]{
    \lclose
  }
    {\lsequent{\rfvar(y)}{\dbox{\pevolvein{\D{x}=\genDE{x}}{\ivr \land (\rfvar \lor \initassum)}}{\rfvar(y)}}}
  }
  {\lsequent{\initassum,\rfvar}{\dbox{\pevolvein{\D{x}=\genDE{x}}{\ivr \land (\rfvar \lor \initassum)}}{(\initassum \limply \rfvar)}}}
\end{sequentdeduction}
}%

From \textcircled{4}, the derivation continues by dualizing to the diamond modality, and using \irref{diareflect} to syntactically reverse the direction of the ODE.
The~\irref{gddR} step weakens the resulting postcondition of the diamond modality with the propositional tautology $x {\neq} y \land \rfvar \limply x {\neq} y$.
{\footnotesizeoff%
\begin{sequentdeduction}[array]
  \linfer[diamond+notr]{
  \linfer[diareflect]{
  \linfer[gddR]{
    \lsequent{\initassum, \lnot{\rfvar}, \ddiamond{\pevolvein{\D{x}=-\genDE{x}}{\ivr \land (\rfvar \lor \initassum)}}{x {\neq} y}}{\lfalse}
  }
    {\lsequent{\initassum, \lnot{\rfvar}, \ddiamond{\pevolvein{\D{x}=-\genDE{x}}{\ivr \land (\rfvar \lor \initassum)}}{(x {\neq} y \land \rfvar)}}{\lfalse}}
  }
    {\lsequent{x{\neq} y, \rfvar, \ddiamond{\pevolvein{\D{x}=\genDE{x}}{\ivr \land (\rfvar \lor \initassum)}}{(\initassum \land \lnot{\rfvar})}}{\lfalse}}
  }
  {\lsequent{x{\neq} y,\rfvar}{\dbox{\pevolvein{\D{x}=\genDE{x}}{\ivr \land (\rfvar \lor \initassum)}}{(\initassum \limply \rfvar)}}}
\end{sequentdeduction}
}%

The diamond modality in the antecedent splits by axiom~\irref{Uniq} into two assumptions with domain constraints $\ivr$ and $\rfvar \lor \initassum$ respectively.
With a use of derived axiom~\irref{bigsmallequiv}, all the antecedents of the bottom premise of rule~\irref{realindin} are gathered, leaving only its succedent.
{\footnotesizeoff%
\begin{sequentdeduction}[array]
  \linfer[Uniq]{
  \linfer[bigsmallequiv]{
    \lsequent{\initassum, \lnot{\rfvar}, \ivr, \dprogressin{\D{x}=-\genDE{x}}{\ivr}, \dprogressin{\D{x}=-\genDE{x}}{\rfvar}}{\lfalse}
  }
    {\lsequent{\initassum, \lnot{\rfvar}, \ddiamond{\pevolvein{\D{x}=-\genDE{x}}{\ivr}}{x \neq y}, \ddiamond{\pevolvein{\D{x}=-\genDE{x}}{\rfvar \lor \initassum}}{x \neq y}}{\lfalse}}
  }
    {\lsequent{\initassum, \lnot{\rfvar}, \ddiamond{\pevolvein{\D{x}=-\genDE{x}}{\ivr \land (\rfvar \lor \initassum)}}{x \neq y}}{\lfalse}}
\end{sequentdeduction}
}%

Continuing with the derived implication~\irref{dualityimp} results in the bottom premise of rule~\irref{realindin}:
{\footnotesizeoff%
\begin{sequentdeduction}[array]
  \linfer[dualityimp]{
  \linfer[notl]{
    \lsequent{\initassum, \lnot{\rfvar}, \ivr, \dprogressin{\D{x}=-\genDE{x}}{\ivr}}{\dprogressin{\D{x}=-\genDE{x}}{\lnot{\rfvar}}}
  }
    {\lsequent{\initassum, \lnot{\rfvar}, \ivr, \dprogressin{\D{x}=-\genDE{x}}{\ivr}, \lnot{\dprogressin{\D{x}=-\genDE{x}}{\lnot{\rfvar}}}}{\lfalse}}
  }
    {\lsequent{\initassum, \lnot{\rfvar}, \ivr, \dprogressin{\D{x}=-\genDE{x}}{\ivr}, \dprogressin{\D{x}=-\genDE{x}}{\rfvar}}{\lfalse}}
\\[-\normalbaselineskip]\tag*{\qedhere}
\end{sequentdeduction}
}%
\end{proof}

The rule \irref{realindin} discards any additional context in the antecedents of its premises.
This is due to the use of \irref{RealIndIn} which focuses on particular states along trajectories of the ODE $\D{x}=\genDE{x}$.
It would be unsound to keep any assumptions about the initial state that depend on $x$ because the state being examined may not be the initial state!
On the other hand, assumptions that do not depend on $x$ remain true along the ODE.
These constant assumptions can be kept with uses of \irref{V} throughout the derivation above or added into $\ivr$ before using \irref{realindin} by a \irref{DC} that proves with \irref{V}.
Such additional steps are elided and rule \irref{realindin} is used directly while keeping these \emph{constant} assumptions around.

\subsection{Axioms and Proof Rules Summary}
\label{app:axiomssummary}

For the sake of a self-contained presentation, this section presents all the propositional and first-order sequent calculus proof rules used in this article.
Propositional and first-order logic is not the focus of this article so a simplified treatment is used.
A full treatment of these rules can be found elsewhere~\cite{DBLP:journals/jar/Platzer08}.
For ease of reference, this section also compactly references all base (and extended) \dL axioms and proof rules but does not repeat them explicitly because of space considerations.
All other (derived) axioms and proof rules in this article derive from this \dL axiomatization.

\paragraph{Propositional and First-Order Proof Rules}
Rule~\irref{qear} for first-order real arithmetic is explained in~\rref{subsec:background-axiomatizaton}.
The following are sound propositional and first-order sequent calculus proof rules.
In rules~\irref{alll+existsr}, an arbitrary (extended) term $\etermA$ can be used to instantiate the respective quantifiers.
In rules~\irref{existsl+allr}, the Skolem variable $y$ is fresh, i.e., does not occur free, in the conclusion:

\begin{calculuscollection}
\begin{calculus}
  \cinferenceRule[notl|$\lnot$\leftrule]{not left}
  {\linferenceRule
    {\lsequent{\Gamma}{\fvarA}}
    {\lsequent{\Gamma,\lnot{\fvarA}}{\lfalse}}
  }{}

  \cinferenceRule[andl|$\land$\leftrule]{and left}
  {\linferenceRule
    {\lsequent{\Gamma,\fvarA_1,\fvarA_2}{\fvarB}}
    {\lsequent{\Gamma,\fvarA_1 \land \fvarA_2}{\fvarB}}
  }{}

  \cinferenceRule[orl|$\lor$\leftrule]{or left}
  {\linferenceRule
    {\lsequent{\Gamma,\fvarA_1}{\fvarB} & \lsequent{\Gamma,\fvarA_2}{\fvarB}}
    {\lsequent{\Gamma,\fvarA_1 \lor \fvarA_2}{\fvarB}}
  }{}
\end{calculus}
\quad
\begin{calculus}
  \cinferenceRule[notr|$\lnot$\rightrule]{not right}
  {\linferenceRule
    {\lsequent{\Gamma,\fvarA}{\lfalse}}
    {\lsequent{\Gamma}{\lnot{\fvarA}}}
  }{}

  \cinferenceRule[andr|$\land$\rightrule]{and right}
  {\linferenceRule
    {\lsequent{\Gamma}{\fvarA_1} & \lsequent{\Gamma}{\fvarA_2}}
    {\lsequent{\Gamma}{\fvarA_1 \land \fvarA_2}}
  }{}

  \cinferenceRule[cut|cut]{cut}
  {\linferenceRule
    {\lsequent{\Gamma}{\fvarB} & \lsequent{\Gamma,\fvarB}{\fvarA}}
    {\lsequent{\Gamma}{\fvarA}}
  }{}
\end{calculus}
\quad
\begin{calculus}
  \cinferenceRule[implyl|$\limply$\leftrule]{imply left}
  {\linferenceRule
    {\lsequent{\Gamma}{\fvarA_1} & \lsequent{\Gamma,\fvarA_2}{\fvarB}}
    {\lsequent{\Gamma,\fvarA_1 \limply \fvarA_2}{\fvarB}}
  }{}

  \cinferenceRule[alll|$\forall$\leftrule]{forall left}
  {\linferenceRule
    {\lsequent{\Gamma,\fvarA(\etermA)}{\fvarB}}
    {\lsequent{\Gamma,\lforall{x}{\fvarA(x)}}{\fvarB}}
  }{}

  \cinferenceRule[existsl|$\exists$\leftrule]{exists left}
  {\linferenceRule
    {\lsequent{\Gamma,\fvarA(y)}{\fvarB}}
    {\lsequent{\Gamma,\lexists{x}{\fvarA(x)}}{\fvarB}}
  }{}
\end{calculus}
\quad
\begin{calculus}
  \cinferenceRule[implyr|$\limply$\rightrule]{imply right}
  {\linferenceRule
    {\lsequent{\Gamma,\fvarA_1}{\fvarA_2}}
    {\lsequent{\Gamma}{\fvarA_1 \limply \fvarA_2}}
  }{}

  \cinferenceRule[allr|$\forall$\rightrule]{forall right}
  {\linferenceRule
    {\lsequent{\Gamma}{\fvarA(y)}}
    {\lsequent{\Gamma}{\lforall{x}{\fvarA(x)}}}
  }{}

  \cinferenceRule[existsr|$\exists$\rightrule]{exists right}
  {\linferenceRule
    {\lsequent{\Gamma}{\fvarA(\etermA)}}
    {\lsequent{\Gamma}{\lexists{x}{\fvarA(x)}}}
  }{}
\end{calculus}
\end{calculuscollection}

\paragraph{Base Axioms and Proof Rules}
All base axioms and proof rules of \dL are in~\rref{thm:baseeqaxioms}.
The base differential axioms are in~\rref{thm:diffaxioms}, while~\rref{lem:noefdifferentials} provides differential axioms for extended terms.
The base differential equation axioms~\irref{DI+DC+DG} are in~\rref{thm:diffeqax}.

\paragraph{Extended Axiomatization}
The extended axioms~\irref{Uniq+Cont+Dadjoint} are in~\rref{lem:uniqcont}.
The extended real induction axiom~\irref{RealInd} and its generalization~\irref{RealIndIn} are in Lemmas~\ref{lem:realindODE} and~\ref{lem:realindODEin} respectively.

\paragraph{Hybrid Program Axioms} These are presented in~\rref{thm:basehpaxioms} for use only in~\rref{app:hybridprograms}.

\section{Completeness}
\label{app:completeness}

This appendix gives completeness proofs for the derived rules \irref{dRI+sAI} and the derived local progress characterization \irref{Lpiff}.
Completeness of \irref{dRI} is proved by showing that \irref{DRI} is a derived axiom.
A similar approach is taken for completeness of \irref{sAI} by showing \irref{SAI} is a derived axiom.

This syntactic approach to proving completeness of \irref{dRI} and \irref{sAI} demonstrates the versatility of the \dL calculus.
It also enables complete \emph{dis}proofs of invariance properties as opposed to just failing to apply a complete proof rule.
To conclude that invariance is \emph{disproved} after applying an algorithmic procedure (like the presentations of (semi)algebraic~\irref{dRI} and~\irref{sAI}~\cite{DBLP:conf/tacas/GhorbalP14,DBLP:conf/emsoft/LiuZZ11}), one would need to trust, in addition to soundness, that no completeness errors are present in the implementation.

Recall that axioms \irref{Cont+RealIndIn} have an additional syntactic requirement, e.g., the presence of $\D{x_1}=1$, which is assumed to be met throughout this appendix, using axiom \irref{DG} if necessary.

\subsection{Progress Formulas}
\label{app:progressformulas}

The following are useful logical rearrangements of the progress formulas for extended term $\etermA$:

\begin{proposition}[Atomic progress formula equivalences]
\label{prop:rearrangement}
Let $N$ be the rank of extended term $\etermA$.
The following are provable equivalences on the progress and differential radical formulas for $\etermA$:
\allowdisplaybreaks
\begin{align}
\sigliedgt{\genDE{x}}{\etermA} \lbisubjunct&~\etermA > 0 \lor (\etermA = 0 \land \lied[]{\genDE{x}}{\etermA} > 0) \label{eq:sigliedgtrearrangement} \lor \dots\\
&\lor \big(\etermA=0 \land \lied[]{\genDE{x}}{\etermA} = 0 \land \dots \land \lied[N-3]{\genDE{x}}{\etermA} = 0 \land \lied[N-2]{\genDE{x}}{\etermA} > 0\big)\nonumber\\
&\lor \big(\etermA=0 \land \lied[]{\genDE{x}}{\etermA} = 0 \land \dots \land \lied[N-2]{\genDE{x}}{\etermA} = 0 \land \lied[N-1]{\genDE{x}}{\etermA} > 0\big)\nonumber\\
\sigliedgeq{\genDE{x}}{\etermA} \lbisubjunct&~\etermA\geq 0 \land \big(\etermA=0 \limply \lied[]{\genDE{x}}{\etermA} \geq 0\big) \label{eq:sigliedgeqrearrangement} \land \dots \\
&\land \big(\etermA=0 \land \lied[]{\genDE{x}}{\etermA} = 0 \land \dots \land \lied[N-3]{\genDE{x}}{\etermA} = 0 \limply \lied[N-2]{\genDE{x}}{\etermA} \geq 0\big)\nonumber\\
&\land \big(\etermA=0 \land \lied[]{\genDE{x}}{\etermA} = 0 \land \dots \land \lied[N-2]{\genDE{x}}{\etermA} = 0 \limply \lied[N-1]{\genDE{x}}{\etermA} \geq 0\big)\nonumber\\
\lnot{(\sigliedgt{\genDE{x}}{\etermA})} \lbisubjunct &~\sigliedgeq{\genDE{x}}{(-\etermA)} \qquad \lnot{(\sigliedgeq{\genDE{x}}{\etermA})} \lbisubjunct ~\sigliedgt{\genDE{x}}{(-\etermA)}  \label{eq:sigliedfullrearrangement} \\
\lnot{(\sigliedzero{\genDE{x}}{\etermA})} \lbisubjunct&~\sigliedgt{\genDE{x}}{\etermA} \lor \sigliedgt{\genDE{x}}{(-\etermA)} \label{eq:sigliedzerorearrangement}
\end{align}
\end{proposition}
\begin{proof}
The equivalences are proved one at a time.
By linearity of Lie derivatives, $\lied[i]{\genDE{x}}{(-\etermA)} = -(\lied[i]{\genDE{x}}{\etermA})$ proves in real arithmetic for any $i$.
The proof also uses these provable real arithmetic equivalences:
\[ \etermA \geq 0 \lbisubjunct \etermA=0 \lor \etermA > 0 \qquad -\etermA \geq 0 \land \etermA \geq 0 \lbisubjunct \etermA=0 \qquad \lnot{(\etermA > 0)} \lbisubjunct -\etermA \geq 0 \]
\begin{itemize}
\item[\rref{eq:sigliedgtrearrangement}] This equivalence follows by real arithmetic, and simplifying with propositional rearrangement as follows (here, the remaining conjuncts of $\sigliedgt{\genDE{x}}{\etermA}$ are abbreviated to $\dots$):
\begin{align*}
\etermA \geq 0 \land \Big((\etermA = 0 \limply \lied[]{\genDE{x}}{\etermA} \geq 0) \land \dots \Big) \lbisubjunct&~\etermA > 0 \land \Big((\etermA = 0 \limply \lied[]{\genDE{x}}{\etermA} \geq 0) \land \dots \Big)\\
&\lor  \etermA=0 \land \Big((\etermA = 0 \limply \lied[]{\genDE{x}}{\etermA} \geq 0) \land \dots \Big)
\end{align*}

The first disjunct on the RHS simplifies by real arithmetic to $\etermA > 0$ since all of its implicational conjuncts contain $\etermA = 0$ on the left of an implication.
The latter disjunct simplifies to $\etermA=0 \land \Big(\lied[]{\genDE{x}}{\etermA} \geq 0 \land \dots\Big)$, yielding the provable equivalence:
\[\sigliedgt{\genDE{x}}{\etermA} \lbisubjunct \etermA > 0 \lor \etermA=0 \land \Big(\lied[]{\genDE{x}}{\etermA} \geq 0 \land \dots\Big)\]

The equivalence $\rref{eq:sigliedgtrearrangement}$ proves by iterating this expansion on the RHS of this equivalence for its nested conjuncts with higher Lie derivatives.

\item[\rref{eq:sigliedgeqrearrangement}] This equivalence proves by expanding the formula $\sigliedgeq{\genDE{x}}{\etermA}$ which yields a disjunction between $\sigliedgt{\genDE{x}}{\etermA}$ and $\sigliedzero{\genDE{x}}{\etermA}$. The latter formula is used to relax the strict inequality in the last conjunct of $\sigliedgt{\genDE{x}}{\etermA}$ to a non-strict inequality.
\item[\rref{eq:sigliedfullrearrangement}] The equivalence for $\lnot{(\sigliedgt{\genDE{x}}{\etermA})}$ follows by negating both sides of equivalence \rref{eq:sigliedgtrearrangement} and moving negations on the RHS inwards propositionally, yielding the provable equivalence:
\begin{align*}
\lnot{(\sigliedgt{\genDE{x}}{\etermA})} \lbisubjunct &\Big( \lnot{(\etermA > 0)} \land (\etermA = 0 \limply \lnot{(\lied[]{\genDE{x}}{\etermA} > 0)}) \land \dots\\
&\land \big(\etermA{=}0 \land \lied[]{\genDE{x}}{\etermA} {=} 0 \land \dots \land \lied[N-2]{\genDE{x}}{\etermA} {=} 0 \limply \lnot{(\lied[N-1]{\genDE{x}}{\etermA} > 0)}\big)\Big)
\end{align*}
The desired equivalence derives by negating the inequalities and by equivalence~\rref{eq:sigliedgeqrearrangement} for $\sigliedgeq{\genDE{x}}{(-\etermA)}$.
The equivalence for $\lnot{(\sigliedgeq{\genDE{x}}{\etermA})}$ derives by negating both sides of the equivalence for $\lnot{(\sigliedgt{\genDE{x}}{\etermA})}$, since $-(-\etermA)=\etermA$.

\item[\rref{eq:sigliedzerorearrangement}] By \rref{eq:sigliedfullrearrangement}, the following equivalence is provable:
\[\lnot{(\sigliedgt{\genDE{x}}{\etermA})} \land \lnot{(\sigliedgt{\genDE{x}}{(-\etermA)})} \lbisubjunct (\sigliedgeq{\genDE{x}}{(-\etermA)}) \land (\sigliedgeq{\genDE{x}}{\etermA}) \]
By rewriting with \rref{eq:sigliedgeqrearrangement}, the RHS of this equivalence is provably equivalent to the formula $\sigliedzero{\genDE{x}}{\etermA}$ in real arithmetic.
Negating both sides yields the provable equivalence \rref{eq:sigliedzerorearrangement}. \qedhere
\end{itemize}
\end{proof}

The provable equivalences \rref{eq:sigliedfullrearrangement} are particularly important, because they underlie the next proposition, from which the complete characterization of local progress follows:

\begin{proposition}[Negated semianalytic progress formula] \label{prop:negationrearrangement}
Let the semianalytic formula $\rfvar$ be in normal form~\rref{eq:normalform}.
Then $\lnot{\rfvar}$ can be put into normal form such that \(\lnot{(\sigliedsai{\genDE{x}}{\rfvar})} \lbisubjunct \sigliedsai{\genDE{x}}{(\lnot{\rfvar})}\) is provable:
\[\lnot{\rfvar} \equiv \lorfold_{i=0}^{N} \Big(\landfold_{j=0}^{a(i)} \etermAA_{ij} \geq 0 \land \landfold_{j=0}^{b(i)} \etermBB_{ij}> 0\Big)\]
\end{proposition}
\begin{proof}
The proof uses propositional tautologies $\lnot{(A \land B)} \lbisubjunct \lnot{A} \lor \lnot{B}$ and $\lnot{(A \lor B)} \lbisubjunct \lnot{A} \land \lnot{B}$.
Formula $\rfvar$ is negated (in normal form~\rref{eq:normalform}) and all sub-terms are negated so the inequalities have $0$ on the RHS, yielding the following provable equivalence.
The resulting RHS is abbreviated by $\fvarA$:
\[\lnot{\rfvar} \lbisubjunct \underbrace{\landfold_{i=0}^{M} \Big(\lorfold_{j=0}^{m(i)} -\etermA_{ij} > 0 \lor \lorfold_{j=0}^{n(i)} -\etermB_{ij} \geq 0 \Big)}_{\fvarA} \]
Negating both sides of the progress formula for $\rfvar$ and simplifying propositionally proves:
\[\lnot{(\sigliedsai{\genDE{x}}{P})} \lbisubjunct
\landfold_{i=0}^{M} \Big(\lorfold_{j=0}^{m(i)} \lnot{(\sigliedgeq{\genDE{x}}{\etermA_{ij}})} \lor \lorfold_{j=0}^{n(i)} \lnot{(\sigliedgt{\genDE{x}}{\etermB_{ij}})}\Big)\]
Rewriting the RHS with equivalences \rref{eq:sigliedfullrearrangement} from \rref{prop:rearrangement} yields the following provable equivalence.
The resulting RHS is abbreviated by $\fvarB$:
\[\lnot{(\sigliedsai{\genDE{x}}{P})} \lbisubjunct \underbrace{\landfold_{i=0}^{M} \Big(\lorfold_{j=0}^{m(i)} \sigliedgt{\genDE{x}}{(-\etermA_{ij})} \lor \lorfold_{j=0}^{n(i)} \sigliedgeq{\genDE{x}}{(-\etermB_{ij})}\Big)}_{\fvarB} \]
Observe that $\fvarA,\fvarB$ have the same conjunctive normal form shape.
Distribute the outer conjunction over the inner disjunctions in $\fvarA$ to obtain the following provable equivalence, whose RHS is a normal form for $\lnot{P}$ (for some indices $N,a(i),b(i)$ and extended terms $\etermAA_{ij},\etermBB_{ij}$):
\[\lnot{\rfvar} \lbisubjunct \lorfold_{i=0}^{N} \Big(\landfold_{j=0}^{a(i)} \etermAA_{ij} \geq 0 \lor \landfold_{j=0}^{b(i)} \etermBB_{ij} > 0\Big)\]
Distribute the disjunction in $\fvarB$ following the same syntactic steps taken for $\fvarA$ to obtain the following provable equivalence:
\[ \fvarB \lbisubjunct \lorfold_{i=0}^{N} \Big(\landfold_{j=0}^{a(i)} \sigliedgeq{\genDE{x}}{\etermAA_{ij}} \lor \landfold_{j=0}^{b(i)} \sigliedgt{\genDE{x}}{\etermBB_{ij}} \Big)\]
Rewriting with the equivalences derived so far, and using the above normal form for $\lnot{\rfvar}$, yields the required, provable equivalence:
\[
\lnot{(\sigliedsai{\genDE{x}}{\rfvar})} \lbisubjunct \sigliedsai{\genDE{x}}{(\lnot{\rfvar})}
\qedhere
\]
\end{proof}

\subsection{Local Progress}
\label{app:localprogress}

This section derives the characterizations of local progress from~\rref{subsec:localprogress}.
These characterizations are used in the completeness proofs for both analytic and semianalytic invariants.

\subsubsection{Atomic Inequalities}

The proof of \rref{lem:localprogresscmp} was outlined in \rref{subsec:localprogress}.
The case where $\cmp$ is $\geq$ is proved first, while the more technical case where $\cmp$ is $>$ is proved subsequently.

\begin{proof}[Proof of \rref{lem:localprogresscmp} (\irref{Lpgeqfull})]
Let $N$ be the rank of extended term $\etermA$ with respect to $\D{x}=\genDE{x}$ from~\rref{eq:differential-rank}.
For the derivation of \irref{Lpgeqfull}, the additional flexibility of the $\ddnext$ modality with a disjunct $x=y$ in the domain constraint is not needed.
This disjunction is removed after unfolding the $\ddnext$ abbreviation using a~\irref{gddR} monotonicity step.
The definition of $\sigliedgeq{\genDE{x}}{\etermA}$ is also unfolded, with both disjuncts handled separately.
The premises are labeled \textcircled{1} and \textcircled{2} respectively.
{\footnotesizeoff%
\begin{sequentdeduction}[array]
\linfer[]{
\linfer[gddR]{
  \linfer[orl]{
  \lsequent{\initassum,\sigliedgt{\genDE{x}}{\etermA}}{\ddiamond{\pevolvein{\D{x}=\genDE{x}}{\etermA \geq 0}}{x \neq y}} !
  \lsequent{\initassum,\sigliedzero{\genDE{x}}{\etermA}}{\ddiamond{\pevolvein{\D{x}=\genDE{x}}{\etermA \geq 0}}{x \neq y}}
  }
  {\lsequent{\initassum,\sigliedgt{\genDE{x}}{\etermA} \lor \sigliedzero{\genDE{x}}{\etermA}}{\ddiamond{\pevolvein{\D{x}=\genDE{x}}{\etermA \geq 0}}{x \neq y}}}
}
  {\lsequent{\initassum,\sigliedgt{\genDE{x}}{\etermA} \lor \sigliedzero{\genDE{x}}{\etermA}}{\ddiamond{\pevolvein{\D{x}=\genDE{x}}{\etermA \geq 0 \lor x=y}}{x \neq y}}}
}
  {\lsequent{\initassum,\sigliedgeq{\genDE{x}}{\etermA}}{\dprogressin{\D{x}=\genDE{x}}{\etermA \geq 0}}}
\end{sequentdeduction}
}%

From \textcircled{2},~\irref{gddR} strengthens the inequality $\etermA \geq 0$ in the domain constraint to an equation $\etermA = 0$.
The derivation continues using \irref{dDR}, because by \irref{dRI}, $\etermA=0$ is provably invariant.
The proof is completed with \irref{Cont} using the trivial arithmetic fact $1 > 0$:
{\footnotesizeoff%
\begin{sequentdeduction}[array]
\linfer[gddR]{
\linfer[dDR]{
  \linfer[dRI]{
    \lclose
  }
  {\lsequent{\sigliedzero{\genDE{x}}{\etermA}}{\dbox{\pevolvein{\D{x}=\genDE{x}}{1 > 0}}{\etermA = 0}}} !
  \linfer[Cont]{
    \lclose
  }
  {\lsequent{\initassum}{\ddiamond{\pevolvein{\D{x}=\genDE{x}}{1 > 0}}{x \neq y}}}
}
  {\lsequent{\initassum,\sigliedzero{\genDE{x}}{\etermA}}{\ddiamond{\pevolvein{\D{x}=\genDE{x}}{\etermA = 0}}{x \neq y}}}
}
  {\lsequent{\initassum,\sigliedzero{\genDE{x}}{\etermA}}{\ddiamond{\pevolvein{\D{x}=\genDE{x}}{\etermA \geq 0}}{x \neq y}}}
\end{sequentdeduction}
}%

From \textcircled{1}, the premise is lined up for the derived step axiom~\irref{Lpgeq}.
The proof proceeds by closing the (left) premises obtained by iterating \irref{Lpgeq} for higher Lie derivatives.
In this way, the derivation continues until the final (rightmost) open premise which is abbreviated here and continued below:
{\footnotesizeoff%
\begin{sequentdeduction}[array]
\linfer[Lpgeq]{
  \linfer[qear]{\lclose}{\lsequent{\sigliedgt{\genDE{x}}{\etermA}}{\etermA \geq 0}} !
  \linfer[Lpgeq]{
  \linfer[qear]{\lclose}{\lsequent{\sigliedgt{\genDE{x}}{\etermA},\etermA=0}{\lied[]{\genDE{x}}{\etermA} \geq 0}} !
  \linfer[Lpgeq]{
    \lsequent{\initassum,\sigliedgt{\genDE{x}}{\etermA},\dots}{\dots}
  }
  {\dots}
}
  {\lsequent{\initassum,\sigliedgt{\genDE{x}}{\etermA},\etermA=0}{\ddiamond{\pevolvein{\D{x}=\genDE{x}}{\lied[]{\genDE{x}}{\etermA} \geq 0}}{x \neq y}}}
}
  {\lsequent{\initassum,\sigliedgt{\genDE{x}}{\etermA}}{\ddiamond{\pevolvein{\D{x}=\genDE{x}}{\etermA \geq 0}}{x \neq y}}}
\end{sequentdeduction}
}%

The open premise corresponds to the last conjunct of $\sigliedgt{\genDE{x}}{\etermA}$.
The implication in the conjunct uses the gathered antecedents $\etermA=0,\dots,\lied[N-2]{\genDE{x}}{\etermA}= 0$ after which~\irref{Cont+gddR} completes the proof:
{\footnotesizeoff%
\begin{sequentdeduction}[array]
\linfer[cut]{
  \linfer[Cont+gddR]{
    \lclose
  }
  {\lsequent{\initassum,\lied[N-1]{\genDE{x}}{\etermA} > 0}{\ddiamond{\pevolvein{\D{x}=\genDE{x}}{\lied[N-1]{\genDE{x}}{\etermA} \geq 0}}{x \neq y}}}
}
  \lsequent{\initassum,\sigliedgt{\genDE{x}}{\etermA},\etermA = 0,\dots,\lied[N-2]{\genDE{x}}{\etermA}= 0}{\ddiamond{\pevolvein{\D{x}=\genDE{x}}{\lied[N-1]{\genDE{x}}{\etermA} \geq 0}}{x \neq y}}
\\[-\normalbaselineskip]\tag*{\qedhere}
\end{sequentdeduction}
}%
\end{proof}

Unlike the non-strict case just derived for~\rref{lem:localprogresscmp}, the strict case (where $\cmp$ is $>$) crucially uses the fact that the $\ddnext$ modality \emph{excludes} the initial state, so that it is possible to locally progress into the strict inequality $\etermA > 0$ without already satisfying it in the initial state.
Topological considerations made this exclusion irrelevant for the non-strict case (see~\rref{subsec:localprogress}), as derived axiom \irref{bigsmallequiv} explains logically.
The idea behind the remaining proof of \rref{lem:localprogresscmp} for the strict case is to syntactically embed this difference into the derivation of~\irref{Lpgtfull}.
Moreover, this syntactic transformation reduces the proof to the non-strict case, so that the derived step axiom~\irref{Lpgeq} can again be used to progressively analyze higher Lie derivatives.
The following proposition is used for the transformation:

\begin{proposition}
\label{prop:leibnizpowers}
Let \(\etermAA = \etermA^{k}\) for some $k \geq 1$ and $\D{x}=\genDE{x}$ be an ODE with extended terms $\etermA,\etermAA,\genDE{x}$.
For each $0 \leq i \leq k-1$, there (computably) exists an extended term cofactor $\cofterm$ such that the following identity is provable in real arithmetic:
\[ \lied[i]{\genDE{x}}{\etermAA} = \cofterm \etermA  \]
\end{proposition}
\begin{proof}
The proof proceeds by induction on $k$.
\begin{itemize}
\item For $k=1$, $\etermAA=\etermA^1$ so $\lied[0]{\genDE{x}}{\etermA}=\etermA$ hence the cofactor $\cofterm \mnodefeq 1$ suffices.

\item For $\etermAA=\etermA^{k+1}$, the $j$-th Lie derivative of $\etermAA$ for $0 \leq j \leq k$ is given by Leibniz's rule:
\[ \lied[j]{\genDE{x}}{\etermAA} = \lie[j]{\genDE{x}}{\etermA^{k}\etermA} = \sum_{i=0}^j {j \choose i} \lied[j-i]{\genDE{x}}{(\etermA^k)} \lied[i]{\genDE{x}}{\etermA} \]
The induction hypothesis implies $\lied[j-i]{\genDE{x}}{(\etermA^k)} = \cofterm_i \etermA$ is a provable identity for some computable extended term cofactor $\cofterm_i$ for each $1 \leq i \leq j$.
The final summand for $i=0$ is:
\[{j \choose 0} \lied[j]{\genDE{x}}{(\etermA^k)} \lied[0]{\genDE{x}}{\etermA} = \lied[j]{\genDE{x}}{(\etermA^k)} \etermA\]
Thus, the cofactor $\cofterm \mnodefeq \lied[j]{\genDE{x}}{(\etermA^k)} + \sum_{i=1}^j {j \choose i} \cofterm_i \lied[i]{\genDE{x}}{\etermA}$ yields the identity $\lied[j]{\genDE{x}}{\etermAA} = \cofterm \etermA$.
This identity is provable because it only depends on first-order properties of real arithmetic. \qedhere
\end{itemize}
\end{proof}

For $\etermAA=\etermA^{k},k \geq 1$,~\rref{prop:leibnizpowers} implies that the formula \(\etermA=0 \limply \landfold_{i=0}^{k-1} \lied[i]{\genDE{x}}{\etermAA} = 0\) is provable in real arithmetic for extended terms $\etermA,\etermAA$.
This enables the remaining proof of \rref{lem:localprogresscmp}.

\begin{proof}[Proof of \rref{lem:localprogresscmp} (\irref{Lpgtfull})]
Let $N \geq 1$ be the rank of extended term $\etermA$ with respect to $\D{x}=\genDE{x}$.
This rank bounds the number of higher Lie derivatives of $\etermA$ that need to be considered.

The derivation starts by unfolding the syntactic abbreviation of the $\ddnext$ modality and reducing to the non-strict case with \irref{gddR} and the real arithmetic fact \(\etermA{-}\etermAA\geq 0 \limply \etermA > 0 \lor x=y\) for the abbreviation \(\etermAA \mdefeq |x-y|^{2N}\), which is a (polynomial) term: \(\big((x_1-y_1)^2 + \dots + (x_n-y_n)^2\big)^N\).

{\footnotesizeoff%
\begin{sequentdeduction}[array]
\linfer[]{
  \linfer[gddR]{
    \linfer[qear]{ \lclose }
    {\lsequent{\etermA{-}\etermAA \geq 0}{\etermA > 0 \lor x=y}}!
    \lsequent{\initassum,\sigliedgt{\genDE{x}}{\etermA}}{\ddiamond{\pevolvein{\D{x}=\genDE{x}}{\etermA{-}\etermAA \geq 0}}{x \neq y}}
  }
  {\lsequent{\initassum,\sigliedgt{\genDE{x}}{\etermA}}{\ddiamond{\pevolvein{\D{x}=\genDE{x}}{\etermA > 0 \lor x=y}}{x \neq y}}}
}
  {\lsequent{\initassum,\sigliedgt{\genDE{x}}{\etermA}}{\dprogressin{\D{x}=\genDE{x}}{\etermA > 0}}}
\end{sequentdeduction}
}%

Next, the initial assumption $\initassum$ in the antecedent is used.
The first cut proves using the formula of real arithmetic: $\initassum \limply |x-y|^2=0$.
As remarked, with $|x-y|^2=0$ and $N \geq 1$, by \rref{prop:leibnizpowers}, $|x-y|^2=0 \limply \landfold_{i=0}^{N-1} \lied[i]{\genDE{x}}{\etermAA} = 0$ is a provable real arithmetic formula.
The second cut proves using this fact. The resulting open premise is abbreviated \textcircled{1}.
{\footnotesizeoff%
\begin{sequentdeduction}[array]
\linfer[cut+qear]{
\linfer[cut+qear]{
  \lsequent{\initassum,\landfold_{i=0}^{N-1} \lied[i]{\genDE{x}}{\etermAA} = 0, \sigliedgt{\genDE{x}}{\etermA}}{\ddiamond{\pevolvein{\D{x}=\genDE{x}}{\etermA{-}\etermAA \geq 0}}{x \neq y}}
}
  {\lsequent{\initassum,|x-y|^2=0,\sigliedgt{\genDE{x}}{\etermA}}{\ddiamond{\pevolvein{\D{x}=\genDE{x}}{\etermA{-}\etermAA \geq 0}}{x \neq y}}}
}
  {\lsequent{\initassum,\sigliedgt{\genDE{x}}{\etermA}}{\ddiamond{\pevolvein{\D{x}=\genDE{x}}{\etermA{-}\etermAA \geq 0}}{x \neq y}}}
\end{sequentdeduction}
}%

To continue from \textcircled{1}, observe that for $0\leq i \leq N-1$, by linearity of the Lie derivative:
\[\lie[i]{\genDE{x}}{\etermA{-}\etermAA} = \lied[i]{\genDE{x}}{\etermA} - \lied[i]{\genDE{x}}{\etermAA}\]

Using the conjunction $\landfold_{i=0}^{N-1} \lied[i]{\genDE{x}}{\etermAA} = 0$ in the antecedents, the formula $\lied[i]{\genDE{x}}{(\etermA{-}\etermAA)} = \lied[i]{\genDE{x}}{\etermA}$ proves by a cut and real arithmetic for $0\leq i \leq N-1$.
This justifies the next real arithmetic step from \textcircled{1}, with the assumptions \(\Gamma_\etermAA \mdefequiv \landfold_{i=0}^{N-1} \lied[i]{\genDE{x}}{(\etermA{-}\etermAA)} = \lied[i]{\genDE{x}}{\etermA}\).
Intuitively, $\Gamma_\etermAA$ allows the derivation to locally work with higher Lie derivatives of $\etermA$ instead of higher Lie derivatives of $\etermA-\etermAA$ in subsequent steps.
{\footnotesizeoff%
\begin{sequentdeduction}[array]
\linfer[cut+qear]{
  \lsequent{\Gamma_\etermAA, \initassum, \sigliedgt{\genDE{x}}{\etermA}}{\ddiamond{\pevolvein{\D{x}=\genDE{x}}{\etermA{-}\etermAA \geq 0}}{x \neq y}}
}
  {\lsequent{\initassum,\landfold_{i=0}^{N-1} \lied[i]{\genDE{x}}{\etermAA} = 0, \sigliedgt{\genDE{x}}{\etermA}}{\ddiamond{\pevolvein{\D{x}=\genDE{x}}{\etermA{-}\etermAA \geq 0}}{x \neq y}}}
\end{sequentdeduction}
}%

The derivation is completed using the same technique of iterating \irref{Lpgeq}, as shown in the earlier proof of~\rref{lem:localprogresscmp} for the non-strict case \irref{Lpgeqfull}.
It starts with a single \irref{Lpgeq} step.
The left premise closes by real arithmetic because $\sigliedgt{\genDE{x}}{\etermA}$ has the conjunct $\etermA\geq 0$, and $\Gamma_\etermAA$ provides $\etermA-\etermAA = \etermA$, which imply \(\etermA-\etermAA\geq0\).
The right premise is abbreviated \textcircled{2}.
{\footnotesizeoff%
\begin{sequentdeduction}[array]
\linfer[Lpgeq]{
  \linfer[qear]{
    \lclose
  }{\lsequent{\Gamma_\etermAA, \sigliedgt{\genDE{x}}{\etermA}}{\etermA{-}\etermAA \geq 0}} !
  \lsequent{\Gamma_\etermAA, \initassum, \sigliedgt{\genDE{x}}{\etermA}, \etermA{-}\etermAA=0}{\ddiamond{\pevolvein{\D{x}=\genDE{x}}{\lied[1]{\genDE{x}}{(\etermA{-}\etermAA)}{\geq} 0}}{x \neq y}}
}
  {\lsequent{\Gamma_\etermAA, \initassum, \sigliedgt{\genDE{x}}{\etermA}}{\ddiamond{\pevolvein{\D{x}=\genDE{x}}{\etermA{-}\etermAA \geq 0}}{x \neq y}}}
\end{sequentdeduction}
}%

From \textcircled{2}, local progress for the first Lie derivative of $\etermA{-}\etermAA$ is shown.
The first step simplifies formula $\etermA{-}\etermAA=0$ in the antecedents using $\Gamma_\etermAA$.
The derived axiom \irref{Lpgeq}, together with $\Gamma_\etermAA$, simplifies and proves the left premise.
The right premise is abbreviated \textcircled{3} (shown and continued below).
{\footnotesizeoff%
\begin{sequentdeduction}[array]
\linfer[qear]{
\linfer[Lpgeq]{
  \linfer[qear]{
  \linfer[qear]{
    \lclose
  }
    {\lsequent{\etermA=0 \limply \lied[]{\genDE{x}}{\etermA} \geq 0, \etermA=0}{\lied[1]{\genDE{x}}{\etermA}\geq 0}}
  }
  {\lsequent{\Gamma_\etermAA, \sigliedgt{\genDE{x}}{\etermA}, \etermA=0}{\lied[1]{\genDE{x}}{(\etermA{-}\etermAA)}{\geq} 0}} !
  \textcircled{3}
}
  {\lsequent{\Gamma_\etermAA, \initassum, \sigliedgt{\genDE{x}}{\etermA}, \etermA=0}{\ddiamond{\pevolvein{\D{x}=\genDE{x}}{\lied[1]{\genDE{x}}{(\etermA{-}\etermAA)}{\geq} 0}}{x \neq y}}}
}
  {\lsequent{\Gamma_\etermAA, \initassum, \sigliedgt{\genDE{x}}{\etermA}, \etermA{-}\etermAA=0}{\ddiamond{\pevolvein{\D{x}=\genDE{x}}{\lied[1]{\genDE{x}}{(\etermA{-}\etermAA)}{\geq} 0}}{x \neq y}}}
\end{sequentdeduction}
}%

The derivation continues from \textcircled{3} similarly for higher Lie derivatives of $\etermA{-}\etermAA$, using $\Gamma_\etermAA$ to replace $\lied[i]{\genDE{x}}{(\etermA{-}\etermAA)}$ with $\lied[i]{\genDE{x}}{\etermA}$, and then using the corresponding conjunct of $\sigliedgt{\genDE{x}}{\etermA}$.
The final open premise obtained from \textcircled{3} by iterating \irref{Lpgeq} corresponds to the last conjunct of $\sigliedgt{\genDE{x}}{\etermA}$:
{\footnotesizeoff\renewcommand{\sigliedgt}[3][]{\siglied[#1]{#2}{#3}{>}0}%
\begin{sequentdeduction}[array]
\linfer[qear]{
\linfer[Lpgeq]{
\linfer[Lpgeq]{
  \lsequent{\Gamma_\etermAA, \initassum,\sigliedgt{\genDE{x}}{\etermA},\etermA{=}0,\dots,\lied[N-2]{\genDE{x}}{\etermA}{=}0}{\ddiamond{\pevolvein{\D{x}=\genDE{x}}{\lied[N-1]{\genDE{x}}{(\etermA{-}\etermAA)}{\geq} 0}}{x \neq y}}
}
  {\dots}
}
  {\lsequent{\Gamma_\etermAA, \initassum, \sigliedgt{\genDE{x}}{\etermA}, \etermA=0,\lied[1]{\genDE{x}}{\etermA}\geq 0}{\ddiamond{\pevolvein{\D{x}=\genDE{x}}{\lied[2]{\genDE{x}}{(\etermA{-}\etermAA)}{\geq} 0}}{x \neq y}}}
}
  {\lsequent{\Gamma_\etermAA, \initassum, \sigliedgt{\genDE{x}}{\etermA}, \etermA{=}0,\lied[1]{\genDE{x}}{(\etermA{-}\etermAA)}{\geq} 0}{\ddiamond{\pevolvein{\D{x}=\genDE{x}}{\lied[2]{\genDE{x}}{(\etermA{-}\etermAA)}{\geq} 0}}{x \neq y}}}
\end{sequentdeduction}
}%

The gathered antecedents $\etermA=0,\dots,\lied[N-2]{\genDE{x}}{\etermA}= 0$ are respectively obtained from $\Gamma_\etermAA$ by real arithmetic.
The proof is closed with \irref{gddR+Cont}, similarly to the non-strict case.
{\footnotesizeoff%
\begin{sequentdeduction}[array]
\linfer[cut]{
  \linfer[gddR]{
  \linfer[cut+qear]{
  \linfer[Cont]{
    \lclose
  }
    {\lsequent{\initassum,\lied[N-1]{\genDE{x}}{(\etermA{-}\etermAA)}{>} 0}{\ddiamond{\pevolvein{\D{x}=\genDE{x}}{\lied[N-1]{\genDE{x}}{(\etermA{-}\etermAA)}{>} 0}}{x \neq y}}}
  }
    {\lsequent{\Gamma_\etermAA,\initassum,\lied[N-1]{\genDE{x}}{\etermA} {>} 0}{\ddiamond{\pevolvein{\D{x}=\genDE{x}}{\lied[N-1]{\genDE{x}}{(\etermA{-}\etermAA)}{>} 0}}{x \neq y}}}
  }
  {\lsequent{\Gamma_\etermAA,\initassum,\lied[N-1]{\genDE{x}}{\etermA} {>} 0}{\ddiamond{\pevolvein{\D{x}=\genDE{x}}{\lied[N-1]{\genDE{x}}{(\etermA{-}\etermAA)}{\geq} 0}}{x \neq y}}}
}
  {\lsequent{\Gamma_\etermAA, \initassum, \sigliedgt{\genDE{x}}{\etermA},\etermA{=}0,..,\lied[N-2]{\genDE{x}}{\etermA}{=}0}{\ddiamond{\pevolvein{\D{x}=\genDE{x}}{\lied[N-1]{\genDE{x}}{(\etermA{-}\etermAA)}{\geq} 0}}{x \neq y}}}
\\[-\normalbaselineskip]\tag*{\qedhere}
\end{sequentdeduction}
}%
\end{proof}

\subsubsection{Semianalytic Formulas}

The proof in the semianalytic case is outlined in~\rref{subsec:localprogress}.
It lifts derived axioms~\irref{Lpgeqfull} and~\irref{Lpgtfull} according to the homomorphic definition of the semianalytic progress formula, using~\irref{Uniq} to prove local progress into a conjunction of two formulas simultaneously.

\begin{proof}[Proof of \rref{lem:localprogresssemialg}]
By congruential equivalence~\cite{DBLP:journals/jar/Platzer17}, assume, without loss of generality, that formula $\rfvar$ is propositionally rewritten to the same normal form~\rref{eq:normalform} as in the corresponding semianalytic progress formula $\sigliedsai{\genDE{x}}{\rfvar}$.
Throughout this proof, similar premises are collapsed in proofs and directly indexed by $i,j$.
The $i$-th disjunct of $\rfvar$ is abbreviated with $\rfvar_i \mdefequiv \landfold_{j=0}^{m(i)} \etermA_{ij} \geq 0 \land \landfold_{j=0}^{n(i)} \etermB_{ij}> 0$.

The derivation starts by splitting the (outermost) disjunction in $\sigliedsai{\genDE{x}}{\rfvar}$ with \irref{orl}.
For each resulting premise (indexed by $i$), local progress is proved for the corresponding disjunct $\rfvar_i$ of $\rfvar$.
The domain change with \irref{gddR} proves because $\rfvar_i \lor x=y \limply \rfvar \lor x=y$ is a propositional tautology for each $i$.
{\footnotesizeoff%
\begin{sequentdeduction}[array]
\linfer[orl]{
\linfer[gddR]{
  \lsequent{\initassum,\landfold_{j=0}^{m(i)} \sigliedgeq{\genDE{x}}{\etermA_{ij}} \land \landfold_{j=0}^{n(i)} \sigliedgt{\genDE{x}}{\etermB_{ij}}}{\dprogressin{\D{x}=\genDE{x}}{\rfvar_i}}
}
  {\lsequent{\initassum,\landfold_{j=0}^{m(i)} \sigliedgeq{\genDE{x}}{\etermA_{ij}} \land \landfold_{j=0}^{n(i)} \sigliedgt{\genDE{x}}{\etermB_{ij}}}{\dprogressin{\D{x}=\genDE{x}}{\rfvar}}}
}
  {\lsequent{\initassum,\sigliedsai{\genDE{x}}{\rfvar}}{\dprogressin{\D{x}=\genDE{x}}{\rfvar}}}
\end{sequentdeduction}
}%

It suffices now to prove local progress in $\rfvar_i$.
The uniqueness axiom \irref{Uniq} splits conjuncts in $\rfvar_i$.
The~\irref{gddR} step distributes $\initassum$ in domain constraint from $\ddnext$ over conjunctions using the propositional tautology $(\rrfvar_1 \land \rrfvar_2) \lor \initassum \lbisubjunct (\rrfvar_1 \lor \initassum) \land (\rrfvar_2 \lor x=y)$.
This leaves premises (indexed by $j$) for the non-strict and strict inequalities of $\rfvar_i$ which are closed by~\irref{Lpgeqfull} and~\irref{Lpgtfull} respectively.
{\footnotesizeoff\renewcommand{\linferPremissSeparation}{~}%
\begin{sequentdeduction}[array]
\linfer[Uniq+andr+gddR]{
  \linfer[Lpgeqfull]{
    \lclose
  }
  {\lsequent{\initassum,\sigliedgeq{\genDE{x}}{\etermA_{ij}}}{\dprogressin{\D{x}{=}\genDE{x}}{\etermA_{ij} {\geq} 0}}}
  !
  \linfer[Lpgtfull]{
    \lclose
  }
  {\lsequent{\initassum,\sigliedgt{\genDE{x}}{\etermB_{ij}}}{\dprogressin{\D{x}{=}\genDE{x}}{\etermB_{ij} {>} 0}}}
}
  {\lsequent{\initassum,\landfold_{j=0}^{m(i)} \sigliedgeq{\genDE{x}}{\etermA_{ij}} {\land} \landfold_{j=0}^{n(i)} \sigliedgt{\genDE{x}}{\etermB_{ij}}}{\dprogressin{\D{x}{=}\genDE{x}}{\rfvar_i}\qedhere}}
\end{sequentdeduction}
}%
\end{proof}

Using~\rref{prop:negationrearrangement}, the implicational semianalytic local progress axiom~\irref{LpRfull} from \rref{lem:localprogresssemialg} is strengthened to an equivalent characterization of semianalytic local progress.
\begin{proof}[Proof of \rref{thm:localprogresscomplete}]
By congruential equivalence~\cite{DBLP:journals/jar/Platzer17}, assume, without loss of generality, that formula $\rfvar$ is propositionally rewritten to the same normal form~\rref{eq:normalform} as in the corresponding semianalytic progress formula $\sigliedsai{\genDE{x}}{\rfvar}$.
By \rref{prop:negationrearrangement}, there is a normal form for $\lnot{\rfvar}$
with the provable equivalence $ \lnot{(\sigliedsai{\genDE{x}}{\rfvar})} \lbisubjunct \sigliedsai{\genDE{x}}{(\lnot{\rfvar})}$.
The ``$\lylpmi$'' direction of the inner equivalence is \irref{LpRfull}.
The derivation in the ``$\limply$'' direction of the inner equivalence starts by reducing to the contrapositive statement by propositional logic transformations.
The final step rewrites the negation in the antecedents using the above normal form for $\lnot{\rfvar}$ from~\rref{prop:negationrearrangement}.
{\footnotesizeoff%
\begin{sequentdeduction}[array]
\linfer[cut+notl+notr]{
\linfer[qear]{
  \lsequent{\initassum,\sigliedsai{\genDE{x}}{(\lnot{\rfvar})}}{\lnot{\dprogressin{\D{x}=\genDE{x}}{\rfvar}}}
}
  {\lsequent{\initassum,\lnot{(\sigliedsai{\genDE{x}}{\rfvar})}}{\lnot{\dprogressin{\D{x}=\genDE{x}}{\rfvar}}}}
}
  {\lsequent{\initassum,\dprogressin{\D{x}=\genDE{x}}{\rfvar}}{\sigliedsai{\genDE{x}}{\rfvar}}}
\end{sequentdeduction}
}%

By the derived axiom \irref{LpRfull} from \rref{lem:localprogresssemialg}, the progress formula for $\lnot{\rfvar}$ in the antecedent implies local progress for $\lnot{\rfvar}$.
The proof is completed with derived axiom~\irref{dualityimp} of \rref{cor:dualityimp}:

{\footnotesizeoff%
\begin{sequentdeduction}[array]
\linfer[LpRfull]{
\linfer[dualityimp]{
  \lclose
}
  {\lsequent{\initassum,\dprogressin{\D{x}=\genDE{x}}{\lnot{\rfvar}}}{\lnot{\dprogressin{\D{x}=\genDE{x}}{\rfvar}}}}
}
  {\lsequent{\initassum,\sigliedsai{\genDE{x}}{(\lnot{\rfvar})}}{\lnot{\dprogressin{\D{x}=\genDE{x}}{\rfvar}}}}
\\[-\normalbaselineskip]\tag*{\qedhere}
\end{sequentdeduction}
}%
\end{proof}

\begin{proof}[Proof of \rref{cor:localprogresscompletedualcong}]

Self-duality axiom \irref{duality} derives by using \irref{Lpiff} twice together with the provable equivalence \(\lnot{(\sigliedsai{\genDE{x}}{\rfvar})} \lbisubjunct \sigliedsai{\genDE{x}}{(\lnot{\rfvar})}\) from \rref{prop:negationrearrangement} (and double negation elimination).
{\footnotesizeoff%
\begin{sequentdeduction}[array]
\linfer[Lpiff]{
\linfer[qear]{
\linfer[Lpiff]{
  \lclose
}
  {\lsequent{\initassum}{\dprogressin{\D{x}=\genDE{x}}{\rfvar} \lbisubjunct \sigliedsai{\genDE{x}}{\rfvar}}}
}
  {\lsequent{\initassum}{\dprogressin{\D{x}=\genDE{x}}{\rfvar} \lbisubjunct \lnot{\sigliedsai{\genDE{x}}{(\lnot{\rfvar})}}}}
}
  {\lsequent{\initassum}{\dprogressin{\D{x}=\genDE{x}}{\rfvar} \lbisubjunct \lnot{\dprogressin{\D{x}=\genDE{x}}{\lnot{\rfvar}}}}}
\end{sequentdeduction}
}%
Local progress congruence rule~\irref{CLP} derives similarly by introducing an initial assumption $\initassum$ with~\irref{cut+qear+existsl}, equivalently rewriting with~\irref{Lpiff}, and congruential equivalence~\cite{DBLP:journals/jar/Platzer17} in the last step.
{\footnotesizeoff%
\begin{sequentdeduction}[array]
\linfer[cut+qear+existsl]{
\linfer[Lpiff]{
\linfer[]{
  \lsequent{} {\rfvar \lbisubjunct \rrfvar}
}
  {\lsequent{} {\dprogressin{\D{x}=\genDE{x}}{\rfvar} \lbisubjunct \dprogressin{\D{x}=\genDE{x}}{\rrfvar}}}
}
  {\lsequent{\initassum} {\sigliedsai{\genDE{x}}{\rfvar} \lbisubjunct \sigliedsai{\genDE{x}}{\rrfvar}}}
}
  {\lsequent{} {\sigliedsai{\genDE{x}}{\rfvar} \lbisubjunct \sigliedsai{\genDE{x}}{\rrfvar}}}
\\[-\normalbaselineskip]\tag*{\qedhere}
\end{sequentdeduction}
}%
\end{proof}

\subsection{Analytic Invariants}
\label{app:alginvariants}

This section derives the analytic completeness axiom~\irref{DRI} (and its generalization~\irref{DRIQ}), thus proving completeness for analytic (Noetherian) invariants and also for analytic postconditions.

\subsubsection{Differential Radical Invariants}

The differential radical invariants proof rule \irref{dRI} derives from \irref{vdbx} by turning the differential radical identity~\rref{eq:differential-rank} into a provable vectorial Darboux equality.

\begin{proof}[Proof of \rref{thm:DRI}]
Let $\etermA$ be an extended term satisfying both premises of the \irref{dRI} proof rule and let $\vecpolyn{\etermA}{x}$ be the vector of extended terms with components $\vecpolyn{\etermA}{x}_i \mdefeq \lied[i-1]{\genDE{x}}{\etermA}$ for $i = 1,2,\dots,N$.
The derivation starts by setting up the premise for an application of derived rule~\irref{vdbx}.
In the first step, axiom \irref{DX} is used to assume that the domain constraint $\ivr$ is true initially.
On the left premise after the cut, arithmetic equivalence \m{\landfold_{i=0}^{N-1} \lied[i]{\genDE{x}}{\etermA} = 0 \lbisubjunct \vecpolyn{\etermA}{x}= 0} is used to rewrite the succedent to the left premise of \irref{dRI}.
On the right premise, monotonicity~\irref{MbW} strengthens the postcondition to $\vecpolyn{\etermA}{x}=0$:
{\footnotesizeoff%
\begin{sequentdeduction}[array]
\linfer[DX]{
\linfer[cut]{
  \linfer[qear]{
    \lsequent{\Gamma,\ivr} {\landfold_{i=0}^{N-1} \lied[i]{\genDE{x}}{\etermA} = 0 }
  }
  {\lsequent{\Gamma,\ivr} {\vecpolyn{\etermA}{x} = 0}} !
  \linfer[MbW]{
    \lsequent{\vecpolyn{\etermA}{x} = 0} {\dbox{\pevolvein{\D{x}=\genDE{x}}{\ivr}}{\vecpolyn{\etermA}{x} = 0}}
  }
  {\lsequent{\vecpolyn{\etermA}{x} = 0} {\dbox{\pevolvein{\D{x}=\genDE{x}}{\ivr}}{\etermA =0}}}
}
  {\lsequent{\Gamma,\ivr} {\dbox{\pevolvein{\D{x}=\genDE{x}}{\ivr}}{\etermA = 0}}}
}
  {\lsequent{\Gamma} {\dbox{\pevolvein{\D{x}=\genDE{x}}{\ivr}}{\etermA = 0}}}
\end{sequentdeduction}
}%

Continuing from the right premise, the component-wise Lie derivative of $\vecpolyn{\etermA}{x}$ is defined as \((\lied[]{\genDE{x}}{\vec{\etermA}})_i = \lie[]{\genDE{x}}{\vec{\etermA}_i} = \lied[i]{\genDE{x}}{\etermA}\).
The vector $\lied[]{\genDE{x}}{\vec{\etermA}}$ will be obtained from $\vecpolyn{\etermA}{x}$ by matrix multiplication with an $N \times N$ extended term cofactor matrix $\coftermC$ with $1$ on its superdiagonal, and the $g_i$ cofactors in the last row:
{\footnotesizeoff%
\[\matpolyn{\coftermC}{x}= \left(\begin{array}{ccccc}
0      & 1      & 0      & \dots & 0      \\
0      & 0      & \ddots & \ddots & \vdots \\
\vdots & \vdots & \ddots & \ddots & 0      \\
0      & 0      & \dots & 0      & 1\\
\cofterm_0    & \cofterm_1    & \dots & \cofterm_{N-2}& \cofterm_{N-1} \end{array}\right),
\quad
\vecpolyn{\etermA}{x} = \left(\begin{array}{l}\etermA\\ \lied[1]{\genDE{x}}{\etermA}\\ \vdots \\\lied[N-2]{\genDE{x}}{\etermA}\\ \lied[N-1]{\genDE{x}}{\etermA}\end{array}\right),
\quad
\lied[]{\genDE{x}}{\vec{\etermA}} = \left(\begin{array}{l} \lied[1]{\genDE{x}}{\etermA} \\ \lied[2]{\genDE{x}}{\etermA}\\ \vdots \\\lied[N-1]{\genDE{x}}{\etermA}\\ \lied[N]{\genDE{x}}{\etermA}\end{array}\right)
\]}%

The vectorial equation $\lied[]{\genDE{x}}{\vec{\etermA}} = \coftermC \vecpolyn{\etermA}{x}$ is provably equivalent to the equation $\lied[N]{\genDE{x}}{\etermA} = \sum_{i=0}^{N-1} \cofterm_i \lied[i]{\genDE{x}}{\etermA}$.
To see this, note that for indices $1 \leq i < N$, matrix multiplication yields:
\[ (\lied[]{\genDE{x}}{\vec{\etermA}})_i = \lied[i]{\genDE{x}}{\etermA} = (\vecpolyn{\etermA}{x})_{i+1} = (\coftermC\itimes\vecpolyn{\etermA}{x})_i \]

Therefore, all but the final component-wise equality are trivially valid and prove by \irref{qear}.
The remaining (non-trivial) equation for $i=N$ is $(\lied[]{\genDE{x}}{\vec{\etermA}})_{N} = (\coftermC\itimes\vecpolyn{\etermA}{x})_N$.
The LHS of this equation simplifies with $(\lied[]{\genDE{x}}{\vec{\etermA}})_{N} = \lied[N]{\genDE{x}}{\etermA}$, while the RHS simplifies by matrix multiplication to:
\[ (\coftermC\itimes\vecpolyn{\etermA}{x})_N = \sum_{i=1}^{N} g_{i-1} (\vec{\etermA})_i = \sum_{i=1}^{N} g_{i-1} \lied[i-1]{\genDE{x}}{\etermA} = \sum_{i=0}^{N-1} g_{i} \lied[i]{\genDE{x}}{\etermA} \]

Hence, real arithmetic equivalently turns the formula $\lied[]{\genDE{x}}{\vec{\etermA}} = \coftermC \vecpolyn{\etermA}{x}$ into the succedent of the right premise of rule \irref{dRI}.
An application of derived rule~\irref{vdbx} from \rref{cor:vdbxscalar} with cofactor matrix $\coftermC$ followed by rule~\irref{qear} yields the remaining right premise of rule~\irref{dRI}, completing the derivation.
{\footnotesizeoff%
\begin{sequentdeduction}[array]
\linfer[vdbx]{
\linfer[qear]{
    \lsequent{\ivr} { \lied[N]{\genDE{x}}{\etermA} = \sum_{i=0}^{N-1} g_{i} \lied[i]{\genDE{x}}{\etermA}}
}
    {\lsequent{\ivr} {\lied[]{\genDE{x}}{\vec{\etermA}}=\matpolyn{G}{x} \itimes\vecpolyn{\etermA}{x}}}
}
    {\lsequent{\vecpolyn{\etermA}{x} = 0} {\dbox{\pevolvein{\D{x}=\genDE{x}}{\ivr}}{\vecpolyn{\etermA}{x}=0}}}
\\[-\normalbaselineskip]\tag*{\qedhere}
\end{sequentdeduction}
}%
\end{proof}

The derivation of rule~\irref{dRI} uses a specific choice of cofactor matrix $\coftermC$ in rule~\irref{vdbx} to prove invariance of the equation $\etermA = 0$.
This suffices for analytic completeness because analytic formulas can always be normalized to a single equation in real arithmetic.
However, such normalization may not yield the computationally most efficient way of proving an analytic invariant (see~\rref{ex:expressivity}).

\subsubsection{Completeness for Analytic Invariants}

The analytic completeness axiom with semianalytic domain constraints \irref{DRIQ} from~\rref{thm:algcompletedom} is derived next, making use of \irref{LpRfull} from \rref{lem:localprogresssemialg}.
The completeness argument can be summarized by contrapositives: if the local progress formula $\sigliedzero{\genDE{x}}{\etermA}$ is false in an initial state, then some higher Lie derivative of $\etermA$ is non-zero and gives a definite (local) sign to the value of $\etermA$, which, by \irref{LpRfull} implies progress to \(\etermA\neq0\).
For completeness, axiom~\irref{DRIQ} also handles the vacuous case where domain constraint $\ivr$ is false initially $(\ivr \limply \dots)$ and the stuck case where the domain constraint is true initially but cannot locally progress ($\sigliedsai{\genDE{x}}{\ivr}  \limply \dots$).

\begin{proof}[Proof of \rref{thm:algcompletedom} (implies \rref{thm:algcomplete})]
For formulas $\ivr$ formed from  conjunctions and disjunctions of strict inequalities, $\ivr \limply \sigliedsai{\genDE{x}}{\ivr}$ is provable in real arithmetic, so axiom~\irref{DRI} follows as an arithmetical corollary of~\irref{DRIQ}.
The derivation of axiom~\irref{DRIQ} starts by rewriting its LHS equivalently with axiom~\irref{DX}.
This is followed by a series of equivalent propositional rewrites that simplify the logical structure of the succedent.
The propositional steps are shown below, first pulling out the common implication $\ivr$ then the common conjunct $\etermA = 0$ as antecedent assumptions.
{\footnotesizeoff%
\begin{sequentdeduction}[array]
\linfer[DX]{
\linfer[]{
\linfer[]{
  \lsequent{\ivr, \etermA = 0}{ \big(\dbox{\pevolvein{\D{x}=\genDE{x}}{\ivr}}{\etermA=0}\big) \lbisubjunct \big(\sigliedsai{\genDE{x}}{\ivr} \limply \sigliedzero{\genDE{x}}{\etermA} \big)}
}
  {\lsequent{\ivr}{ \big(\etermA = 0 \land \dbox{\pevolvein{\D{x}=\genDE{x}}{\ivr}}{\etermA=0}\big) \lbisubjunct \big(\etermA = 0 \land (\sigliedsai{\genDE{x}}{\ivr} \limply \sigliedzero{\genDE{x}}{\etermA}) \big)}}
}
  {\lsequent{}{ \big(\ivr \limply \etermA = 0 \land \dbox{\pevolvein{\D{x}=\genDE{x}}{\ivr}}{\etermA=0}\big) \lbisubjunct \big(\ivr \limply \etermA = 0 \land (\sigliedsai{\genDE{x}}{\ivr} \limply \sigliedzero{\genDE{x}}{\etermA}) \big)}}
}
  {\lsequent{}{\dbox{\pevolvein{\D{x}=\genDE{x}}{\ivr}}{\etermA=0} \lbisubjunct \big(\ivr \limply \etermA = 0 \land (\sigliedsai{\genDE{x}}{\ivr} \limply \sigliedzero{\genDE{x}}{\etermA} ) \big)}}
\end{sequentdeduction}
}%

Next, a cut of the first-order formula $\lexists{y}{\initassum}$ proves trivially in real arithmetic and Skolemizing it with~\irref{existsl} yields an initial state assumption ($\initassum$ for fresh variables $y$).
To make use of this initial state assumption, the derivation continues with a classical case split on whether the semianalytic progress formula $\sigliedsai{\genDE{x}}{\ivr}$ is true initially.
The resulting premises are labeled \textcircled{1} (for the $\sigliedsai{\genDE{x}}{\ivr}$ disjunct) and \textcircled{2} (for the $\lnot(\sigliedsai{\genDE{x}}{\ivr})$ disjunct) and continued below.
{\footnotesizeoff%
\begin{sequentdeduction}[array]
\linfer[cut+qear]{
\linfer[existsl]{
\linfer[cut]{
\linfer[orl]{
  \textcircled{1} ! \textcircled{2}
}
  {\lsequent{\initassum, \ivr, \etermA = 0, \sigliedsai{\genDE{x}}{\ivr}\lor\lnot{(\sigliedsai{\genDE{x}}{\ivr})}}{ \big(\dbox{\pevolvein{\D{x}=\genDE{x}}{\ivr}}{\etermA=0}\big) \lbisubjunct \big(\sigliedsai{\genDE{x}}{\ivr} \limply \sigliedzero{\genDE{x}}{\etermA} \big)}}
}
  {\lsequent{\initassum, \ivr, \etermA = 0}{ \big(\dbox{\pevolvein{\D{x}=\genDE{x}}{\ivr}}{\etermA=0}\big) \lbisubjunct \big(\sigliedsai{\genDE{x}}{\ivr} \limply \sigliedzero{\genDE{x}}{\etermA} \big)}}
}
  {\lsequent{\lexists{y}{\initassum}, \ivr, \etermA = 0}{ \big(\dbox{\pevolvein{\D{x}=\genDE{x}}{\ivr}}{\etermA=0}\big) \lbisubjunct \big(\sigliedsai{\genDE{x}}{\ivr} \limply \sigliedzero{\genDE{x}}{\etermA} \big)}}
}
  {\lsequent{\ivr, \etermA = 0}{ \big(\dbox{\pevolvein{\D{x}=\genDE{x}}{\ivr}}{\etermA=0}\big) \lbisubjunct \big(\sigliedsai{\genDE{x}}{\ivr} \limply \sigliedzero{\genDE{x}}{\etermA} \big)}}
\end{sequentdeduction}
}%

The premise \textcircled{2} corresponds to the case where solutions are \emph{stuck} in the initial state because no local progress in the domain constraint $\ivr$ is possible.
Topologically, this corresponds to the situation where initial states are on the boundary of the set characterized by $\ivr$ (and also in $\ivr$)\footnote{%
This situation is impossible for domain constraints $\ivr$ characterizing topologically open sets, which is the semantical reason for derived axiom~\irref{DRI} having a simpler RHS characterization than~\irref{DRIQ}.}
but the ODE locally leaves $\ivr$.
Since $\etermA=0$ is already true in this stuck state, it trivially remains true for all solutions staying in domain constraint $\ivr$.
The derivation from~\textcircled{2} starts with a propositional simplification of the succedent since its RHS is vacuously equivalent to $\ltrue$ by assumption $\lnot(\sigliedsai{\genDE{x}}{\ivr})$.
The local progress characterization axiom~\irref{Lpiff} equivalently rewrites the sub-formula $\sigliedsai{\genDE{x}}{\ivr}$ to local progress for $\ivr$ before axiom~\irref{diamond} unfolds the $\ddnext$ modality turning it into a box modality formula.
{\footnotesizeoff%
\begin{sequentdeduction}[array]
\linfer[qear]{
\linfer[Lpiff]{
\linfer[diamond]{
  \lsequent{\initassum, \etermA = 0, \dbox{\pevolvein{\D{x}=\genDE{x}}{\ivr \lor x=y}}{x = y}}{\dbox{\pevolvein{\D{x}=\genDE{x}}{\ivr}}{\etermA=0}}
}
  {\lsequent{\initassum, \etermA = 0, \lnot{(\dprogressin{\D{x}=\genDE{x}}{\ivr})}}{\dbox{\pevolvein{\D{x}=\genDE{x}}{\ivr}}{\etermA=0}}}
}
  {\lsequent{\initassum, \etermA = 0, \lnot{(\sigliedsai{\genDE{x}}{\ivr})}}{\dbox{\pevolvein{\D{x}=\genDE{x}}{\ivr}}{\etermA=0}}}
}
  {\lsequent{\initassum, \etermA = 0, \lnot{(\sigliedsai{\genDE{x}}{\ivr})}}{\big(\dbox{\pevolvein{\D{x}=\genDE{x}}{\ivr}}{\etermA=0}\big) \lbisubjunct \big(\sigliedsai{\genDE{x}}{\ivr} \limply \sigliedzero{\genDE{x}}{\etermA} \big)}}
\end{sequentdeduction}
}%
By axiom~\irref{V}, the constant assumption $\etermA(y)=0$ (with $y$ in place of $x$) strengthens the postcondition of the antecedent box modality to $\etermA=0$ using the provable arithmetic formula \(\etermA(y)=0 \land x=y \limply \etermA=0\).
A subsequent~\irref{DMP+dW} step finishes the proof using the propositional tautology $\ivr \limply \ivr \lor x=y$.
{\footnotesizeoff%
\begin{sequentdeduction}[array]
\linfer[V]{
\linfer[DMP+dW]{
\linfer{%
  \lclose
}
  {\lsequent{\ivr}{\ivr \lor x=y}}
}
  {\lsequent{\dbox{\pevolvein{\D{x}=\genDE{x}}{\ivr \lor x=y}}{\etermA=0}}{\dbox{\pevolvein{\D{x}=\genDE{x}}{\ivr}}{\etermA=0}}}
}
  {\lsequent{\initassum, \etermA = 0, \dbox{\pevolvein{\D{x}=\genDE{x}}{\ivr \lor x=y}}{x = y}}{\dbox{\pevolvein{\D{x}=\genDE{x}}{\ivr}}{\etermA=0}}}
\end{sequentdeduction}
}%

From premise \textcircled{1}, the succedent propositionally simplifies to \(\big(\dbox{\pevolvein{\D{x}=\genDE{x}}{\ivr}}{\etermA=0}\big) \lbisubjunct \sigliedzero{\genDE{x}}{\etermA}\) by assumption $\sigliedsai{\genDE{x}}{\ivr}$.
The two directions of this simplified succedent are proved separately.
In the ``$\lylpmi$'' direction, the derivation uses rule~\irref{dRI} by setting $N$ to the rank of extended term $\etermA$, so that the succedent of its left premise is exactly $\sigliedzero{\genDE{x}}{\etermA}$.
The right premise resulting from~\irref{dRI} closes by real arithmetic, since $N$ is the rank of $\etermA$, it must, by definition satisfy the provable rank identity \rref{eq:differential-rank}.
{\footnotesizeoff%
\renewcommand{\arraystretch}{1.3}%
\begin{sequentdeduction}[array]
\linfer[dRI]{
    \linfer[qear]{\lclose}
    {\lsequent{}{\lied[N]{\genDE{x}}{\etermA} = \sum_{i=0}^{N-1} g_i \lied[i]{\genDE{x}}{\etermA}}}
  }
  {\lsequent{\sigliedzero{\genDE{x}}{\etermA}}{\dbox{\pevolvein{\D{x}=\genDE{x}}{\ivr}}{\etermA=0} }}
\end{sequentdeduction}%
}%

The derivation in the ``$\limply$'' direction starts by reducing to the contrapositive statement with duality \irref{diamond} and propositional logical manipulation.
{\footnotesizeoff%
\begin{sequentdeduction}[array]
\linfer[diamond]{
\linfer[notl+notr]{
  \lsequent{\initassum, \ivr, \sigliedsai{\genDE{x}}{\ivr}, \lnot{(\sigliedzero{\genDE{x}}{\etermA})}}{\ddiamond{\pevolvein{\D{x}=\genDE{x}}{\ivr}}{\etermA \neq 0}}
}
  {\lsequent{\initassum, \ivr, \sigliedsai{\genDE{x}}{\ivr}, \lnot\ddiamond{\pevolvein{\D{x}=\genDE{x}}{\ivr}}{\etermA \neq 0}}{\sigliedzero{\genDE{x}}{\etermA}}}
}
  {\lsequent{\initassum, \ivr, \sigliedsai{\genDE{x}}{\ivr}, \dbox{\pevolvein{\D{x}=\genDE{x}}{\ivr}}{\etermA = 0}}{\sigliedzero{\genDE{x}}{\etermA}}}
\end{sequentdeduction}
}%
By derived axiom~\irref{Lpiff}, the antecedent assumption $\sigliedsai{\genDE{x}}{\ivr}$ is equivalently rewritten to local progress for $\ivr$ before~\irref{bigsmallequiv} (\rref{cor:bigsmallequiv}) is used to strengthen it with the assumption $\ivr$.
The final step rewrites the resulting negated differential radical formula in the antecedents to two progress formulas by \rref{eq:sigliedzerorearrangement} from \rref{prop:rearrangement}.
Subsequent splitting with \irref{orl} yields two premises, which are abbreviated \textcircled{3} (for disjunct $\sigliedgt{\genDE{x}}{\etermA}$) and \textcircled{4} (for disjunct $\sigliedgt{\genDE{x}}{(-\etermA)}$) respectively, and continued below.
{\footnotesizeoff%
\begin{sequentdeduction}[array]
\linfer[Lpiff]{
\linfer[bigsmallequiv]{
\linfer[qear]{
\linfer[orl]{
  \textcircled{3} ! \textcircled{4}
}
  {\lsequent{\initassum, \ddiamond{\pevolvein{\D{x}=\genDE{x}}{\ivr}}{x \neq y}, \sigliedgt{\genDE{x}}{\etermA} \lor \sigliedgt{\genDE{x}}{(-\etermA)}}{\ddiamond{\pevolvein{\D{x}=\genDE{x}}{\ivr}}{\etermA \neq 0}}}
}
  {\lsequent{\initassum, \ddiamond{\pevolvein{\D{x}=\genDE{x}}{\ivr}}{x \neq y}, \lnot{(\sigliedzero{\genDE{x}}{\etermA})}}{\ddiamond{\pevolvein{\D{x}=\genDE{x}}{\ivr}}{\etermA \neq 0}}}
}
  {\lsequent{\initassum, \ivr, \dprogressin{\D{x}=\genDE{x}}{\ivr}, \lnot{(\sigliedzero{\genDE{x}}{\etermA})}}{\ddiamond{\pevolvein{\D{x}=\genDE{x}}{\ivr}}{\etermA \neq 0}}}
}
  {\lsequent{\initassum, \ivr, \sigliedsai{\genDE{x}}{\ivr}, \lnot{(\sigliedzero{\genDE{x}}{\etermA})}}{\ddiamond{\pevolvein{\D{x}=\genDE{x}}{\ivr}}{\etermA \neq 0}}}
\end{sequentdeduction}
}%

Continuing from \textcircled{3}, the assumption $\sigliedgt{\genDE{x}}{\etermA}$ is rewritten with \irref{Lpgtfull} to obtain local progress for $\etermA > 0$.
Unfolding the $\ddnext$ abbreviation, the uniqueness axiom \irref{Uniq} combines the two diamond modality formulas in the antecedent:
{\footnotesizeoff%
\begin{sequentdeduction}[array]
\linfer[Lpgtfull]{
\linfer[]{
\linfer[Uniq]{
  \lsequent{\ddiamond{\pevolvein{\D{x}=\genDE{x}}{\ivr \land (\etermA > 0 \lor x=y)}}{x \neq y} }{ \ddiamond{\pevolvein{\D{x}=\genDE{x}}{\ivr}}{\etermA \neq 0}}
}
  {\lsequent{\ddiamond{\pevolvein{\D{x}=\genDE{x}}{\ivr}}{x\neq y},\ddiamond{\pevolvein{\D{x}=\genDE{x}}{\etermA > 0 \lor x=y}}{x \neq y} }{ \ddiamond{\pevolvein{\D{x}=\genDE{x}}{\ivr}}{\etermA \neq 0}}}
}
  {\lsequent{\ddiamond{\pevolvein{\D{x}=\genDE{x}}{\ivr}}{x\neq y},\dprogressin{\D{x}=\genDE{x}}{\etermA > 0} }{ \ddiamond{\pevolvein{\D{x}=\genDE{x}}{\ivr}}{\etermA \neq 0}}}
}
  {\lsequent{\initassum,\ddiamond{\pevolvein{\D{x}=\genDE{x}}{\ivr}}{x\neq y},\sigliedgt{\genDE{x}}{\etermA}}{ \ddiamond{\pevolvein{\D{x}=\genDE{x}}{\ivr}}{\etermA \neq 0}}}
\end{sequentdeduction}
}%

The succedent's domain constraint is strengthened to match the antecedent's using rule~\irref{gddR} since $\ivr \land (\etermA > 0 \lor x=y) \limply \ivr$ is a propositional tautology.
The Kripke axiom \irref{Kd} reduces the succedent to the box modality, after which the proof finishes with a \irref{dW} step because the formula $\etermA > 0 \lor x=y$ in the domain constraint implies the succedent by real arithmetic.
{\footnotesizeoff%
\begin{sequentdeduction}[array]
\linfer[gddR]{
\linfer[Kd]{
\linfer[dW]{
\linfer[qear]{ \lclose }
  {\lsequent{\ivr \land (\etermA > 0 \lor x=y)}{(x\neq y \limply \etermA \neq 0)}}
}
  {\lsequent{}{\dbox{\pevolvein{\D{x}=\genDE{x}}{\ivr \land (\etermA > 0 \lor x=y)}}{(x\neq y \limply \etermA \neq 0)}}}
}
  {\lsequent{\ddiamond{\pevolvein{\D{x}=\genDE{x}}{\ivr \land (\etermA > 0 \lor x=y)}}{x \neq y} }{\ddiamond{\pevolvein{\D{x}=\genDE{x}}{\ivr \land (\etermA > 0 \lor x=y)}}{\etermA \neq 0}}}
}
  {\lsequent{\ddiamond{\pevolvein{\D{x}=\genDE{x}}{\ivr \land (\etermA > 0 \lor x=y)}}{x \neq y} }{\ddiamond{\pevolvein{\D{x}=\genDE{x}}{\ivr}}{\etermA \neq 0}}}
\end{sequentdeduction}
}%

The remaining premise \textcircled{4} follows similarly, except that the progress formula $\sigliedgt{\genDE{x}}{(-\etermA)}$ enables the cut $\dprogressin{\D{x}=\genDE{x}}{{-}\etermA > 0}$ which leads to the same postcondition $\etermA \neq 0$ in the succedent.
\end{proof}

\subsection{Completeness for Semianalytic Invariants with Semianalytic Evolution Domains}
\label{app:semialginvariants}

The following generalized version of rule \irref{sAI} from \rref{thm:sAI} additionally handles evolution domain constraints.
It derives from derived rule \irref{realindin} and derived axiom \irref{Lpiff}.

\begin{theorem}[Semianalytic invariants with semianalytic domain constraints]
The semianalytic invariant proof rule with semianalytic domain constraints~\irref{sAIQ} derives from $\irref{RealIndIn+Dadjoint+Cont}$, $\irref{Uniq}$ for semianalytic formulas $\rfvar,\ivr$.
\[
\dinferenceRule[sAIQ|sAI{$\&$}]{Semianalytic invariant with domains}
{\linferenceRule
  {
  \lsequent{\rfvar,\ivr,\sigliedsai{\genDE{x}}{\ivr}}{\sigliedsai{\genDE{x}}{\rfvar}} &
  \lsequent{\lnot{\rfvar},\ivr,\sigliedsai[-]{-\genDE{x}}{\ivr}}{\sigliedsai[-]{-\genDE{x}}{(\lnot{\rfvar})}}
  }
  {\lsequent{\rfvar}{\dbox{\pevolvein{\D{x}=\genDE{x}}{\ivr}}{\rfvar}}}
}{}\]
\end{theorem}
\begin{proof}[Proof (implies~\rref{thm:sAI})]
Rule \irref{sAIQ} derives from rule \irref{realindin} derived in \rref{cor:realindin} and the characterization of semianalytic local progress \irref{Lpiff} derived in \rref{thm:localprogresscomplete}.
The $\initassum$ assumptions provided by \irref{realindin} are used to convert between local progress modalities and the semianalytic progress formulas by \irref{Lpiff}, but, by weakening, $\initassum$ can be elided again in the premises of \irref{sAIQ}.
\end{proof}

Recalling the earlier discussion for derived rule \irref{realindin}, axiom \irref{V} can be used, as usual, to keep \emph{constant} context assumptions that do not depend on variables $x$ for the ODEs $\D{x}=\genDE{x}$ in rule \irref{sAIQ}, because it immediately derives from \irref{realindin}, which supports constant contexts.
The proof rule~\irref{sAIQ} is complete for invariance properties.
This is proved syntactically, enabling complete disproofs of invariance.
The completeness of~\irref{sAI} from \rref{thm:semialgcompleteness} follows as a special case, where $\ivr \equiv \ltrue$.

\begin{theorem}[Semianalytic invariant completeness with semianalytic domains]
\label{thm:semialgcompletenessdom}
The semianalytic invariant axiom with semianalytic domain constraints~\irref{SAIQ} derives from $\irref{RealIndIn+Dadjoint+Cont}$, $\irref{Uniq}$ for semianalytic formulas $\rfvar,\ivr$.
\[
\dinferenceRule[SAIQ|SAI{$\&$}]{Semianalytic invariant with domains axiom}
{\linferenceRule[equiv]
  {
  \underbrace{\lforall{x}{\big(\rfvar \land \ivr \land \sigliedsai{\genDE{x}}{\ivr} {\limply} \sigliedsai{\genDE{x}}{\rfvar}\big)}}_{\textcircled{a}}
  \land
  \underbrace{\lforall{x}{\big(\lnot{\rfvar} \land \ivr \land \sigliedsai[-]{-\genDE{x}}{\ivr} {\limply} \sigliedsai[-]{\-genDE{x}}{(\lnot{\rfvar})}\big)}}_{\textcircled{b}}
  }
  {\axkey{\lforall{x}{(\rfvar \limply \dbox{\pevolvein{\D{x}=\genDE{x}}{\ivr}}{\rfvar})}}}
}{}
\]
\end{theorem}
\begin{proof}[Proof (implies \rref{thm:semialgcompleteness}).]
The left and right conjunct on the RHS of \irref{SAIQ} are abbreviated \textcircled{a} and \textcircled{b} respectively.
The ``$\lylpmi$'' direction derives by \irref{sAIQ}.
The antecedents \textcircled{a} and \textcircled{b} are first-order formulas quantified over $x$, the variables evolved by the ODE $\D{x}=\genDE{x}$.
They are kept as constant context in the antecedents of the premises when applying rule \irref{sAIQ} and later instantiated by \irref{alll}.
{\footnotesizeoff%
\begin{sequentdeduction}[array]
\linfer[allr+implyr]{
\linfer[sAIQ]{
  \linfer[alll+implyl]{
  \lclose
  }
  {\lsequent{\textcircled{a},\rfvar,\ivr,\sigliedsai{\genDE{x}}{\ivr}}{\sigliedsai{\genDE{x}}{\rfvar}}}
  !
  \linfer[alll+implyl]{
  \lclose
  }
  {\lsequent{\textcircled{b},\lnot{\rfvar},\ivr,\sigliedsai[-]{-\genDE{x}}{\ivr}}{\sigliedsai[-]{-\genDE{x}}{(\lnot{\rfvar})}}}
}
  {\lsequent{\textcircled{a}, \textcircled{b},\rfvar}{\dbox{\pevolvein{\D{x}=\genDE{x}}{\ivr}}{\rfvar}}}
}
  {\lsequent{\textcircled{a}, \textcircled{b}}{\lforall{x}{(\rfvar {\limply} \dbox{\pevolvein{\D{x}=\genDE{x}}{\ivr}}{\rfvar})}}}
\end{sequentdeduction}
}%

In the ``$\limply$'' direction, the derivation proceeds by contraposition in both cases after~\irref{andr}.
For \textcircled{b}, the derived invariant reflection axiom (\irref{reflect}) is used to syntactically turn the invariance assumption for the forwards ODE into an invariance assumption for the backwards ODE.
{\footnotesizeoff%
\begin{sequentdeduction}[array]
\linfer[andr]{
  \lsequent{\lforall{x}{(\rfvar {\limply} \dbox{\pevolvein{\D{x}=\genDE{x}}{\ivr}}{\rfvar})}}{\textcircled{a}}
  !
  \linfer[reflect]{
  \lsequent{\lforall{x}{(\lnot{\rfvar} {\limply} \dbox{\pevolvein{\D{x}=-\genDE{x}}{\ivr}}{\lnot{\rfvar}})}}{\textcircled{b}}
  }
  {\lsequent{\lforall{x}{(\rfvar {\limply} \dbox{\pevolvein{\D{x}=\genDE{x}}{\ivr}}{\rfvar})}}{\textcircled{b}}}
}
  {\lsequent{\lforall{x}{(\rfvar {\limply} \dbox{\pevolvein{\D{x}=\genDE{x}}{\ivr}}{\rfvar})}}{\textcircled{a} \land \textcircled{b}}}
\end{sequentdeduction}
}%

Continuing from the left premise (with \textcircled{a} in its succedent), standard logical manipulation is used to dualize both sides of the sequent.
The~\irref{existsl} step Skolemizes the existential in the antecedent, with the resulting $x$ used to witness the (then) existentially quantified succedent with~\irref{existsr}:
{\footnotesizeoff%
\begin{sequentdeduction}[array]
\linfer[diamond+notl+notr]{
\linfer[existsl]{
\linfer[existsr+andr]{
  \lsequent{\rfvar, \ivr, \sigliedsai{\genDE{x}}{\ivr}, \lnot{(\sigliedsai{\genDE{x}}{\rfvar})}}{\ddiamond{\pevolvein{\D{x}=\genDE{x}}{\ivr}}{\lnot{\rfvar}}}
}
  {\lsequent{\rfvar, \ivr, \sigliedsai{\genDE{x}}{\ivr}, \lnot{(\sigliedsai{\genDE{x}}{\rfvar})}}{\lexists{x}{(\rfvar \land \ddiamond{\pevolvein{\D{x}=\genDE{x}}{\ivr}}{\lnot{\rfvar}})}}}
}
  {\lsequent{\lexists{x}{\big(\rfvar \land \ivr \land \sigliedsai{\genDE{x}}{\ivr} \land \lnot{(\sigliedsai{\genDE{x}}{\rfvar})}\big)}}{\lexists{x}{(\rfvar \land \ddiamond{\pevolvein{\D{x}=\genDE{x}}{\ivr}}{\lnot{\rfvar}})}}}
}
  {\lsequent{\lforall{x}{(\rfvar {\limply} \dbox{\pevolvein{\D{x}=\genDE{x}}{\ivr}}{\rfvar})}}{\lforall{x}{\big(\rfvar \land \ivr \land \sigliedsai{\genDE{x}}{\ivr} {\limply} \sigliedsai{\genDE{x}}{\rfvar}\big)}}}
\end{sequentdeduction}
}%

Next, an initial state assumption $\initassum$ is introduced by a~\irref{cut+qear} followed by~\irref{existsl} to Skolemize the resulting existential quantifier.
The antecedent assumption $\lnot{(\sigliedsai{\genDE{x}}{\rfvar})}$ is replaced with $\sigliedsai{\genDE{x}}{(\lnot{\rfvar})}$ equivalently, by \rref{prop:negationrearrangement}.
Both local progress formulas in the antecedents are then replaced equivalently with the local progress modalities using the derived axiom~\irref{Lpiff}.
{\footnotesizeoff%
\begin{sequentdeduction}[array]
\linfer[cut+qear]{
\linfer[existsl]{
\linfer[qear]{
\linfer[Lpiff]{
  \lsequent{\initassum,\rfvar, \ivr, \dprogressin{\D{x}=\genDE{x}}{\ivr}, \dprogressin{\D{x}=\genDE{x}}{\lnot{\rfvar}}}{\ddiamond{\pevolvein{\D{x}=\genDE{x}}{\ivr}}{\lnot{\rfvar}}}
}
  {\lsequent{\initassum,\rfvar, \ivr, \sigliedsai{\genDE{x}}{\ivr}, \sigliedsai{\genDE{x}}{(\lnot{\rfvar})}}{\ddiamond{\pevolvein{\D{x}=\genDE{x}}{\ivr}}{\lnot{\rfvar}}}}
}
  {\lsequent{\initassum,\rfvar, \ivr, \sigliedsai{\genDE{x}}{\ivr}, \lnot{(\sigliedsai{\genDE{x}}{\rfvar})}}{\ddiamond{\pevolvein{\D{x}=\genDE{x}}{\ivr}}{\lnot{\rfvar}}}}
}
  {\lsequent{\lexists{y}{\initassum},\rfvar, \ivr, \sigliedsai{\genDE{x}}{\ivr}, \lnot{(\sigliedsai{\genDE{x}}{\rfvar})}}{\ddiamond{\pevolvein{\D{x}=\genDE{x}}{\ivr}}{\lnot{\rfvar}}}}
}
  {\lsequent{\rfvar, \ivr, \sigliedsai{\genDE{x}}{\ivr}, \lnot{(\sigliedsai{\genDE{x}}{\rfvar})}}{\ddiamond{\pevolvein{\D{x}=\genDE{x}}{\ivr}}{\lnot{\rfvar}}}}
\end{sequentdeduction}
}%

By~\irref{bigsmallequiv} from \rref{cor:bigsmallequiv}, local progress for $\ivr$ is strengthened, while $\ddnext$ can only be unfolded for $\lnot{\rfvar}$.
The two resulting diamond modality formulas are combined with~\irref{Uniq}:
{\footnotesizeoff%
\begin{sequentdeduction}[array]
\linfer[bigsmallequiv]{
\linfer[Uniq]{
  \lsequent{\ddiamond{\pevolvein{\D{x}=\genDE{x}}{\ivr \land (\lnot{\rfvar} \lor x=y)}}{x \neq y}}{\ddiamond{\pevolvein{\D{x}=\genDE{x}}{\ivr}}{\lnot{\rfvar}}}
}
  {\lsequent{\rfvar, \ddiamond{\pevolvein{\D{x}=\genDE{x}}{\ivr}}{x \neq y}, \ddiamond{\pevolvein{\D{x}=\genDE{x}}{\lnot{\rfvar} \lor x = y}}{x \neq y}}{\ddiamond{\pevolvein{\D{x}=\genDE{x}}{\ivr}}{\lnot{\rfvar}}}}
}
  {\lsequent{\initassum,\rfvar, \ivr, \dprogressin{\D{x}=\genDE{x}}{\ivr}, \dprogressin{\D{x}=\genDE{x}}{\lnot{\rfvar}}}{\ddiamond{\pevolvein{\D{x}=\genDE{x}}{\ivr}}{\lnot{\rfvar}}}}
\end{sequentdeduction}
}%

With a~\irref{Kd+dW} step, the diamond modality in the antecedent strengthens to $\lnot{\rfvar}$ in its postcondition with the propositional tautology $\ivr \land (\lnot{\rfvar} \lor x=y) \limply (x \neq y \limply \lnot{\rfvar})$.
A~\irref{gddR} step completes the proof using the propositional tautology $\ivr \land (\lnot{\rfvar} \lor x=y) \limply \ivr$.
{\footnotesizeoff%
\begin{sequentdeduction}[array]
\linfer[Kd+dW]{
\linfer[gddR]{
  \lclose
}
  {\lsequent{\ddiamond{\pevolvein{\D{x}=\genDE{x}}{\ivr \land (\lnot{\rfvar} \lor x=y)}}{\lnot{\rfvar}}}{\ddiamond{\pevolvein{\D{x}=\genDE{x}}{\ivr}}{\lnot{\rfvar}}}}
}
  {\lsequent{\ddiamond{\pevolvein{\D{x}=\genDE{x}}{\ivr \land (\lnot{\rfvar} \lor x=y)}}{x \neq y}}{\ddiamond{\pevolvein{\D{x}=\genDE{x}}{\ivr}}{\lnot{\rfvar}}}}
\end{sequentdeduction}
}%

The remaining derivation from the right premise (with \textcircled{b} in its succedent) is similar using local progress for the already reflected backwards differential equations instead.
\end{proof}

\section{Hybrid Programs}
\label{app:hybridprograms}

This appendix is devoted to proving~\rref{cor:testfree}.
The syntax of hybrid programs and the axioms of \dL for handling the hybrid program operators are recalled~\cite{DBLP:conf/lics/Platzer12b,DBLP:journals/jar/Platzer17} below.
Readers are referred to the literature~\cite{DBLP:conf/lics/Platzer12b,DBLP:journals/jar/Platzer17} for a complete definition of hybrid program semantics.

The syntax of hybrid programs $\alpha,\beta$ is given by the following grammar.
The formula $\ivr$ is semianalytic in both tests $\ptest{\ivr}$ and in evolution domain constraints of ODEs $\pevolvein{\D{x}=\genDE{x}}{\ivr}$:
\begin{equation}
\alpha,\beta \bebecomes \pumod{x}{\etermA}  \alternative \ptest{\ivr} \alternative \pevolvein{\D{x}=\genDE{x}}{\ivr} \alternative \pchoice{\alpha}{\beta} \alternative \alpha;\beta \alternative \prepeat{\alpha}
\label{eq:hybridprogramsyntax}
\end{equation}
The following \dL axioms are sound for reasoning about hybrid programs~\cite{DBLP:journals/jar/Platzer17}:

\begin{theorem}[Hybrid program axioms~\cite{DBLP:journals/jar/Platzer17}]
\label{thm:basehpaxioms}
The following are sound axioms of \dL.

    \begin{calculuscollection}
    \begin{calculus}
      \cinferenceRule[assignb|$\dibox{:=}$]{assignment / substitution axiom}
      {\linferenceRule[equiv]
        {\fvarA(e)}
        {\axkey{\dbox{\pupdate{\umod{x}{\etermA}}}{\fvarA\argx}}}
      }
      {\text{$\etermA$ free for $x$ in $\fvarA$}}
      \cinferenceRule[testb|$\dibox{?}$]{test}
      {\linferenceRule[equiv]
        {(\ivr \limply \fvarA)}
        {\axkey{\dbox{\ptest{\ivr}}{\fvarA}}}
      }{}
      \cinferenceRule[choiceb|$\dibox{\cup}$]{axiom of nondeterministic choice}
      {\linferenceRule[equiv]
        {\dbox{\alpha}{\fvarA} \land \dbox{\beta}{\fvarA}}
        {\axkey{\dbox{\pchoice{\alpha}{\beta}}{\fvarA}}}
      }{}
    \end{calculus}\quad
    \begin{calculus}
      \cinferenceRule[composeb|$\dibox{{;}}$]{composition} %
      {\linferenceRule[equiv]
        {\dbox{\alpha}{\dbox{\beta}{\fvarA}}}
        {\axkey{\dbox{\alpha;\beta}{\fvarA}}}
      }{}
      \cinferenceRule[iterateb|$\dibox{{}^*}$]{iteration/repeat unwind} %
      {\linferenceRule[equiv]
        {\fvarA \land \dbox{\alpha}{\dbox{\prepeat{\alpha}}{\fvarA}}}
        {\axkey{\dbox{\prepeat{\alpha}}{\fvarA}}}
      }{}
      \cinferenceRule[I|I]{loop induction}
      {\linferenceRule[equiv]
        {\fvarA \land \dbox{\prepeat{\alpha}}{(\fvarA\limply\dbox{\alpha}{\fvarA})}}
        {\axkey{\dbox{\prepeat{\alpha}}{\fvarA}}}
      }{}

    \end{calculus}
    \end{calculuscollection}
\end{theorem}

The loop induction rule derives from induction axiom \irref{I} using G\"odel's generalization rule \irref{G}~\cite{DBLP:conf/lics/Platzer12b}:
\[
\dinferenceRule[loop|loop]{}
{\linferenceRule
  {\lsequent{\fvarA}{\dbox{\alpha}{\fvarA}}}
  {\lsequent{\fvarA}{\dbox{\prepeat{\alpha}}{\fvarA}}}
}{}
\]

\rref{cor:testfree} is proved using the equivalent characterization of future truth of analytic formulas along ODEs from~\rref{thm:algcomplete}.
It applies to the fragment of \dL programs generated by the grammar~\rref{eq:hybridprogramsyntax}, where formula $\ivr$ in tests and domain constraints are  negations of analytic formulas (so $\etermAA \neq 0$ for some extended term $\etermAA$).
The loop-free fragment further removes the $\prepeat{\alpha}$ clause from~\rref{eq:hybridprogramsyntax}.

The set of extended terms always forms a ring because the $+,\cdot$ operations are interpreted as the usual real-valued addition and multiplication and hence obey the ring axioms~\cite{MR1727221} for $\reals$.
The extended term language is said to be \emph{Noetherian} if its corresponding ring of extended terms is Noetherian.
Similar to the computable differential radicals condition~\rref{itm:reqdiffradical}, an algorithm is assumed that decides (and proves) ideal membership in the ring of extended terms.

In particular, polynomial term languages without extensions are Noetherian and have decidable ideal membership.
Extended term languages are not necessarily Noetherian, even if all additional fixed function symbols are Noetherian functions.
For example, extended term language~\rref{eq:extlang}, even only with $\exp$, is not Noetherian~\cite[Remark 1.4.2]{Terzo}.

\begin{proof}[Proof of \rref{cor:testfree}]
The analytic formula $\rfvar$ is provably equivalent in real arithmetic to a formula $\etermA=0$ for some extended term $\etermA$ by normalizing with the provable arithmetic equivalences $\etermA=0 \land \etermB =0 \lbisubjunct \etermA^2 + \etermB^2 = 0$ and $\etermA=0 \lor \etermB=0 \lbisubjunct \etermA\etermB = 0$.
Assume without loss of generality that it is already written in this form
(and accordingly for the negative analytic formulas in $\alpha$).

Following~\cite[Theorem 1]{DBLP:conf/lics/Platzer12b}, the proof proceeds by structural induction on the form of $\alpha$, showing that for some (computable) extended term $\etermB$, the equivalence \(\dbox{\alpha}{\etermA=0} \lbisubjunct \etermB=0\) is derivable in \dL.
\begin{itemize}
\item Case $ \pevolvein{\D{x}=\genDE{x}}{\etermAA \neq 0}$.
The formula $\etermAA \neq 0$ is a strict inequality so \rref{thm:algcomplete} derives \(\dbox{\pevolvein{\D{x}=\genDE{x}}{\etermAA \neq 0}}{\etermA = 0} \lbisubjunct (\etermAA \neq 0 \limply \sigliedzero{\genDE{x}}{\etermA})\).
Let $N$ be the rank of $\etermA$ so that $\sigliedzero{\genDE{x}}{\etermA}$ expands to \(\landfold_{i=0}^{N-1}  \lied[i]{\genDE{x}}{\etermA} = 0\) and let \(\etermB \mnodefeq \etermAA(\sum_{i=0}^{N-1} (\lied[i]{\genDE{x}}{\etermA})^2)\), giving the provable real arithmetic equivalence \((\etermAA \neq 0 \limply \sigliedzero{\genDE{x}}{\etermA}) \lbisubjunct \etermB = 0\).
Rewriting with this derives in \dL the equivalence:
\[ \dbox{\pevolvein{\D{x}=\genDE{x}}{\etermAA \neq 0}}{\etermA=0} \lbisubjunct \etermB=0 \]

\item Case $\pumod{x}{e}$. By axiom \irref{assignb}, $\dbox{\pumod{x}{\etermA}}{\etermB(x)=0} \lbisubjunct \etermB(\etermA)=0$ derives and $\etermB(\etermA)$ is an extended term.

\item Case $\ptest{\etermAA \neq 0}$. By axiom \irref{testb}, $\dbox{\ptest{\etermAA \neq 0}}{\etermA=0} \lbisubjunct (\etermAA \neq 0 \limply \etermA= 0)$ derives.
Rewriting with the provable real arithmetic equivalence $(\etermAA \neq 0 \limply \etermA = 0)\lbisubjunct \etermAA\etermA=0$ derives the equivalence:
\[ \dbox{\ptest{\etermAA \neq 0}}{\etermA=0} \lbisubjunct \etermAA \etermA=0 \]

\item Case $\pchoice{\alpha}{\beta}$. By axiom \irref{choiceb}, $\dbox{\pchoice{\alpha}{\beta}}{\etermA=0} \lbisubjunct \dbox{\alpha}{\etermA=0} \land \dbox{\beta}{\etermA=0}$ derives.
By the induction hypothesis on $\alpha,\beta$, the equivalences $\dbox{\alpha}{\etermA=0} \lbisubjunct \etermB_1=0$ and $\dbox{\beta}{\etermA=0} \lbisubjunct \etermB_2 =0$ derive for some extended terms $\etermB_1,\etermB_2$.
Moreover, $\etermB_1=0 \land \etermB_2=0 \lbisubjunct \etermB_1^2+\etermB_2^2=0$ is  provable in real arithmetic.
Rewriting with the derived equivalences derives the equivalence:
\[\dbox{\pchoice{\alpha}{\beta}}{\etermA=0} \lbisubjunct \etermB_1^2+\etermB_2^2=0\]

\item Case $\alpha;\beta$. By axiom \irref{composeb}, $\dbox{\alpha;\beta}{\etermA=0} \lbisubjunct \dbox{\alpha}{\dbox{\beta}{\etermA=0}}$ derives.
By the induction hypothesis on $\beta$, the equivalence $\dbox{\beta}{\etermA=0} \lbisubjunct \etermB_2=0$ derives for some extended term $\etermB_2$.
Rewriting with this equivalence derives $\dbox{\alpha;\beta}{\etermA=0} \lbisubjunct \dbox{\alpha}{\etermB_2=0}$.
By the induction hypothesis on $\alpha$, the equivalence $\dbox{\alpha}{\etermB_2=0} \lbisubjunct \etermB_1=0$ derives for some extended term $\etermB_1$.
Rewriting with the chain of derived equivalences derives the equivalence:
\[\dbox{\alpha;\beta}{\etermA=0} \lbisubjunct \etermB_1=0\]

\item Case $\prepeat{\alpha}$.
This case crucially requires that the extended term language is Noetherian.
First, construct the sequence of terms $\etermB_i$ defined inductively with $\etermB_0 \mdefeq \etermA$ and $\etermB_{i+1}$ is the term satisfying the derived equivalence \(\dbox{\alpha}{\etermB_{i}=0} \lbisubjunct \etermB_{i+1} = 0 \) obtained by applying the induction hypothesis on $\alpha$ with postcondition $\etermB_{i}=0$ for $i=0,1,2,\dots$.
Since the term language is assumed to be Noetherian, the following ascending chain of ideals stabilizes:
\[ \ideal{\etermB_0} \subseteq \ideal{\etermB_0,\etermB_1} \subseteq \ideal{\etermB_0,\etermB_1,\etermB_2} \subseteq \cdots \]

Moreover, by decidable ideal membership for the extended term language, there is a (smallest) computable $k$ such that $\etermB_k$ satisfies the following provable identity, with cofactor terms $\cofterm_i$:
\begin{equation}
  \etermB_k = \sum_{i=0}^{k-1} \cofterm_i \etermB_i
  \label{eq:HPNoetherianalgebraicloop}
\end{equation}

The following equivalence is provable by real arithmetic: \(\sum_{i=0}^{k-1} \etermB_i^2 = 0 \lbisubjunct \landfold_{i=0}^{k-1} \etermB_i=0\).
Thus, to derive the equivalence \(\dbox{\prepeat{\alpha}}{\etermA=0} \lbisubjunct \sum_{i=0}^{k-1} \etermB_i^2 = 0\),
it suffices to show that the equivalence \(\dbox{\prepeat{\alpha}}{\etermA=0} \lbisubjunct \landfold_{i=0}^{k-1} \etermB_i=0\) is derivable.
The two directions of this claim are shown separately:

\begin{itemize}
\item[``$\limply$''] This direction is straightforward using $k$ times the iteration axiom \irref{iterateb} together with derived axiom \irref{band}.
By construction, the formulas $\dbox{\alpha}{\etermB_i=0}$ are provably equivalent to $\etermB_{i+1}=0$ using derived equivalences, which derives the required implication:
{\footnotesizeoff\renewcommand*{\arraystretch}{1.3}%
\begin{sequentdeduction}[array]
\linfer[iterateb]{
\linfer[iterateb+band]{
\linfer[iterateb+band]{
\linfer[iterateb+band]{
\linfer[]{
\linfer[]{
  \lclose
}
  {\lsequent{\etermB_0=0 \land \etermB_1=0 \land \etermB_2=0 \land \cdots \land \etermB_{k-1} = 0}{\landfold_{i=0}^{k-1} \etermB_i=0}}
}
  {\lsequent{\etermA=0 \land \dbox{\alpha}{\etermA=0} \land \dbox{\alpha}{\dbox{\alpha}{\etermA=0}} \land \cdots \land \underbrace{\dbox{\alpha}{\dots\dbox{\alpha}{}}}_{k-1~\text{times}}{\etermA=0}}{\landfold_{i=0}^{k-1} \etermB_i=0}}
}
  {\cdots}
}
  {\lsequent{\etermA=0 \land \dbox{\alpha}{\etermA=0} \land \dbox{\alpha}{\dbox{\alpha}{\dbox{\prepeat{\alpha}}{\etermA=0}}}}{\landfold_{i=0}^{k-1} \etermB_i=0}}
}
  {\lsequent{\etermA=0 \land \dbox{\alpha}{\dbox{\prepeat{\alpha}}{\etermA=0}}}{\landfold_{i=0}^{k-1} \etermB_i=0}}
}
  {\lsequent{\dbox{\prepeat{\alpha}}{\etermA=0}}{\landfold_{i=0}^{k-1} \etermB_i=0}}
\end{sequentdeduction}
}%

\item[``$\lylpmi$''] The postcondition of the box modality is strengthened to \(\landfold_{i=0}^{k-1} \etermB_i=0\) by monotonicity with~\irref{K+G},
recalling that $\etermB_0 \mdefeq \etermA$, so \(\landfold_{i=0}^{k-1} \etermB_i=0 \limply \etermA=0\) is a propositional tautology.
Subsequently, the~\irref{loop} rule is used to prove \(\landfold_{i=0}^{k-1} \etermB_i=0\) is a loop invariant of $\prepeat{\alpha}$:
{\footnotesizeoff\renewcommand*{\arraystretch}{1.3}%
\begin{sequentdeduction}[array]
\linfer[K+G]{
\linfer[loop]{
  \lsequent{\landfold_{i=0}^{k-1} \etermB_i=0}{\dbox{\alpha}{\landfold_{i=0}^{k-1} \etermB_i=0}}
}
  {\lsequent{\landfold_{i=0}^{k-1} \etermB_i=0}{\dbox{\prepeat{\alpha}}{\landfold_{i=0}^{k-1} \etermB_i=0}}}
}
  {\lsequent{\landfold_{i=0}^{k-1} \etermB_i=0}{\dbox{\prepeat{\alpha}}{\etermA=0}}}
\end{sequentdeduction}
}%

By axiom \irref{band} and \irref{andr}, each conjunct of the postcondition (indexed by $0 \leq i \leq k-1$) is proved separately.
By construction, each $\dbox{\alpha}{\etermB_i=0}$ is provably equivalent to $\etermB_{i+1}=0$:
{\footnotesizeoff\renewcommand*{\arraystretch}{1.3}%
\begin{sequentdeduction}[array]
\linfer[band]{
\linfer{
\linfer[qear+andr]{
  \lclose
}
  {\lsequent{\landfold_{i=0}^{k-1} \etermB_i=0}{\landfold_{i=0}^{k-1} \etermB_{i+1}=0}}
}
  {\lsequent{\landfold_{i=0}^{k-1} \etermB_i=0}{\landfold_{i=0}^{k-1} \dbox{\alpha}{\etermB_i=0}}}
}
  {\lsequent{\landfold_{i=0}^{k-1} \etermB_i=0}{\dbox{\alpha}{\landfold_{i=0}^{k-1} \etermB_i=0}}}
\end{sequentdeduction}
}%

The premises for indices $0 \leq i < k-1$ all close trivially because $\etermB_{i+1}=0$ is already in the antecedent.
The last premise for index $i=k-1$ has succedent $\etermB_k=0$. However, this follows (by construction and~\irref{qear}) from the antecedent using the provable identity \rref{eq:HPNoetherianalgebraicloop}.
\qedhere
\end{itemize}
\end{itemize}
\end{proof}

\end{document}